\documentclass[a4paper,UKenglish,cleveref, autoref, thm-restate, numberwithinsect]{lipics-v2021}
\nolinenumbers
\title{A Complexity Bound for Determinisation of Min-Plus Weighted Automata}

\author{Shaull Almagor}{Department of Computer Science, Technion, Israel \and \url{https://shaull.cswp.cs.technion.ac.il/}} {shaull@technion.ac.il}{https://orcid.org/0000-0001-9021-1175}{}%TODO mandatory, please use full name; only 1 author per \author macro; first two parameters are mandatory, other parameters can be empty. Please provide at least the name of the affiliation and the country. The full address is optional. Use additional curly braces to indicate the correct name splitting when the last name consists of multiple name parts.

\author{Guy Arbel}{Department of Computer Science, Technion, Israel} {guy.arbel@campus.technion.ac.il}{}{}

\author{Sarai Sheinvald}{Department of Computer Science, Technion, Israel} {surke@technion.ac.il}{https://orcid.org/0000-0002-0524-7390}{}

\authorrunning{S. Almagor,  G. Arbel and S. Sheinvald} %TODO mandatory. First: Use abbreviated first/middle names. Second (only in severe cases): Use first author plus 'et al.'

\Copyright{Shaull Almagor, Guy Arbel and Sarai Sheinvald} %TODO mandatory, please use full first names. LIPIcs license is "CC-BY";  http://creativecommons.org/licenses/by/3.0/

 %TODO mandatory: Please choose ACM 2012 classifications from https://dl.acm.org/ccs/ccs_flat.cfm 


\category{} %optional, e.g. invited paper

\relatedversion{} 
\acknowledgements{This research was supported by the ISRAEL SCIENCE FOUNDATION (grant No. 989/22)}%optional

\EventEditors{Claudia Faggian and Joost-Pieter Katoen}
\EventNoEds{2}
\EventLongTitle{41st Annual Symposium on Logic in Computer Science (LICS 2026)}
\EventShortTitle{LICS 2026}
\EventAcronym{LICS}
\EventYear{2026}
\EventDate{July 20--23, 2026}
\EventLocation{Lisbon, Portugal}
\EventLogo{}
\SeriesVolume{380}
\ArticleNo{16}
\usepackage[dvipsnames]{xcolor}
\usepackage{todonotes}
\usepackage{thmtools}
\usepackage{graphicx}
\usepackage{amsmath}
\usepackage{stmaryrd}
\usepackage{mathtools}
\usepackage{tikz}
\usetikzlibrary{decorations.pathmorphing, backgrounds,shapes.geometric, arrows, positioning, calc}
\usetikzlibrary{decorations.pathreplacing}
\usetikzlibrary{automata, positioning}
\usepackage{amssymb}
\usepackage{etoolbox}
\usepackage{xspace}
\usepackage{upgreek}
\usepackage{subcaption}
\usepackage{circledsteps}
\usepackage{halloweenmath}
\usepackage{cancel}
\usepackage{float}
\usepackage{caption}
\tikzstyle{block} = [rectangle, draw, fill=blue!20, 
    text centered, rounded corners, minimum height=1cm]
\tikzstyle{line} = [draw, -latex']

\newcommand{\bs}[1]{\boldsymbol{#1}}

\newcommand{\bbN}{\mathbb{N}}
\newcommand{\bbZ}{\mathbb{Z}}

\newcommand{\bbB}{\mathbb{B}}
\newcommand{\bbNinf}{\mathbb{N}_{\infty}}
\newcommand{\bbZinf}{\mathbb{Z}_{\infty}}
\newcommand{\cA}{\mathcal{A}}
\newcommand{\cB}{\mathcal{B}}
\newcommand{\cD}{\mathcal{D}}

\newcommand{\RefStates}{\textsc{RefStates}}
\newcommand{\MinRefStates}{\textsc{MinStates}}

\newcommand{\GroundPairs}{\textsc{GrnPairs}}

\newcommand{\tup}[1]{\langle #1 \rangle}

\newcommand{\init}{\texttt{init}}

\newcommand{\pref}{\mathsf{pref}}
\newcommand{\weight}{\mathsf{wt}}
\newcommand{\xconf}{\mathsf{xconf}}
\newcommand{\maxdom}{\mathsf{max\_dom}}
\newcommand{\domval}{\mathsf{dom\_val}}

\newcommand{\ST}{\text{ s.t. }}
\newcommand{\runsto}[1]{\xrightarrow{#1}}
\newcommand{\minweight}{\mathsf{mwt}}
\newcommand{\wmax}[1]{{\mathsf{max\_wt}(#1)}}
\newcommand{\maxeff}[1]{{\mathsf{max\_eff}(#1)}}

\newcommand{\augA}{\widehat{\cA}}
\newcommand{\augStates}{S}

\newcommand{\augInitState}{s_0}

\newcommand{\augTrans}{\widehat{\Delta}}

\newcommand{\rhobase}{\rho_{\mathrm{base}}}
\newcommand{\baseshift}[2]{\mathsf{bshift}[{#1}/{#2}]}

\newcommand{\stab}{\mathsf{stab}}
\newcommand{\rebase}{\mathsf{rebase}}
\newcommand{\jl}{\mathsf{jump}}

\newcommand{\depth}{\mathsf{dep}}

\newcommand{\jlset}{\mathsf{JL}}

\newcommand{\boundedCacti}[1]{\Upsilon_{#1}}

\newcommand{\simpL}{L_{\mathsf{simp}}}
\newcommand{\generL}{L_{\mathsf{gen}}}

\newcommand{\simpLfuncLength}[2]{\texttt{Len}(#1,#2)} %d then i 
\newcommand{\simpLfuncAmp}[2]{\texttt{Amp}(#1,#2)} 
 
\newcommand{\simpLfuncCover}[2]{\texttt{Cov}(#1,#2)}
\newcommand{\simpLfuncTypes}[2]{\texttt{Typ}(#1,#2)}
\newcommand{\simpLfuncMaxW}[1]{\texttt{MaxWt}(#1)}

\newcommand{\generLfuncLength}[2]{\texttt{GLen}(#1,#2)} %d then i 
\newcommand{\generLfuncAmp}[2]{\texttt{GAmp}(#1,#2)} 
 
\newcommand{\generLfuncCover}[2]{\texttt{GCov}(#1,#2)}
\newcommand{\generLfuncTypes}[2]{\texttt{GTyp}(#1,#2)}
\newcommand{\generLfuncMaxW}[1]{\texttt{GMaxWt}(#1)}

\newcommand{\ramsey}{\mathsf{Ramsey}}

\newcommand{\unfold}{\mathsf{unfold}}

\newcommand{\booltrans}{\delta_\bbB}
\renewcommand{\vec}[1]{\boldsymbol{#1}}

\newcommand{\supp}{\mathsf{supp}}
\newcommand{\flatten}{\mathsf{flatten}}

\newcommand{\near}{\mathsf{near}}
\newcommand{\difftypeset}{\mathsf{DIFF\_TYPES}}

\newcommand{\difftype}{\mathsf{diff\_type}}

\newcommand{\sign}{\mathsf{sign}}

\newcommand{\pot}{\upphi}
\newcommand{\charge}{\uppsi}

\newcommand{\bigM}{{\bs{\mathfrak{m}}}}

\newcommand{\ghostTrans}{\delta_{\mathghost}}

\newcommand{\boundedCac}[1]{{\Upsilon^0_{#1,0}}}
\newcommand{\boundedCacReb}[1]{{\Upsilon^1_{#1,0}}}
\newcommand{\boundedCacRebCac}[1]{{\Upsilon^1_{#1,1}}}
\newcommand{\boundedCacRebJump}[1]{{\Upsilon^1_{#1,0+\jlset}}}

\newcommand{\simpAbsCac}{\Upsilon_{\mathsf{simp}}}

\newcommand{\genAbsCac}{\Upsilon_{\mathsf{gen}}}

\newcommand{\simpCac}{{\boundedCac{\mathsf{simp}}}}
\newcommand{\simpCacReb}{{\boundedCacReb{\mathsf{simp}}}}
\newcommand{\simpCacRebCac}{{\boundedCacRebCac{\mathsf{simp}}}}
\newcommand{\simpCacRebJump}{{\boundedCacRebJump{\mathsf{simp}}}}

\newcommand{\absCac}{\Upsilon}
\newcommand{\absDomL}{L_{\mathrm{dom}}}

\newcommand{\genCac}{{\boundedCac{\mathsf{gen}}}}
\newcommand{\genCacReb}{{\boundedCacReb{\mathsf{gen}}}}
\newcommand{\genCacRebCac}{{\boundedCacRebCac{\mathsf{gen}}}}
\newcommand{\genCacRebJump}{{\boundedCacRebJump{\mathsf{gen}}}}
\newcommand{\budH}{\textbf{H}}

\newcommand{\extraSize}{32\bigM^2}

\newcommand{\slope}{\mathsf{slope}}
\newcommand{\minslope}{\mathsf{min\_slope}}

\newcommand{\gregF}{\mathbf{F}}
\newcommand{\prim}{\mathsf{prim}}

\theoremstyle{definition}
\newtheorem{problem}{Problem}
\begin{document}

\maketitle
\begin{abstract}
  The determinisation problem for min-plus (tropical) weighted automata was recently shown to be decidable. However, the proof is purely existential, relying on several non-constructive arguments. Our contribution in this work is twofold: first, we present the first complexity bound for this problem, showing it is primitive recursive. Second, our techniques introduce a versatile framework to analyse runs of weighted automata in a constructive manner. In particular, this simplifies the previous decidability argument and provides a tighter analysis, thus serving as a critical step towards a tight complexity bound. 
\end{abstract}

\section{Introduction}
\label{sec:abs:intro}
Tropical Weighted Finite automata (WFAs) are a popular quantitative computational model~\cite{droste2009handbook,schutzenberger1961definition,chatterjee2010quantitative,Almagor2020Whatsdecidableweighted}, defining functions from words to values over the \emph{tropical semiring} (also called the \emph{min-plus} semiring). There, a WFA $\cA$ assigns to a word $w$ the \emph{minimal} weight of a run of $\cA$ on $w$, where the weight of a run is the \emph{sum} of the weights along the run, hence min-plus. 

For example, consider the tropical WFA in \cref{fig:min num a num b}. For every word $w\in \{a,b\}^*$ there are two runs: one cycles in $q_a$ and counts the number of $a$'s, and the other cycles in $q_b$ and counts the number of $b$'s. Thus, the weight of $w$ is the minimum between the number of $a$'s and the number of $b$'s.
\begin{figure}[H]
    \centering
    \includegraphics[width=0.3\linewidth]{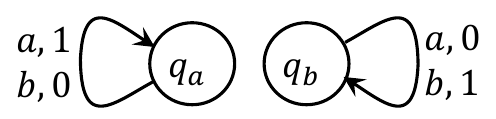}
        \caption{A tropical WFA that assigns to $w\in \{a,b\}^*$ the minimum between the number of $a$'s and $b$'s in $w$.}
        \label{fig:min num a num b}
\end{figure}

Applications of such weighted automata fall on a wide spectrum including verification, rewriting systems, tropical algebra, Boolean language theory, and speech and image processing (in the pre-LLM era). We refer the reader to \cite{kirsten2005distance,Has82,Has00,LP04,daviaud2020register,droste2009handbook,Almagor2020Whatsdecidableweighted,almagor2026determinization} and references therein for applications and surveys. 

Unlike Boolean automata, nondeterministic WFAs are strictly more expressive than their deterministic fragment.  
Accordingly, a natural problem for WFAs is the \emph{determinisation problem}: given a WFA $\cA$, is there a deterministic WFA $\cD$ such that $L_{\cA}\equiv L_{\cD}$. 
The deterministic fragment of WFAs admits much better algorithmic properties than general WFAs~\cite{Almagor2020Whatsdecidableweighted}, making this problem important in order to increase the practical usages of this model.
The determinisability problem was recently shown to be decidable in~\cite{almagor2026determinization}, after being open for over 30 years~\cite{mohri1994compact,mohri1997finite,klimann2004deciding,kirsten2009deciding}. 

Unfortunately, the proof of decidability in~\cite{almagor2026determinization} is non-constructive (i.e., it is shown via two semi-algorithms). Moreover, it is heavily based on several inherently non-constructive techniques. %Therefore obtaining a constructive complexity bound is not merely a matter of ``accounting''.
Specifically, \cite{almagor2026determinization} introduces two important concepts that pave the way to decidability: \emph{Cactus letters} and the \emph{Zooming technique}. Cactus letters are a method for (carefully) constructing words in a recursive manner to allow pumping in a ``safe'' manner (as pumping arguments are notoriously difficult for WFAs~\cite{chattopadhyay2021pumping}). 
The Zooming technique is a way of untangling the run structure of WFAs in order to find well-behaved infixes that can be used to obtain cactus letters.  
Both techniques, while useful, rely on non-constructive arguments. In a nutshell, the former constructs an infinite alphabet, and the latter relies on Ramsey-arguments for infinite sequences, and moreover -- requires an infinite set of words to be applied.

Determinisability of WFAs has a well-known characterisation by means of runs, and of how far two potentially-minimal runs can get away from one another, which we commonly refer to as a \emph{gap}. 
Intuitively, $\cA$ is determinisable if and only if there is some bound $B\in \bbN$ such that if two runs on a word $x$ obtain weights that are more than $B$ apart, but the ``upper'' run can still become minimal over some suffix. For example, consider the WFA in \cref{fig:min num a num b} and the word $a^nb^{n+1}$: upon reading $a^n$, the $q_a$ component accumulates weight $n$ and $q_b$ accumulates $0$, making it minimal and far ``below'' $q_a$. However, the suffix $b^{n+1}$ makes $q_b$ increase to weight $n+1$ and surpass $q_a$. Thus, there is a gap of size $n$ for every $n$, showing that this WFA is not determinisable.

This characterisation (potentially) gives rise to a simple algorithm for determinisability: if we could show a ``small-model'' property for the $B$ above (i.e., an a-priori upper bound on it, if it exists) it is enough to check whether there is some gap above this upper bound. This can be effectively checked. The only component missing is such an upper bound.

A disappointing surprise from the algorithm in~\cite{almagor2026determinization} is that it does not provide such an upper bound, rendering the algorithm even stranger.

\subparagraph*{Our Contribution}
Our main result is that the determinisation problem is decidable in primitive-recursive complexity, thus providing the first complexity upper bound for it. Moreover, our algorithm is the ``expected'' one: an upper bound on the gaps. 

Technically, we build upon the outline of~\cite{almagor2026determinization} for the basic framework, but deviate at the non-constructive steps. Concretely, for cactus letters we specialise the definition so that we only consider a finite alphabet, whereas for the Zooming technique we keep the intuition, but develop a new quantitative approach which enables us to apply the technique on a single word instead of an infinite sequence. The heart of our contribution is technical, and we explore it throughout this extended abstract. We remark that our work is not an ``accounting'' layer over~\cite{almagor2026determinization}, but rather introduces new structures and concepts.

The new tools we develop in this work present generic ideas, which we hope can be used to tackle other problems and settings. In particular, the quantitative zooming technique shows how to untangle the run tree for a concrete word, rather than an infinite sequence, which is arguably far more useful.
We also generalise and demystify some of the ideas in~\cite{almagor2026determinization} to make them more accessible and applicable.

\subparagraph*{Paper Organisation}
The following extended abstract is an informal survey of the main work, starting at \cref{sec:prelim}. Inevitably, we recall much of the results of~\cite{almagor2026determinization}, as we build upon them and they are not (yet) standardised in literature. As we progress, we present our modifications, until reaching our central technical contributions.

In \cref{sec:abs:WFAs det and gaps} we define the model and state basic properties. In \cref{sec:abs:effective cactus} we recall the cactus-letters framework, and take the first step towards making it finite in~\cref{sec:abs:effective cactus letters}. In \cref{sec:abs:potential and charge,sec:abs:witnesses} we recall fundamental notions from~\cite{almagor2026determinization}, adapted to our new cactus letters. In \cref{sec:abs:two length functions} we deviate from these notions and present important elements in the proof. 
In \cref{sec:abs:SRI} we introduce our new key construction -- \emph{separated repeating infix}, and present our tools for reasoning about it. In \cref{sec:abs:quantitative zooming} we present the quantitative zooming technique, and in \cref{sec:abs:proof outline} we combine all the various elements to construct the main proof. Throughout, we refer to the relevant sections in the technical appendix where the precise arguments can be found. Often, the precise details differ wildly from their informal description, due to omitted details. Finally, in \cref{sec:prelim} we start the technical appendix with the complete proof.

\section{WFAs, Determinisability, Gaps, and the Baseline-Augmented Construction} 
\label{sec:abs:WFAs det and gaps}
\subparagraph*{WFAs}
A WFA $\cA=\tup{Q,\Sigma,q_0,\Delta}$ is a finite automaton where $Q$ is a set of states, $\Sigma$ is an alphabet, $q_0\in Q$ is an initial state\footnote{It is equivalent to consider a set of initial states, which we sometimes do for illustrations.},
and $\Delta\subseteq Q\times \Sigma\times \bbZinf\times Q$ is a \emph{weighted} transition relation, assigning every $p,q,\sigma$ a weight in $\bbZinf=\bbZ\cup \{\infty\}$. Intuitively, weight $\infty$ means that there is no transition. 

A \emph{run} of $\cA$ on a word $w$ is a sequence of transitions reading $w$, and its {\em weight} is the sum of the weights along its transitions. We denote a run $\rho$ from state $q$ to state $p$ reading $w$ by $\rho:q\runsto{w}p$, and its weight by $\weight(\rho)$.
The {\em weight of a word} $w$, is the minimal weight of a run on $w$ (where weight $\infty$ signifies having no run on $w$). More generally, we denote the minimal weight of a run on $w$ from state $q$ to $p$ by $\minweight(q\runsto{w}p)$, and the minimal weight of a run between sets of states $Q'$ and $P'$ by $\minweight(Q'\runsto{w}P')$. An example of a WFA is depicted in \cref{fig:min num a num b}.

A \emph{configuration} of $\cA$ is a vector $\vec{c}\in \bbZinf^Q$ which, intuitively, describes for each $q\in Q$ the weight $\vec{c}(q)$ of a minimal run to $q$ thus far. We denote by $\xconf(q_0,w)$ the configuration reached after reading $w$ starting from $q_0$. 

We denote by $\booltrans:Q\times \Sigma\to 2^Q$ the $\bbB$oolean transition function that corresponds to the underlying NFA induced by finite-weight transitions (and naturally extend it to $\booltrans:2^Q\times \Sigma \to 2^Q$).
A WFA is \emph{deterministic} if it has a single initial state, 
%and from every state there exists at most one finite-weight transition on each letter. 
and $|\booltrans(q,\sigma)|\le 1$ for every $q\in Q,\sigma\in \Sigma$.
It is \emph{determinisable} if it has an equivalent deterministic WFA (i.e., they describe the same function) and is otherwise \emph{undeterminisable}.
It is well-known that not every WFA is determinisable~\cite{chatterjee2010quantitative}. 
The \emph{determinisability} problem asks, given a WFA $\cA$, whether it is determinisable.
In this paper, we give a complexity bound on the determinisability problem.
% \begin{figure}[H]
%     \centering
%     \includegraphics[width=0.7\linewidth]{fig/abs/wfa_example.pdf}
%     \caption{A tropical WFA $\cA$. All the states are initial. \cref{fig:abs:running dag example} depicts a run DAG of $\cA$.}
%     %Indeed, the configurations along the runs are:
%     \label{fig:running example}
% \end{figure}

% \begin{figure}[H]
%   \begin{center}
%     \includegraphics[width=0.45\textwidth]{fig/abs/run_dag.pdf}
%   \end{center}
%   \caption{After reading $bab^2ab^3ab^4ab^5a$, we arrive at the configuration $(0,0,5)$, as depicted. We then show how reading $b^6a$ increases the weight of $r$ by $1$. Note that this cannot be done with any shorter word.}
%   \label{fig:abs:running dag example}
% \end{figure}

We remark that there are several variants of WFAs: with or without initial and final weights, and with or without distinguished accepting states. 
The determinisability problems for these variations are inter-reducible (see \cref{rmk: initial and final weights}). Thus, 
our definition is of the simplest variant: no initial and final weights, and all states are accepting. 

\subsection{Gaps and Determinisability}
The starting point for discussing determinisation is the folklore characterisation by means of \emph{gaps} discussed in \cref{sec:abs:intro} and depicted in \cref{fig:abs:gap witness}.
\begin{figure}
    \centering
    \includegraphics[width=0.4\linewidth]{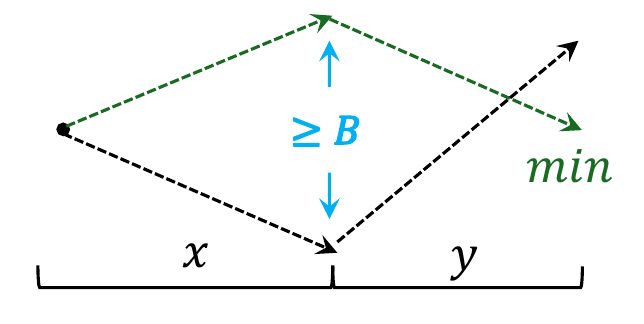}
    \caption{A $B$-gap witness: upon reading $x$, the minimal run on $x$ is at least $B$ below another run, but after reading $y$, the upper run becomes minimal.}
    \label{fig:abs:gap witness}
\end{figure}
\begin{definition}[$B$-Gap Witness]
\label{def:abs: B gap witness}
    For $B\in \bbN$, a \emph{$B$-gap witness} consists of a pair of words $x,y$ and a state $q\in Q$ such that
    \begin{itemize}
        \item $\minweight(q_0\runsto{x\cdot y} Q)=\minweight(q_0\runsto{x}q\runsto{y}Q)<\infty$, but 
        \item $\minweight(q_0\runsto{x} q)-\minweight(q_0\runsto{x} Q)>B$.
    \end{itemize}
\end{definition}
\noindent Then, the characterisation is as follows.
\begin{theorem}[\cite{filiot2017delay,almagor2026determinization}]
    \label{thm:abs:det iff bounded gap}
    Consider a WFA $\cA$, then $\cA$ is determinisable if and only if there is some $B\in \bbN$ such that there is no $B$-gap witness.
 \end{theorem}

\subsection{Baseline-Augmented Construction (\cref{sec:augmented construction})}
\label{sec:abs:baseline augmented}
The first step in~\cite{almagor2026determinization}, which we now briefly recall, is the introduction of the \emph{Baseline-Augmented Construction} $\augA$. 
Intuitively, $\augA$ reads a word $w$, but is also given a run $\rho$ of $\cA$ on $w$. It then follows the exact same runs as $\cA$, but normalises their weight relatively to $\rho$. In addition, it keeps track of the reachable states. 

This simple construction merely makes some information explicit, which is otherwise implicit in the runs of $\cA$. However, it turns out that this information greatly assists in manipulating and reasoning about runs. Thus, we put it forward as our starting point of this work.

Technically, the states $S$ of $\augA$ are tuples $(p,q,T)$ where $p\in Q$ is the current state of $\cA$, $q\in Q$ is the \emph{baseline state}, which is the state in the run that is being read (the \emph{baseline run}), and $T\subseteq Q$ is the reachable set of states. 
$\augA$ actually reads only a run (i.e., a sequence of transitions of $\cA$), and the word is extracted from it. Thus, upon reading a transition $(q,\sigma,c,q')$, the state is updated to $(r,q',T')$ as follows: $q'$ follows from $q$, as the transition dictates. $r$ can be any state such that $p\runsto{\sigma}r$, and $T'=\booltrans(T,\sigma)$. The weight of this transition is $d-c$, where $d=\minweight(p\runsto{\sigma}r)$, i.e., we normalise according to the given transition.
In particular, the run of $\augA$ that follows the baseline in the first two components always has weight $0$, and is called the \emph{baseline run}. We denote the alphabet of $\augA$ by $\Gamma$ and its initial state by $s_0$.

A key property of $\augA$ is the ability to \emph{shift the baseline}: instead of feeding $\augA$ a run $\rho:s_0\runsto{x}S$, one may choose any other run $\rho':s_0\runsto{x}S$, preserving the gaps between the runs of $\augA$, and changing only their concrete weights. This is called \emph{baseline shift}, and we heavily rely on it, as well as develop new tools for it in \cref{sec:abs:baseline shift,sec: baseline shift}.

Since $\augA$ mirrors the run structure of $\cA$, it is easily shown in~\cite{almagor2026determinization} that $\augA$ is determinisable if and only if $\cA$ is determinisable. Therefore, we can reason about $\augA$.

% Recall of the augmented construction from \cite{almagor2026determinization}.
% No new technical content here.

\section{Cactus Letters}
\label{sec:abs:effective cactus}
The standard approach in reasoning about automata is via ``pumping arguments'', which come in various forms. For the purpose of determinisability, pumping corresponds to increasing the gap between two runs to obtain unbounded $B$-gap witnesses. 
The main reason that decidability of determinisability has remained open for so long until~\cite{almagor2026determinization} is (arguably) that known pumping techniques do not seem to work for increasing gaps in this manner.

The first novelty in~\cite{almagor2026determinization} is the introduction of \emph{cactus letters}: these are letters built up from ``nesting'' words, when these words exhibit certain nice structural properties in $\augA$. This nesting then allows us to pump a word while keeping track of which runs can increase unboundedly, and which runs do not. In a way, this is a finer notion than the \emph{stabilisation monoid}~\cite{daviaud2023big,lombardy2006sequential}.

We recall the fundamentals of cactus letters, and then explain where they fall short for obtaining complexity bounds, and how we modify them. 
\subsection{Reflexive and Stable Cycles (\cref{sec:stable cycles and bounded behaviours})}
\label{sec:abs:reflexive and stable cycles}
A \emph{reflexive cycle} in $\augA$ is a pair $(S',w)$ where $S'\subseteq S$ is a subset of the states of $\augA$, and $w$ is a word such that $\booltrans(S',w)\subseteq S'$, i.e., $w$ causes a ``cycle'' on $S'$. We assume $S'$ induces a baseline run on $w$, which therefore has weight $0$. We then ask whether there is some $s\in S'$ such that $\minweight(s\runsto{w^k}s)<0$ for some $k\in \bbN$. If not, we call $(S',w)$ a \emph{stable cycle}. Thus, in a stable cycle the minimal cycles on $w^k$ for any $k$ are of weight $0$. 

The fundamental property of stable cycles is that for each pair of states $s,t\ S'$, either $\minweight(s\runsto{w^k}t)$ tends to $\infty$ as $k$ increases, or $\minweight(s\runsto{w^k}t)$ is bounded. In the latter case, we say that $(s,t)$ are a \emph{grounded pair}, and the value  $\minweight(s\runsto{w^{2k\bigM}}t)$ is constant for all $k$, where $\bigM=|S|\cdot |S|!$ is the \emph{stabilisation constant}. Moreover, we can exactly characterise the grounded pairs: $(s,t)$ are grounded if there is some state $r\in S'$ such that $s\runsto{x^\bigM}r\runsto{x^{\bigM}}t$ and there is a run $\tau:r\runsto{x^{\bigM}}r$ such that $\weight(\tau)=0$, i.e., $r$ is a \emph{minimal reflexive state}.

To gain some intuition on stable cycles, consider the runs of $\augA$ on some long word. Eventually, the subset $S'$ of states reachable by $\augA$ repeats on some infix $x$. Then, $(S',x)$ is certainly a \emph{reflexive} cycle. The question of whether it is stable depends on whether enough repetitions of $x$ yield any negative cycles in $S'$. If there are no negative cycles, and the minimal run has weight $0$, then it is a stable cycle. 
As we discuss in \cref{sec:abs:baseline shift}, $\augA$ allows us to manually change the baseline run in order to guarantee that the minimal run indeed has weight $0$, thus obtaining a stable cycle. 

\subsection{Original Cactus Letters (\cref{sec:cactus extension})}
\label{sec:abs:cactus letters}
Consider a stable cycle $(S',w)$. Since we understand the behaviour of runs upon reading $w^{2k\bigM}$ for any $k$, we can define a new letter $\alpha_{S',w}$ such that for every grounded pair $(s,t)$ we have $\minweight(s\runsto{\alpha_{S',w}}t)=\minweight(s\runsto{w^{2\bigM}}t)$, and otherwise we have $\minweight(s\runsto{\alpha_{S',w}}t)=\infty$. Intuitively, non-grounded pairs jump to $\infty$. Thus, in a way, cactus letters capture unboundedly many iterations of $w$.

Notice that in this way we potentially introduce infinitely many letters (as $w$ is unbounded). The next step, however, is even more dramatic: we again consider stable cycles over this new infinite alphabet, define the corresponding cactus letters (which now nest within them cactus letters), and continue in this fashion ad-infinitum, obtaining the \emph{cactus alphabet} $\Gamma_\infty$. We define $\depth(\alpha_{S',w})$ to be the depth of nesting of cactus letters in $\alpha_{S',w}$ (with ``standard'' letters being of depth $0$). This is depicted in \cref{fig:abs:cactus}.

\begin{figure}[H]
        \centering
        \includegraphics[width=0.95\linewidth]{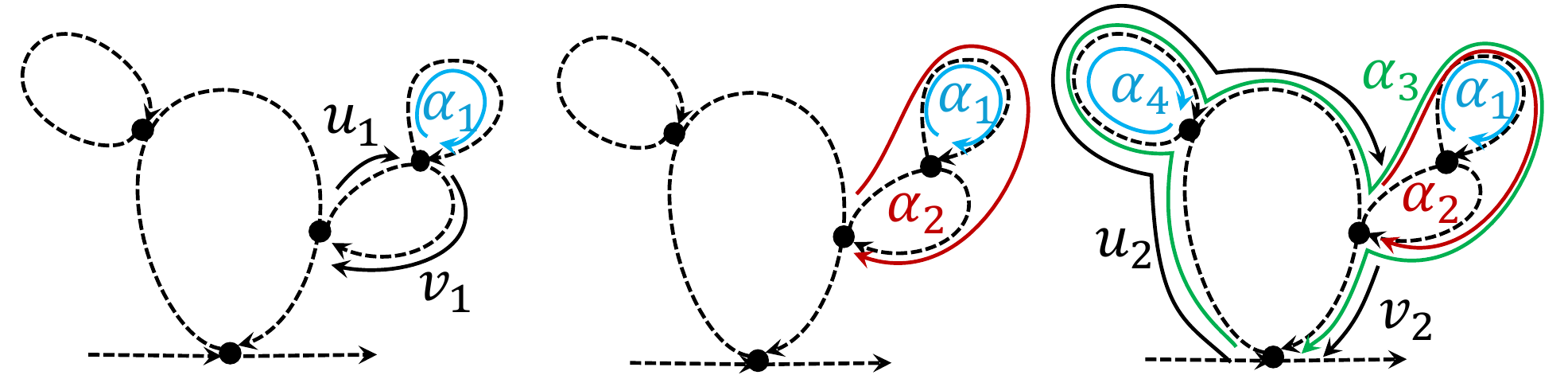}
        \caption{Cactus letters: $\alpha_1$ is of depth $1$ (and its inner word is of depth $0$). The word $u_1\alpha_1v_1$ (left) forms a cactus letter $\alpha_2$ of depth $2$ (center). Then, $u_2\alpha_2v_2$ forms a cactus letter $\alpha_3$ of depth $3$ (right). Notice that $u_2$ contains also the cactus letter $\alpha_4$.}
        \label{fig:abs:cactus}
    \end{figure}

A-priori, the alphabet $\Gamma_\infty$ is infinite both due to arbitrarily long words in cactus letters, and both due to unbounded depth. Fortunately, the latter is already handled (in a sense) in~\cite{almagor2026determinization}, by showing that for depth greater than $|S|$, any cactus letter contains a \emph{degenerate} cycle -- one whose reachable state pairs are all grounded, and therefore does not introduce any ``interesting'' behaviours. 
However, the words in the cactus letters may still be arbitrarily long.
This is the first place where the proof in~\cite{almagor2026determinization} becomes non-constructive, and where we begin our contribution.

\subsection{Effective Cactus Letters (\cref{sec:effective cactus})}
\label{sec:abs:effective cactus letters}
Our first goal is to limit the alphabet of cactus letters to a finite subset. Naturally, this should be done in a way that facilitates the remainder of the proof. 
Therefore, we need to bound not only the depth of the letters we consider, but also the length of their ``inner'' word. That is, bound $|w|$ for the letter $\alpha_{S',w}$.
The catch is what bound we can give so that we can still find some structure in the runs that can be used for some form of pumping. 
Since such a bound depends on the proof structure, we leave it unspecified at this point, and accumulate requirements on it as we progress through the proof, eventually building the concrete bound. 
At this stage, we only explain what this length bound depends on. 

The length of words we allow is not a fixed constant, but depends on the depth of the cactus letters being considered. 
The technical reason is the inductive nature of the proof, but it is also conceptually sensible: consider e.g., a letter $\alpha_{S_1,ab{\alpha_{S_2,w_2}cd}}$, then we would expect to be able to ``unfold'' the inner cactus $\alpha_{S_2,w_2}$ at least to its stable version $w_2^{2\bigM}$ and still have a legal cactus letter $\alpha_{S_1,abw^{2\bigM}cd}$. To do so, the length bound for the outer letter must be much higher than the inner letter.

In the following, we therefore have a function $L:\bbN\to \bbN$ such that for every cactus letter $\alpha_{S',w}$, if $\depth(\alpha_{S',w})=d$ then $|w|\le L(d)$. Moreover, we limit ourselves to letters of depth $|S|$. This renders the alphabet finite, as we required. We denote this alphabet by $\absCac$. 

Unfortunately, we recommend not getting too attached to $L$, since in 
\cref{sec:abs:effective functions} we actually need \emph{two} such bounding functions and two corresponding alphabets in order to obtain our results, but we defer this problem for now.

\subsection{Cactus Unfolding (\cref{sec: cactus unfolding})}
\label{sec:abs:unfolding}
Cactus letters are useful to capture runs on $w^{2\bigM k}$ for any $k$. However, we sometimes want to extract these concrete runs back. Intuitively, we can ``unfold'' a cactus letter $\alpha_{S',w}$ into $w^{2\bigM k}$ for any $k$. Due to the characterisation of grounded pairs, we know what this unfolding does: the grounded pairs maintain their weight, whereas ungrounded pairs increase with $k$.

We borrow this tool with only cosmetic changes from~\cite{almagor2026determinization}, and therefore bring only the essentials here.
We have two operators: the first is $\unfold(u\alpha_{S',w}v\wr F)=uw^{2\bigM M_0}v$, which unfolds a single cactus letter in such a way that all ungrounded pairs yield runs with weight greater than $F$ after the $w^{2\bigM M_0}$. 
The second is $\flatten(w\wr F)$ which recursively applies $\unfold$ to all cactus letters. The resulting word is in $\Gamma^*$. We remark that here we sweep under the rug many important details.

Due to unfolding a new entity is introduced, called \emph{ghost states}: an ungrounded pair $(s,t)$ of $\alpha_{S',w}$ does not yield runs over $\alpha_{S',w}$, but does yield them on $w^{2\bigM k}$ for every $k$. We think of such states as ``ghosts'' floating high above the other runs. When unfolding, we can increase their weight arbitrarily. The states that are reachable as ghosts upon reading $x$ are denoted $\ghostTrans(s_0,x)$. These play a significant role later on.

\subsection{Baseline Shifts}
\label{sec:abs:baseline shift}
As mentioned in \cref{sec:abs:baseline augmented}, $\augA$ offers a convenient way to shift perspective between runs, called \emph{baseline shifts}, by changing the word fed to $\augA$ to represent a different run of $\cA$. 
Once cactus letters are introduced, however, it is no longer clear what this means, since cactus letters are not transitions of $\cA$. Fortunately, in~\cite{almagor2026determinization} some tools are developed to support these shifts also in the presence of cactus letters. Technically, this involves the introduction of a new set of letters called \emph{rebase letters} (\cref{sec:rebase letters}), and a new hierarchy on top of the cactus hierarchy. Our contribution uses this framework, but does not develop it further, and therefore we sweep it under the rug in this abstract, leaving the details to the technical appendix. 

We thus assume there is an operator $\baseshift{\cdot}{\cdot}$ with the following usage and properties: given a word $w\in \absCac^*$ and some run $\rhobase:s_0\runsto{w}S$ of $\augA$, there is a word $w'=\baseshift{w}{\rhobase}$ whose baseline run is $\baseshift{\rhobase}{\rhobase}$, and for every two runs $\pi,\tau:s_0\runsto{w}S$ of $\augA$, we have corresponding runs $\pi'=\baseshift{\pi}{\rhobase}$ and $\tau'=\baseshift{\tau}{\rhobase}$ such that 
$\weight(\pi)-\weight(\tau)=\weight(\pi')-\weight(\tau')$.
Moreover, this holds for every prefix of the runs and the words. Intuitively, the entire run structure is maintained, but the baseline changes. This is illustrated in \cref{fig:abs:baseline shift}.
\begin{figure}
    \centering
    \includegraphics[width=0.9\linewidth]{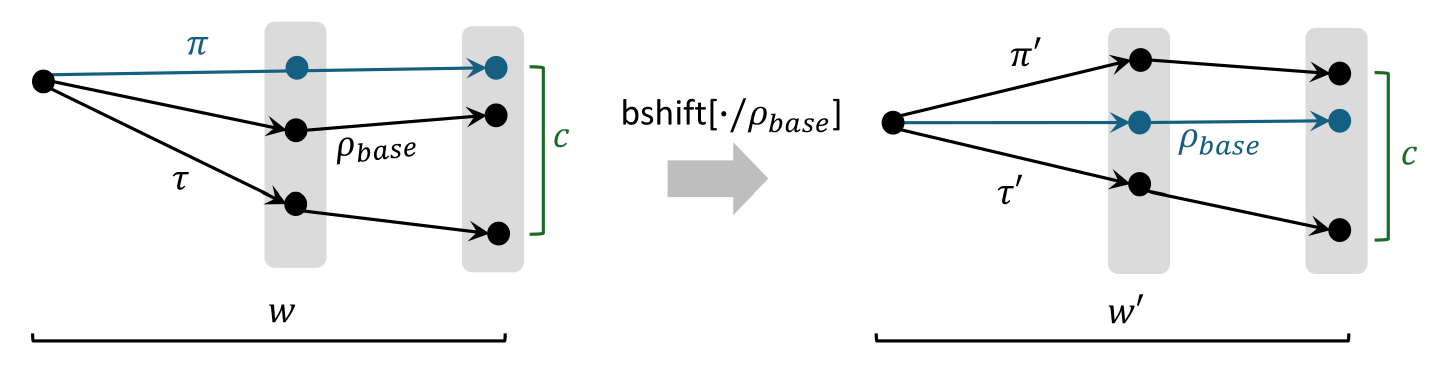}
    \caption{Shifting the baseline: $\rhobase$ becomes the baseline run instead of $\pi$; the word $w$ changes to $w'=\baseshift{w}{\rhobase}$, and the runs $\pi$ and $\tau$ change accordingly to $\pi' = \baseshift{\pi}{\rhobase}$ and $\tau'=\baseshift{\tau}{\rhobase}$, respectively. The shifted runs keep their relative distance from one another.}
    \label{fig:abs:baseline shift}
 \end{figure}

A minor contribution of this work is an extension of baseline shifts beyond runs and words to e.g., sets and cycles, so we can write $\baseshift{(S',w)}{\rhobase}$ to mean a shift of a reflexive cycle (\cref{def: baseline shift on states sets and config}). In addition, we establish new results about the invariance of certain structures to baseline shifts (\cref{rmk:baseline shift is right absorbing,prop: dominance invariant to baseline shifts}). We believe that this framework is extremely convenient, and may prove useful in other cases of reasoning about WFAs and Counter Automata.

\subsection{From Reflexive Cycles to Stable Cycles (\cref{sec:from reflexive to stable})}
\label{sec:abs:reflexive to stable cycles}
Finding stable cycles is a difficult task, due to their strict requirement on minimal cycles. This is a source of complication in~\cite{almagor2026determinization}, and leads to a lot of branching in the proof structure. Moreover, there a stable cycle is found by considering infinite families of words, which is a non-constructive technique we cannot afford. 

Our first contribution is that given a reflexive cycle, we can shift it to find a stable cycle, at the cost of possibly elongating the word by at most $|S|$ repetitions. We sketch the idea as a gentle introduction to baseline shifts.
\begin{lemma}
    \label{lem:abs:reflexive to stable}
    Consider a reflexive cycle $(S',w)$, then there exists $k\le |S|$ and a cycle $\rho:s\runsto{w^k}s$ with $s\in S'$ such that $\baseshift{(S',w)}{\rho}$ is a stable cycle.
\end{lemma}
\begin{proof}[Proof Sketch]
    Assume by way of contradiction that $(S',w)$ is not a stable cycle. Thus, there is some run $\rho:s\runsto{w^k}s$ with $s\in S'$ and $\weight(\rho)<0$. We look for a run whose \emph{slope} is minimal, i.e., it has the steepest descent. We show that this minimal slope is attained by $k\le |S|$.

    We then perform a baseline shift on this $\rho$, using the word $w^k$. We show that $\baseshift{(S',w^k)}{\rho}$ is a reflexive cycle, which has no negative-slope cycles, and is therefore a stable cycle.
\end{proof}

The implication of this result is that it is now fairly easy to find stable cycles: in long-enough words, we must reach a reflexive cycle by simple pigeonhole principle. Then, we can elongate the word somewhat and perform a baseline shift to obtain a stable cycle.

% New development: disciplined use of baseline shifts.
% Explain interaction with effective bounds and later constructions.

\section{(Effective) Dominance, Charge and Potential (\cref{sec:dominance and potential})}
\label{sec:abs:potential and charge}
Recall that our goal is to reason about gap witnesses (\cref{thm:abs:det iff bounded gap}). There, we compare a minimal run on a prefix $x$ with a higher run that can potentially become minimal by reading $y$. 
A key insight from~\cite{almagor2026determinization} (that is not made explicit there) is that the two attributes -- being minimal versus being \emph{potentially} minimal, are very different. As it turns out, separating them opens up new avenues of reasoning that ultimately pave the way to decidability of determinisation. Technically, they are separated by incorporating the baseline run as a separator between them.
To this end, the concepts of \emph{charge} ($\charge$) and \emph{potential} ($\pot$) are introduced. The charge simply measures the weight of the minimal run. Technically, it is negated to keep it positive (as the minimal run is always below the baseline, which is $0$). The potential is somewhat more involved, and measures how high a run can get, so that there is a suffix from which it reaches a finite weight, but all runs below it jump to $\infty$. 

In this work, the definition is adapted to account for the finite alphabet. Specifically, we no longer allow arbitrary suffixes in the definition of the potential, but rather very specific ones. We give an intuitive version of the definitions.
Consider a configuration $\vec{c}$. We define the \emph{charge} $\charge(\vec{c})=\min\{\vec{c}(q)\mid q\in S\}$ to be the minimal weight. 
Next, we fix a language $\absDomL$ of ``legal suffixes'' (see \cref{sec:dominance and potential} for details).
We say that a state $q\in S$ is \emph{dominant} in $\vec{c}$ if there is a suffix $z\in \absDomL$ such that $\minweight(q\runsto{z}S)<\infty$ but $\minweight(q'\runsto{z}S)=\infty$ for every $q'$ such that $\vec{c}(q')<\vec{c}(q)$. If $q$ has the highest weight in the configuration with this property, it is \emph{maximal dominant}, and we define $\pot(\vec{c})=\vec{c}(q)$. For a word $w$, we define $\pot(w)$ as the potential on the configuration reached after reading $w$.

The new notions come with some important tools.
First, as part of our development of baseline shifts, we show that dominance is maintained through baseline shift. That is, if $q$ is dominant in $\vec{c}$, then $\baseshift{q}{\rho}$ is dominant in $\baseshift{\vec{c}}{\rho}$. This may seem like a trivial result by shifting the suffix ``accordingly''. Unfortunately, there is no clear way to define this. The technicalities of this proof is what gives rise to the exact definition of $\absDomL$, and we refer the reader to \cref{prop: dominance invariant to baseline shifts} for the details.
Note that the potential and charge may change under baseline shift, due to the change in baseline. 

The potential satisfies a very nice property called \emph{bounded growth}, namely that $|\pot(w\cdot \sigma)-\pot(w)|$ is bounded by a constant depending on the alphabet $\absCac$. This renders the potential well-behaved, and is used extensively later on.
We remark that this is a useful simplification from~\cite{almagor2026determinization}, as there this property fails in general, due to the infinite alphabet. 

Unfortunately, this property does not hold for charge: the minimal run can increase abruptly by not being able to continue, and making a higher run suddenly minimal. This is where $\charge$ and $\pot$ significantly differ, and why separating them is so useful.
In order to overcome this, we slightly sharpen a result from~\cite{almagor2026determinization}, stating the following:
\begin{lemma}
\label{lem:abs:charge decrease to high potential}
    For every $P\in \bbN$ and $u\sigma\in \absCac^*$ with $\depth(\sigma)=d$, if $\charge(u)-\charge(u\sigma)>P+\generLfuncMaxW{d}$, then there is a word $w$ such that $\pot(w)>P$.
\end{lemma}
$\generLfuncMaxW{d}$ is a function that bounds the maximal weight over any transition on a letter from $\absCac$ of depth $d$. In order to define it, we first need to establish the length bound $L(d)$ which we are yet to do.
This important lemma gives a precise and effective connection between the potential and the charge, and is used to split our main proof into two conceptual parts: one dealing with the case where the charge has bounded growth, and one dealing with the case where the potential can be very high.

Another small observation is that $\charge$ plays nice with unfolding, in the sense that if $x=\unfold(u\alpha_{S',w}v\wr F)$ then $\charge(x)\ge \charge(u\alpha_{S',w}v)$. We remark that a dual property for $\pot$ does not hold, and this is of central importance in the following.

\section{Witnesses (\cref{sec:witness})}
\label{sec:abs:witnesses}
% Witness framework largely follows \cite{almagor2026determinization},
% adapted to the new definitions and effective constructions.
The surprising element in~\cite{almagor2026determinization} is that the characterisation of determinisability is not via unbounded gaps, but rather through a new notion called \emph{witness}. 
In this work, we restrict the notion of witnesses in a way that retains this characterisation, but works over our finite alphabet. In a nutshell, the changes in this definition are a direct result of the changes in the definition of dominance. We bring a semi-precise definition (see \cref{def:witness}).
\begin{definition}
    \label{def:abs:witness}
    A \emph{witness} consists of $w_1,\alpha_{S_1,w_2},w_3$ over $\absCac$ such that $S_1=\ghostTrans(s_0,w_1)$, and the following holds:
    \begin{itemize}
        \item $\minweight(s_0\runsto{w_1\alpha_{S_1,w_2}w_3}S)=\infty$, but
        \item $\minweight(s_0\runsto{w_1w_2^{2\bigM}w_3}S)<\infty$.
    \end{itemize}
\end{definition}
Intuitively, since $w_1\alpha_{S_1,w_2}w_3$ gets weight $\infty$, but $w_1w_2^{2\bigM}w_3$ get finite weight, it means that the minimal runs on the latter go through ghost states. More precisely, it means that the maximal dominant state after reading $w_1w_2^{2\bigM}$ is a ghost state in $\ghostTrans(s_0,w_1\alpha_{S_1,w_2})$. 

Since ghost states can be raised arbitrarily high using flattening, this serves to give the intuition that a witness implies unbounded gaps.
Indeed, the main result of~\cite{almagor2026determinization} is that a witness exists if and only if $\augA$ is undeterminisable, which readily leads to decidability via two semi algorithms (with the ``if'' direction being the bulk of the work). 

We bring a crucial result (heavily abridged, see~\cref{lem:unfolding maintains potential}) relating to witnesses, whose proof is similar (up to alphabet-based changes) to the analogue in~\cite{almagor2026determinization}.
\begin{lemma}
    \label{lem:abs:unfolding maintains potential}
        Either there is a witness, or the following holds.
        Consider a word $u \alpha_{B,x} v\in \absCac^*$ and let $ux^{2\bigM k}v=\unfold(u,\alpha_{B,x},v \wr F)$ for large-enough $F$.
        For every prefix $v'$ of $v$ we have $\pot(ux^{2\bigM k}v')=\pot(u\alpha_{B,x}v')$.
\end{lemma}
The way this should be read is: ``either the potential unfolds nicely, or we are done''. Indeed, once we establish the existence of a witness, we know $\augA$ is undeterminisable, which is our goal. This lemma plays a central role in the main proof.

\subsection{On $\pot$, $\charge$ and two length functions}
\label{sec:abs:two length functions}
We can now make the following vague observation: from \cref{lem:abs:charge decrease to high potential} we know that in some cases, properties of $\charge$ can induce high potential. Similarly, from \cref{lem:abs:unfolding maintains potential} we know that in some cases, properties of the potential induce a witness. In \cref{sec:abs:proof outline} our proof takes us from high $\charge$ to high $\pot$, and from there to witnesses. However, these two steps have several technical differences. Henceforth, we often separate definitions and results to $\charge$ and $\pot$.

We start by separating the length functions mentioned in \cref{sec:abs:effective cactus letters}: we now have two length functions, \emph{simple} ($\simpL$) and \emph{general} ($\generL$), which give rise to two effective cactus alphabets $\simpAbsCac$ and $\genAbsCac$. The constraints on each function become clearer in \cref{sec:abs:SRI,sec:abs:proof outline}.

\section{Separated Repeating Infixes (\cref{sec:separated repeating infix})}
\label{sec:abs:SRI}
We can finally start describing our contribution in earnest. 
The first part of our technical contribution is the introduction of \emph{Separated Repeating Infixes (SRI)}. Intuitively, these are words whose run structure exhibits certain organised behaviours, that can be (carefully) used for pumping and for shortening words. The idea is based on the notion of separated increasing infixes from~\cite{almagor2026determinization}, but the technical details differ in many significant ways, and the tools we develop here are new. See \cref{rmk: SRI v.s. original separated increasing infix} for a detailed comparison.

We start with the basic definition.
An SRI is a word $uxyv$ with the following structure (see \cref{fig:separated repeating infix}):
\begin{itemize}
    \item Reading $u$ from $s_0$ reaches a configuration that is partitioned to ``sub-configurations'' that have a huge \emph{gap} $G$ between them, we call the state sets in these sub configurations $V_1,\ldots, V_{\ell}$. 
    \item Next, reading $x$ maintains the exact structure of these $V_j$ sets, and shifts the weight of the states in each set $V_j$ by a constant $k_{j,x}$. In particular, the entire configuration maintains its support after $x$.

    An additional property of $x$ is that it is quite short -- enough so that if needed, we can turn $x^k$ into a stable cycle as per~\cref{lem:abs:reflexive to stable} without exceeding the length bound $L(\depth(x))$.
    
    \item Reading $y$ has very similar behaviour to $x$, in the sense that the $V_j$ are maintained, and each is increased by a constant $k_{j,y}$, and the sign of $k_{j,y}$ is the same as that of $k_{j,x}$, so each $V_j$ is decreasing/increasing on both $x$ and $y$. 

    \item Finally, $v$ is just a harmless suffix.
\end{itemize}

\begin{figure}[H]
    \centering
    \includegraphics[width=1\linewidth]{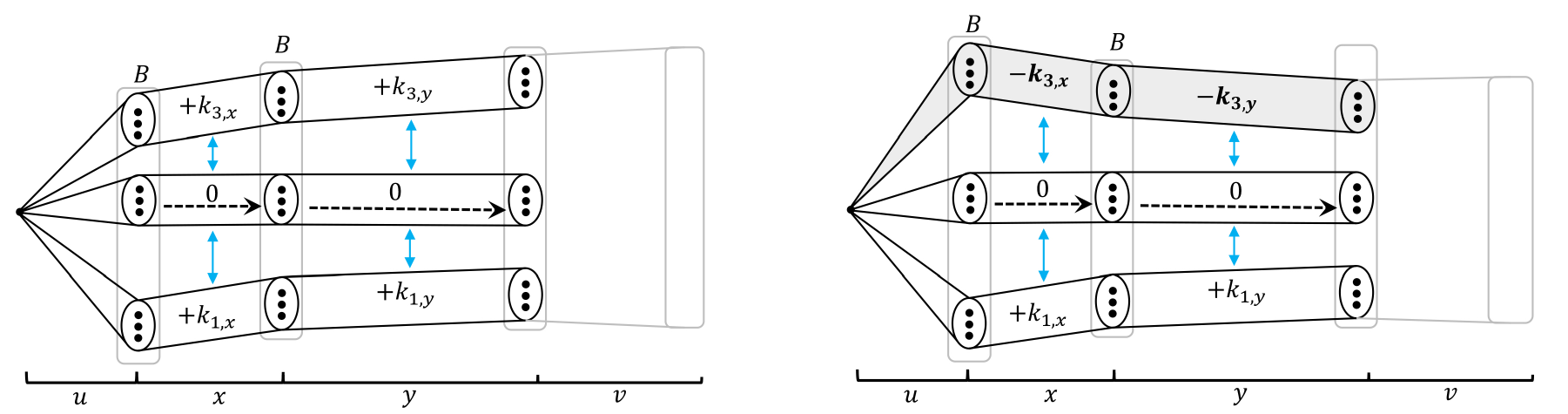}%[scale=0.5]
    \caption{An SRI. After reading $u$, the configuration is separated to sub-configurations corresponding to the $V_i$. Reading $x$ maintains this partition and the exact distances within each set. The same holds for $y$, but with different shifts. 
    The gaps between the sub-configurations are much larger than the effect of $x$ and $y$.
    On the left is a positive SRI: all sub-configurations are trending upwards. On the right is a negative SRI: There exists a sub-configuration (the upper one in the right figure) with a downwards slope.
    }
    \label{fig:separated repeating infix}
\end{figure}

An underlying theme of this work is that increasing and decreasing runs are inherently different. This manifests also in SRI, where sub-configurations with $k_{j,x}<0$ are different to those with $k_{j,x}\ge 0$. 
We formalise and extend this by defining different \emph{flavours} of SRI.
An SRI is:
\begin{itemize}
    \item \emph{Negative} if $k_{j,x}<0$ for some $j$.
    \item \emph{Positive} if $k_{j,x}\ge 0$ for all $j$.
    \item \emph{Stable} if $(\ghostTrans(s_0,u),x)$ is a stable cycle. In this case it is \emph{Degenerate} if $(\ghostTrans(s_0,u),x)$ is degenerate, and \emph{Non-degenerate} otherwise.
\end{itemize}
We also distinguish between SRIs that pertain to $\pot$ and to $\charge$, and specialise the definition for each one. The definitions here (and also in other places) make use of the \emph{maximal effect} of a word $w$, denoted $\maxeff{w}$, which is an upper bound on the change in weight upon reading $w$. Also, recall that $G$ is the gap size of the SRI, and that $\simpL,\generL$ are two length functions giving rise to two alphabets $\simpAbsCac,\genAbsCac$ which we do not elaborate on here.
\begin{itemize}
    \item A \emph{simple} SRI (SSRI) is over $\simpAbsCac$, has $G=4\bigM \maxeff{xy}$, and requires $|x|\le \frac{1}{\bigM}\simpL(\depth(x)+1)$ and $\pot(u)\le \pot(ux)\le \pot(uxy)$.
    \item A \emph{general} SRI (GSRI) is over $\genAbsCac$, has $G=4\bigM \maxeff{xy}+\budH$, and requires $|x|\le \frac{1}{\bigM}\generL(\depth(x)+1)$ and $\charge(u)\ge \charge(ux)\ge \charge(uxy)$, where $\budH$ is a constant explained later.
\end{itemize}
The bounds involved in the definition (including $\budH$) arise from the proof as we go along. 
In the following, we use SRI as a general term for either an SSRI or GSRI.

SRIs serve as a bridge between two parts of the proof: on the one hand, if a word can be decomposed as an SRI, we show that we can use its structure to obtain certain properties (mostly -- shortening the word while maintaining $\pot$ or $\charge$). Showing such properties is technically challenging, and makes extensive use of baseline shifts and dominance.
On the other hand, showing that SRIs can be found at all is challenging, and requires us to develop a \emph{quantitative zooming technique}. 

Didactically, the correct way to proceed with the proof is to ask how large a $B$-gap witness we need to guarantee the existence of an SRI, and then show how an SRI helps us once we find it. It is more concise, however, to develop our tools ``bottom-up'' and combine everything at the end. We therefore start by showing what we can do given an SRI.

\subsection{SRI -- A Toolbox (\cref{sec: basic results for SRI})}
Our first result identifies what kinds of runs are allowed when reading $x$ and $y$ in an SRI.  Intuitively, it shows that there are no runs from a ``low'' $V_j$ to a ``higher'' $V_{j'}$.
\begin{proposition}
\label{prop:abs:SRI run structure}
    Consider an SRI $uxyv$ and $\rho:p\runsto{x}q$ with $p\in V_j$, then $q\in V_{j'}$ with $j'\le j$. Moreover, there exists some $p'\in V_{j'}$ such that $p'\runsto{x}q$. 
    We also have $|k_{j,x}|\le \maxeff{x}$.

    The same holds for $y$.
\end{proposition}
\begin{proof}[Proof Sketch]
The idea is simple, and appears already in~\cite{almagor2026determinization}: if there is a run that goes from e.g., $V_2$ to $V_3$, then the gap between $V_2$ and $V_3$ (after reading $x$ or $y$) would be at most e.g., $2\maxeff{x}$. Indeed, two runs on $x$ stemming from the same state cannot diverge by more than $2\maxeff{x}$.

Note that in order to make the proof work, we must ensure that the gap between the sets is much higher than $2\maxeff{x}$, which is part of the reason behind the choice of $G$ above.
\end{proof}
We now consider Stable SRIs, which are the strictest flavour. We show in \cref{sec: basic results for SRI} that every Stable SRI is in particular Positive. This is not surprising, since stable cycles do not allow negative cycles.
For Degenerate Stable SRIs, we have that $x$ is a degenerate stable cycle. As the name suggests, this induces a degenerate behaviour. Concretely, we show that we can get rid of $x$ entirely.
\begin{proposition}
\label{prop:abs:degenerate SRI shorten}
    For a degenerate SRI $uxyv$ we have $\charge(uxyv)=\charge(uyv)$ and $\pot(uxyv)=\pot(uyv)$.
\end{proposition}
For non-degenerate SRIs, things are more involved. Fortunately, techniques for this fragment are developed in~\cite{almagor2026determinization}, and require only small adaptations. Intuitively, we show that we can replace $x$ with $\alpha_{S',x}$ and drop $y$ entirely, while keeping $\pot$ and $\charge$ well behaved.
We remark that the following result is what gives rise to $G\ge 4\bigM\maxeff{xy}$ for both SSRI and GSRI.
\begin{lemma}
\label{lem:abs:pumping stable SRI}
Consider a Stable Non-degenerate SRI $w=uxyv$ and let $w'=u\alpha_{S',x}v$, then 
\begin{itemize}
    \item $\charge(w)\ge \charge(w')$, and
    \item either there is a witness, or $\pot(w)\le \pot(w')$.
\end{itemize}
\end{lemma}
Naturally, the ``either there is a witness'' clause stems from \cref{lem:abs:unfolding maintains potential}.

Our first major roadblock is Negative SRIs. Indeed, suppose $V_3$ has $k_{3,x}<0$ and contains a maximal dominant state, but $k_{2,x}=0$ and $V_2$ contains the baseline run. Then, $\pot$ may \emph{decrease} with pumping. Dually, if a negative $V_j$ is below the baseline, this may cause $\charge$ to \emph{increase}. Such behaviours are very problematic and we would like to get rid of them. 
The problem goes even further: even in a Positive SRI, there may be negative \emph{ghost runs} on $x$ (c.f., \cref{sec:abs:unfolding}). These also cause problems if we want to turn $x$ to a stable cycle.

The following result (\cref{lem: shift on negative run in SRI kills positive sets}) handles these scenarios using careful baseline shifts, and identifies what happens to the resulting SRI. Intuitively, the lemma considers a setting where there is a negative cycle on some $x^k$ from a reachable or ghost-reachable state $s$. Then, by \cref{lem:abs:reflexive to stable} we can shift $x^k$ to obtain a reflexive cycle $(S'',x')$. 
We then ask what this stable cycle does to the states in various $V_j$'s. The lemma says that if $V_1,\ldots,V_{\ell'}$ (i.e., the lowest $\ell'$ sets, for some $\ell'\le \ell$) all have $k_{i,x}\ge 0$, then upon reading $\alpha_{S'',x'}$ they all jump to $\infty$. 
\begin{lemma}
    \label{lem:abs:shift on negative run in SRI kills positive sets}
     Consider an SRI $uxyv$ where $k_{j,x}\ge 0$ for all $0\le j\le \ell'$ for some $\ell'\le \ell$ (i.e., the first $\ell'$ $k_{j,x}$ are positive).
    Assume there is a minimal-slope run $\rho:s\runsto{x^k}s$ with $s\in \ghostTrans(s_0,u)$, $k\le |S|$ and $\weight(\rho)<0$. 
    Then $(S'',x')=\baseshift{(\ghostTrans(s_0,u),x^k)}{\rho}$ is a stable cycle, and for every state $s'\in V_j$ with $j\le \ell'$ we have $\minweight(\baseshift{s'}{\rho}\runsto{\alpha_{S'',x'}}S)=\infty$.
\end{lemma}
\begin{proof}[Proof Sketch]
    Intuitively, we proceed in four steps as depicted in \cref{fig:abs:shift on neg kills pos sets proof}. First, we baseline shift $x^k$ on the negative run $\rho$, resulting in a stable cycle $(S'',x')$ as per \cref{lem:abs:reflexive to stable}. 
    
    We want to use the cactus letter $\alpha_{S'',x'}$. However, we need it to be effective, so as not to go beyond our effective alphabet. This is where the requirement on the length of $x$ in an SRI comes into play: since $k<|S|\ll \bigM$ and $|x|\le \frac{1}{\bigM}\simpL(\depth(x)+1)$ (or similarly for $\generL$), then $|x^k|\le \simpL(\depth(x)+1)$, so we can consider the corresponding cactus letter as per \cref{sec:abs:effective cactus letters}.
    
    We now consider a state $s_1$ in some $V_j$ with $j\le \ell'$, i.e., one of the ``low and non-negative'' $V_j$'s. 
    The shifted version of $s_1$ is $t_1=\baseshift{s_1}{\rho}$, and we examine what happens when it reads $\alpha_{S'',x'}$. 
    If, by way of contradiction, there is a (finite-weight) run $t_1\runsto{\alpha_{S'',x'}}t_2$, then $(t_1,t_2)$ is a grounded pair for $(S'',x')$ (as per \cref{sec:abs:reflexive and stable cycles}). This means that there exists a minimal reflexive state $t_3\in S''$ and a run $\tau:t_3\runsto{x'^{\bigM}}t_3$ with $\weight(\tau)=0$ such that $t_1\runsto{x'}t_3$.

    We now ``pull back'' $t_3$ from the baseline $\rho$ to our original baseline of $uxyv$ using another baseline shift, denoting the corresponding state $t'_3$. Thus, there is a run $s_1\runsto{x^k}t'_3$. By \cref{prop:abs:SRI run structure} we have that $t'_3\in V_{j'}$ for some $j'\le j$ (recall that runs cannot go to ``higher'' sub-configuration). 
    In particular, by the premise of the lemma we have $k_{j',x}\ge 0$, so there are no negative cycles on $x$ generated from $t'_3$.
    
    However, we now consider the run $\tau$, and similarly pull it back to a run $\tau':t'_3\runsto{x^k}t'_3$.
    Since baseline shifts preserve the run structure and since $\weight(\tau)=0$, then $\weight(\tau')$ must be ``$0$ with respect to $\weight(\rho)$''. More precisely, we have
    \[\weight(\rho)-\weight(\tau')=\weight(\baseshift{\rho}{\rho})-\weight(\tau)=0\]
    Thus, $\weight(\tau')=\weight(\rho)<0$. 
    This means that there is a negative cycle from $t'_3$ on $x^k$. It is now not difficult to reach a contradiction with the fact that $k_{j',x}\ge 0$ (see the full proof for details).

    \begin{figure}[H]
        \centering
        \includegraphics[width=1\linewidth]{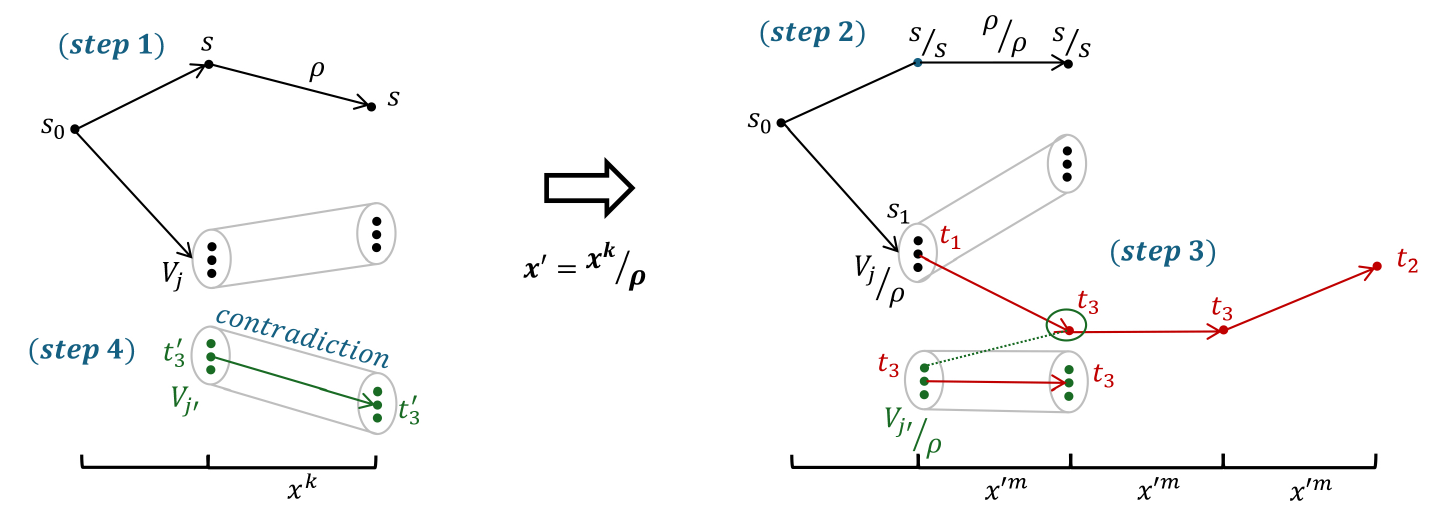}
        \caption{Proof of \cref{lem:abs:shift on negative run in SRI kills positive sets}. In step 1 we consider the decreasing run $\rho$ over $x^k$. We baseline shift to it to obtain the figure on the right. In step 2 we consider some increasing $V_j$, and assume there is a finite-weight run from it (the red run). We then identify a state $t_3$ with a $0$-weight run. In step 3 we pull $t_3$ back to the left drawing ($t'_3$) so it becomes a negative run. In step 4 we reach a contradiction, since this negative run is lower than $V_j$.}
        \label{fig:abs:shift on neg kills pos sets proof}

    \end{figure}
\end{proof}

Equipped with general tools for SRI, we proceed to develop tools specific for SSRI and GSRI.

\subsection{SSRI -- A Toolbox (\cref{sec:SSRI toolbox})}
Recall that SSRI pertain to $\pot$. In the main proof (\cref{sec:abs:proof outline}) when we reach an SSRI our goal is either to find a witness, or to shorten the SSRI while not decreasing $\pot$.
We present the relevant tools to achieve this.
Equipped with \cref{lem:abs:shift on negative run in SRI kills positive sets}, we start by tackling Negative SSRI.
\begin{lemma}
\label{lem:abs:negative SSRI to type 1 witness}
    If there exists a Negative SSRI, then there exists a witness.
\end{lemma}
\begin{proof}[Proof Sketch]
    Consider a Negative SSRI $w=uxyv$. The proof proceeds in four steps.
    
    Recall that a dominant state has a suffix it can read with finite weight, and that all states below it cannot. 
    In Step 1 we show that the maximal dominant states at $\vec{c_u}$ and $\vec{c_{ux}}$ are the same. 
    This is simple, since the ordering of the states induced by $\vec{c_u}$ and $\vec{c_{ux}}$ is maintained, due to the large gaps between sub-configurations, and the identical sub-configurations before and after reading $x$.

    In Step 2 we start our path towards a witness by flattening the prefix $u$ (as per \cref{sec:abs:unfolding}). We then ask what this does to $\pot$. Recall that flattening is repeated unfolding. We therefore call upon \cref{lem:abs:unfolding maintains potential}, and get that either there is a witness (in which case we are done), or that flattening does not affect the potential and the dominant states. 
    Note that after flattening, the ghost states are reachable.

    In Step 3 we consider the minimal-slope run on the $x$ infix (rather, on any $x^k$), and baseline shift on this run, thus inducing a stable cycle $x'$ as per \cref{lem:abs:reflexive to stable}. Note that since the SSRI is negative, this run has negative slope indeed.
    We also shift the (flattened) prefix $u$ so that it ``glues'' properly to the shifted cycle. We then keep track again of the maximal dominant state, and show (using our result that dominance is shift-invariant, described in \cref{sec:abs:potential and charge}) that it corresponds to the shift of the dominant state at $\vec{c_u}$. We also designate a suffix $z$ that shows it is dominant.

    Finally, in Step 4 we show that the resulting tuple of (prefix,stable cycle,suffix) forms a witness. Specifically, we consider $(u,\alpha_{S'',x'},z)$. By \cref{def:abs:witness} we want to show that $\minweight(s_0\runsto{u\alpha_{S'',x'}z})=\infty$, but that 
    $\minweight(s_0\runsto{ux'^{2\bigM}z})<\infty$.
    First, since in an SSRI $x$ does not change the reachable states, this also holds for $x'$. Thus, the dominant state $t$ in $\vec{c_u}$ is reachable also by $ux'$. By definition, $z$ can be read from $t$ (as it is the separating suffix showing that $t$ is dominant). We therefore have
    \[\minweight(s_0\runsto{ux'^{2\bigM}z})\le \minweight(s_0\runsto{ux'}t\runsto{z}S)<\infty\]
    Assume by way of contradiction that there is a finite weight run
    $\tau: s_0\runsto{u}s'_1\runsto{\alpha_{S'',x'}}s'_2\runsto{z}s'_3$.
    This means that $(s'_1,s'_2)$ is a grounded pairs. Using similar baseline shifts as in the proof of \cref{lem:abs:shift on negative run in SRI kills positive sets}, we show that this puts $s'_2$ in a set $V_j$ with negative $k_{j,x}$. 

    We then compare the weight of $s'_2$ to the weight of the maximal dominant $t$. Intuitively, we show the following:
    \begin{itemize}
        \item If $s'_2$ is below $t$, then it cannot read $z$ by the definition of dominance (which is a contradiction to the assumption).
        \item If $s'_2$ is in the same set $V_j$ as $t$, then since $k_{j,x}$ is negative, this would mean that the $\pot$ is \emph{decreasing} along $x$, in contradiction to the definition of SSRI.
        \item If $s'_2$ is way above $t$, then we can show it is maximal dominant, contradicting the maximal dominance of $t$.
    \end{itemize}
\end{proof}

Having handled Negative SSRI, we now turn our attention to Positive SSRI. 
\begin{lemma}
    \label{lem:abs:positive SSRI to Stable SSRI}
    Either there is a witness, or every Positive SSRI is a Stable SSRI.
\end{lemma}
\begin{proof}[Proof Sketch]
    Consider a Positive SSRI $uxyv$. If $(\ghostTrans(s_0,u),x)$ is a stable cycle, then $uxyv$ is a Stable SSRI and we are done. 
    Assume $(\ghostTrans(s_0,u),x)$ is not a stable cycle, then there is some negative-slope run $\rho:s\runsto{x^k}s$ for some $k$ and $s\in \ghostTrans(s_0,u)$. 
    However, since $x$ is part of a Positive SSRI, then no reachable state has such a cycle, i.e., $s\in \ghostTrans(s_0,u)\setminus\booltrans(s_0,u)$ is a proper ghost state.

    We are now within the conditions of \cref{lem:abs:shift on negative run in SRI kills positive sets}: $k_{j,x}\ge 0$ for all $j$ (since we are in a Positive SSRI), and there is a negative run. We apply it and get that $(S'',x')=\baseshift{(\ghostTrans(s_0,u),x^k)}{\rho}$ is a stable cycle, and that $\minweight(\booltrans(s_0,u)\runsto{\alpha_{S'',x'}}S)=\infty$ (since now all reachable states jump to $\infty$ upon reading $\alpha_{S'',x'}$)\footnote{The diligent reader may notice that we cannot actually read $\alpha_{S'',x'}$ after $u$, due to a mismatch in baselines. This is handled carefully in the full proof.}. 
    On the other hand, $\minweight(\baseshift{s}{\rho}\runsto{x'}S)<\infty$, since we shift on $\rho$. This is exactly the structure of a witness (\cref{def:abs:witness}), so we conclude the claim.

    We remark that the precise definition of witness (\cref{def:witness}) is somewhat more elaborate. Thus, the full proof has some additional steps that we omit here.
\end{proof}
We are now equipped to handle SSRI on all their flavours (where Stable SSRI are handled as SRI). 

\subsection{GSRI -- A Toolbox (\cref{sec:GSRI toolbox})}
We proceed to develop a single tool for GSRI. Recall that now our main focus is $\charge$, i.e., the distance of the minimal run from the baseline ($=0$). Also recall that in a GSRI we have $\charge(u)\ge \charge(ux)\ge \charge(uxy)$, implying that the minimal run increases. It is not hard to show that this implies $k_{1,x}\ge 0$.
Further recall that the gaps in GSRI have an additional $+\budH$ on top of their SSRI counterparts, where $\budH$ is yet to be defined, but should be thought of as a very large constant.

The next lemma is the central property of GSRI, and is the counterpart of \cref{lem:abs:positive SSRI to Stable SSRI}. Conceptually, it is what links the charge to the potential.
\begin{lemma}[From GSRI to Stable GSRI or High Potential]
\label{lem:abs:GSRI to stable or high potential} 
    Consider a GSRI $w=uxyv$, then it is a Stable GSRI, or there is a word $w'$ such that $\pot(w') \ge \budH$.
\end{lemma}
\begin{proof}[Proof Sketch]
    If $(\ghostTrans(s_0,u),x)$ is a stable cycle then $uxyv$ is a Stable GSRI and we are done.
    
    If $(\ghostTrans(s_0,u),x)$ is not a stable cycle, then there is some minimal negative-slope run $\rho$ on some $x^k$. Since $k_{1,x}\ge 0$, this cycle cannot stem from a $V_1$ state, so it stems from a higher (possibly ghost) state $s$. 
    Using $V_1$ as the ``low'' sub-configurations, we can again apply \cref{lem:abs:shift on negative run in SRI kills positive sets} to obtain a cactus letter $\alpha_{S'',x'}$ that has no finite-weight transition from $V_1$.
    By the gap property of GSRI, $s$ is more than $\budH$ higher than $V_1$.

    Intuitively, this can be used to show that the maximal dominant state $s'$ after reading $u$ is not in $V_1$, and is therefore at least $\budH$ above it (the precise details are not so straightforward, but we omit them here).

    We are now almost where we want: a very high dominant state. In order to make $\pot$ high, we need to make sure this dominant state is also high above the baseline. 
    Fortunately, this can be done by baseline-shifting $u$ to a word $u'$ so that $V_1$ is on the baseline, rendering $\pot(u')>\budH$.
    We depict the proof only in~\cref{fig:pos GSRI to high potential}, since the depiction uses elements that are glossed over in this sketch.
\end{proof}

\section{Quantitative Zooming (\cref{sec:existence of SSRI,sec:existence of GSRI})}
\label{sec:abs:quantitative zooming}
% Introduce zooming as a method to find SRIs effectively.
% Explain how constants are chosen and how they interact.
We now come to the second challenge of showing that under certain conditions, an SSRI or GSRI exists. The technical contribution here is the \emph{quantitative zooming technique}. We start with a rough intuition.
Consider a long word $w$, and suppose (unreasonably) there is some infix $w_2$ of $w$ where every state in $S$ is reachable, and the gap between each two states is very large. We refer to this setting as $|S|$ \emph{independent runs}. Moreover, suppose $w_2$ is very long with this property. 
By the pigeonhole principle, at some point (i.e., after $2|S|!+1$ steps) the reachable states repeat three times in the same order. The infixes $x$ and $y$ between these three repetitions can almost serve as an SRI, except we also need to track for every run whether it is increasing, decreasing, or $0$, in order to have $\sign(k_{j,x})=\sign(k_{j,y})$. Adding this information, we can find an SRI in this simplest of cases. Note that this does not give any guarantees about $\pot$ or $\charge$. Also note that in this case, each $V_j$ is just a single state.

We now ask what happens when there are only $|S|-1$ states that are far away from each other, and another reachable state $s$. In this case there are two options: if $s$ is always close up to a constant to a single independent run, we say that $s$ is \emph{covered}. We can then track the exact distance from $s$ to its covering run, and using the pigeonhole principle we can again find an SRI, this time with one $V_j$ capturing $s$ and its covering run.

The second option is that $s$ gets far away from every independent run (depicted in \cref{fig:abs:zooming}). In this case, we observe that in order to get \emph{really} far away, $s$ must be \emph{somewhat} far away for a long while. We can then ``zoom in'' on this part of the run, and we see that now $s$ essentially generates its own independent run. This brings us back to the $|S|$ independent runs case. Thus, we have an inductive scheme. 
Note that this still does not track the $\pot$ and $\charge$, which we need to add somehow. However, this is a small issue compared to the following.

\begin{figure}
  \begin{center}
    \includegraphics[width=1\textwidth]{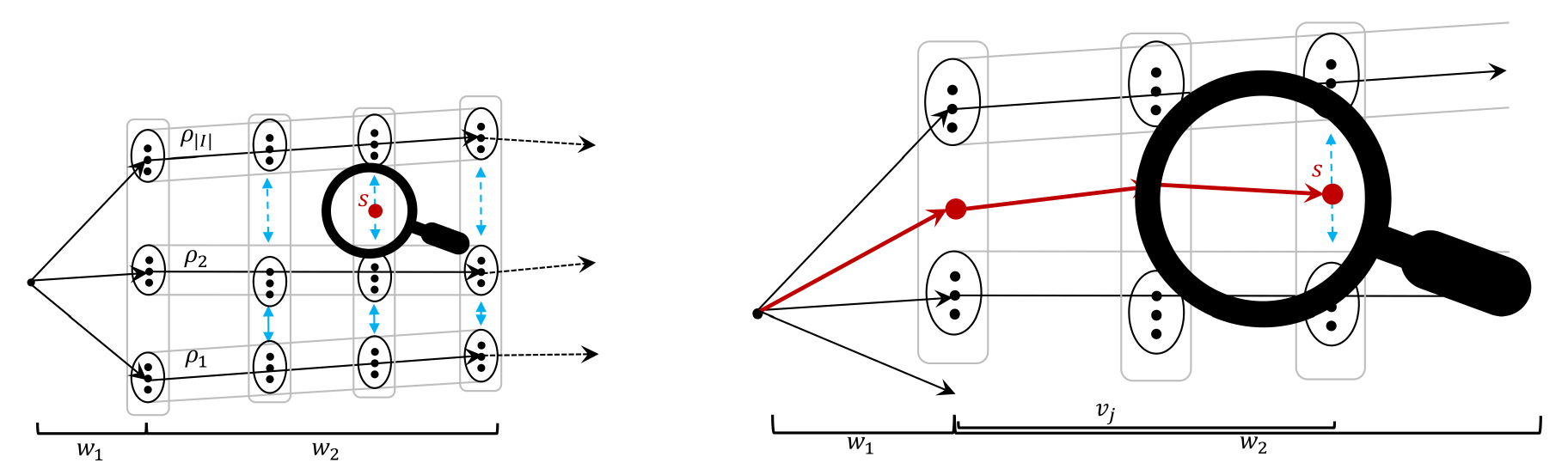}
  \end{center}
  \caption{The state $s$ is not near any of the independent runs. We can then zoom in on $s$, and construct a new independent run that reaches it.}
  \label{fig:abs:zooming}
\end{figure}

\subparagraph*{A Major Quantitative Problem.}
We are now facing a significant problem: when zooming in we need to modify the definition of ``far away'', to a smaller gap. This needs to be done in a way such that at the end of the induction we indeed have sufficient gaps for an SRI.
It is tempting to now say that we can define the gap arbitrarily large, and just take long enough words. However, to do so we need to allow long words in our cactus alphabet (inductively). This, in turn, increases the \emph{weight} of nested cactus letters (longer inner word $\implies$ higher weight). But this defeats the purpose: if the weight of letters is higher, we need the gap to be even larger, creating a cyclic dependence. 

One way to overcome this is to lose the effectiveness. Indeed, the approach taken in~\cite{almagor2026determinization} assumes infinitely many words of increasing lengths and gaps. Then all these problems go away. However, this renders the proof non-constructive.

Here we carefully tailor the length function together with: the gap sizes, cover sizes, maximal weight, and the numbers required to apply the pigeonhole principle in such a way that all the quantitative arguments fall into place.
Crucially, all of these functions are now parametrised not just by the depth of the word, but also by the number of \emph{independent runs}, to facilitate the induction.

\subsection{Recursive Functions (\cref{sec:effective cactus})}
\label{sec:abs:effective functions}
In order to reason about the various properties required for zooming, we introduce two families of recursive functions. We bring here their intuitive description and dependencies, without the explicit definitions. The parameters $d$ and $i$ represent the \emph{depth} of a word and the number of \emph{independent runs}, respectively. The inductions are always decreasing with $d$ (smaller depth is simpler) and increasing with $i$ (more independent runs is simpler).

$\bullet\ \simpLfuncLength{d}{i}$: the maximal length we allow for a word $w$ in a cactus letter $\alpha_{S',w}$ of depth $d$ with $i$ independent runs. Depends on: $\simpLfuncLength{d}{i+1}$, $\simpLfuncTypes{d}{i}$, and $\simpLfuncAmp{d}{i}$

$\bullet\ \simpLfuncCover{d}{i}$:  the maximal distance allowed between a run and its closest independent run (i.e., the \emph{cover}). Depends on: $\simpLfuncMaxW{d-1}$, $\simpLfuncLength{d}{i+1}$ and $\simpLfuncCover{d}{i+1}$.

$\bullet\ \simpLfuncMaxW{d}$: the maximal weight of a letter of depth $d$ (does not depend on $i$). Depends on: $\simpLfuncMaxW{d-1}$ and $\simpLfuncLength{d}{1}$ (note that $1$ is the ``worst'' $i$).

$\bullet\ \simpLfuncAmp{d}{i}$: A lower bound on what is considered ``high values'' of $\pot$, dubbed \emph{amplitude}. Depends on: $\simpLfuncMaxW{d-1}$, $\simpLfuncTypes{d}{i}$, and $\simpLfuncLength{d}{i+1}$.

$\bullet\ \simpLfuncTypes{d}{i}$: An upper bound on the number of configurations representable with $i$ runs, such that all other runs are within $\simpLfuncCover{d}{i}$ from an independent run. These are the ``pigeonholes''.
Depends on $\simpLfuncCover{d}{i}$. 

We can now define some well-awaited terms: the function $\simpL(d)=\simpLfuncLength{d}{1}$ for all $d$, and the constant $\budH=\simpLfuncAmp{|S|}{0}$.

The functions above pertain to $\pot$. A similar set of functions (prefixed by $\texttt{G}$ for ``general'') pertains to $\charge$. Their definitions, however, depend also on $\budH$. In particular, we can define $\generL(d)=\generLfuncLength{d}{1}$.

\subsection{From Increasing $\pot$ to SSRI (\cref{sec:existence of SSRI})}
\label{sec:abs:inc pot to SSRI}
Consider a word $w=w_1w_2$ with $\pot(w_1)\le \pot(w_2)$, and such that along $w_2$ we identify a set $I$ of runs that maintain some gaps between them. Denote $i=|I|$ and $d=\depth(w_2)$.
Assume that $w_2$ is fairly long, specifically, $|w_2|=\beta\cdot  \simpLfuncLength{d+1}{i}$ for some small fixed $0<\beta<1$. We now distinguish between two cases: the case where $\pot(w_1w_2)>\pot(w_1)+\simpLfuncAmp{d+1}{i}$ (\emph{high-amplitude}), and when $\pot(w_1z)$ remains within $\pm \simpLfuncAmp{d+1}{i}$ for every prefix $z$ of $w_2$ (\emph{bounded amplitude}). These do not cover all possible cases, but they suffice.

\begin{figure}
  \begin{center}
    \includegraphics[width=0.8\textwidth]{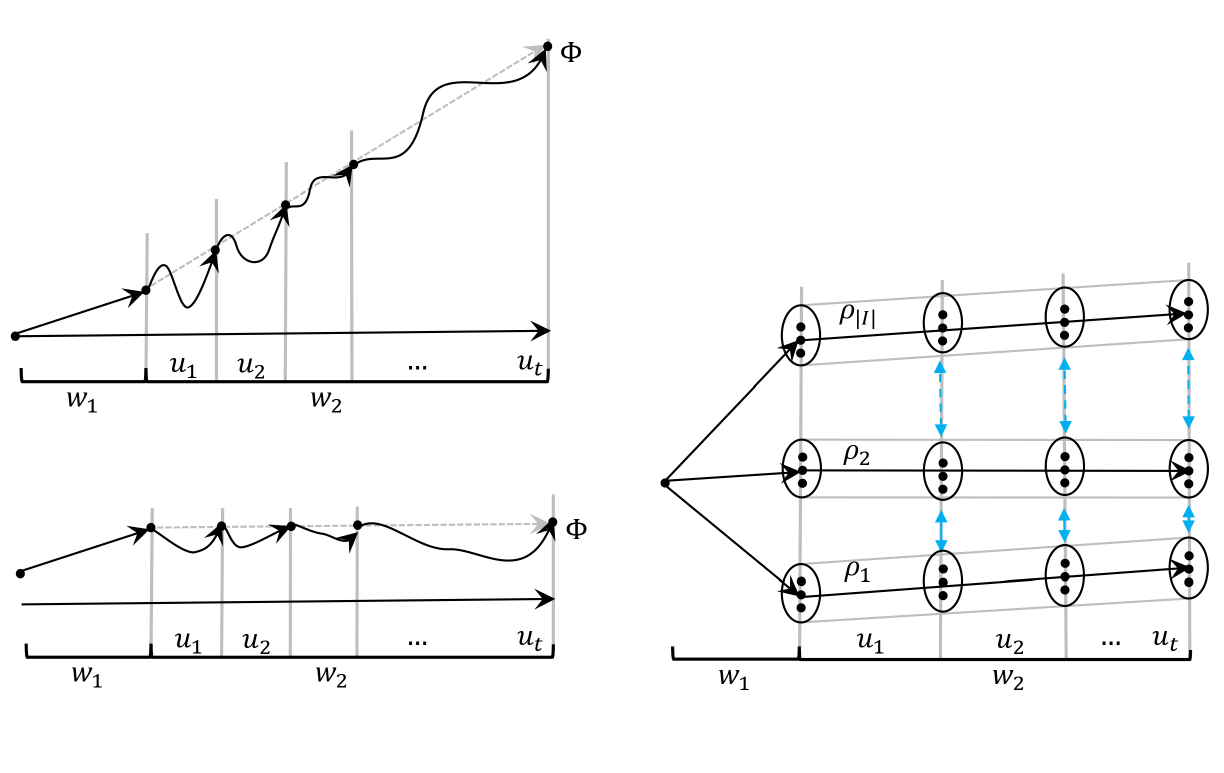}
  \end{center}
  \caption{In a $\pot$-increasing decomposition (top left), the potential gradually increases, and within each segment, it never surpasses the potential at its end. In a $\pot$-bounded decomposition (bottom left), the potential keeps its value at the end of the segments, and certain points within the segments never surpass it. In a cover decomposition (right), the states within the sub-configurations at the end of each segment are close to one of the independent runs.}
  \label{fig:abs:decompositions}
\end{figure}

We then proceed to define two types of \emph{decompositions} (see \cref{sec:existence of SSRI}): in the high-amplitude case, we divide the word to many infixes, such that $\pot$ strictly increases between each two infixes, and never exceeds $\pot$ of the next segment. In addition, we require that these segments are not too short (at least $\simpLfuncLength{d}{1}$), and that there are many of them -- so many that $\simpLfuncTypes{d}{i}$ can serve as pigeonholes for a certain argument (see \cref{fig:abs:decompositions}).

In the bounded-amplitude case we also show a decomposition, still requiring long segments, and that the potential is equal after each segment. We remark that showing these decompositions is technically challenging (see \cref{prop: pot high amplitude to pot increasing decomposition,prop: pot bounded amp to bounded decomposition}) and accounts for much in the recursive definitions.

Next, we call a decomposition a \emph{cover}, if after each segment, all states are within $\simpLfuncCover{d+1}{i}$ from some independent run. 
Crucially, the number of segments in the decomposition is chosen so that if it is also a cover, then we can apply the pigeonhole principle and find $3$ repetitions of a configurations, using which we can obtain an SSRI. 

With these definitions in place, we can apply the zooming technique described above: start with a single run in $I$, and decompose the word. If its decomposition is a cover, we find an SSRI and we are done. Otherwise, there is some state that goes outside the cover. Zooming in on it and the run that leads to it, we can induce a new independent run. An important point is that the carefully-chosen lengths, gaps and amplitudes actually yield all the correct assumptions for $i+1$, and we can proceed. 

A caveat to the above is that while $\depth(w_2)=d$, it may be the case that the infix $u$ we zoom on has $\depth(u)<d$. This is a problem since $u$ is now too \emph{long} to serve as an effective cactus letter. To this end, we tailor the functions even further so that in such cases we can use induction with $d-1$ rather than $i+1$.

Proceeding with the induction up to depth $|S|$, we establish the following (see \cref{cor: potential increasing to SSRI} for the precise formulation).
\begin{lemma}
    \label{lem:pot increase to SSRI}
    Consider a word $w_1w_2$ with $\pot(w_1)\le \pot(w_1w_2)$ and $|w_2|=\beta\cdot \simpLfuncLength{|S|}{1}$, then $w_1w_2=uxyv$ where $uxyv$ is an SSRI.
\end{lemma}

\subsection{From Decreasing $\charge$ to GSRI (\cref{sec:existence of GSRI})}
\label{sec:abs:decreasing charge to GSRI}
Using similar techniques to \cref{sec:abs:inc pot to SSRI}, we define decompositions and cover of $\charge$-decreasing words. 
Here, however, we cannot assume that $\charge$ has bounded growth, so \cref{cor: charge decrease to GSRI or high potential} looms over us. 

We condition on $\charge$ having bounded growth, and choose a constant depending on $\budH$ for this bound. This allows us on the one hand to obtain a GSRI, by relying on the definition of the General recursive functions, which account for this $\budH$ factor, and on the other hand if the conditioning fails, we obtain a word with high enough $\pot$ amplitude. We show the following.
\begin{lemma}
    \label{lem:charge decrease to GSRI}
    Either there is a word $w\in \simpAbsCac^*$ with $\pot(w)> \budH$, or for every
    word $w_1w_2\in \genAbsCac^*$ with $\charge(w_1)\ge \charge(w_1w_2)$ and $|w_2|=\beta\cdot \generLfuncLength{|S|}{1}$, we have $w_1w_2=uxyv$ where $uxyv$ is a GSRI.
\end{lemma}

\section{Outline of the Proof}
\label{sec:abs:proof outline}
We are finally ready to put all the components together. The proof outline is depicted in \cref{fig:proof flow}.
\begin{figure}[H]
    \centering
    \includegraphics[width=0.95\linewidth]{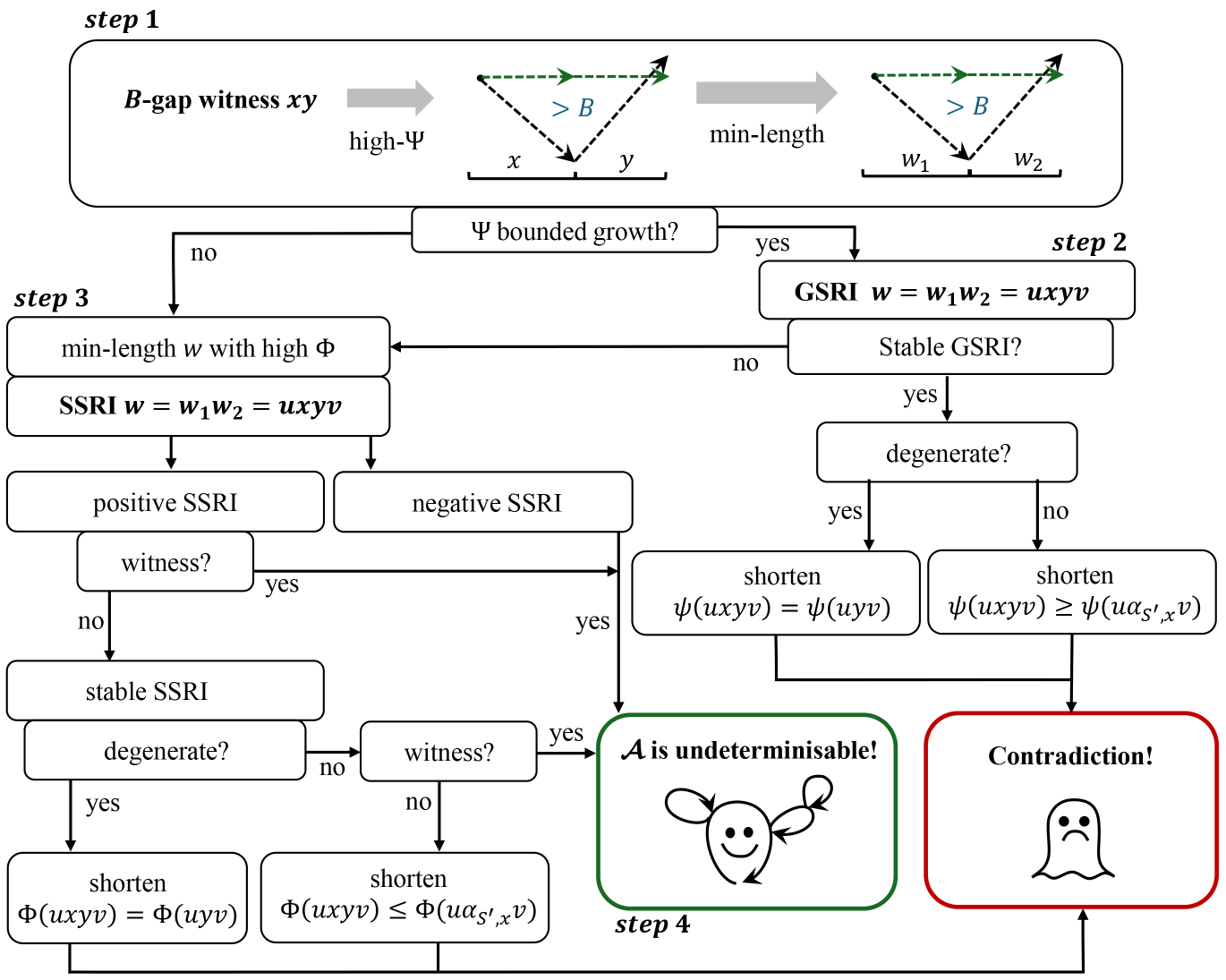}
    \caption{A summary of our proof flow.}
    \label{fig:proof flow}
\end{figure}
The precise main result is the following.
\begin{theorem}
    \label{thm:abs:main large gap to undet}
    Consider a WFA $\cA$, then $\cA$ is undeterminisable if and only if there is a $B$-gap witness with $B=\generLfuncAmp{|S|}{0}$ where $|S|$ is the size of $\augA$.
\end{theorem}
Note that the direction $\cA$ is undeterminisable $\implies$ $B$-gap witness is trivial from \cref{thm:abs:det iff bounded gap}. We handle the converse in several steps.

\subparagraph*{Step 1: From a Gap witness to a $\charge$-Decreasing Word.}
Consider a $B$-gap witness $xy$ for $B=\generLfuncAmp{|S|}{0}$ as per \cref{thm:det iff bounded gap}. By choosing the minimal run on $xy$ as baseline, we have another run that goes far below it on $x$, yielding $\charge(x)>B$ but $\charge(xy)=0$. That is, a significant decrease in $\charge$ upon reading $y$.

\subparagraph*{Step 2: From $\charge$-Decreasing to $\pot$-Increasing}
Thus far, $xy$ do not have cactus letters. Let $w_1w_2$ be a \emph{minimal length} word over $\genAbsCac^*$ such that $\charge(w_1)-\charge(w_1w_2)>\generLfuncAmp{|S|}{0}$. Note that we now allow cactus letters, and also require length minimality. There exists such a word because $xy$ satisfies this condition. By the construction of $\genAbsCac$, we have $\depth(w_1w_2)\le |S|-1$.

We now ask whether there is a word $w'\in \simpAbsCac^*$ such that $\pot(w')> \budH$. If there is, we proceed to Step 3.
Otherwise, we assume by way of contradiction that there is no such word.

By \cref{lem:abs:charge decrease to high potential}, this assumption tells us that the $\charge$ has bounded decrease over $\genAbsCac$. Since $w_1w_2$ has a huge charge drop, then $w_2$ must be very long. We use this with \cref{lem:charge decrease to GSRI} to show that $w_1w_2$ is a GSRI. 
Then, by \cref{lem:abs:GSRI to stable or high potential} and the assumption that there is no word with $\pot(w')> \budH$, we get that this GSRI is stable. In particular, it is a Stable SRI $w_1w_2=uxyv$. 
Then, we have the following:
\begin{itemize}
    \item If it is a Degenerate SRI then by \cref{prop:abs:degenerate SRI shorten} we can shorten $w_1w_2$ while maintaining $\charge$, contradicting the minimality.
    \item If it is Nondegenerate, then by \cref{lem:abs:pumping stable SRI} we can shorten it to $u\alpha_{S',x}v$ (which is indeed shorter, since $xy$ is replaced with a single letter), while decreasing $\charge$ on the suffix. This is again a contradiction to the minimality of $w_1w_2$.
\end{itemize}
Thus, either way we reach a contradiction. So there must exist a word $w'\in \simpAbsCac^*$ with $\pot(w')> \budH$.

\subparagraph*{Step 3: Increasing Potential to Witness}
Intuitively, we now repeat Step 2 with $\pot$ and SSRI instead of $\charge$ and GSRI, and this results in a witness, as follows.

By the existence of $w'$, let $w\in \simpAbsCac^*$ be the shortest word for which $\pot(w)>\budH$. Note that we again require minimality, now over $\simpAbsCac$.
By the bounded growth of $\pot$ (which is unconditional, see \cref{sec:abs:potential and charge}) we can show that $w$ is long enough to apply  \cref{sec:abs:inc pot to SSRI}. 
We can therefore write $w_1w_2=uxyv$ where $uxyv$ is an SSRI.

The case split now is somewhat more intricate than in GSRI. In each case, if we reach a witness then we proceed to Step 4.
\begin{itemize}
    \item If $uxyv$ is a Negative SSRI, then by \cref{lem:abs:negative SSRI to type 1 witness} there is a witness.
    \item If $uxyv$ is not a Negative SSRI, then it is a Positive SRI. By \cref{lem:abs:positive SSRI to Stable SSRI}, either there is a witness or $uxyv$ is a Stable SSRI, and we proceed to the next item.
    \item If $uxyv$ is a Stable SSRI, we split to cases again.
    \begin{itemize}
        \item If $uxyv$ is a Degenerate SSRI, then by \cref{prop:abs:degenerate SRI shorten} we have $\pot(uxyv)=\pot(uyv)$, contradicting the minimality of $w$.
        \item If $uxyv$ is a Non-degenerate SSRI, then by \cref{lem:abs:pumping stable SRI}, either there is a witness or $\pot(uxyv)\le \pot(u\alpha_{S',x}v)$, again contradicting the minimality of $w$.
    \end{itemize}
\end{itemize}
We conclude that in all cases we either reach a contradiction, or we have a witness.

\subparagraph*{Step 4: Witness to Undeterminisability}
Recall that if there is a witness, then $\cA$ is undeterminisable. This concludes the proof of \cref{thm:abs:main large gap to undet}.
\hfill \qed

\subsection{Algorithm and Complexity}
\label{sec:abs:algorithm and complexity}
The proof of the gap characterisation of \cref{thm:abs:det iff bounded gap} is constructive~\cite{filiot2017delay,almagor2026determinization}, in the following sense:
\begin{theorem}
\label{thm:abs:det construction for B bounded gap}
    Consider a WFA $\cA$ and $B\in \bbN$. We can construct a deterministic WFA $\cA|_B$ such that $\cA$ is equivalent to $\cA|_B$ if and only if $\cA$ has $B$-bounded gaps.
    Moreover, $\cA|_B$ can be computed in time $B^{O(|\cA|)}$.
\end{theorem}
A determinisability algorithm is now simple: given a WFA $\cA$, compute $B=\generLfuncAmp{|S|}{0}$, construct $\cA|_B$, and check whether $\cA$ is equivalent to $\cA|_B$ (which is decidable by \cite{Almagor2020Whatsdecidableweighted}). 
If it is, it is determinisable, and $\cA|_B$ is a deterministic equivalent. Otherwise, it is not determinisable by \cref{thm:main large gap to undet}.

The complexity of this algorithm is dominated by the computation of $B$. In \cref{sec:complexity} we analyse the recursive structure of our functions, placing it in the 6th
level of the fast-growing hierarchy~\cite{grzegorczyk1953some, schmitz2016complexity,tourlakis1984computability}, i.e., non-elementary (elementary is the 3rd level of this hierarchy), but primitive recursive (well below the Ackermann function). We do not bring the ideas of the analysis here. In broad terms, the bound is obtained by formulating upper bounds for our functions in terms that are easy to manage using standard operations of the fast-growing functions, namely substitutions and bounded primitive recursion.

\section{Discussion and Future Research}
\label{sec:abs:discussion}
The decidability of determinisation of WFAs has been out of reach for over 30 years until recently, when it was shown to be decidable. This raises the hope that this problem might not be as intractable as its open status may have suggested.

In this paper, we make the first step towards establishing the complexity of this problem, by showing that it is primitive recursive. This is particularly good news, as it is already far below some problems in related counter-based models, specifically reachability in Vector Addition Systems and Universality of One Counter Nets, which are both Ackermann-complete~\cite{czerwinski2022reachability,leroux2022reachability,hofman2018trace}.

A natural question is what machinery is needed to improve the complexity. To this end, our work restores the faith in the original approach: improve the bound on the maximal gap $B$, and get a better algorithm. On the negative side, our approach seems limited for further significant improvements. Indeed, the presence of cactus letters quickly leads to recursive Ramsey-style arguments, which lead to a tower of exponents. A possible direction would be to use only constant-depth cactus letters. This, however, is already dangerously close to working with the original alphabet, and approach that was unsuccessful for decades.

Going beyond determinisability, our work introduces simpler and more accessible tools for analysing the run-tree of a WFA (namely Zooming and SRI). We believe these can be reused and adapted for other counter-based methods, in particular certain open problems for One-Counter Nets and One-Counter Automata.

\bibliography{main}

\pagebreak
\section*{Appendix}
\section{Preliminaries}
\label{sec:prelim}
For an alphabet $\Sigma$, we denote by $\Sigma^*$ 
%(resp. $\Sigma^+$) 
the set of finite words 
%(resp. non-empty finite words) 
over $\Sigma$. 
For a word $w\in \Sigma^*$, we denote its length by $|w|$ and the set of its prefixes by $\pref(w)$. We write $w[i,j]=\sigma_i\cdots \sigma_j$ for the infix of $w$ corresponding to $1<i<j\le |w|$.

We denote by  $\bbNinf$ and $\bbZinf$ the sets $\bbN\cup \{\infty\}$ and $\bbZ\cup \{\infty\}$, respectively. We extend the addition and $\min$ operations to $\infty$ in the natural way: $a+\infty=\infty$ and $\min\{a,\infty\}=a$ for all $a\in \bbZinf$. By $\arg\min\{f(x)\mid x\in A\}$ we mean the set of elements in $A$ for which $f(x)$ is minimal for some function $f$ and set $A$.

\paragraph*{Weighted Automata}
A \emph{$(\min,+)$ Weighted Finite Automaton} (WFA for short) is a tuple $\cA= \tup{Q,\Sigma, q_0, \Delta}$ with the following components:
\begin{itemize}
    \item $Q$ is a finite set of \emph{states}.
    \item $q_0\in Q$ is the \emph{initial} state.\footnote{Having a set of initial states does not add expressiveness, as it can be replaced by a single initial state that simulates the first transition from the entire set.}
    \item $\Sigma$ is an \emph{alphabet}. We remark that we sometimes use an \emph{infinite alphabet} $\Sigma$, hence we do not require it to be finite. 
    \item $\Delta\subseteq Q\times \Sigma\times \bbZinf\times Q$ is a transition relation such that for every $p,q\in Q$ and $\sigma\in \Sigma$ there exists exactly\footnote{This is without loss of generality: if there are two transitions with different weights, the higher weight can always be ignored in the $(\min,+)$ semantics. Missing transitions can be introduced with weight $\infty$.} one weight $c\in \bbZinf$ such that $(p,\sigma,c,q)\in \Delta$.
    %is relation $\Delta \subseteq Q \times \Sigma \times \bbZ \times Q$.
\end{itemize}
If for every $p\in Q$ and $\sigma\in \Sigma$ there exists at most one transition $(p,\sigma,c,q)$ with $c\neq \infty$, then $\cA$ is called \emph{deterministic}.

\begin{remark}[On initial and final weights, and accepting states]
    \label{rmk: initial and final weights}
    Weighted automata are often defined with initial and final weights, and with accepting states. In \cite{almagor2026determinization} it is shown that the determinisation problem is logspace-equivalent to our model.
\end{remark}

\paragraph*{Runs}
A \emph{run} of $\cA$ is a sequence of transitions $\rho=t_1,t_2,\ldots,t_m$ where $t_i=(p_i,\sigma_i,c_i,q_i)$ such that $q_i=p_{i+1}$ for all $1\le i<m$ and $c_i <\infty$ for all $1\le i \le m$.
We say that $\rho$ is a run \emph{on the word $w=\sigma_1\cdots\sigma_m$ from $p_1$ to $q_m$}, and we denote $\rho:p_1\runsto{w}q_m$. 
For an infix $x=w[i,j]$ we denote the corresponding infix of $\rho$ by $\rho[i,j]=t_i,\ldots,t_j$ (and sometimes by $\rho(x)$, if this clarifies the indices).
The \emph{weight} of the run $\rho$ is $\weight(\rho)=\sum_{i=1}^m c_i$. 

For a word $w$, the weight assigned by $\cA$ to $w$, denoted $\weight_{\cA}(w)$ is the minimal weight of a run of $\cA$ on $w$. For convenience, we introduce some auxiliary notations.

For a word $w\in \Sigma^*$ and sets of states $Q_1,Q_2\subseteq Q$, denote
\[\minweight_\cA(Q_1\runsto{w} Q_2)=\min\{\weight(\rho)\mid \exists q_1\in Q_1,q_2\in Q_2,\ \rho:q_1\runsto{w}q_2\}\]
If $Q_1$ or $Q_2$ are singletons, we denote them by a single state (e.g., $\minweight_\cA(P\runsto{w} q)$ for some set $P\subseteq Q$ and state $q$).
Then, we can define
$\cA(w)=\minweight_\cA(q_0\runsto{w} F)$.
If there are no accepting runs on $w$, then $\cA(w)=\infty$.
The function $\weight_\cA:\Sigma^*\to \bbZinf$ can be seen as the weighted analogue of the \emph{language} of an automaton.
We omit the subscript $\cA$ when clear from context.

For a word $w\in \Sigma^*$ and a state $q\in Q$, of particular interest are runs from $q_0$ to $q$ on $w$ that remain minimal throughout. Indeed, other runs are essentially ``cut short'' by some lower run. Formally, let $w=\sigma_1\cdots \sigma_n$ and consider a run $\rho:q_0\runsto{w}q$ with $\rho=t_1,t_2,\ldots,t_n$ and $t_i=(q_i,\sigma_i,e_i,q_{i+1})$. We say that $\rho$ is \emph{seamless} if for every $1\le j\le n$ it holds that $\weight(t_1,\ldots,t_j)=\minweight(\sigma_1\cdots \sigma_j,q_0\to q_{j+1})$.

We define $\wmax{w}$ to be the maximal weight (in absolute value) occurring in any possible transition on the letters of $w$, i.e., 
\[\wmax{w}=\max\{|c|<\infty \mid \exists p,q\in Q, 1\le i\le n,\  (p,\sigma_i,c,q)\in \Delta\}\] 
Additionally, the \emph{maximal effect} of $w$ is an upper bound on the change in weight that can happen upon reading $w$, defined as $\maxeff{w}=\sum_{i=1}^n\wmax{\sigma_i}$.
Thus, any finite-weight run $\rho$ on $w$ satisfies $|\weight(\rho)|\le \maxeff{w}\le \wmax{w}|w|$.

We write $p\runsto{w}q$ when there exists some run $\rho$ %with $\weight(\rho)<\infty$ 
such that $\rho:p\runsto{w}q$.
We lift this notation to concatenations of runs, e.g., $\rho: p\runsto{x}q\runsto{y}r$ means that $\rho$ is a run on $xy$ from $p$ to $r$ that reaches $q$ after the prefix $x$. We also incorporate this to $\minweight$ by writing e.g., $\minweight(q\runsto{x}p\runsto{y}r)$ to mean the minimal weight of a run $\rho:q\runsto{x}p\runsto{y}r$.

For a set of states $Q'\subseteq Q$ and a word $w\in \Sigma^*$, we define the {\em reachable set of states from $Q'$ upon reading $w$} as $\delta_{\bbB}(Q',w)=\{q\in Q\mid \exists p\in Q',\ p\runsto{w}q\}$ (where $\bbB$ stands for ``$\bbB$oolean'').

A WFA is \emph{trim} if every state is reachable from $q_0$ by some run. Note that states that do not satisfy this can be found in polynomial time (by simple graph search), and can be removed from the WFA without changing the weight of any accepted word. Throughout this paper, we assume that all WFAs are trim.

\paragraph*{Determinisability}
We say that WFAs $\cA$ and $\cB$ are \emph{equivalent} if $\weight_\cA\equiv \weight_\cB$. We say that $\cA$ is \emph{determinisable} if it is equivalent to some deterministic WFA. Our central object of study is the following problem, shows to be decidable in~\cite{almagor2026determinization}.
\begin{problem}[WFA determinisability]
\label{prob:determinisability}
    Given a WFA $\cA$ over a finite alphabet $\Sigma$, decide whether $\cA$ is determinisable.
\end{problem}

\paragraph*{Configurations}
A \emph{configuration} of $\cA$ is a vector $\vec{c}\in \bbZinf^Q$ which, intuitively, describes for each $q\in Q$ the weight $\vec{c}(q)$ of a minimal run to $q$ thus far (assuming some partial word has already been read). 
For a state $q$ we define the configuration $\vec{c_q}$ that assigns $0$ to $q$ and $\infty$ to $Q\setminus\{q_0\}$.
%that assigns $0$ to $q_0$ and $\infty$ to $Q\setminus\{q_0\}$. 

We adapt our notations to include a given starting configuration $\vec{c}$, as follows.
Given a configuration $\vec{c}$ and a word $w$, they induce a new configuration $\vec{c'}$ by assigning each state the minimal weight with which it is reachable via $w$ from $\vec{c}$. We denote this by 
$\xconf(\vec{c},w)(q)=\minweight_{\vec{c}}(Q\runsto{w}q)$ for every $q\in Q$.
In particular, $\xconf(\vec{c}_{\init},w)$ is the configuration that $\cA$ reaches by reading $w$ along a minimal run. For a state $q$ we denote $\xconf(q,w)=\xconf(\vec{c_q},w)$.

We use the natural component-wise partial order on configurations. That is, for two configuration $\vec{c},\vec{d}$, we say that $\vec{d}$ is \emph{superior} to $\vec{c}$, denoted $\vec{c}\le \vec{d}$, if $\vec{c}(q)\le \vec{d}(q)$ for every $q\in Q$. We also denote the \emph{support} of a configuration by $\supp(\vec{c})=\{q\in Q\mid \vec{c}(q)<\infty\}$.

For a run $\rho:p\runsto{w}q$ and configuration $\vec{c}$, we define the \emph{weight of $\rho$ from $\vec{c}$} to be $\weight_{\vec{c}}(\rho)=\vec{c}(p)+\weight(\rho)$. We similarly extend the notation $\minweight$ to include configurations, by writing
\[
\minweight_{\vec{c}}(Q_1\runsto{w} Q_2)=\min\{\weight_{\vec{c}}(\rho) \mid \exists q_1\in Q_1,q_2\in Q_2,\ \rho:q_1\runsto{w}q_2\}
\]
(we remark that we never need both the subscript $\cA$ and $\vec{c}$ together).
We say that $\rho$ is \emph{seamless from $\vec{c}$} if for every $1\le j\le n$ it holds that $\weight_{\vec{c}}(t_1,\ldots,t_j)=\minweight_{\vec{c}}(T\runsto{\sigma_1\cdots \sigma_j} q_{j+1})$ where $T=\supp(\vec{c})$.
%%%%%We use ``Seamless from \vec{c}'' in the monotone potential/charge lemma. 

\subsection{A Characterisation of determinisability}
\label{sec:determinisability equiv characterisation}
The following theorem gives an equivalent characterisation of determinisability by means of runs, and of how far two potentially-minimal runs can get away from one another, which we commonly refer to as a \emph{gap}. 
Intuitively, we show that $\cA$ is determinisable if and only if there is some bound $B\in \bbN$ such that if two runs on a word $x$ obtain weights that are far from each other, but the ``upper'' run can still become minimal over some suffix, 
%and the minimal run can become accepting on some suffix, 
then the runs are at most $B$ apart. This is depicted in \cref{fig:gap witness}.
\begin{figure}
    \centering
    \includegraphics[width=0.4\linewidth]{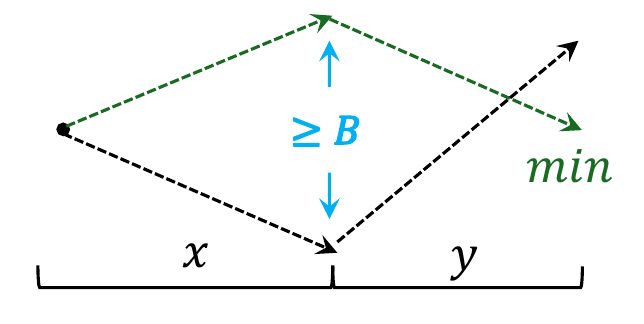}
    \caption{A $B$-gap witness: upon reading $x$, the minimal run on $x$ is at least $B$ below another run, but after reading $y$, the upper run becomes minimal.}
    \label{fig:gap witness}
\end{figure}

\begin{definition}[$B$-Gap Witness]
\label{def: B gap witness}
    For $B\in \bbN$, a \emph{$B$-gap witness over alphabet $\Sigma'$} consists of a pair of words $x,y\in \Sigma'^*$ and a state $q\in Q$ such that
    \begin{itemize}
        \item $\minweight(q_0\runsto{x\cdot y} Q)=\minweight(q_0\runsto{x}q\runsto{y}Q)<\infty$, but 
        \item $\minweight(q_0\runsto{x} q)-\minweight(q_0\runsto{x} Q)>B$.
    \end{itemize}
\end{definition}
\noindent Then, the characterisation is as follows.
\begin{theorem}[\cite{filiot2017delay,almagor2026determinization}]
    \label{thm:det iff bounded gap}
    Consider a WFA $\cA$, then the following hold.
    \begin{enumerate}
        \item If there exists $B\in \bbN$ such that there is no $B$-gap witness over $\Sigma$, then $\cA$ is determinisable.
        \item If there is a finite alphabet $\Sigma'\subseteq \Sigma$ such that for every $B\in \bbN$ there is a $B$-gap witness over $\Sigma'$, then $\cA$ is undeterminisable.
    \end{enumerate}
 \end{theorem}

Observe that for WFAs over a finite alphabet, \cref{thm:det iff bounded gap} provides an exact characterisation of determinisability.

\section{The Baseline-Augmented Subset Construction}
\label{sec:augmented construction}
We recall a construction from~\cite{almagor2026determinization} of a WFA $\augA$ that augments $\cA$ with additional information regarding its runs. This information is kept in both the states and in the alphabet.
Technically, the alphabet of $\augA$ consists of the \emph{transitions} of $\cA$, i.e. $\Delta$. Each state of $\augA$ contains information about the current state, the current reachable subset of states, as well as a special \emph{baseline} state. 
Intuitively, instead of only reading a word $w$, we read $w$ as well as a ``suggested'' run of $\cA$ on $w$. Then, we normalise all transitions so that this suggested run has constant weight $0$, and all other runs are shifted accordingly. Thus, we do not change the \emph{relative} structure of the run-tree. In particular, this does not affect determinisability. It does, however, offer great convenience in shifting the point of view to a specific run. We recall the construction and some of its important properties. See~\cite{almagor2026determinization} for further intuition.

Let $\cA= \tup{Q,\Sigma, q_0, \Delta}$. We define $\augA=\tup{\augStates,\Delta,\augInitState,\augTrans}$ with the following components:
\begin{itemize}
    \item The states are 
    $\augStates=\{(q,p,T) \in Q \times Q \times 2^Q \mid q,p \in T\}$.
    We refer to the components of a state $(q,p,T)\in S$ as the \emph{inner state} $q$, the \emph{baseline component} $p$, and the \emph{reachable subset} $T$.
    \item The alphabet of $\augA$ is the set of transitions $\Delta$ of $\cA$. 
    \item The initial state is $(q_0,q_0,\{q_0\})$.
    \item The transitions $\augTrans$ are defined as follows. 
    Consider states $(p_1,q_1,T_1)$ and $(p_2,q_2,T_2)$ in $\augStates$ and a letter $(q_1,\sigma,c,q_2)\in \Delta$ with $c\neq \infty$
    (i.e., the transition in $\cA$ from $q_1$ to $q_2$ on some $\sigma\in \Sigma$). 
    If $T_2=\booltrans(T_1,\sigma)$, then consider the unique transition $(p_1,\sigma,c',p_2)\in \Delta$. We add to $\augTrans$ the transition $((p_1,q_1,T_1),(q_1,\sigma,c,q_2),c'-c,(p_2,q_2,T_2))$.

    Finally, we complete any undefined transitions with $\infty$. 
\end{itemize}
We briefly outline the intuition behind $\augA$. Starting in $(s_0,s_0,\{s_0\})$, $\augA$ reads a word $x$ which is a run of $\cA$ on some word $w$. Note that $w$ can be inferred from the run. Then, the second component of $\augA$ simply tracks this run, whereas the second component acts identically to $\cA$ on $w$ (i.e., may nondeterministically choose any run on $w$), but the weights prescribed by $\cA$ are shifted by the prescribed run. Specifically, if both components track the same run, the weight is $0$. The third component tracks the reachable states. Note that the second and third components are deterministic.

The motivation behind $\augA$ is to simply allow us to change our ``point of view'' of the run tree of $\cA$ by fixing a certain run $\rho$ and looking at all other runs as if $\rho$ is constantly $0$. This is formalised and extended in \cref{sec: baseline shift}.

Of particular importance are states of the form $(q_1,q_1,T_1)$, which are reachable and can read $(q_1,\sigma,c,q_2)$. Then, $\augA$ may move to $(q_2,q_2,T_2)$ (as well as possibly to other states). Due to the weight normalisation, the weight of this transition in $\augA$ is $c-c=0$. Such states are central to our analysis.
\begin{definition}[Baseline States and Baseline Runs]
    \label{def:baseline states}
    A state in $S$ is a \emph{baseline} if it is of the form $(p,p,T)$ for some $p\in T$. 
    %States in $S_0$ are also baseline.
    A run $\rho=t_1\cdots t_m$ of $\augA$ with $t_i=(s_i,\tau_i,d,s_{i+1})$ is called a \emph{baseline} if $s_i$ is a baseline for all $i$.
\end{definition}

\begin{observation}
    \label{obs:baseline runs have weight 0}
    Recall that the weight of a transition in $\augA$ on letter $(p,\sigma,c,q)$ is normalised by reducing $c$ from any transition. In particular, a transition between baseline states $(p,p,T)\runsto{\sigma}(q,q,T)$ has weight $0$. Consequently, if $\rho$ is a baseline run, then $\weight(\rho)=0$. Note that there may still be other runs of both positive and negative weights.
\end{observation}
% The intuition behind the baseline run is that it forms a ``point of view'' of the various runs of $\cA$ on a given word, in the sense that the weights of the transitions in $\cA$ are shifted so that the baseline run has weight $0$, and all other runs are shifted accordingly. 

Baseline runs can be chosen arbitrarily. However, a natural choice is to take a seamless run as baseline. This is a natural assumption in many results in the paper. We say that a word $w$ \emph{has a seamless baseline run} if its baseline run is seamless.

This intuition behind the baseline augmented construction is formalised in~\cite{almagor2026determinization}, and is used to show the following.
\begin{lemma}[\cite{almagor2026determinization}]
    \label{lem:A det iff augA det}
    $\cA$ is determinisable if and only if $\augA$ is determinisable.
\end{lemma}

\section{Stable Cycles and Bounded behaviours}
\label{sec:stable cycles and bounded behaviours}
The (first) crux of the decidability proof of~\cite{almagor2026determinization} is the introduction of \emph{Stable Cycles}. Intuitively, these are a refined notion of more standard stabilisation operators, which allows a form of pumping. We recall the definitions and their basic properties, and tweak some definitions slightly to fit this work.

The basic structure that allows some form of pumping is that of \emph{stable cycles}. Intuitively a stable cycle is a pair $(S',w)$ where $S'$ is a set of states and $w$ is a word such that reading $w$ any number of times from the states in $S'$ leads back to $S'$, and the weighted behaviours of runs are ``stable'', in the sense that we can identify which states generate minimal runs, and which tend to unbounded weights, regardless of the starting weights. 

\subsection{Stable Cycles}
\label{sec:stable cycles}
The definition of Stable Cycle is made up of two components, as follows.
\begin{definition}[Reflexive Cycle]
    \label{def:reflexive cycle and state}
    Consider a set of states $S'\subseteq S$ and $w\in \Delta^*$. We say that $(S',w)$ is a \emph{reflexive cycle} if there exists $q\in Q$ and $T\subseteq Q$ such that $S'=\{(p,q,T)\mid p\in T\}$ and $\booltrans(S',w)\subseteq S'$.

    A state $s\in S'$ is a \emph{reflexive state} if $s\runsto{w}s$. We denote the set of reflexive states of $(S',w)$ as $\RefStates(S',w)=\{s\in S'\mid s\runsto{w}s\}$.

    $(S',w)$ is \emph{proper} if its baseline state is reflexive, i.e., $(q,q,T)\runsto{w}(q,q,T)$.
\end{definition}

That is, the states in $S'$ share the last two components, and $T$ is ``saturated'' in the sense that every state $p\in T$ represented in $S'$ as $(p,q,T)$. Reflexivity requires that reading $w$ from $S'$ leads back to $S'$, or a subset thereof.
Note that if $(S',w)$ is a reflexive cycle, then so is $(S',w^n)$ for every $n\in \bbN$ (this is immediate by the requirement $\booltrans(S',w)\subseteq S'$).
In the following, we always consider proper reflexive cycles.
Of particular importance are reflexive cycles whose cycles have minimal weight:
\begin{definition}
    \label{def:minimal reflexive states}
    Consider a reflexive cycle $(S',w)$. We say that a state $s\in \RefStates(S',w)$ is \emph{minimal reflexive} if $\minweight(s\runsto{w} s)=\min\{\minweight(r\runsto{w} r)\mid r\in \RefStates\}$.    
    We then define 
    \[\MinRefStates(S',w)=\{s\in S'\mid s\text{ is minimal reflexive.}\}\]
\end{definition}
That is, the cycle on $w$ from $s$ to $s$ has minimal weight among all reflexive cycles on $w$.
Recall that $S'=\{(p,q,T)\mid p\in T\}$  for some $T\subseteq Q$. In particular, $(q,q,T)\in S'$ is the unique baseline state in $S'$. 
We now ask whether $(q,q,T)$ is a minimal reflexive state, and whether its minimal reflexive run is a baseline (and therefore by \cref{obs:baseline runs have weight 0} has weight $0$). Moreover, we ask whether this holds for $w^n$ for every $n\in \bbN$. In case this holds, we say that $(S',w)$ is a \emph{Stable Cycle}. 
\begin{definition}[Stable Cycle]
\label{def:stable cycle}
A reflexive cycle $(S',w)$ with baseline state $s=(q,q,T)$ is a \emph{stable cycle} if for every $n\in \bbN$ it holds that $s\in \MinRefStates(S',w^n)$ and $\minweight(s\runsto{w^n}s)=0$.
\end{definition}

Note that if $(S',w)$ is a stable cycle, then the baseline run $\rho:s\runsto{w^n} s$ is a cycle of minimal weight on any state $s'\in S'$, and its weight is $0$.
There could be other runs from $S'$ on some $w^n$ that are negative, but not cycles.

\subsection{Stable Cycles and Bounded behaviours}
\label{sec:bounded states}
Having defined stable cycles, we now recall some results about the behaviour of runs in their presence. The proofs are all in~\cite{almagor2026determinization}.

A-priori, some states in a reflexive/stable cycle $(S',w)$ might not be reflexive, but become so for $w^n$ with $n$ large enough. The relevant constant is used extensively throughout the paper.
\begin{definition}[Stabilisation Constant]
    \label{def:bigM}
    Denote 
    %$\bigN= |S|!$ and 
    $\bigM=|S| \cdot |S|!$.
\end{definition}
\begin{proposition}
    \label{prop: transitions stabilise at M}
    For every word $w$, states $s,r\in S$ and $n\in \bbN$ we have $s\runsto{w^\bigM}r$ if and only if $s\runsto{w^{\bigM n}}r$.
\end{proposition}

% Recall that our motivation in this section is to establish a ``pumpable'' structure. Let $(S',w)$ be a stable cycle, and consider the long-term behaviour of runs on $w^m$ for large $m$, starting from states in $S'$.
Stable cycles encapsulate two types of runs on $w^n$ for large $n$: those that can go through a minimal cycle (and are therefore bounded), and those that do not. The following captures this intuition.

\begin{definition}[Grounded Pairs]
    \label{def:grounded pairs}
    For a stable cycle $(S',w)$ the \emph{grounded pairs} are
        \begin{align*}
        \GroundPairs(S',w)= & \{(s,r)\in S'\times S'\mid \exists g\in \MinRefStates(S',w^{\bigM})\\ \ST 
        & s\runsto{w^{\bigM}}g\runsto{w^{\bigM}}r\}
        \end{align*}
        For $(s,r)\in \GroundPairs(S',w)$, the state $g\in \MinRefStates(S',w^{\bigM})$ for which $\weight(s\runsto{w^{\bigM}}g\runsto{w^{\bigM}}r)$ is minimal (among all states in $\MinRefStates(S',w^\bigM)$) is the \emph{grounding state of $(s,r)$.}
\end{definition}

\begin{proposition}
    \label{prop:grounding states reachable with any bigM k}
    Consider a stable cycle $(S',w)$. For every $s,r,g\in S'$ with $g\in \MinRefStates(S',w^{\bigM})$, if $s\runsto{w^{\bigM i}}g\runsto{w^{\bigM j}}r$ for some $i,j\ge 1$, then $(s,r)\in \GroundPairs(S',w)$.
\end{proposition}
A key result is that with enough iterations of $w$ (namely $w^{2\bigM}$), the grounded states determine the asymptotic behaviour of runs, as follows.

\begin{lemma}[Pumping Grounded Pairs]
\label{lem:pumping grounded pairs}
Consider a stable cycle $(S',w)$. For every $n\in \bbN$ there exists $M_0\in \bbN$ (efficiently computable) such that for all $m\ge M_0$ the following hold.\footnote{Note that Item 2 does not depend on $n$.}  
\begin{enumerate}
    \item If $(s,r)\notin \GroundPairs(S',w)$ then $\minweight(s\runsto{w^{m\cdot 2\bigM},} r)>n$.
    
    \item If $(s,r)\in \GroundPairs(S',w)$ with grounding state $g$, then 
    \[
    \minweight(s\runsto{w^{m\cdot 2\bigM}} r)=\minweight(s\runsto{w^{M_0\cdot 2\bigM}}r)=\minweight(s\runsto{w^{\bigM}} g\runsto{w^{\bigM}}  r)
    \]
    Moreover, we have $\minweight(s\runsto{w^{m'\cdot 2\bigM}} r)\le \minweight(s\runsto{w^{\bigM}} g\runsto{w^{\bigM}}  r)$ for every $m'\in \bbN$.
\end{enumerate}
\end{lemma}

\section{The Cactus Extension}
\label{sec:cactus extension}
Intuitively, when reaching a stable cycle, we wish to mimic its long-term behaviour with a single letter. This is the motivation behind the \emph{cactus letters} introduced in~\cite{almagor2026determinization}. In the following we recall the definition. Note, however, that a crucial step in our effective procedure is to restrict the general cactus letters to a smaller fragment. This is done in \cref{sec:effective cactus}.

Consider a baseline-augmented subset construction $\augA=\tup{S,\Gamma,\augInitState,\augTrans}$. Note that we now denote the alphabet by $\Gamma$; we no longer assume $\Gamma$ is finite. 
More precisely, we do not assume the precise transitions defined in \cref{sec:augmented construction}, but the following:
\begin{itemize}
    \item The states are of the form $(q,p,T)$ where $q,p\in T$.
    \item The transitions are deterministic with respect to the second and third components (referred to as \emph{baseline component} and \emph{reachable set}, respectively).
\end{itemize}
We repeatedly extend the alphabet and transitions of $\augA$ to obtain a sequence of automata over the same state space with the properties above. In a nutshell, this is done by converting a stable cycle to a single letter, and then finding stable cycles on the new alphabet and repeating this ad-infinitum, thus obtaining the \emph{cactus letters}.

A crucial tool we need (which is at the heart of the baseline augmented construction) is the ability to \emph{shift} between baseline runs, to reason about runs under different baselines. This is easy in $\augA$, but is not natively supported by cactus letters. 
In order to maintain the ability to shift the baseline, we introduce some additional types of letters (namely \emph{rebase} and \emph{jump}).

\subsection{Cactus Letters}
\label{sec:cactus letters}
Cactus letters are defined as the infinite union of an inductively defined set. The induction steps are defined by adding new letters and transitions with the following \emph{stabilisation} operator.
%At the base of the induction, for every stable cycle $(S',w)$ in $\augA$ we introduce a new letter and transitions, as follows.
\begin{definition}[Stabilisation]
\label{def:stabilisation}
    For each stable cycle $(S',w)$ we introduce a new letter $\alpha_{S',w}$ with the following transitions: for every $(s,r)\in \GroundPairs(S',w)$ with grounding state $g$ we have the transition $(s,\alpha_{S',w},\minweight(s\runsto{w^\bigM}g\runsto{w^\bigM} r),r)$. We refer to this operation as \emph{stabilisation}, and denote the new letters by $\stab_{\augA}(\Gamma)$ and the new transitions by $\stab_{\augA}(\augTrans)$.
\end{definition}

We now apply stabilisation inductively, and take the infinite union, as follows.
\begin{definition}[stabilisation Closure]
\label{def:stab closure}
For every $k\in \bbNinf$ we define  
$\stab_k(\augA) = \augA_k =\tup{S,\Gamma_k,\augInitState,\augTrans_k}$ inductively as follows.
\begin{itemize}
    \item For $k=0$ we define $\Gamma_0=\Gamma$, $\augTrans_0=\augTrans$.
    \item For $k>0$ we define $\Gamma_k=\Gamma_{k-1}\cup \stab_{\augA_{k-1}}(\Gamma_{k-1})$ and $\augTrans_k=\augTrans_{k-1}\cup \stab_{\augA_{k-1}}(\augTrans_{k-1})$ 
    \item For $k=\infty$ we define $\Gamma_\infty=\bigcup_{k\in \bbN}\Gamma_k$ and $\augTrans_\infty=\bigcup_{k\in \bbN}\augTrans_k$.
\end{itemize}
\end{definition}
We refer to the letters in $\Gamma_\infty$ as \emph{cactus letters of $\Gamma$} and we denote 
$\augA_\infty = \stab_\infty(\augA) = \bigcup_{k\in \bbN}\augA_k$. Notice that while the sets of letters and transitions of $\augA_\infty$ are infinite, its set of states is $S$.  
\begin{example}[Cactus Letters]
    \label{xmp:cactus letters}
    Suppose $\Gamma=\{a,b\}$, we construct some cactus letters (See \cref{fig:cactus}). 
    For the sake of the example we do not have a concrete WFA, as we just illustrate the syntax of cactus letters. 
    Let ${\color{Cerulean}w_1=aaa}$ and ${\color{Cerulean}w_4=bbb}$, we define ${\color{Cerulean}\alpha_1=\alpha_{S_1,w_1}}$ and ${\color{Cerulean}\alpha_4=\alpha_{S_4,w_4}}$. 
    We then recursively define ${\color{red} w_2=ab{\color{Cerulean}\alpha_1 }ab}$ with ${\color{red}\alpha_2=\alpha_{S_2, w_2}}$. 
    We go up another level: ${\color{Green} w_3=baa{\color{Cerulean}\alpha_4}abba{\color{red}\alpha_2}ba}$ with ${\color{Green}\alpha_3=\alpha_{S_3, w_3}}$. We can consider another word outside, e.g., $ab{\color{Green}\alpha_3} aa$.
    Note that we could write the entire word as a single expression, but it is very cumbersome:
    \[ 
    ab \cdot {\color{Green}\alpha_{S_3,baa\cdot{\color{Cerulean}\alpha_{S_4,bbb}}\cdot abba\cdot{\color{red}\alpha_{S_2,ab\cdot{\color{Cerulean}\alpha_{S_1,aaa}}\cdot ab}}\cdot ba}}\cdot aa
    \]
\end{example}

\begin{figure}[H]
        \centering
        \includegraphics[width=0.95\linewidth]{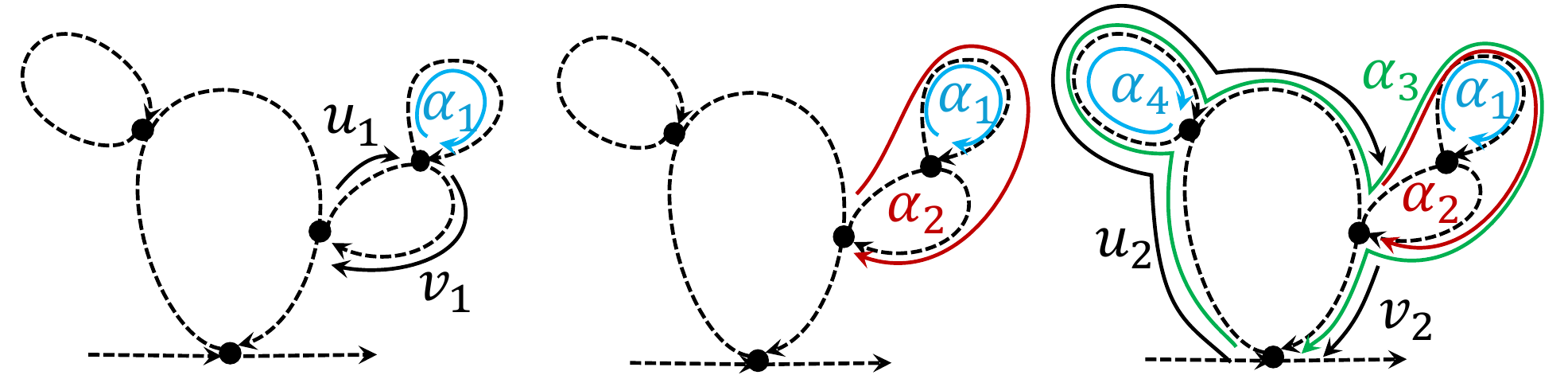}
        \caption{Cactus letters and cactus chains: $\alpha_1$ is of depth $1$ (and its inner word is of depth $0$). The word $u_1\alpha_1v_1$ (left) forms a cactus letter $\alpha_2$ of depth $2$ (center). Then, $u_2\alpha_2v_2$ forms a cactus letter $\alpha_3$ of depth $3$ (right). Notice that $u_2$ contains the cactus letter $\alpha_4$. The sequence $\alpha_3,\alpha_2 ,\alpha_1$ is a cactus chain, and so is $\alpha_3,\alpha_4$.}
        \label{fig:cactus}
    \end{figure}
    
\begin{remark}
\label{rmk:baseline transitions with cactus are zero}
Recall that in $\augA$, transitions between baseline states have weight 0. We remark that this is preserved in $\augA_\infty$. Indeed, a transition on $\alpha_{S',w}$ between baseline states $(p,p,T)$ and $(q,q,T')$ implies that $p=q$, and that $w$ prescribes a $0$-weight self loop on $p$, since $S'$ is a stable cycle and $p\in \MinRefStates(S',w^{\bigM})$ (see \cref{def:stable cycle}). 
\end{remark}

Consider a word $w\in \Gamma_\infty^*$, then $w$ is a concatenation of letters from various $\Gamma_k$'s. This gives a natural measure of \emph{depth} as the maximal $k$ used in the letters of $w$. This is illustrated in \cref{fig:cactus}. We remark that this definition  differs from that of~\cite{almagor2026determinization}, which becomes significant later on.
\begin{definition}
\label{def:depth of cactus letter}
For a word $w\in \Gamma_\infty$ we define its \emph{depth} $\depth(w)$ inductively on the structure of $\Gamma_\infty^*$ as follows:
\begin{itemize}
    \item If $w=a\in \Gamma_0=\Gamma$ then $\depth(w)=0$.
    \item If $w=\alpha_{S',x}\in \Gamma_\infty$ then $\depth(w)=1+\depth(x)$.
    \item If $w = \sigma_1 \cdots \sigma_m \in \Gamma_\infty^*$ for $m\ge 2$ then \\$\depth(w)=\max_{1\le i\le m}(\depth(\sigma_i))$.
\end{itemize}
\end{definition}

\subsection{Rebase Letters}
\label{sec:rebase letters}
In order to support baseline shifts (\cref{sec: baseline shift}), we need to introduce a new type of letters, called \emph{rebase letters}. Intuitively, these are cactus letters that also change the baseline component. Their introduction comes at a cost that becomes apparent in \cref{lem:flattening essence}.

\begin{definition}[Rebase]
\label{def:rebase}
    Let $\augA_\infty$ be the stabilisation closure of $\augA$.
    Consider a cactus letter $\alpha_{S',w}\in \Gamma_\infty$ and write $S'=\{(q,p,T)\mid q\in T\}$. 
    For each grounded pair $(s,r)\in \GroundPairs(S',w)$ where $s=(q_1,p,T)$ and $r=(q_2,p,T)$ with the transition $(s,\alpha_{S',w},c,r)\in \augTrans_\infty$, we add a letter $\beta_{S',w,s\to r}$ with the following transitions:
 \[ \begin{split}
    & \{((q',q_1,T),\beta_{S',w,s\to r},d-c,(q'',q_2,T))\mid \\&((q',p,T),\alpha_{S',w},d,(q'',p,T))\in \augTrans_\infty\}
 \end{split}
 \]
We refer to the resulting WFA as $\rebase(\augA_\infty)$. 
%, and denote the new letters by $\rebase_{\augA}(\Gamma)$ and the new transitions by $\rebase_{\augA}(\augTrans)$.
\end{definition}
\begin{figure}[H]
    \centering
    \includegraphics[width=0.85\linewidth]{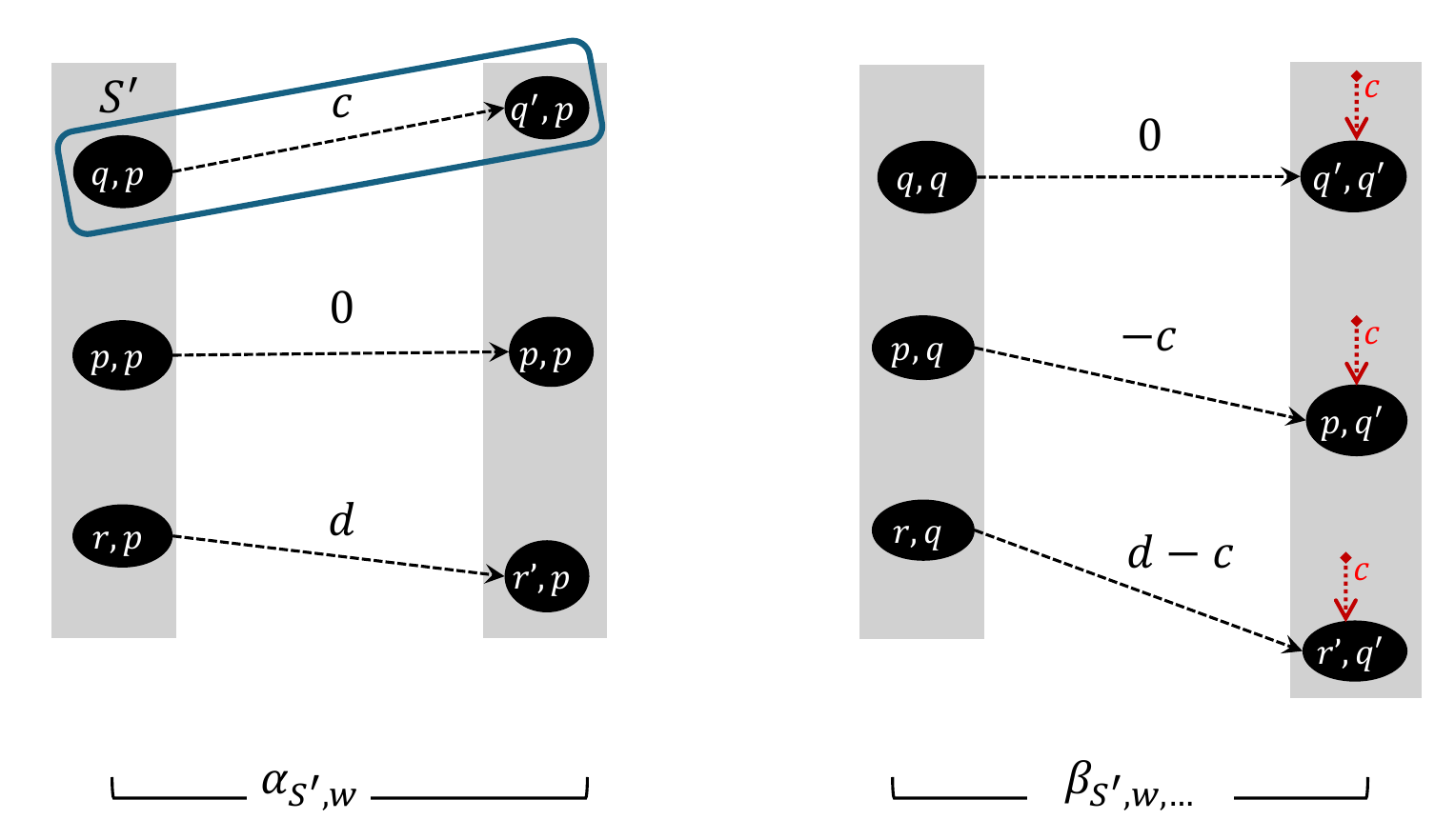}
    \caption{Rebase letter $\beta_{S',w,s\to r}$ with $s=(q,p,T)$ and $r=(q',p,T)$ induced by the $\alpha_{S',w}$ transition from $(q,p)$ to $(q',p)$. The weight of the resulting transitions, as well as the baseline run, shifts down by $c$, in accordance with the new base transition from $q$ to $q'$. }
    \label{fig:rebase}
\end{figure}
We illustrate the behaviour of a rebase letter in \cref{fig:rebase}.

\begin{remark}[Cactus as Rebase]
\label{rmk:cactus as rebase}
Consider a cactus letter $\alpha_{S',w}$ where $S'=\{(q,p,T)\mid q\in T\}$. We observe that for the baseline state $s_p=(p,p,T)$, the transitions defined on the letter $\beta_{S',w,s_p\to s_p}$ are identical to those on $\alpha_{S',w}$.

Indeed, $(s_p,\alpha_{S',w},0,s_p)\in \augTrans_\infty$, so plugging in $0$ for $c$ in \cref{def:rebase} gives identical weights to those of $\alpha_{S',w}$.
\end{remark}

We now alternate the stabilisation closure and rebase ad-infinitum. We remark that we only need one such alternation in practice, but we keep the definition to make it easier citing the relevant results from~\cite{almagor2026determinization}.
\begin{definition}[Cactus Extension]
\label{def:cactus extension}
Define $\augA^0_0=\augA$. 
For every $k,j\in \bbNinf$ we define 
$\augA^j_k=\tup{S,\Gamma^j_k,\augInitState,\augTrans^j_k}$ inductively as follows.
\begin{itemize}
    \item We define $\augA^0_k=\stab_k(\augA)$ as per \cref{def:stab closure}.
    \item For $j>0$ we define $\augA^j_0=\rebase(\augA^{j-1}_\infty)$.
    \item For $j>0$ and $k>0$ (including $k=\infty$) we define $\augA^j_k = \stab_k(\augA^j_0)$.
    \item For $j=\infty$ and every $k$ we define $\augA^\infty_k$ as the union (of alphabets and transitions) of $\augA^j_\infty$ for $j\in\bbN$. 
\end{itemize}
\end{definition}
\begin{remark}
\label{rmk:rebase not relevant in infinite level}
    Observe that by definition $\augA^\infty_k=\augA^\infty_\infty$ for all $k$. Intuitively, this reflects that fact that applying stabilisation on the infinite union does not add new letters (since stabilisation is applied to words of some finite depth anyway, and is therefore already in the union).
\end{remark}

We lift the definition of depth to rebase letters by defining $\depth(\beta_{S',w,s\to r})=1+\depth(w)$.

\subsection{Jump Letters}
\label{sec:jump letters}
Our final type of letters is a technical addition that is needed for \cref{lem:flattening essence}, which enables a central technique called \emph{flattening}.
These letters, referred to as \emph{jump letters}, allow us to change the baseline state without other effect. 
\begin{definition}[Jump Letters, $\jlset$]
    \label{def:jump letters}
    Consider a cactus extension $\augA^j_k$. For states $s=(q,p_1,T), r=(q,p_2,T)\in S$ and $T\subseteq Q$ (recall that $Q$ is the state space of the original WFA $\cA$) we introduce a \emph{jump letter} $\jl_{s\to r}$ with the following transitions. For states $t=(q',p_1,T)$ and $g=(q',p_2,T)$ (note the same \emph{first} and third component), we add the transition $(t,\jl_{s\to r},0,g)$.
    For states not of the form $(\cdot,p_1,T)$, the transitions have weight $\infty$.
    Note that the transitions on $\jl_{s\to r}$ are deterministic.
    We denote the set of Jump letters by $\jlset$.
\end{definition}

Note that since the baseline component is deterministic, then upon reading $\jl_{s\to r,T}$, all reachable states change their baseline. 
Also observe that jump can only change the baseline component to a state in $T$. However, as we discuss in the following (\cref{sec:reachable ghost states}), some states in $T$ may be unreachable. And indeed, $\jl$ letters may allow the baseline to be at an unreachable state, thus making the baseline run not seamless.

\begin{remark}
    \label{rmk:jump letters first component}
    Note that in \cref{def:jump letters}, the component $q$ in $s$ and $r$ is ignored (indeed, the same transitions are induced regardless of $q$). Thus, we sometimes denote the states $s,r$ of $\jl_{s\to r}$ as e.g., $s=(\cdot,p_1,T)$, where $\cdot$ stands for some arbitrary state.
\end{remark}

Finally, due to the ``violent'' nature of jump letters, we emphasise whenever they are present in a run, and in particular in a seamless run. We say that a run is \emph{jump free} if its transitions do not contain a jump letter.
%We do not annotate whether an alphabet contains jump letters or not, but we state so explicitly when relevant.

\subsection{Ghost Reachable States}
\label{sec:reachable ghost states}
Consider a word $w$. Recall that $\booltrans(s_0,w)$ is the set of states reachable from $s_0$ by reading $w$. 
In $\augA$, before the introduction of cactus letters, it holds that 
$\booltrans(s_0,w)=\{(p,q,T)\mid p\in T\}$ for some set $T$. That is, the reachable first components are exactly captured by the third component $T$.
After introducing cactus letters, however, this no longer holds. Indeed, \cref{def:cactus extension,def:rebase} limit the transitions to certain states (namely those obtained from grounded pairs). In particular, we have that
$\booltrans(s_0,w)=\{(p,q,T)\mid p\in T'\}$ for some $T'\subseteq T$, but equality does not necessarily hold.
Intuitively, the states in $T\setminus T'$ are 
those for which repetitions of cactus infixes lead to unbounded increase in weight, and are therefore abstracted to $\infty$ in the stabilisation.

Nonetheless, our analysis often requires us to reason about these states, which we term \emph{ghost states}. 
\begin{definition}[Ghost-Reachable States $\protect\mathghost$]
    \label{def:ghost states}
    Consider a state $s=(p_1,q_1,T_1)$ and word $w$ with $\booltrans(s,w)=\{(p_2,q_2,T'_1)\mid p_2\in T_2\}$ for some $T_2\subseteq T'_1$. 
    We define the set \emph{ghost-reachable states of $w$ from $s$} as
$\ghostTrans(s,w)=\{(p_2,q_2,T'_1)\mid p_2\in T'_1\}$.
\end{definition}
%We therefore define for  
\begin{remark}[Ghost States are implied by Reachable States]
\label{rmk:booltrans determines ghosttrans}
Note that $\ghostTrans(s,w)$ includes the ghost states as well as the standard reachable states. That is, $\booltrans(s,w)\subseteq \ghostTrans(s,w)$.
Moreover, observe that $\ghostTrans(s,w)$ is determined by $\booltrans(s,w)$, and therefore if $\booltrans(s,w)=\booltrans(s',w')$ then $\ghostTrans(s,w)=\ghostTrans(s',w')$. 
\end{remark}

\section{A Toolbox for Cacti}
\label{sec:cactus toolbox}
In~\cite{almagor2026determinization}, a set of tools is developed for using cactus, rebase and jump letters. The overall gist of the tools is to allow 
 ``folding'' parts of words into a cactus, ``unfolding'' other parts, possibly in a nested fashion, and shifting the baseline to a certain run as a ``point-of-view change''. These operations demonstrate the versatility of the cactus framework. 
 
In this section, we recall much of this toolbox (often succinctly, citing only the tools we require). More importantly, we start laying down the foundations for our results, building upon these tools and extending them.

\subsection{Degeneracy and A Bound on Cactus Depth}
\label{sec:degeneracy and cactus depth bound}
Consider words built up by repeatedly nesting cactus letters, as the following definition captures (see \cref{fig:cactus}).
\begin{definition}[Cactus Chain]
\label{def:cactus chain}
    A \emph{cactus chain} is a finite sequence of letters $\alpha_{S'_1,w_1},\ldots,\alpha_{S'_n,w_n}$ such that for every $1\le i<n$ we have $w_i=u_i \alpha_{S'_{i+1},w_{i+1}}v_i$ for some $u_i,v_i\in \Gamma_\infty^0$.
\end{definition}
Intuitively, a cactus chain singles out a ``branch'' within a cactus (see \cref{xmp:cactus letters,fig:cactus}). 
A-priori, cactus chains can be arbitrarily deep. However, a nontrivial result of~\cite{almagor2026determinization} is that deep-enough chains contain degenerate behaviours, and are therefore ``not interesting'', in a sense.
% Our goal in this section is to show that while this is the case, only a bounded number of ``layers'' can include interesting behaviours. In order to formalise this, we introduce the notions of \emph{degenerate} and \emph{non-degenerate} stable cycles, as follows.
\begin{definition}[(Non-)Degenerate Stable Cycle]
\label{def:degenerate stable cycle}
    A stable cycle $(S',w)$ is \emph{degenerate} if for every $s\in S'$, $t\in \RefStates(S',w^{2\bigM})$, if there is a run $s\runsto{w^{2\bigM k}}t$ for some $k\ge 1$, then $(s,t)\in \GroundPairs(S',w)$.

    Otherwise, $(S',w)$ is \emph{non-degenerate}, and a state $s\in S'$ such that there exists $t\in \RefStates(S',w^{2\bigM})$ and $k\ge 1$ for which $s\runsto{w^{2\bigM k}}t$ and $(s,t)\notin \GroundPairs(S',w)$ is called a \emph{non-degenerate state}.
    %we define degenerate stable cycle $(S',w)$ as stable cycle in which for every $s \in S', t \in \RefStates(S',w^{2 \bigM})$ if there is a run $s \runsto{w^{2\bigM i}} t$ (for some $i$) than $(s,t) \in \GroundPairs(S',w)$.
\end{definition}
Intuitively, the ``canonical'' degenerate cycle is one where all the reflexive states are also minimal-reflexive (i.e., have a $0$-cycle on $w^{2\bigM}$). The formal definition generalises this slightly.
The central result about degeneracy is the following, which serves as a basis for our effective alphabet in \cref{sec:effective cactus}.
\begin{lemma}
    \label{lem:deep cactus chain has degenerate cycle}
    Consider a cactus chain $\alpha_{S'_1,w_1},\ldots,\alpha_{S'_n,w_n}$. 
    If $n\ge |S|$ then\footnote{In~\cite{almagor2026determinization} this is proved for $n>|S|$. However, the proof actually works already for $n=|S|$.}
    %\shtodo{If reviewers ask, the reason it works for $n=|S|$ is that in this case, the last cycle has $n$ SCCs, so all states are minimal, and it is therefore degenerate.}
    there exists $1\le i\le n$ such that $\alpha_{S'_i,w_i}$ is a degenerate stable cycle.
\end{lemma}

\subsection{Baseline Shift}
\label{sec: baseline shift}
Perhaps the most versatile feature of the cactus framework is the ability to shift from one baseline to another. We bring the main tools from~\cite{almagor2026determinization}, and develop some new tools that make this even more flexible.

The intuition behind baseline shifts is simple: we consider some word/run/state with some baseline component, and we change the baseline. To account for this change, we re-normalise the weights according to the new baseline. 

We bring the precise definition, the proof that it is well-defined is in~\cite{almagor2026determinization}.
Consider a word $w=\tau_1\cdots\tau_k$ where for every $1\le i\le k$ we have that 
\[\tau_i\in \{(p_{i-1},\sigma_i,c_i,p_i),\alpha_{S'_i,w_i},\beta_{S'_i,w_i,s_i\to r_i}\}\]
i.e., each letter is either in $\Delta$, a cactus, or a rebase.
Let $\rho_0=t_1,\ldots,t_k$ be a run of $\augA_{\infty}^\infty$ on $w$, and denote 
$t_i=((q_{i-1},p_{i-1},T_{i-1}),\tau_i,d_i,(q_i,p_i,T_i))$ (where the $d_i$ are the weights).

We define the word $w'=\gamma_1\cdots \gamma_k$, where we split the definition of $\gamma_i$ to cases according to whether $\tau_i$ is a standard letter or a cactus/rebase.
\begin{itemize}
    \item If $\tau_i=(p_{i-1},\sigma_i,c_i,p_i)$, then $d_i=c'_i-c_i$ where $c'_i$ is the weight of the transition $(q_{i-1},\sigma_i,c'_i,q_i)\in \augTrans$ (by \cref{sec:augmented construction}).
    In this case, we define $\gamma_i=(q_{i-1},\sigma_i,c'_i,q_i)$.
    \item If $\tau_i=\alpha_{S'_i,w_i}$ or $\tau_i=\beta_{S'_i,w_i,r_i\to s_i}$, denote $S'_i=\{(q,p,T)\mid q\in T\}$ for some $p,T$. We define $\gamma_i=\beta_{S'_i,w_i,(q_{i-1},p,T)\to (q_i,p,T)}$. In case $q_{i-1}=q_i=p$, then we identify $\gamma_i=\alpha_{S'_i,w_i}$ as per \cref{rmk:cactus as rebase}.
\end{itemize}
Having defined $w'$, we now have a bijection between runs on $w$ and on $w'$.
\begin{proposition}
    \label{prop:baseline shift run bijection}
    In the notations above, for every sequence of states $f_1,\ldots,f_k\in Q$ (states of $\cA$), the following are equivalent.
    \begin{itemize}
        \item $\eta=x_1,\ldots,x_k$ with $x_i=((f_{i-1},p_{i-1},T_{i-1}),\tau_i,a_i,(f_i,p_i,T_i))$ is a run of $\augA_\infty^\infty$ on $w$.
        \item $\mu=y_1,\ldots,y_k$ with $y_i=((f_{i-1},q_{i-1},T_{i-1}),\gamma_i,a_i-d_i,(f_i,q_i,T_i))$ is a run of $\augA_\infty^\infty$ on $w'$ (recall that $d_i$ are the weights in $\rho_0$).
    \end{itemize}
\end{proposition}
We now wrap the definitions above in a concise baseline-shift operator, and overload it to runs, states, sets and configurations. Note that some of these extensions are new, and make shifting more convenient in the following proofs.

%We wrap up this section by introducing the baseline-shift operator, and providing a concise formalism for using it.
For a word $w$ and $\rho_0$ as above, we denote the constructed word $w'$ by 
$\baseshift{w}{\rho_0}$.
We overload this notation to runs: for a run $\eta=x_1,\ldots,x_k$ with $x_i=((f_{i-1},p_{i-1},T_{i-1}),\gamma_i,a_i,(f_i,p_i,T_i))$ we denote by $\baseshift{\eta}{\rho_0}$ the run $\mu=y_1,\ldots,y_k$ with $y_i=((f_{i-1},q_{i-1},T_{i-1}),\gamma_i, \\ a_i-d_i,(f_i,q_i,T_i))$. Note that a run already encapsulates the information about the word, and therefore we do not need to specify $w$ in the latter notation.

\cref{prop:baseline shift run bijection} readily implies the following (see~\cite{almagor2026determinization}):
%section, namely that the baseline shift of $\rho_0$ by itself is a seamless run, and that baseline shifts maintain the gaps between runs.
\begin{corollary}
    \label{cor:baseline shift to seamless run}
    Consider a word $w$ and a run $\rho_0$ on $w$, then $\baseshift{\rho_0}{\rho_0}$ is a seamless run (and in particular gains $0$ weight in every transition).
\end{corollary}
\begin{corollary}
    \label{cor:baseline shift maintains gaps}
    Consider a word $w$ and a run $\rho_0$ on $w$. Let $\rho_1,\rho_2$ be runs on $w$, and consider the shifted runs $\mu_1=\baseshift{\rho_1}{\rho_0}$ and $\mu_2=\baseshift{\rho_2}{\rho_0}$ then for every $0\le i\le |w|$ it holds that $\weight(\rho_1[1,i])-\weight(\rho_2[1,i])=\weight(\mu_1[1,i])-\weight(\mu_2[1,i])$.
\end{corollary}

We further notice that in a sequence of baseline shifts, only the last one matters:
\begin{remark}[Baseline Shift is Right-Absorbing]
\label{rmk:baseline shift is right absorbing}
Baseline shift is a right-absorbing operator in the following sense: 
\[\baseshift{\baseshift{w}{\rho_1}}{\rho_2}=\baseshift{w}{\rho_2}\]
(and similarly for runs). 
This follows immediately from the definition.  
\end{remark}

Finally, we lift the baseline shift notation also to states, sets and configurations as follows. Note that these are exactly the shifts that occur in the states in \cref{prop:baseline shift run bijection}. 
\begin{definition}
    \label{def: baseline shift on states sets and config}
    For states $s=(p,q,T)$ and $s'=(r,q,T)$ we define $\baseshift{s}{s'}=(p,r,T)$ (i.e., the first component of $s'$ indicates the new baseline component).\\ 
    For a set $R\subseteq S$ we lift this to $\baseshift{R}{s'}=\{\baseshift{r}{s'}\mid r\in R\}$.\\
    For a configuration $\vec{c}$ we lift this to $\vec{c'}=\baseshift{\vec{c}}{s'}$ where $\vec{c'}(q')=\vec{c}(q)$ where $q'=\baseshift{q}{s'}$.

    Lastly, we sometimes specify $s'$ as the first state in some run. That is, let $\rho:s'\runsto{w}s''$ for some $w,s''$, then we denote $\baseshift{\cdot}{\rho}=\baseshift{\cdot}{s'}$.
Finally, 
\end{definition}

\subsection{From Reflexive Cycles to Stable Cycles}
\label{sec:from reflexive to stable}
Recall from \cref{def:stable cycle} that the difference between a (proper) reflexive cycle and a stable cycle is that in the latter there are no negative cycles. 
Equipped with baseline shifts (\cref{sec: baseline shift}) we can now convert any reflexive cycle into a stable cycle by changing the baseline to the one with the ``steepest'' negative self-cycles. We formalise this intuition in this section. We remark that this is where our novel contribution begins in earnest. 

Consider a reflexive cycle $(S',w)$ with baseline run $\rhobase$, and let $\rho:r\runsto{w}r$ be a reflexive run (i.e., $r\in \RefStates(S',w)$). 
We define $\baseshift{(S',w)}{\rho}=(\baseshift{S'}{\rho},\baseshift{w}{\rho})$.  
Note that $\baseshift{(S',w)}{\rho}$ is then also a reflexive cycle, with baseline run $\baseshift{\rhobase}{\rho}$.

We now claim that by choosing the right $\rho$, we can actually induce a stable cycle (possibly by iterating $w$ several times).
We start with an auxiliary lemma.

For every $n\in \bbN$ and run $\rho:s\runsto{w^n}s$ for $s\in S$ we define the \emph{slope} of $\rho$ to be\footnote{Strictly speaking, the slope should be with respect to $\rho$ and $w$ and $n$ (e.g., if $\rho:s\runsto{(abab)^3}s$ it is not clear if $w=abab$ and $n=3$ or $w=ab$ and $n=6$), but we omit this and assume $w$ and $n$ are clear from $\rho$.} 
$\slope(\rho)=\frac{1}{n}\weight(\rho)$. We also define $\minslope(S',w)=\min \{\slope(\rho)\mid \rho:s\runsto{w^n}s,\ s\in S,n\in \bbN\}$. We show that this minimum is  attained by using at most $|S|$ repetitions of $w$. This is a fairly standard argument, but we give the proof for completeness.
\begin{proposition}
\label{prop: reflexive cycle negative slope short}
Consider a reflexive cycle $(S',w)$. There exists $k\le |S|$, a state $s\in S'$ and a run $\rho:s\runsto{w^k}s$ such that $\slope(\rho)=\minslope(S',w)$.
%for every $n\in \bbN$, $s'\in S$ and run $\pi:s'\runsto{w^n}s'$ it holds that $\slope(\rho)\le \slope(\pi)$.
\end{proposition}
\begin{proof}
    Consider a run $\pi:s'\runsto{w^n}s'$ with $n>|S|$.
    We show that there is a run $\rho:s\runsto{w^k}s$ for some $k\le |S|$ and $s\in S'$ with $\slope(\rho)\le \slope(\pi)$.

    Indeed, since $n>|S|$ by the pigeonhole principle we can write $\pi:s'\runsto{w^{k_1}}s\runsto{w^{k_2}}s\runsto{w^{k_3}}s'$ for some $s\in S'$ with $k_2\le |S|$ (note that $s\in S'$ since $(S',w)$ is a reflexive cycle).
    There are now two cases: if the infix $\pi[w^{k_2}]:s\runsto{w^{k_2}}s$ satisfies $\slope(\pi[w^{k_2}])\le \slope(\pi)$ then we are done.

    Otherwise, $\slope(\pi[w^{k_2}])> \slope(\pi)$. We claim that we can then cut out the $w^{k_2}$ infix to obtain a shorter run with lower slope. Indeed, denote $\pi[w^{k_1}w^{k_3}]:s'\runsto{w^{k_1}}s\runsto{w^{k_3}}s'$, then expanding $\slope(\pi[w^{k_2}])> \slope(\pi)$, we have
    \begin{align*}
    &\frac{\weight(\pi[w^{k_2}])}{k_2}>\frac{\weight(\pi)}{n}
    \;\Longrightarrow\;
    \weight(\pi[w^{k_2}])\cdot n > \weight(\pi)\cdot k_2
    && \\
    &\;\xRightarrow{-\weight(\pi)\cdot n}\;
    \weight(\pi[w^{k_2}])\cdot n - \weight(\pi)\cdot n
    >
    \weight(\pi)\cdot k_2 - \weight(\pi)\cdot n
    &&  \\
    &\;\Longrightarrow\;
    n\bigl(\weight(\pi[w^{k_2}])-\weight(\pi)\bigr)
    >
    \weight(\pi)\,(k_2-n)
    &&  \\
    &\;\xRightarrow{\text{multiply by } -1}\;
    n\bigl(\weight(\pi)-\weight(\pi[w^{k_2}])\bigr)
    <
    \weight(\pi)\,(n-k_2)
    &&  \\
    &\;\xRightarrow{\text{since } n=k_1+k_2+k_3}\;
    n\,\weight(\pi[w^{k_1}w^{k_3}])
    <
    \weight(\pi)\,(k_1+k_3)
    &&  \\
    &\;\Longrightarrow\;
    \frac{\weight(\pi[w^{k_1}w^{k_3}])}{k_1+k_3}
    <
    \frac{\weight(\pi)}{n}
    &&  \\
    &\;\xRightarrow{\text{definition of slope.}}\;
    \slope(\pi[w^{k_1}w^{k_3}])<\slope(\pi)
    && 
    \end{align*}
We can now repeat this argument until either finding an infix shorter than $|S|$ with a smaller slope, or by repeatedly cutting infixes until the remaining cycle is shorter than $|S|$.
\end{proof}

We turn to show that with a small number of repetitions of $w$, we can shift $(S',w)$ to become a stable cycle.
\begin{lemma}
    \label{lem: shift reflexive cycle to stable cycle}
    Consider a reflexive cycle $(S',w)$, and let $\rho:s\runsto{w^k}s$ with $\slope(\rho)=\minslope(S',w)$, $s\in S'$ and $k\le |S|$. Then $(S'',u)=\baseshift{(S',w^k)}{\rho}$ is a stable cycle. Moreover, if $\slope(\rho)\ge 0$ then $(S',w)$ is already a stable cycle.
\end{lemma}
\begin{proof}
    Denote $S'=\{(p,q,T) \mid p \in T\}$ and let $\rho:s\runsto{w^k}s$ be a minimal-slope run as per \cref{prop: reflexive cycle negative slope short}. In particular $k\le |S|$. 
    We start with the ``Moreover'' part: if $\slope(\rho)\ge 0$ then for every run $\rho':s'\runsto{w^n}s'$ with $s'\in S$ we have $\slope(\rho')\ge 0$, so $\weight(\rho')\ge 0$. Thus, $(S',w)$ is a stable cycle (note that in this case we actually have $\slope(\rho)=0$, due to the baseline run).

    We proceed to the general case. 
    Recall that $(S',w^k)$ is also a reflexive cycle (see \cref{def:reflexive cycle and state}), and denote its baseline run by with baseline run $\rhobase:(q,q,T)\runsto{w^k}(q,q,T)$. 
    Let $(S'',u)=\baseshift{(S',w^k)}{\rho}$ be the shifted reflexive cycle. We claim that $(S'',u)$ is a stable cycle. 
    To this end, we need to show that there are no negative cycles on $u^{m}$ for any $m\in \bbN$. 
    
    Assume by way of contradiction that $\tau:s\runsto{u^m}s$ satisfies $\weight(\tau)<0$ for some $s\in S''$. The proof now proceeds in two ``realms'': that of $(S'',u^m)$ (whose baseline is $\rho^m$, i.e., $\rho$ repeated $m$ times) and that of $(S',w^{km})$ (whose baseline is $\rhobase^m$). 
    We use baseline shifts to move between these realms.
    
    Denote $\tau'=\baseshift{\tau}{\rhobase^m}$. By \cref{rmk:baseline shift is right absorbing} we have that $\baseshift{u}{\rhobase}=\baseshift{w^k}{\rhobase}=w^k$. We therefore have that $\tau':s'\runsto{w^{km}}s'$ for some $s'\in S'$ (in the $(S',w^{km})$ realm). We now consider the two runs $\tau'$ and $\rho^m$ (both on $w^{km}$) and their shifted counterparts $\tau$ and $\baseshift{\rho^m}{\rho^m}$ (both on $u^m$) in the $(S'',u^m)$ realm.

    By \cref{cor:baseline shift maintains gaps}, we have
    \[
    \weight(\baseshift{\rho^m}{\rho^m})-\weight(\tau)=\weight(\rho^m)-\weight(\tau')
    \]
    By \cref{cor:baseline shift to seamless run} we have $\weight(\baseshift{\rho^m}{\rho^m})=0$. Recall our assumption that $\weight(\tau)<0$, we then have that the left-hand side of the equation is positive. Rearranging, we get $\weight(\tau')<\weight(\rho^m)$. Both $\tau'$ and $\rho^m$ are runs on $w^{km}$, and therefore we can divide by $km$ to obtain their slope, giving $\slope(\tau')<\slope(\rho^m)$.
    However, we have 
    \[\slope(\rho^m)=\frac{\weight(\rho^m)}{km}=\frac{m\cdot \weight(\rho)}{km}=\frac{\weight(\rho)}{k}=\slope(\rho)\]
    Thus, $\slope(\tau')<\slope(\rho)$, in contradiction to the choice of $\rho$ as the minimal-slope run.

    We conclude that there are no negative cycles on any $u^m$ from $S''$, and therefore $(S'',u)$ is a stable cycle, as required.
\end{proof}

\subsection{Cactus Flattening}
\label{sec: cactus unfolding}
Intuitively, cactus (and rebase) letters encapsulate many repetitions of the same word, where rebase also allows a change in the baseline. Naturally, we wish to consider the corresponding words that are obtained by repeatedly ``unfolding'' cactus and rebase letters, in a process we refer to as \emph{flattening}. 
We make extensive use of flattening throughout the paper.

Recall that for a word $u$ the value $\maxeff{u}$ is an upper-bound on the absolute value of the change in weight that a run can incur upon reading $u$.

\begin{definition}[Cactus-Unfolding]
    \label{def:unfolding function}
    Consider a cactus letter $\alpha_{S',w}$, a word $u=x \cdot \alpha_{S',w} \cdot y$ and $F\in \bbN$.

    By \cref{lem:pumping grounded pairs} there exists $M_0\in \bbN$ such that for every $k\ge M_0$ and every $s,r\in S'$ with $(s,r)\notin\GroundPairs(S',w)$ it holds that $\minweight(s\runsto{w^{k\cdot 2\bigM}}r)>F$.
    
    The \emph{Unfolding of $u$ with constant $F$} is then $\unfold(x,\alpha_{S',w},y \wr F)=xw^{2\bigM M_0}y$, and any word $xw^{2\bigM m}y$ with $m\ge M_0$ is \emph{an unfolding of $u$ with constant $F$.}
\end{definition}
\begin{remark}[``The'' v.s. ``an'' unfolding]
\label{rmk:increasing repetitions in unfolding}
We remark that the distinction between ``the'' unfolding and ``an'' unfolding merely means that we can add repetitions of $w$ (in multiples of $2\bigM$), and still be referred to as an unfolding.

In the following, we mostly refer to $M_0$ as the number of repetitions, but we sometimes need to increase this number, in which case we remark that all the results still hold.
\end{remark}

Intuitively, while $\alpha_{S',w}$ makes transitions $s\runsto{\alpha_{S',w}}r$ attain value $\infty$ if $(s,r)\notin \GroundPairs(S',w)$, the unfolding repeats $w$ enough times such that these runs may exist (and reach ghost states), but attain a value as high as we like (namely $F$). 
We always choose $F$ large enough such that any continuation of such runs is irrelevant for the value of the word. Therefore, the values of runs after unfolding are the same as those before unfolding (see~\cref{fig:unfolding}). We capture this as follows.

\begin{proposition}
    \label{prop:unfolding cactus maintains seamlesss gaps}
    Consider a cactus letter $\alpha_{S',w}$, a word $x \cdot \alpha_{S',w} \cdot y$ and $F> 2\maxeff{x \cdot \alpha_{S',w} \cdot y}$.
    
    Let $\rho$ be a seamless run on $x \cdot \alpha_{S',w} \cdot y$, then there exists a \emph{seamless} run $\mu$ on $\unfold(x,\alpha_{S',w},y \wr F)=xw^{2\bigM M_0}y$ such that 
    $\mu[1,|x|]=\rho[1,|x|]$, and for every $i\ge 0$, if $z=x\cdot\alpha_{S',w}\cdot y[1,i]$ is a prefix of $x \cdot \alpha_{S',w} \cdot y$ and the corresponding prefix of $\unfold(x,\alpha_{S',w},y \wr F)$ is $z'=x\cdot w^{2\bigM M_0}\cdot  y[1,i]$, then it holds that 
    $\weight(\rho(z))=\weight(\mu(z'))$.
\end{proposition}
 
      \begin{figure}[H]
        \centering
        \begin{subfigure}{0.45\textwidth}
        \includegraphics[width=1.0\linewidth]{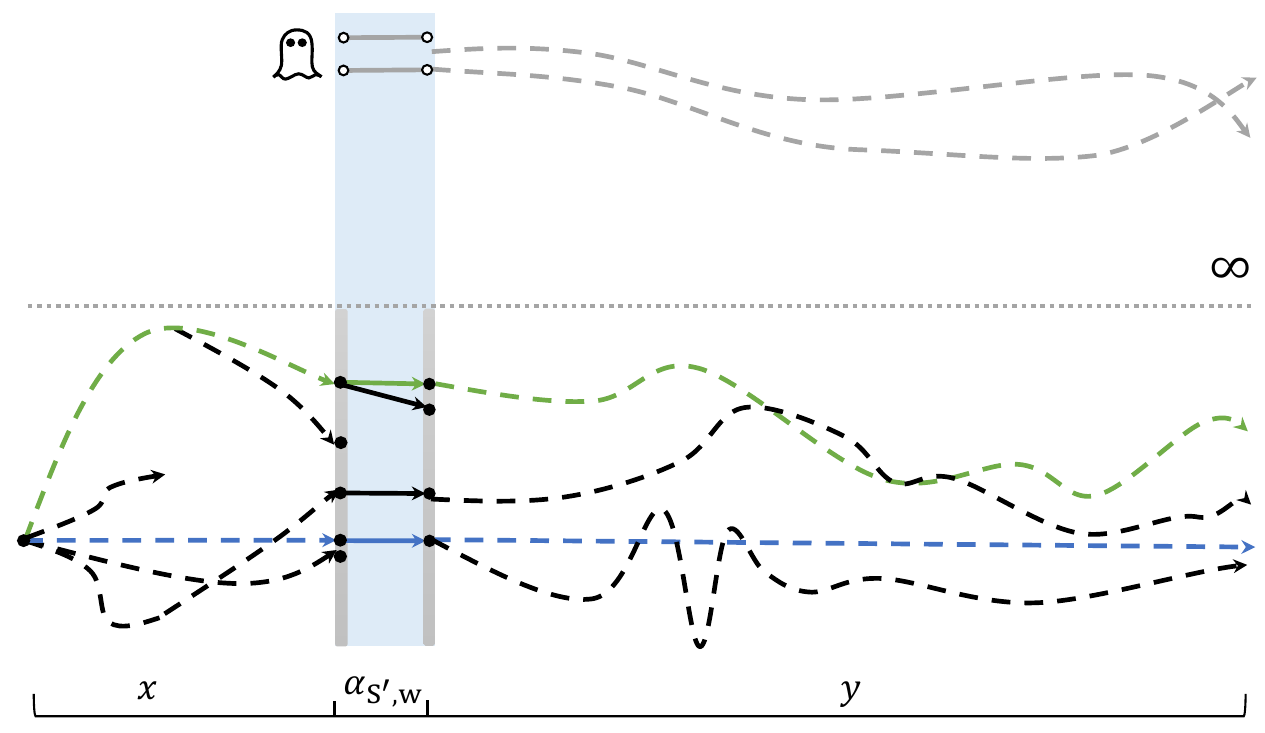}
        \caption{The cactus letter $\alpha_{S',w}$ induces ghost states and runs. The baseline run is in blue, and the green run is seamless.}
        \label{fig:unfolding before}
        \end{subfigure}
        \qquad
        \begin{subfigure}{0.45\textwidth}
        \includegraphics[width=1.0\linewidth]{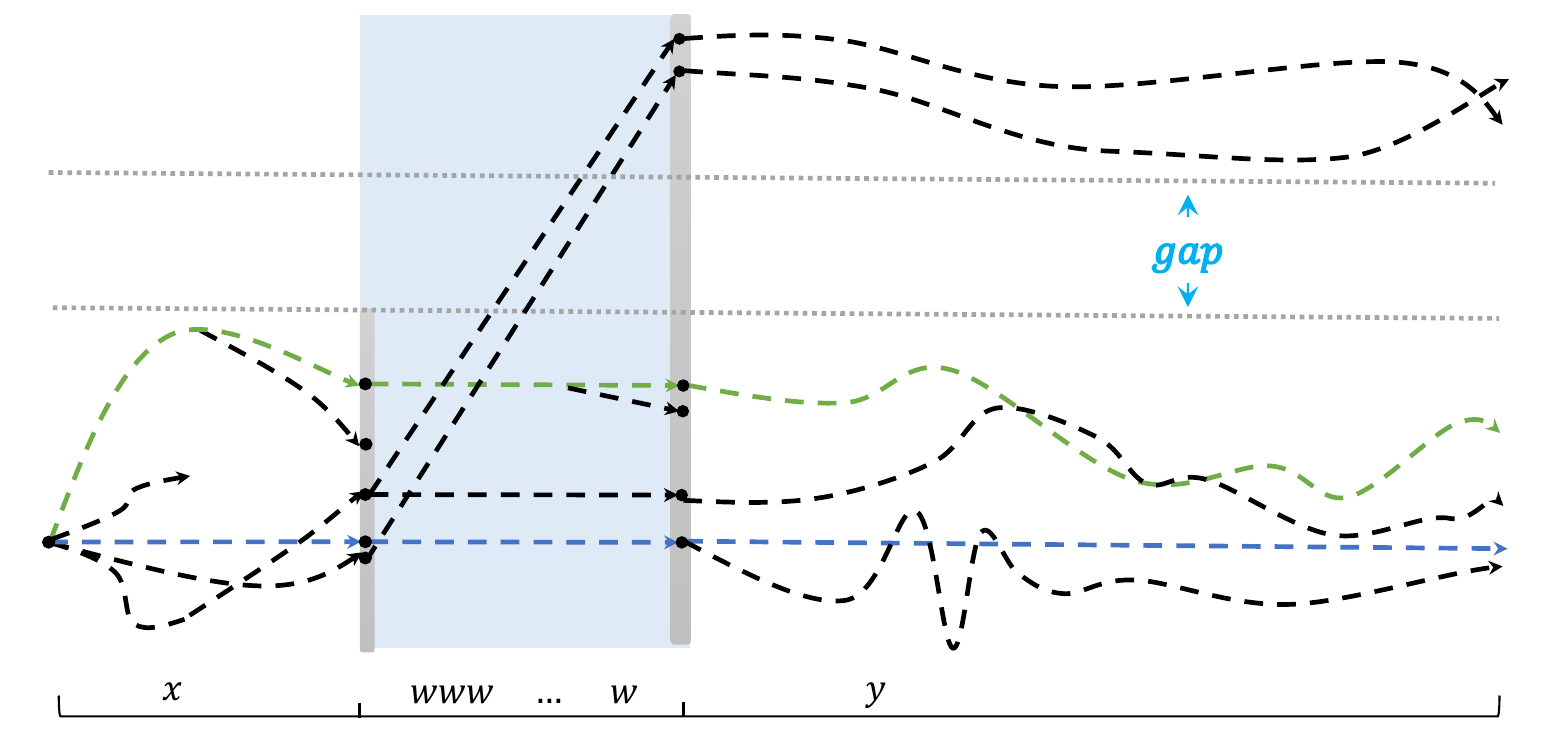}
        \caption{After unfolding, $\alpha_{S',w}$ is replaced by repetitions of $w$. The ghost runs now become concrete, and are far above the original runs. The updated green run is still seamless.}
        \label{fig:unfolding after}
        \end{subfigure}
        \caption{Cactus letter unfolding, before (\cref{fig:unfolding before}) and after (\cref{fig:unfolding after}).}
        \label{fig:unfolding}
    \end{figure}

The main result from~\cite{almagor2026determinization} concerning unfolding is the following.
\begin{lemma}[The Effect of Unfolding]
    \label{lem:unfolding configuration characterisation}
    Consider a cactus letter $\alpha_{S',w}$, a word $x\cdot \alpha_{S',w}\cdot y$ and $F>2\maxeff{x\cdot \alpha_{S',w}\cdot y}$, and let $\unfold(x,\alpha_{S',w},y \wr F)=xw^{2\bigM M_0}y$.
    Let $y'$ be a prefix of $y$, and consider the configurations $\vec{c_1}=\xconf(s_0,x\cdot \alpha_{S',w}\cdot y')$ and $\vec{c_2}=\xconf(s_0,xw^{2\bigM M_0}y')$, then $\supp(\vec{c_1})\subseteq \supp(\vec{c_2})$ and for every state $q\in S$ we have the following.
    \begin{enumerate}
        \item If $q\in \supp(\vec{c_1})$ then $\vec{c_2}(q)=\vec{c_1}(q)$.
        \item If $q\notin \supp(\vec{c_1})$ then $\vec{c_2}(q)>F-\maxeff{x\alpha_{S',w}y}$ (and could be $\infty$).
    \end{enumerate}
\end{lemma}

Unfolding cactus letters is fairly ``safe'': $\alpha_{S',w}$ is replaced with many iterations of $w$. 
However, we also want to unfold rebase letters. This turns out to be more involved. We refer to~\cite{almagor2026determinization} for the details, as we do not need them here. Having the ability to unfold, we can apply it recursively until there are no cactus or rebase letters left. We bring only the necessary essentials of this operation.
\begin{lemma}[Flattening]
    \label{lem:flattening essence}
    Consider a word $u\in (\Gamma_\infty^\infty)^*$ and $F> 2\maxeff{u}$. There exists a word denoted $\flatten(u\wr F)$ with the following properties:
    \begin{enumerate}
        \item If $u\in (\Gamma^0_\infty)^*$ then $\flatten(u\wr F)\in (\Gamma_0^0)^*$ (i.e., if $u$ has only cactus letters, the flattening has only $\Gamma_0^0$ letters).
        \item If $u$ has rebase letters, then $\flatten(u\wr F)\in (\Gamma_0^0\cup \jlset)^*$ (i.e., flattening rebase letters introduces jump letters).
    \end{enumerate}
    The following holds: let $\vec{c}=\xconf(s_0,u)$ and $\vec{d}=\xconf(s_0,\flatten(u \wr F))$.
    Then 
    \begin{align*}
    \booltrans(s_0,u)=\supp(\vec{c})\subseteq \supp(\vec{d})=\booltrans(s_0,\flatten(u \wr F))\\=\ghostTrans(s_0,\flatten(u \wr F))=\ghostTrans(s_0,u) 
      \end{align*} 
    
    and for every $q\in \supp(\vec{d})$, if $q\in \supp(\vec{c})$ then $\vec{c}(q)=\vec{d}(q)$, and otherwise $\vec{d}(q)\ge \max\{\vec{c}(p)\mid p\in \supp(\vec{c})\}+F$.
\end{lemma}

Flattening is used in~\cite{almagor2026determinization} to show the following.
\begin{corollary}
    \label{cor:aug inf inf is det iff A is det}
     $\augA_\infty^\infty$ is determinisable if and only if $\cA$ is determinisable.
\end{corollary}

\section{Effective Cactus Letters}
\label{sec:effective cactus}
The cactus letters defined in \cref{sec:cactus letters} are built of words of arbitrary length. In particular, this leads to an infinite alphabet which is the source of much technical difficulty, and is a central reason that the decidability result of \cite{almagor2026determinization} does not provide a complexity bound. 
We now take the first step towards circumventing this problem, by restricting the cactus alphabet to a finite set.

Intuitively, we are given a function $L:\{0,\ldots,|S|\}\to \bbN$ that prescribes for every depth $d$ up to $|S|$, the maximal length $L(d)$ of a word $w$ such that $w$ can appear in a cactus letter of depth $d$. We also restrict cactus letter to be non-degenerate, thus bounding the depth of cactus nesting. 

We then proceed to instantiate this type of hierarchy with two different $L$ functions, which we refer to as \emph{simple} ($\simpL$) and \emph{general} $\generL$, that are used in distinct places in the proof.

\begin{definition}[$L$-Bounded Cactus Letters]
\label{def:L bounded cactus letters}
Consider a function $L:\{0,\ldots,|S|\}\to \bbN$. For every $k\in \{0,\ldots, |S|\}$ we define $\boundedCacti{L,k}$ inductively as follows.
\begin{itemize}
    \item $\boundedCacti{L,0}=\Gamma^0_0$ (i.e., standard letters from $\augA$).
    \item 
    \begin{align*}
    \boundedCacti{L,k}= &\{\alpha_{S',w}\mid \depth(\alpha_{S',w})=k,\ |w|\le L(k), \text{ and } \alpha_{S',w} \\&\text{ is non-degenerate}\}.
    \end{align*}
\end{itemize}
We then define four classes of letters
\begin{itemize}
    \item $\boundedCac{L}=\bigcup_{k=0}^{|S|-1}\boundedCacti{L,k}$ are the \emph{bounded cactus letters}.
    \item $\boundedCacReb{L}=\rebase(\boundedCacti{L})$ adds rebase letters to $\boundedCac{L}$.
    \item $\boundedCacRebCac{L}=\{\alpha_{S',w}\mid w\in \boundedCacReb{L}^*\wedge |w|\le L(\depth(w)+1)\}$. 
    These are bounded cactus letters built from a single stabilisation of a letter in $\boundedCacReb{L}$.
    \item $\boundedCacRebJump{L}=\boundedCacReb{L}\cup \jlset$ adds jump letters to $\boundedCacReb{L}$.
\end{itemize}
\end{definition}
An immediate observation is that all classes above are finite, which is a crucial simplification of the set of all cactus letters. Note that the $1$ and $0$ sub/superscripts correspond to those of \cref{def:cactus extension}, e.g., $\boundedCacRebCac{L}\subseteq \Gamma^1_1$, as these letters are composed of a depth-$1$ stabilisation over rebase letters, which are in turn over cactus letters without further rebase.
The specific roles of the classes is explained as they are used. Intuitively, the main objective is to keep the alphabets finite, and as restricted as possible. 
Also note that by \cref{lem:deep cactus chain has degenerate cycle}, considering letters up to depth $|S|-1$ is exhaustive, since any deeper words contain a degenerate stable cycle. 

\begin{remark}[Depth v.s. Length Bound]
\label{rmk:depth vs length bound}
Note that according to \cref{def:L bounded cactus letters}, for word $w\in (\boundedCac{L})^*$ with $\depth(w)=d$, each cactus letter $\alpha_{S',x}$ in $w$ of depth $k\le d$ satisfies $|x|\le L(k)$, but $\depth(x)=k-1$. This may lead to some confusion, since when we construct words that are turned into cactus letters, we maintain $|x|\le L(\depth(x)+1)$. We recommend keeping this in mind throughout the proof.    
\end{remark}

We now turn to define two concrete instantiations of such $L$ functions. In broad terms, the \emph{simple} bound function pertains to analysis of the potential $\pot$ (e.g., \cref{sec:SSRI toolbox}), whereas the \emph{general} function is for the charge $\charge$ (e.g., \cref{sec:GSRI toolbox}). For the definitions of potential and charge, see \cref{sec:dominance and potential}. For now, we recommend viewing them as fairly arbitrary functions, and their definitions are demystified in the proof.

These concrete instances depend inductively on certain parameters. We explain the intuitive meaning behind them, although this is not necessary in order to follow the definition. When we reason about these functions, we elaborate on the role of each parameter.
\begin{itemize}
    \item The parameter $d$ represents the depth of a word.
    \item The parameter $i$ represents the number of \emph{independent runs} over a certain word. Intuitively, these are runs whose weights are very far from each other (see \cref{def:independent run}). 
\end{itemize}
The inductive definitions decrease with $d$ and increase with $i$.

In \cite{almagor2026determinization}, a fundamental (non-constructive) step in the proof is the usage of Ramsey's theorem on colouring of an infinite sequence. In this constructive version, we replace this with finitary Ramsey argument, as follows. 
For a set $S$, denote by ${S \choose 2}=\{\{i,j\}\mid i\neq j\in S\}$ the set of size $2$ subsets of $S$.
\begin{definition}[The Ramsey Function]
    \label{def:ramsey function}
    The function $\ramsey:\bbN\times \bbN\to \bbN$ gives for every $k,n\in \bbN$ the number $R=\ramsey(k,n)$ such that for every colouring of $\{1,\ldots, R\}\choose 2$ with $k$ colours there exists a set $S\subseteq \{1,\ldots, R\}$ with $|S|\ge n$ such that the colouring restricted to $S\choose 2$ is constant.
    A folklore upper bound~\cite{erdos1935combinatorial} is $\ramsey(k,n)\le k^{kn}$.
\end{definition}

\subsection{Simple Length Bound $\simpL$}
\label{sec: simple length bound function}
The Simple Length Bound function $\simpL$ is defined inductively by means of several functions, which are used in their own right in later sections. We start by defining the functions, and then illustrate their inductive dependency to show there is no cyclic definition in \cref{fig:simple length inductive dependency graph}. 
In the following we always have $d,i\in\{0,\ldots, |S|+1\}$ and $i\ge 1$ (where the $|S|+1$ case only serves as a base for the induction). Note that some functions do not depend directly on themselves, and are therefore not given a base.
\begin{definition}[$\simpLfuncLength{d}{i}$]
    \label{def: simpL func Length}
    \textbf{Intuition:} The number $\simpLfuncLength{d}{i}$ is the maximal length we allow for a word $w$ in a cactus letter $\alpha_{S',w}$ of depth $d$, assuming there are at least $i$ independent runs on $w$.
    \emph{
    \begin{itemize}
        \item \textbf{Base:} $\simpLfuncLength{0}{i}=1$ for every $i$ and $\simpLfuncLength{d}{|S|+1}=0$ for every $d$.
        \item \textbf{Step:} $\simpLfuncLength{d+1}{i}=\extraSize\cdot\simpLfuncLength{d}{1}\cdot\simpLfuncLength{d+1}{i+1}\cdot (\ramsey(\simpLfuncTypes{d+1}{i},3)+2)^{2\simpLfuncAmp{d+1}{i}+1}$.
    \end{itemize}
    }
\end{definition}
\begin{definition}[$\simpLfuncCover{d}{i}$]
    \label{def: simpL func Cover}
    %USED to be LH
    \textbf{Intuition:} The number $\simpLfuncCover{d}{i}$ represents the maximal distance we allow between a run and its closest independent run, i.e., the ``cover'' of an independent run.
    \emph{
    \begin{itemize}
        \item \textbf{Base:} $\simpLfuncCover{d}{|S|}=1$ for all $d$ (note that the induction only increases $i$).
        \item \textbf{Step:} $\simpLfuncCover{d+1}{i}=8\bigM^2(\simpLfuncMaxW{d}\simpLfuncLength{d}{1}\simpLfuncLength{d+1}{i+1}+\simpLfuncCover{d+1}{i+1})$.
    \end{itemize}
    }
\end{definition}
\begin{definition}[$\simpLfuncMaxW{d}$]
    \label{def: simpL func max weight}
    \textbf{Intuition:} The number $\simpLfuncMaxW{d}$ represents the maximal weight of a simple cactus letter of depth $d$.
    \emph{
    \begin{itemize}
        \item \textbf{Base:} $\simpLfuncMaxW{0}=\max_{\gamma\in \Gamma_0}\wmax{\gamma}$ (maximal weight of a transition in $\augA$).
        \item \textbf{Step:} $\simpLfuncMaxW{d+1}=2\bigM \simpLfuncMaxW{d}\simpLfuncLength{d+1}{1}$.
    \end{itemize}
    }
\end{definition}
\begin{definition}[$\simpLfuncAmp{d}{i}$]
    \label{def: simpL func Amp}
    %Used to be Up
    \textbf{Intuition:} The number $\simpLfuncAmp{d}{i}$ is a bound on the \emph{amplitude} -- an increase in the potential, to be defined in \cref{sec:dominance and potential}. 
    \emph{
    \begin{align*}
    \simpLfuncAmp{d+1}{i} =& \extraSize(2\simpLfuncMaxW{d})(\ramsey(\simpLfuncTypes{d+1}{i},3)\\ &+\simpLfuncLength{d+1}{i+1}+\simpLfuncLength{d}{1})
    \end{align*}}
\end{definition}
\begin{definition}[$\simpLfuncTypes{d}{i}$]
    \label{def: simpL func Types}
    \textbf{Intuition:} The number $\simpLfuncTypes{d}{i}$ represents the number of possible configurations of the states, where we keep track only of the weights within a window of $\simpLfuncCover{d}{i}$ from each independent run.
\emph{
\[\simpLfuncTypes{d+1}{i}=3^i\cdot (2\simpLfuncCover{d+1}{i}+2)^{2|S|i}\]
}
\end{definition}

With these functions in place, we can define the overall length-bound function.
\begin{definition}[$\simpL$]
\label{def: simpL length bound}
    The function $\simpL$ is defined as \emph{$\simpL(d)=\simpLfuncLength{d}{1}$}.
\end{definition}

We can then instantiate the bounded cactus letters of \cref{def:L bounded cactus letters} with the function $\simpL$ to obtain two concrete classes of bounded cactus letters.
\begin{definition}[$\simpL$-Bounded Cactus Letters]
    \label{def: simpL bounded cactus letters}
    Denote $\simpCac=\boundedCac{\simpL}$, $\simpCacReb=\boundedCacReb{\simpL}$, $\simpCacRebCac=\boundedCacRebCac{\simpL}$ and $\simpCacRebJump=\boundedCacRebJump{\simpL}$. 
\end{definition}

\begin{figure}[H]
    \centering

  \begin{tikzpicture}[
  >=latex,
  node distance=12mm and 18mm,
  every node/.style={draw, rounded corners, align=center, font=\small},
  dep/.style={->, thick}
]

% =======================
% Left column
% =======================
\node (Len_d0) {\texttt{Len}$(d,1)$};

% =======================
% Top layer: depth d+1
% =======================
\node (Len_di)   [right=of Len_d0] {\texttt{Len}$(d\!+\!1,i)$};
\node (Len_dip1) [right=of Len_di] {\texttt{Len}$(d\!+\!1,i\!+\!1)$};

\node (Typ_di)   [below=of Len_di] {\texttt{Typ}$(d\!+\!1,i)$};
\node (Amp_di)   [below=of Len_dip1] {\texttt{Amp}$(d\!+\!1,i)$};

\node (Cov_di)   [below=of Typ_di] {\texttt{Cov}$(d\!+\!1,i)$};
\node (Cov_dip1) [below=of Cov_di] {\texttt{Cov}$(d\!+\!1,i\!+\!1)$};

% =======================
% Right column: MaxWt
% =======================
\node (Max_d)   [below=of Amp_di] {\texttt{MaxWt}$(d)$};
\node (Max_dm1) [below=of Max_d] {\texttt{MaxWt}$(d\!-\!1)$};

% =====================================================
% Edges: Len(d+1,i)
\draw[dep] (Len_di) -- (Typ_di);
\draw[dep] (Len_di) -- (Amp_di);
\draw[dep] (Len_di) -- (Len_dip1);
\draw[dep] (Len_di) -- (Len_d0);

% Edges: Amp(d+1,i)
\draw[dep] (Amp_di) -- (Typ_di);
\draw[dep] (Amp_di) -- (Len_dip1);
\draw[dep] (Amp_di) -- (Max_d);
\draw[dep] (Amp_di) -- (Len_d0);

% Edges: Typ(d+1,i)
\draw[dep] (Typ_di) -- (Cov_di);

% Edges: Cov(d+1,i)
\draw[dep] (Cov_di) -- (Cov_dip1);
\draw[dep] (Cov_di) to[bend left=12] (Len_dip1);
\draw[dep] (Cov_di) -- (Max_d);
\draw[dep] (Cov_di) -- (Len_d0);

% Edges: MaxWt
\draw[dep] (Max_d) -- (Max_dm1);

\end{tikzpicture}
    
    \caption{The inductive dependency graph for \cref{sec: simple length bound function}, demonstrating that the dependency is indeed acyclic.
    }
    \label{fig:simple length inductive dependency graph}
\end{figure}
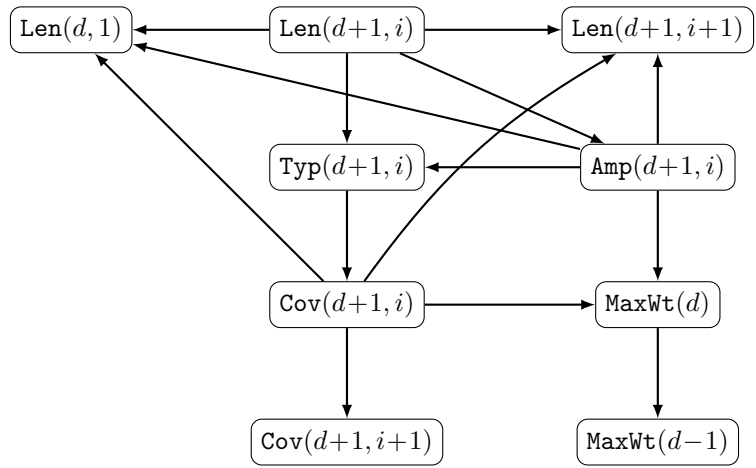

\subsection{General Length Bound $\generL$}
\label{sec: general length bound function}
We now present a similar set of functions as \cref{sec: simple length bound function}, but that are based on a constant $\budH$ defined using the ``simple'' functions of \cref{sec: simple length bound function}. 
Technically, the need for these function arises when considering the charge ($\charge$) instead of the potential ($\pot$). We remark that we could define both sets of functions uniformly with an additional parameter, but choose to separate them for clarity, and since $\budH$ is not actually a parameter, but a constant.
Specifically, denote $\budH=\simpLfuncAmp{|S|}{0}$.
\begin{definition}[$\generLfuncLength{d}{i}$]
    \label{def: genL func Length}
    \textbf{Intuition:} The number $\generLfuncLength{d}{i}$ is the maximal length we allow for a word $w$ in a cactus letter $\alpha_{S',w}$ of depth $d$, assuming there are at least $i$ independent runs on $w$.
    \emph{
    \begin{itemize}
        \item \textbf{Base:} $\generLfuncLength{0}{i}=1$ for every $i$ and $\generLfuncLength{d}{|S|+1}=0$ for every $d$.
        \item \textbf{Step:} $\generLfuncLength{d+1}{i}=\extraSize\cdot\generLfuncLength{d}{1}\cdot\generLfuncLength{d+1}{i+1}\cdot (\ramsey(\generLfuncTypes{d+1}{i},3)+2)^{2\generLfuncAmp{d+1}{i}+1}$.
    \end{itemize}
    }
\end{definition}
\begin{definition}[$\generLfuncCover{d}{i}$]
    \label{def: genL func Cover}
    %USED to be LH
    \textbf{Intuition:} The number $\generLfuncCover{d}{i}$ represents the maximal distance we allow between a run and its closest independent run, i.e., the ``cover'' of an independent run.
    \emph{
    \begin{itemize}
        \item \textbf{Base:} $\generLfuncCover{d}{|S|}=1$ for all $d$ (note that the induction only increases $i$).
        \item \textbf{Step:} $\generLfuncCover{d+1}{i}=8\bigM^2(\generLfuncMaxW{d}\generLfuncLength{d}{1}\generLfuncLength{d+1}{i+1}+\generLfuncCover{d+1}{i+1} + \budH)$.
    \end{itemize}
    }
\end{definition}
\begin{definition}[$\generLfuncMaxW{d}$]
    \label{def: genL func max weight}
    \textbf{Intuition:} The number $\generLfuncMaxW{d}$ represents the maximal weight of a general cactus letter of depth $d$.
    \emph{
    \begin{itemize}
        \item \textbf{Base:} $\generLfuncMaxW{0}=\max_{\gamma\in \Gamma_0}\wmax{\gamma}$ (maximal weight of a transition in $\augA$).
        \item \textbf{Step:} $\generLfuncMaxW{d+1}=2\bigM \generLfuncMaxW{d}\generLfuncLength{d+1}{0}$.
    \end{itemize}
    }
\end{definition}
\begin{definition}[$\generLfuncAmp{d}{i}$]
    \label{def: genL func Amp}
    \textbf{Intuition:} The number $\generLfuncAmp{d}{i}$ is a lower bound on the \emph{amplitude} -- an increase in the charge, to be defined in \cref{sec:dominance and potential}. 
    \emph{
           \begin{align*}
           \generLfuncAmp{d+1}{i}= &\extraSize(\budH+ 2\generLfuncMaxW{d})( \ramsey(\generLfuncTypes{d+1}{i},3)\\&+\generLfuncLength{d+1}{i+1}+\generLfuncLength{d}{1})
           \end{align*}
    }
\end{definition}
\begin{definition}[$\generLfuncTypes{d}{i}$]
    \label{def: genL func Types}
    \textbf{Intuition:} The number $\generLfuncTypes{d}{i}$ represents the number of possible configurations of the states, where we keep track only of the weights within a window of $\generLfuncCover{d}{i}$ from each independent run.
    \emph{
    \[\generLfuncTypes{d+1}{i}=3^i\cdot (2\generLfuncCover{d+1}{i}+2)^{2|S|i}\]
    }
\end{definition}
\begin{definition}[$\generL$]
The function $\generL$ is defined as $\generL(d)=\generLfuncLength{d}{1}$.
\end{definition}

\begin{definition}[$\generL$-Bounded Cactus Letters]
    \label{def: generL bounded cactus letters}
    Denote $\genCac=\boundedCac{\generL}$, $\genCacReb=\boundedCacReb{\generL}$, $\genCacRebCac=\boundedCacRebCac{\generL}$ and $\genCacRebJump=\boundedCacRebJump{\generL}$. 
\end{definition}

\subsection{Properties of Effective Cactus Letters}
\label{sec:properties of effective cactus}
We establish some basic properties of the functions defined above.
The first is that $\simpLfuncMaxW{d}$ and $\generLfuncMaxW{d}$ are indeed bounds on the effect of a letter.
\begin{proposition}
    \label{prop:simp maxw bounds maxeff}
    \label{prop:gen maxw bounds maxeff}
    For every $\sigma\in \simpCac$ (resp. $\sigma\in \genCac$) with $\depth(\sigma)=d$ we have $\maxeff{\sigma}\le \simpLfuncMaxW{d}$ (resp. $\maxeff{\sigma}\le \generLfuncMaxW{d}$).
\end{proposition}
\begin{proof}
We prove the case of $\simpCac$, the proof of $\genCac$ is identical by inserting $\mathtt{G}$ before the relevant functions.
    By \cref{def:L bounded cactus letters} we have $\sigma\in \boundedCacti{\simpL,d}$. 
    We prove the claim by induction on $d$.
    For $d=0$ we have $\sigma\in \Gamma_0$ (by \cref{def:L bounded cactus letters}). Then by \cref{def: simpL func max weight} we have $\simpLfuncMaxW{0}=\max_{\gamma\in \Gamma_0}\wmax{\gamma}\ge \maxeff{\gamma}$.

    For $d>0$ we have $\sigma=\alpha_{S',w}$ for $|w|\le \simpL(d)=\simpLfuncLength{d}{1}$ and $\depth(w)= d-1$. Recall that the weight of a transition on cactus letter is at most that of a run on $w^{2\bigM}$ (by \cref{def:stabilisation}). 
    By the inductive hypothesis we therefore have 
    \[\begin{split}
    &\maxeff{\sigma}\le 2\bigM \maxeff{w}\le 2\bigM |w|\maxeff{\boundedCacti{\simpL,d-1}}\le \\
    &2\bigM|w|\simpLfuncMaxW{d-1}\le 2\bigM \simpLfuncLength{d}{1}\simpLfuncMaxW{d-1}\le \simpLfuncMaxW{d}
    \end{split}
    \]
    and we are done.    
\end{proof}
Next, we establish a relationship between cover and length. This may seem unmotivated, but is at the technical heart of our \emph{Zooming Technique} in \cref{lem: d i fit to SSRI or lower depth,lem: charge d i fit to GSRI or lower depth or high potential}.
\begin{proposition}
    \label{prop: relating cover length and maxW} for every $d,i$ we have
    \[\begin{split}
    &\simpLfuncCover{d+1}{i}-\frac{1}{\extraSize}\simpLfuncLength{d+1}{i+1}\simpLfuncMaxW{d}\cdot 2 \\ & \ge \frac12\simpLfuncCover{d+1}{i+1}\ 
    \text{, and}\\   
    &\generLfuncCover{d+1}{i}-\frac{1}{\extraSize}\generLfuncLength{d+1}{i+1}\generLfuncMaxW{d}\cdot 2 \\& \ge \frac12\generLfuncCover{d+1}{i+1}
    \end{split}
    \]
\end{proposition}
\begin{proof}
We prove for $\simpLfuncCover{d+1}{i}$, where the case of $\generLfuncCover{d+1}{i}$ is identical by inserting $\mathtt{G}$ before the relevant functions.
    Expanding from the left-hand side, we have
\[
\begin{split}
    &\simpLfuncCover{d+1}{i}-\frac{1}{\extraSize}\simpLfuncLength{d+1}{i+1}\simpLfuncMaxW{d}\cdot 2\\
    &=\underbrace{8\bigM^2(\simpLfuncMaxW{d}\simpLfuncLength{d}{1}\simpLfuncLength{d+1}{i+1}+\simpLfuncCover{d+1}{i+1})}_{\simpLfuncCover{d+1}{i}}- \\&\frac{1}{\extraSize}\simpLfuncLength{d+1}{i+1}\simpLfuncMaxW{d}\cdot 2
    =8\bigM^2\simpLfuncCover{d+1}{i+1}+ \\
    &\underbrace{(8\bigM^2\simpLfuncLength{d}{1}-\frac{2}{32\bigM^2})}_{\gg 0}\simpLfuncLength{d+1}{i+1}\simpLfuncMaxW{d}\\
    &\ge \frac12\simpLfuncCover{d+1}{i+1}
\end{split}
\]
\end{proof}

\section{Dominance, Potential and Charge}
\label{sec:dominance and potential}
A central tool introduced in~\cite{almagor2026determinization} is a means to separate gaps (as per \cref{def: B gap witness}) to two parts: above the baseline and below it. Technically, this is captured by two concepts -- the \emph{potential} ($\pot$), which measures how high a state can be above the baseline and still induce a finite-value run while all states below it jump to $\infty$, and the \emph{charge} ($\charge$), which measures how far below the baseline run a state can be.

As it turns out, studying these concepts independently is a key component in obtaining decidability of determinisation.
In this work, we need to refine these concepts so that they can be applied to effective cactus letters. In the following, we introduce these notions and results pertaining to them. As usual, we remark when the new proofs are close to their analogues in~\cite{almagor2026determinization}, and when they are novel.

\subsection{Dominant States, Potential and Charge -- Definitions}
\label{sec:dominant states}
Consider a configuration $\vec{c}$ of $\augA^\infty_\infty$. 
The first definition (\emph{dominance}) describes when a state $q$ is still ``relevant'' in $\vec{c}$, in the sense that some suffix read from $\vec{c}$ may make the run starting from $q$ minimal, or at least better (i.e., lower) than states that are currently lower. 
A crucial difference from the analogous definition in~\cite{almagor2026determinization} is that the suffixes we consider are not from the general alphabet, but rather from a certain fragment of effective cactus letters.  
We make this precise as follows. 
% \gatodo{change to suffix $w\in \frgAlphabetSmall^* \cdot \frgAlphabetBigRebase$}
% \gatodo{update definition of potential}
\begin{definition}[Dominant State]
    \label{def:dominant state}
    Consider a configuration $\vec{c}\in \bbZinf^S$. We say that a state $q\in S$ is \emph{dominant} if there exists $w\in (\simpCacRebJump)^*\cdot \genCacRebCac$
    such that 
    $\minweight(q\runsto{w} S)<\infty$ 
    and for every $p\in S$, if $\vec{c}(p)<\vec{c}(q)$ then 
    $\minweight(p\runsto{w} S)=\infty$.

    We say that $q$ is \emph{maximal dominant} if $$\vec{c}(q)=\max\{\vec{c}(p)\mid p\text{ is dominant}\}$$. We then denote this maximal value by $\domval(\vec{c})=\vec{c}(q)$ and the set of maximal dominant states by $\maxdom(\vec{c})$.
\end{definition}
Intuitively, $q$ is dominant if it can yield a run with finite weight, whereas all states below $q$ only yield infinite weight runs (i.e, they cannot complete a valid run). 

Note that the suffix alphabet is very specific, consisting of a word over the simple alphabet, possibly with rebase and jumps, and then a single letter from the general alphabet, which is a cactus letter possibly containing rebase letters. The specificity is used to allow our proofs to carry through on the one hand, and to allow the effective bounds using the recursive functions on the other. 

Note that states $q'$ with $\vec{c}(q')\ge \vec{c}(q)$ are not considered in the definition, and may yield even lower runs than those starting from $q$. Thus, we distinguish the maximal dominant state as well.
Also note that the potential depends on the baseline run, and is therefore a very different notion to gap witness.
The precise connection between dominance and gap witnesses/determinisability is made in \cref{sec:witness}.

Next, recall from \cref{sec:augmented construction} that a word $w$ has a seamless baseline run if for every state $q$ visited by the baseline run after a prefix, the baseline run is also minimal to that state (although there may be other, lower runs, that go through other states). Moreover, the baseline run has constant weight $0$. We use this to capture the growth rate above and below the baseline, as follows.

\begin{definition}[Potential and Charge]
    \label{def:potential}
    \label{def:charge}
    Consider a jump-free word $w$ that has a seamless baseline run and let $\vec{c}_w=\xconf(s_0,w)$.
    \begin{enumerate}
        \item The \emph{potential} of $w$ is  $\pot(w)=\domval(\vec{c}_w)$.
        \item The \emph{charge} of $w$ is  $\charge(w)=-\min\{\vec{c}_w(q)\mid q\in Q\}$.
    \end{enumerate}
\end{definition}
Note that $\pot(w)$ and $\charge(w)$ are always nonnegative (when they are defined), since the minimal run is of weight at most $0$, due to the seamless baseline run. The notions of dominant state, potential and charge are depicted in \cref{fig:dom charge potential}

 \begin{figure}[H]
        \centering
        \includegraphics[width=0.75\linewidth]{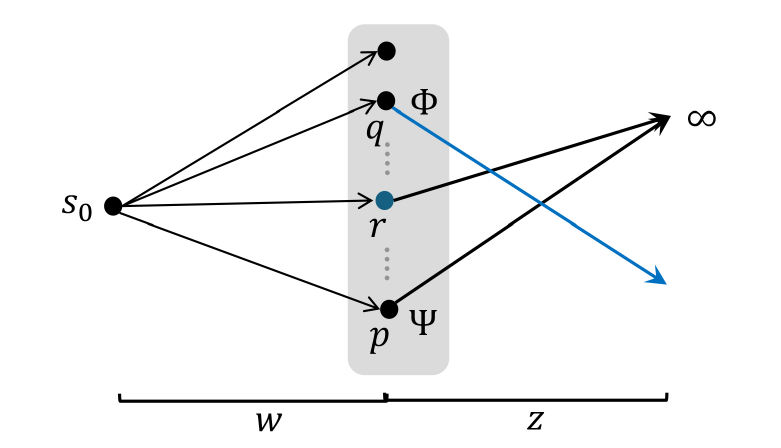}
        \caption{The state $q$ reached after reading $w$ is dominant: there exists a finite-run on a suitable suffix $z$ from $q$, whereas all runs on $z$ from states below $q$ are infinite-weight. Assuming $q$ is also maximal-dominant, its weight determines the potential $\pot(w)$. The weight of the lowest state $p$ determines the charge $\charge(w)$}
        \label{fig:dom charge potential}
    \end{figure}

\subsection{A Toolbox for Potential and Charge}
\label{sec:growth of potential and charge}
We turn to give some fundamental results about potential and charge, as well as link them to the tools of \cref{sec:cactus toolbox} -- specifically to baseline shift and unfolding.

The first result is that dominance is invariant to baseline shifts.
\begin{proposition}
    \label{prop: dominance invariant to baseline shifts}
    Consider a configuration $\vec{c}$ and a dominant state $s$, and let $s'$ be a state. Denote $\vec{c'}=\baseshift{\vec{c}}{s'}$, then for every pair of states $q,q'$ we have $\vec{c}(q)-\vec{c}(q')=\vec{c'}(\baseshift{q}{s'})-\vec{c'}(\baseshift{q'}{s'})$. 
    Moreover, $\baseshift{s}{s'}$ is a dominant state in $\vec{c'}$.
\end{proposition}
\begin{proof}
    \cref{def: baseline shift on states sets and config} immediately implies that 
    $\vec{c}(q)-\vec{c}(q')=\vec{c'}(\baseshift{q}{s'})-\vec{c'}(\baseshift{q'}{s'})$. 
    Then, if $s$ is dominant, there exists a suffix $w\in \simpCacRebJump^*\cdot \genCacRebCac$ such that $\minweight(s,w)<\infty$ but for every $r$ such that $\vec{c}(r)<\vec{c}(s)$ we have $\minweight(r,w)=\infty$.

    It is tempting to claim that we can somehow baseline-shift $w$ to get a separating run from $\baseshift{s}{s'}$. However, there is no natural run to shift $w$ on, and moreover -- the last letter in $w$ is from $\genCacRebCac$, which is not closed under baseline shifts.
    
    Instead, we ``violently'' insert a jump letter before $w$, as follows. Consider $w'=\jl_{s'\to s}w$, and note that $\xconf(\vec{c'},\jl_{s'\to s})=\vec{c}$ and that $w'\in \simpCacRebJump^*\cdot \genCacRebCac$. 
    Then we readily have \\ $\minweight(\baseshift{s}{s'},w')=\minweight(s,w)<\infty$ and for every $q''$ such that $\vec{c'}(q'')<\vec{c'}(\baseshift{s}{s'})$ we can write $q''=\baseshift{s''}{s'}$ for some $s''$ such that $\vec{c}(s'')<\vec{c}(s)$, so 
    $\minweight(q'',w')=\minweight(s'',w)=\infty$.
\end{proof}

In order to analyse the behaviour of $\pot$ and $\charge$ along words, it would be convenient to bound their change upon reading a letter. For the potential, such a bound is already established in~\cite{almagor2026determinization}, but in a non-uniform way due to the infinite alphabet. Here, we specialise it to our effective setting.
\begin{lemma}[Non-Uniform Potential Bounded Growth]
    \label{lem:general bounded growth potential}
    For every finite alphabet $\Gamma\subseteq \Gamma^\infty_\infty$ let $k=2\maxeff{\Gamma}$ then for every $w\in \Gamma^*$ and $\sigma\in \Gamma$, if $w\cdot \sigma$ has a seamless baseline run then $\pot(w\cdot \sigma)-\pot(w)\le k$.
\end{lemma}
For our setting, we need an effective version of this lemma. This is obtained immediately using \cref{prop:simp maxw bounds maxeff}.
\begin{corollary}[Effective Potential Bounded Growth]
    \label{cor:effective bounded growth potential}
    For every $w\in (\simpCac)^*$ such that $\depth(w)=d$ and for every $\sigma\in \simpCac$, if $w\cdot \sigma$ has a seamless baseline run then $\pot(w\cdot \sigma)-\pot(w)\le 2\simpLfuncMaxW{d}$.
\end{corollary}

In contrast to $\pot$, there is no corresponding bound for $\charge$. Indeed, 
$\charge$ can \emph{decrease} abruptly, if a minimal run cannot continue and a much-higher run becomes minimal. The following lemma links charge and potential by showing that when this happens, we can find a high-potential word. This is a central tool in our proof, and is the first usage of our baseline-shift toolbox.
\begin{lemma}[Charge Decrease to High Potential]
    \label{lem: charge decrease to high potential}
    For every $P\in \bbN$, if there exists a word $u\in (\genCac)^*$ and a letter $\sigma\in \genCac$ with $\depth(\sigma)=d$ such that $\charge(u)-\charge(u\sigma)> P+2\generLfuncMaxW{d}$ then there exists a word $w\in (\Gamma_0^0)^*$ such that $\pot(w)> P$.
\end{lemma}
\begin{proof}
    The proof is a quantitative and simplified version of the analogous claim from \cite{almagor2026determinization}.
    Intuitively, since $\sigma$ causes a sudden decrease in $\charge$, this means that $\sigma$ cannot be read by the minimal run after $u$, and that the only run that can read it is much higher. This already suggests that the potential is unbounded. 
    Technically, all that needs to be done is to perform a baseline shift (\cref{sec: baseline shift}) on a minimal run on $u$. This, in turn, also requires us to first flatten (\cref{sec: cactus unfolding}) this prefix, so that we remain in $\Gamma_0^0$. 

    By \cref{prop:gen maxw bounds maxeff} the maximal weight appearing in a transition over $\sigma$ is at most $\generLfuncMaxW{d}$. In particular, note that any two runs can become farther apart by at most $2\generLfuncMaxW{d}$ upon reading $\sigma$.
    We show that the potential can be greater than $P$. 
    By the assumption, there exist $u,\sigma$ such that $\charge(u)> \charge(u\sigma)+ P+2\generLfuncMaxW{d}$.
    
    Since $\charge$ is the negation of the minimal weight, there exists a seamless run $\rho_1:s_0\runsto{u} S$ such that $\weight(\rho_1)=-\charge(u)=\minweight(s_0\runsto{u} S)$. 
    Let $\vec{c_u}=\xconf(s_0,u)$, and consider the minimal state $s_1$ that can read $\sigma$. That is, $s_1\in \arg\min\{s\in \supp(\vec{c_u})\mid \minweight(s_1\runsto{\sigma}S)<\infty\}$. It must hold that $\vec{c_u}(s_1)>\weight(\rho_1)+P$ as otherwise the minimal run on $u\sigma$ has weight at most 
    \[-\charge(u)+P+\generLfuncMaxW{d}< -\charge(u)+P+2\generLfuncMaxW{d}<  -\charge(u\sigma)\]
    where the last inequality follows by our requirement on $\charge(u\sigma)$. But this is a contradiction, since $\minweight(s_0,u\sigma)=-\charge(u\sigma)$.
    Put simply -- if the charge has a large jump on $\sigma$, then the minimal state that can read $\sigma$ is much higher than the minimal state up to $\sigma$. Note that $s_1$ is therefore dominant (\cref{def:dominant state}). This idea is depicted in \cref{fig:charge drop to big potential}
    Let $\rho_2:s_0\runsto{u}s_1$ be a seamless run to $s_1$.

 \begin{figure}[H]
        \centering
        \includegraphics[width=0.9\linewidth]{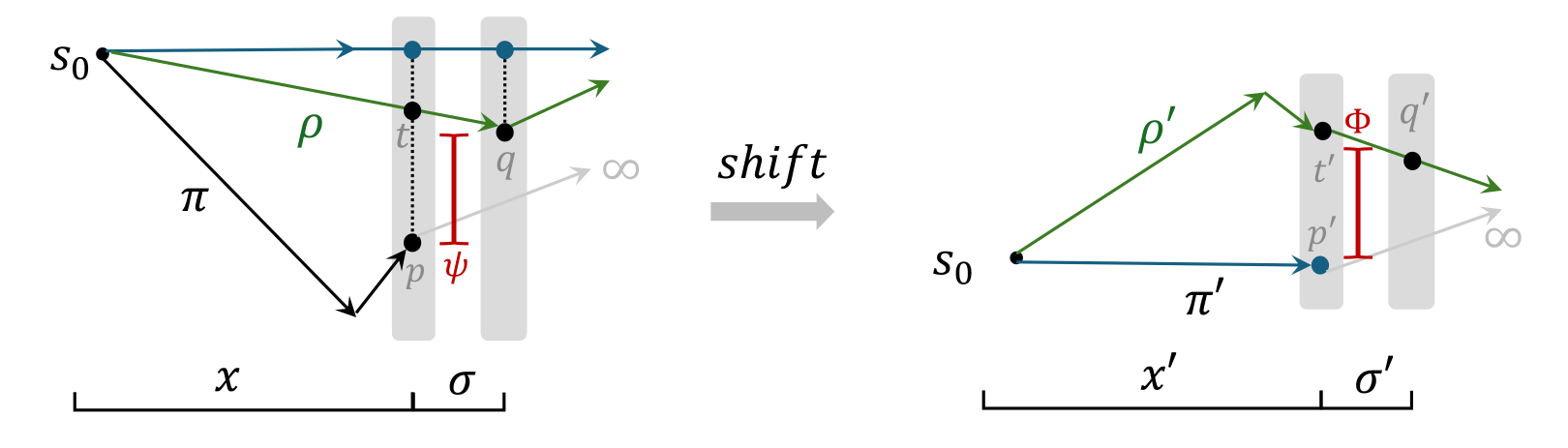}
        \caption{Upon reading $\sigma$, the charge drops drastically. This means that $\rho$ became minimal, instead of $\pi$. By shifting the baseline, the state $t'$ is dominant, setting a lower bound for the potential.}
        
        \label{fig:charge drop to big potential}
    \end{figure}

    We now turn to flattening the prefix (\cref{lem:flattening essence}).
    Let $x_0=\flatten(u \wr 2\generLfuncMaxW{|S|}|u\sigma|)$. Since $\genCac$ does not contain rebase letters, we have $x_0\in (\Gamma_0^0)^*$ (since flattening may only introduce jump letters in case of rebase letters).
    Let $\mu_1:s_0\runsto{x_0}S$ be a seamless run such that $\weight(\mu_1)=-\charge(x_0)$ and similarly $\mu_2:s_0\runsto{x_0}s_1$ a seamless run such that $\weight(\mu_2)=\weight(\rho_2)$. These exist by 
    \cref{lem:flattening essence}.
    By the same lemma, for every state $s'$ with $\minweight(s_0\runsto{x_0} s')<\weight(\mu_2)$ we have $\minweight(s'\runsto{\sigma} S)=\infty$ i.e., $s_1$ is still the minimal state from which $\sigma$ can be read, and $\weight(\mu_2)-\weight(\mu_1)> P$.
    
    We now perform a baseline shift on $\mu_1$, and by \cref{prop: dominance invariant to baseline shifts} we have that $\baseshift{\mu_2}{\mu_1}$ ends in a dominant state with weight at least $P$ (since $\baseshift{\mu_1}{\mu_1}$ is now the baseline with weight $0$). It follows that $\pot(\baseshift{x_0}{\mu_1})> P$, and we are done.
\end{proof}

Lastly, the next result from~\cite{almagor2026determinization} is that both potential and charge are monotone with respect to the initial configuration, as follows.
\begin{lemma}
\label{lem:potential and charge are monotone}
        Consider configurations $\vec{c},\vec{d}$ where $\vec{c}\le \vec{d}$ such that $\supp(\vec{c})=\supp(\vec{d})$, there is a single baseline state $q\in \supp(\vec{c})$, and $\vec{c}(q)=\vec{d}(q)=0$.
        
        Then for every $w\in \Sigma^*$ with a seamless baseline run from $\vec{c}$ it holds that:
        \begin{enumerate}
            \item The baseline run of $\xconf(\vec{d},w)$ is seamless.
            \item $\pot(\xconf(\vec{c},w))\le \pot(\xconf(\vec{d},w))$.
            \item $\charge(\xconf(\vec{c},w))\ge \charge(\xconf(\vec{d},w))$
        \end{enumerate}
    \end{lemma}

\section{A Witness for Undeterminisability}
\label{sec:witness}
The characterisation of undeterminisability given in~\cite{almagor2026determinization} is by means of \emph{witnesses} --  a word containing a cactus letter whose various unfoldings serve to induce arbitrarily large gap witnesses as per \cref{def: B gap witness}.
The leap we take in this work is to significantly restrict the set of witnesses to the finite alphabet of effective cactus letters. We no longer use these witnesses algorithmically, but only as part of the proof (although an algorithmic use is possible).

Intuitively, a witness consists of a prefix $w_1$, a ``pumpable infix'' $w_2$, and a suffix $w_3$, where the idea is that after enough iterations of $w_2$ the potential grows unboundedly, and $w_3$ serves to separate the dominant state from the baseline, thus demonstrating unbounded gaps, which imply undeterminisability by \cref{thm:det iff bounded gap}.

Consider the WFA $\augA_\infty^\infty=\tup{S,\Gamma_{\infty}^\infty,s_0,\augTrans_\infty^\infty}$, and recall that $\Gamma_{\infty}^\infty$ is obtained by alternating the stabilisation closure (\cref{def:stab closure}) and rebase (\cref{def:rebase}) (and it also includes jump letters, as per \cref{def:jump letters}). 
The importance of $\augA_\infty^\infty$ is captured by the following result from~\cite{almagor2026determinization}:
\begin{theorem}
    \label{thm:aug A inf inf determinisable iff A determinisable}
    $\augA_\infty^\infty$ is determinisable if and only if $\cA$ is determinisable.
\end{theorem}
We remark that the concept of a determinisable WFA over an infinite alphabet is a bit delicate. Nonetheless, the theorem holds in a strong sense: if $\cA$ is determinisable (over its finite alphabet), then $\augA_\infty^\infty$ is determinisable over its infinite alphabet. The converse is much simpler.

\begin{definition}[Witness]
    \label{def:witness}
    A \emph{witness of type-$0$ (resp. type-$1$)} is a tuple $w_1,\alpha_{S_1,w_2},w_3$ such that the following hold.
    \begin{enumerate}
        \item $w_1 \in (\simpCac)^*$, $S_1=\ghostTrans(s_0,w_1)$, $\alpha_{S_1,w_2}\in \simpCac$ (resp. $\alpha_{S_1,w_2}\in \simpCacRebCac$ for type-$1$), and $w_3\in(\simpCacRebJump)^* \cdot \genCacRebCac$.
        \item $\booltrans(s_0,w_1)=\booltrans(s_0,w_1 \cdot w_2^{2\bigM})$ 
        (i.e., $w_2^{2\bigM}$ cycles on the reachable states).
        \item $\minweight(s_0\runsto{w_1 \cdot w_2^{2\bigM} \cdot w_3} S)<\infty$ whereas $\minweight(s_0\runsto{w_1\alpha_{S_1,w_2}w_3} S)=\infty$.
    \end{enumerate}
\end{definition}

The main characterisation proved in~\cite{almagor2026determinization} is that $\augA^\infty_\infty$ is undeterminisable if and only if it has a witness. Our modified definition is more restrictive than in~\cite{almagor2026determinization}, hence the challenge is in proving the completeness of this characterisation. The soundness, however, follows directly from~\cite{almagor2026determinization}.
\begin{lemma}
    \label{lem:witness implies nondet}
    If there is a witness, then $\augA^\infty_\infty$ is undeterminisable.
\end{lemma}
The decidability of determinisation in~\cite{almagor2026determinization} relies on the decidability of checking witnesses in $\augA^\infty_\infty$.  
We could follow the same approach for this work. However, our new framework actually leads to a far simpler algorithm (\cref{sec:algorithm and complexity}), and we therefore do not concern ourselves with algorithmic properties of witnesses.

Our next result is key to many of our proofs, and asks what happens to the potential $\pot$ when we unfold a cactus letter. Intuitively, we want to argue that unfolding does not affect the potential. However, this is generally not true. Fortunately, when it is not true, then a type-0 witness exists.
\begin{lemma}[Unfolding Maintains Potential]
    \label{lem:unfolding maintains potential}
        Consider a word $u \alpha_{B,x} v\in \simpCac^*$ and let $F\ge 2\maxeff{u\alpha_{B,x}v}$ and $ux^{2\bigM M_0}v=\unfold(u,\alpha_{B,x},v \wr F)$ as per \cref{def:unfolding function}. For every prefix $v'$ of $v$ we have:
        \begin{enumerate}
            \item $\xconf(s_0,ux^{2\bigM M_0}v')\le \xconf(s_0,u\alpha_{B,x}v')$.
            \item Either $\augA_\infty^\infty$ has a type-$0$ witness, or $\pot(ux^{2\bigM M_0}v')=\pot(u\alpha_{B,x}v')$.
        \end{enumerate}
\end{lemma}
\begin{proof}
    The proof structure is almost identical to the analogous claim in~\cite{almagor2026determinization}, but some changes need to be made to account for the new definitions of dominance (\cref{def:dominant state}) and of witnesses. Naturally, we tailor both definitions so that this proof carries out smoothly.
    
    Item 1 follows immediately from \cref{lem:unfolding configuration characterisation}. Indeed, let $\vec{d}=\xconf(s_0,u\alpha_{B,x}v')$ and $\vec{d'}=\xconf(s_0,ux^{2\bigM M_0}v')$, then  we have 
    $\supp(\vec{d})\subseteq  \supp(\vec{d'})$ and for every $q\in \supp(\vec{d})$ we have $\vec{d'}(q)=\vec{d}(q)$ (and anyway $\vec{d}(q)=\infty$ otherwise). 

    It is now tempting to use \cref{lem:potential and charge are monotone} and use the superiority of the configurations to conclude the proof. Unfortunately, the condition there also requires that the configurations have equal supports, which fails in our case. We therefore need to work much harder to show the lemma.

    Intuitively, the proof proceeds as follows. We consider the unfolding $ux^{2\bigM M_0}v=\unfold(u,\alpha_{B,x},v \wr F)$, and specifically its prefix $ux^{2\bigM M_0}$. 
    By \cref{prop:unfolding cactus maintains seamlesss gaps}, the unfolding preserves the reachable configuration, i.e. the weights with which states in $\booltrans(s_0,u\alpha_{B,x})$ are reached. However, we might also get finite weights on other states in $\ghostTrans(s_0,u\alpha_{B,x})$, due to the unfolding. 
    The question now is whether these newly-reachable states interfere with the potential. 

    In order to reason about the potential, we consider a maximal dominant state $q$, and split to two cases. The first case is when $q$ is in $\booltrans(s_0,ux^{2\bigM M_0})\setminus \booltrans(s_0,u\alpha_{B,x})$ (namely -- very high, reached by pumping $x$ on non-grounded pairs). 
    Then, we show that we can construct a type-0 witness. Specifically, the witness is 
    %$(u x^{2\bigM M_0},x^{2\bigM M_0},z)$, 
    $(u x^{2\bigM M_0},\alpha_{B,x},z)$
    where $z$ 
    is the suffix that shows $q$ is dominant. Proving that this is indeed a witness is technical, and uses the intricacies of our definitions of stable cycles and grounded states. 
    The intuition, however, is not too complicated: if the dominant state is high, then replacing $x^{2\bigM M_0}$ by its corresponding cactus letter sends all states above $q$ to $\infty$, and the remaining states are sent to $\infty$ by $z$.

    The second case is when $q$ is in $\booltrans(s_0,u\alpha_{B,x})$, namely low, reachable by grounded states. Then we use the unfolding toolbox of \cref{sec: cactus unfolding} to show that the potential is indeed maintained.

    We now turn to formalise this argument.
    %Assume that $\augA_\infty^\infty$ does not have a $0$-type witness. 
    Consider $ux^{2\bigM M_0}v=\unfold(u,\alpha_{B,x},v \wr F)$, and assume that the potential is defined (otherwise we are done).
    Consider the following configurations and sets of states:
    % \[
    % \begin{split}
    %     &\vec{c_1}=\xconf(s_0,u)\quad \vec{c_2}=\xconf(s_0,u\alpha_{B,x})\quad \vec{c'_2}=\xconf(s_0,ux^{2\bigM M_0})\\
    %     &S_1=\ghostTrans(s_0,u)\quad S_2=\booltrans(s_0,u\alpha_{B,x})\quad S'_2=\booltrans(s_0,ux^{2\bigM M_0})
    % \end{split}
    % \]
    \[
\begin{array}{ l l l }
\vec{c_1} = \xconf(s_0, u) & \vec{c_2} = \xconf(s_0, u\alpha_{B,x}) & \vec{c'_2} = \xconf(s_0, ux^{2\bigM M_0}) \\ 
S_1 = \ghostTrans(s_0, u) & S_2 = \booltrans(s_0, u\alpha_{B,x}) & S'_2 = \booltrans(s_0, ux^{2\bigM M_0}) \\ 
\end{array}
\]
    Note that $S_2=\supp(\vec{c_2})$ and $S'_2=\supp(\vec{c'_2})$. 
    By \cref{lem:unfolding configuration characterisation} we have that 
    $S_2\subseteq S'_2$, for every $s\in S_2$ it holds that $\vec{c'_2}(s)=\vec{c_2}(s)$ and 
    for every $s'\notin S_2$ it holds that $\vec{c'_2}(s')> \vec{c_2}(s)$ (where $s'\notin S_2$ means either $s'\in  S'_2\setminus S_2$, or $s'\notin S'_2$).

    Consider a maximal-dominant state $q\in S$ in $\vec{c'_2}$. 
    By \cref{def:dominant state}, this means that there exists %$z\in \Gamma^\infty_\infty$ 
    $z\in (\simpCacRebJump)^*\genCacRebCac$
    such that 
    $\minweight(q\runsto{z}S)<\infty$ 
    whereas for every $p\in S$ with $\vec{c'_2}(p)<\vec{c'_2}(q)$ we have 
    %$\minweight_{\vec{c'_2}}(z,p\to S)=\infty$. 
    $\minweight(p\runsto{z} S)=\infty$. We now split to cases according to whether $q\in S'_2\setminus S_2$ (i.e., $q$ is ``very high'') or $q\in S_2$ (i.e., $q$ is ``low'').

    \paragraph*{Dominant State is Very High}
    If $q\in S'_2\setminus S_2$, we claim that \\
    $(u x^{2\bigM M_0},\alpha_{B,x},z)$
    %$(u x^{2\bigM M_0},x^{2\bigM M_0},z)$  
    is a type-0 witness. 
    Indeed, following the requirements in \cref{def:witness}, we have:
    \begin{enumerate}
        \item $ux^{2\bigM M_0}\in (\simpCac)^*,\alpha_{B,x}\in \simpCac$ and $z\in(\simpCacRebJump)^*\genCacRebCac$. 
        Also, note that $B=\ghostTrans(s_0,u)=\ghostTrans(s_0,ux)=\ghostTrans(s_0,ux^{2\bigM M_0})$ since $\alpha_{B,x}$ is a stable cycle readable after $u$.
        \item By \cref{prop: transitions stabilise at M} we have that \\ $\booltrans(s_0,u x^{2\bigM M_0})=\booltrans(s_0,u x^{2\bigM (M_0+1)})=\booltrans(s_0,u x^{2\bigM M_0}x^{2\bigM})$.
        \item Since $q\in S'_2=\booltrans(s_0,u x^{2\bigM M_0})$ then by Item 2 above we have $q\in \booltrans(s_0,u x^{2\bigM M_0}x^{2\bigM})$. In particular, $\minweight(s_0\runsto{u x^{2\bigM M_0}x^{2\bigM}} q)<\infty$. Since $z$ satisfies that $\minweight(q\runsto{z} S)<\infty$, we conclude that:
        \[ \begin{split}
        &\minweight(s_0\runsto{u x^{2\bigM M_0}\cdot x^{2\bigM}\cdot z} S)\le\\& \minweight(s_0\runsto{u x^{2\bigM M_0}x^{2\bigM}} q)+\minweight(q\runsto{z} S)<\infty
        \end{split}
        \]

        It therefore remains to prove that $\minweight(s_0\runsto{u x^{2\bigM M_0}\cdot \alpha_{B,x}\cdot z} S)=\infty$. 
        To this end, %we first notice that by \cref{prop:cactus letters stabilises at 2M}, the transitions on $\alpha_{B,x}$ and $\alpha_{B,x^{2\bigM M_0}}$ are identical. We therefore proceed with $\alpha_{B,x^{2\bigM M_0}}$. We
        we claim that for every state $p$, if $\minweight(s_0\runsto{u x^{2\bigM M_0}\cdot \alpha_{B,x}} p)<\infty$, then $\vec{c'_2}(p)<\vec{c'_2}(q)$ (from which we can obtain the desired result).
        Indeed, if \\$\minweight(s_0\runsto{u x^{2\bigM M_0}\cdot \alpha_{B,x}} p)<\infty$ then there exists a run 
        \[
        \rho:s_0\runsto{u}s_1\runsto{x^{2\bigM M_0}}p'\runsto{\alpha_{B,x}}p
        \]
        By \cref{def:stabilisation}, this implies that $(p',p)\in \GroundPairs(\alpha_{B,x})$. Then, however, it follows by \cref{prop:grounding states reachable with any bigM k} that $(s_1,p)\in \GroundPairs(B,x)$ (by using the same grounding state of $(p',p)$). 
        Thus, we have the run
        $s_0\runsto{u}s_1\runsto{\alpha_{B,x}}p$
        so $p\in S_2$. 
        As noted above, the configurations $\vec{c_2}$ and $\vec{c'_2}$ coincide on all the states in $S_2$, so $\vec{c_2}(p)=\vec{c'_2}(p)$. On the other hand, $q\in S'_2\setminus S_2$, and again as noted above this means that $\vec{c'_2}(q)>\vec{c_2}(p)=\vec{c'_2}(p)$.

        We are therefore in the following setting:
        \begin{itemize}
            \item If $\minweight(s_0\runsto{u x^{2\bigM M_0}\cdot \alpha_{B,x}} p)<\infty$, then $\vec{c'_2}(p)<\vec{c'_2}(q)$, but then $\minweight(p\runsto{z} S)=\infty$, so 
            \[\minweight(s_0\runsto{u x^{2\bigM M_0}\cdot \alpha_{B,x}}p\runsto{z}S)=\infty\]
            \item If $s_0\runsto{\minweight(u x^{2\bigM M_0}\cdot \alpha_{B,x}} p)=\infty$ then trivially 
            \[\minweight(s_0\runsto{u x^{2\bigM M_0}\cdot \alpha_{B,x}}p\runsto{z}S)=\infty\]
        \end{itemize}
        We conclude that anyway we have
        $\minweight(s_0\runsto{u x^{2\bigM M_0}\cdot \alpha_{B,x}\cdot z} S)=\infty$
    \end{enumerate}
    This shows that all the requirements of \cref{def:witness} hold, so there is a type-0 witness.

    \paragraph*{Dominant State is Low}
    If $q\notin S'_2\setminus S_2$, then $q\in S_2$ (since $q\in \supp(\vec{c'_2})=S'_2$). Then, by the observations above, we have $\vec{c'_2}(q)=\vec{c_2}(q)$. We claim that for every prefix $v'$ of $v$ it holds that $\pot(ux^{2\bigM M_0}v')=\pot(u\alpha_{B,x}v')$. 
    Indeed, fix such a prefix $v'$ and denote
    \[
    \begin{split}
        &\vec{c_3}=\xconf(s_0,u\alpha_{B,x} v')=\xconf(\vec{c_2},v') \quad S_3=\supp(\vec{c_3})\\
        &\vec{c'_3}=\xconf(s_0,u x^{2\bigM M_0} v')=\xconf(\vec{c'_2},v') \quad S'_3=\supp(\vec{c_3})
    \end{split}
    \]
    Similarly to the above, we have $S_3\subseteq S'_3$, and due to the choice of $F$ as the unfolding constant (\cref{def:unfolding function,lem:unfolding configuration characterisation}) for every $p\in S_3$ it holds that $\vec{c_3}(p)=\vec{c'_3}(p)$, and for every $p\in S'_3\setminus S_3$ we have $\vec{c'_3}(p)>\max\{\vec{c_3}(r)\mid r\in S_3\}+\maxeff{u\alpha_{B,x}v}$.

    Let $p,p'$ be maximal dominant states in $\vec{c_3}$ and $\vec{c'_3}$, respectively. We claim that $p,p'\in S_3$. Indeed, $p\in S_3$ by definition. If $p'\in S'_3\setminus S_3$, then there exists a suffix $z$ such that $\minweight(p'\runsto{z} S)<\infty$ and for every $r'$ with $\vec{c'_3}(r')<\vec{c'_3}(p')$ it holds that $\minweight(r'\runsto{z} S)=\infty$. In particular, for every $r'\in S_3$ this holds. 
    Since $p'\in S'_3\setminus S_3$, then $p'\in \booltrans(S'_2,v')\setminus \booltrans(S_2,v')$. But then $v'z$ is a suffix that exhibits a dominant state in $S'_2\setminus S_2$, and in particular the maximal dominant state of $\vec{c'_2}$ is in $S'_2\setminus S_2$, so we are back in the previous case of a very high dominant state (which we assume does not hold now).

    Thus, $p'\in S_3$. It now readily follows that $\pot(ux^{2\bigM M_0}v')=\pot(u\alpha_{B,x}v')$: restricted to $S_3$, the configurations $\vec{c_3}$ and $\vec{c'_3}$ are equivalent. Moreover, for every $r'\in S'_3\setminus S_3$ we have $\vec{c'_3}(r')>\vec{c'_3}(p')$, so does not come into consideration for the dominance of $p'$. 

    It thus follows that $\vec{c'_3}(p')=\vec{c_3}(p)$ (in fact, we can assume without loss of generality that $p=p'$, since we can choose any maximal-dominant state, if there are several), and since the baseline run is reachable and its state is also in $S_3$, we conclude that
    \[
    \pot(ux^{2\bigM M_0}v')=\vec{c'_3}(p')=\vec{c_3}(p)=\pot(u\alpha_{B,x}v').
    \]
\end{proof}
By applying \cref{lem:unfolding maintains potential} recursively, we obtain a similar result for flattening.
\begin{corollary}
    \label{cor:flatttening maintains potential}
    Either $\augA_\infty^\infty$ has a type-$0$ witness, or the following holds.
        Consider an initial state $s_0$ and $w \in (\simpCac)^*$. 
        Let  $F\ge 2\maxeff{w}$ and denote $w'=\flatten(w\wr F)$, then:
        \begin{enumerate}
            \item $\xconf(s_0,w)\le \xconf(s_0,w')$.
            \item $\pot(w)=\pot(w')$ (if the potential is defined, with respect to the given initial state $s_0$).
        \end{enumerate}
\end{corollary}

\begin{remark}[The Prefix ``Either $\augA_\infty^\infty$ has a type-0 witness, or...'']
\label{rmk:the either witness prefix}
    Most of our reasoning about potential in the coming sections is ultimately based on \cref{lem:unfolding maintains potential}. 
    Therefore, the prefix ``Either $\augA_\infty^\infty$ has a type-0 witness, or the following holds'' appears in many statements. Since having such a witness immediately implies nondeterminziability (by \cref{lem:witness implies nondet}), this condition can be read as ``either we are done, or the following holds''.
\end{remark}

\section{Separated Repeating Infixes (SRI)}
\label{sec:separated repeating infix}
In this section we introduce the ``pumpable'' structure we work with, called \emph{separated repeating infix (SRI)}. The remainder of the work, and the main technical contribution, is split to two parts: first, we show that once an SRI is found in some word $w$, we can shorten $w$ while maintaining certain desirable properties. This is later used by assuming a minimal-length $w$ with these properties, and reaching a contradiction by finding an SRI.
Second, we show under which conditions an SRI exists.

Technically, SRI come in two variants: \emph{Simple SRI (SSRI)} relating to the potential ($\pot$), and \emph{General SRI (GSRI)} relating to the charge ($\charge$).
As a result, we need to develop several new tools for reasoning about the usage of SRI. Fortunately, these can be obtained using our new developments in baseline shifts (see \cref{sec: baseline shift}). 
While some of the tools work uniformly for both variants, others require separate reasoning. 

The main takeaway from this section can be summarised as follows: if we have an SRI, then we either have a witness, or we can shorten the word without causing too much havoc. This is later used as a crucial step in the proof.

\subsection{Separated Repeating Infixes -- Definition}
\label{sec:separated repeating infix definition}
We start with an intuitive overview of the properties we require, followed by the formal definition.

Consider the WFA $\augA_\infty^\infty$, and focus on the alphabet $\simpCac$ or $\genCac$ (i.e., without rebase and jump letters). 
An SRI consists of a concatenation $uxyv$ with the following structure (see \cref{fig:separated repeating infix}): 
\begin{itemize}
    \item Reading $u$ from $s_0$ reaches a configuration that is partitioned to ``sub-configurations'' that have a huge gap between them, we call the state sets in these sub configurations $V_1,\ldots, V_{\ell}$. 
    \item Next, reading $x$ maintains the exact structure of these $V_j$ sets, and shifts the weight of the states in each set $V_j$ by a constant $k_{j,x}$. In particular, the entire configuration maintains its support after $x$.

    An additional property of $x$ is that it is quite short -- enough so that if needed, we can turn $x^k$ into a stable cycle as per~\cref{sec:from reflexive to stable}.
    
    \item Reading $y$ has very similar behaviour to $x$, in the sense that the $V_j$ are maintained, and each is increased by a constant $k_{j,y}$, and the sign of $k_{j,y}$ is the same as that of $k_{j,x}$, so each $V_j$ is decreasing/increasing on both $x$ and $y$. 

    \item Finally, $v$ is just a harmless suffix.
\end{itemize}

We now turn to the formal definition. Let $L\in \{\simpL,\generL\}$ be a length function as per \cref{sec:effective cactus}.
\begin{definition}[Separated Repeating Infix (SRI) with respect to $L$]
    \label{def:separated repeating infix}
    Consider a word $w\in \boundedCac{L}$
    %Consider a word $w\in \simpCac^*$ 
    decomposed as $w=u\cdot x\cdot y\cdot v$. 
    Denote $\vec{c_u}=\xconf(s_0,u)$, $\vec{c_{ux}}=\xconf(s_0,ux)$ and $\vec{c_{uxy}}=\xconf(s_0,uxy)$.
    We say that $w$ is a \emph{separated repeating infix (SRI)} with respect to $L$ if the following conditions hold.
    \begin{enumerate}
        \item \label{itm:separated infix reachable states} $\supp(\vec{c_u})=\supp(\vec{c_{ux}})=\supp(\vec{c_{uxy}})=B$ for some subset $B\subseteq S$. 
        That is, the subsets of states reached after $u,ux,uxy$ are the same..
        \item 
        $|x| \le \frac{1}{\bigM}L(\depth(x)+1)$.
        \item \label{itm:separated infix partition}
        Define $G=\begin{cases}
            4|S|\bigM\maxeff{xy} & \text{ if } L=\simpL,\\
            4|S|\bigM\maxeff{xy}+\budH &\text{ if } L=\generL.
        \end{cases}$. Then $B$ can be partitioned to $B=V_1\cup \ldots \cup V_\ell$ for some $\ell$ such that the following hold. 
        \begin{enumerate}
            \item 
            \label{itm:separated infix gaps V}
            For every $1\le j<\ell$, every $s\in V_j,s'\in V_{j+1}$ and  every $\xi\in \{u,ux,uxy\}$ we have 
            $\vec{c_\xi}(s')-\vec{c_\xi}(s)> G$. 
            That is, the values assigned to states in $V_{j+1}$ are significantly larger than those assigned to $V_{j}$ in the configurations $\vec{c_{u}}, \vec{c_{ux}}$ and $\vec{c_{uxy}}$.
            \item 
            \label{itm:separated infix linear k}
            For every $1\le j\le \ell$ there exist $k_{j,x},k_{j,y}\in \bbZ$ such that for every $s\in V_j$ we have  $\vec{c_{ux}}(s)=\vec{c_u}(s)+k_{j,x}$ and $\vec{c_{uxy}}(s)=\vec{c_{ux}}(s)+k_{j,y}$.

            Moreover, $\sign(k_{j,x})=\sign(k_{j,y})$ (where $\sign:\bbZ\to \{-,0,+\}$ is the sign function).
        \end{enumerate}
        \item \label{itm:separated infix seamless baseline}  There is a seamless baseline run $\rho:s_0\runsto{uxyv}s$ for some $s$.
    \end{enumerate}
\end{definition}
The definition of SRI is structural. We now specialise it with a restriction on $\pot$ or on $\charge$, as follows.
\begin{definition}
    \label{def: SSRI and GSRI}
    Consider an SRI $uxyz$ over $L$. 
    \begin{itemize}
        \item It is a \emph{Simple SRI (SSRI)} if $L=\simpL$ and in addition $\pot(u)\le \pot(ux)\le \pot(uxy)$, and $\pot(u) = \pot(ux)$ iff $\pot(ux) = \pot(uxy)$.
        \item It is a \emph{General SRI (GSRI)} if $L=\generL$ and in addition $\charge(u)\ge \charge(ux)\ge \charge(uxy)$, and $\charge(u) = \charge(ux)$ iff $\charge(ux) = \charge(uxy)$.
    \end{itemize}
\end{definition}
In the following, when we write SRI we mean either an SSRI or a GSRI.
We further distinguish \emph{flavours} of SRI, each leading to its own results in the following.
\begin{definition}
    \label{def:flavours of inc inf}
    Consider an SRI $w=uxyv$ as per \cref{def:separated repeating infix}. We say that $w$ is:
    \begin{itemize}
        \item \emph{Negative} if $k_{j,x},k_{j,y}<0$ for some $j$.
        \item \emph{Positive} if $k_{j,x},k_{j,y}\ge 0$ for all $j$.
        \item \emph{Stable} if $(\ghostTrans(s_0,u),x)$ is a stable cycle. In this case it is \emph{Degenerate} if $(\ghostTrans(s_0,u),x)$ is degenerate, and \emph{Non-degenerate} otherwise.
    \end{itemize}   
\end{definition}

\begin{remark}[Comparison of SRI with~\cite{almagor2026determinization}]
    \label{rmk: SRI v.s. original separated increasing infix}
    \cite{almagor2026determinization} introduces \emph{Separated Increasing Infixes}, which are similar to \emph{Positive} SRI. Technically, they do not require $x$ to be an effective cactus letter (only a general cactus), but place much stronger restrictions on $x,y$, that have to do with a notion of \emph{cost}, which we completely avoid in this work. 
    Moreover, the notion there is purely structural, with no regard for potential or charge. The incorporation of those into SSRI and GSRI, respectively, renders them far more useful (but technically harder to find).
    
    Additionally, our notion of SRI is more versatile, with its different flavours.    
    Few of the results we give on SRI have similar proofs to those of~\cite{almagor2026determinization}. We comment when this is the case.
\end{remark}

\subsection{SRI -- A Toolbox}
\label{sec: basic results for SRI}
We present several tools for SRI that apply to both SSRI and GSRI. 
We start with some results about the structure of runs in an SRI.

The partition and gaps described in \cref{itm:separated infix partition,itm:separated infix gaps V} of \cref{def:separated repeating infix} give rise to a restriction on the runs on the infixes $x,y$. Specifically, there are only runs from $V_j$ to $V_{j'}$ for $j\ge j'$, and there is always a run from $V_{j'}$ to $V_{j'}$. 
\begin{proposition}
\label{prop:run characteristics of SRI}
    Consider an SRI $uxyv$ and let $B=V_1\cup\ldots \cup V_\ell$ and $k_{j,x},k_{j,y}$ as per \cref{def:separated repeating infix}. Then the following hold for every $1\le j\le \ell$,.
    \begin{itemize}
        \item $|k_{j,x}|\le \maxeff{x}$ and $|k_{j,y}|\le \maxeff{y}$.
        \item For every $z\in \{x,y\}$ and $s\in S$, if $s\in V_{j}$ then there exists $s'\in V_{j}$ such that $s'\runsto{z}s$.
        \item For every $z\in \{x,y\}$ and $s,s'\in S$ such that $s\runsto{z^k}s'$ for some $1\le k\le |S|\bigM$, if $s'\in V_{j'}$ then $s\in V_{j}$ for some $j\ge j'$.
    \end{itemize}
\end{proposition}
\begin{proof}
    The proof follows the same lines as its analogue in~\cite{almagor2026determinization}, with the main changes being the Negative SRI case, and the extension to $z^k$ for $k\ge 1$.
    
    We prove the claim for $z=x$, and the case of $y$ identical up to working with $\vec{c_{uxy}}$ instead of $\vec{c_{ux}}$. 
    
    We start with the first item, and we split to positive and negative $k_{j,x}$. Assume by way of contradiction that there exists $1\le j\le \ell$ such that $k_{j,x}>\maxeff{x}$, and let $j$ be maximal with this property (i.e., the ``highest'' part in the partition where this holds). We depict the proof in \cref{fig:increasing infix run characterisation}.
    \begin{figure}[H]
        \centering
        \includegraphics[width=0.75\linewidth]{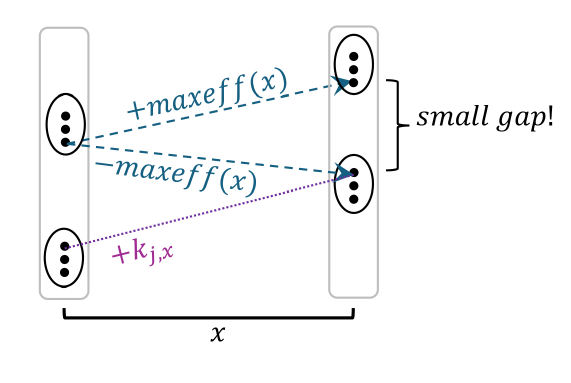}
        \caption{The part $V_j$ and the state $s$ whose minimal run stems from $V_{j'}$, implying a small gap.}
        \label{fig:increasing infix run characterisation}
    \end{figure}
    
    Consider a state $s\in \arg\max\{\vec{c_u}(q)\mid q\in V_j\}$. We then have $\vec{c_{ux}}(s)>\vec{c_u}(s)+\maxeff{x}$. This means that the minimal seamless run $\rho:s'\runsto{x}s$ is from $s'\in V_{j'}$ with $j'>j$. Indeed, no run from any state in $V_j$ or lower can get to weight $\vec{c_u}(s)+\maxeff{x}$. In particular, we get $j<\ell$.
    We therefore have $\vec{c_{ux}(s)}\ge \vec{c_{u}}(s')-\maxeff{x}$.

    In the following, it is useful to think of $j'$ as $j+1$ (although we do not actually assume that). 
    Consider $k_{j',x}$. By the maximality assumption on $j$, we have that $k_{j',x}\le \maxeff{x}$. Therefore, $\vec{c_{ux}}(s')\le \vec{c_{u}}(s')+\maxeff{x}$. This, however, contradicts the gap in $\vec{c_{ux}}$: we have that $s'\in V_{j'}$ and $s\in V_j$, but 
    \[
    \vec{c_{ux}}(s')-\vec{c_{ux}}(s)\le 2\maxeff{x} < 4\maxeff{xy}
    \]
    
    in contradiction to \cref{itm:separated infix gaps V} of \cref{def:separated repeating infix}.
    We conclude that $k_{j,x}\le \maxeff{x}$ for all $j$, and similarly for $y$.

    The case where $k_{j,x}<0$ is completely dual: 
    assume by way of contradiction that there exists $1\le j\le \ell$ such that $k_{j,x}<-\maxeff{x}$, and let $j$ be minimal with this property (i.e., the ``lowest'' part in the partition where this holds). 
    
    Consider a state $s\in \arg\min\{\vec{c_u}(q)\mid q\in V_j\}$. We then have $\vec{c_{ux}}(s)<\vec{c_u}(s)-\maxeff{x}$. 
    This means that the minimal seamless run $\rho:s'\runsto{x}s$ is from $s'\in V_{j'}$ with $j'<j$. Indeed, no run from any state in $V_j$ or higher can get to weight $\vec{c_u}(s)-\maxeff{x}$. In particular, we get $1<j$.
    We therefore have $\vec{c_{ux}(s)}\le \vec{c_{u}}(s')+\maxeff{x}$.

    Consider $k_{j',x}$. By the minimality assumption on $j$, we have that $k_{j',x}\ge -\maxeff{x}$. Therefore, $\vec{c_{ux}}(s')\ge \vec{c_{u}}(s')-\maxeff{x}$. This, however, contradicts the gap in $\vec{c_{ux}}$: we have that $s'\in V_{j'}$ and $s\in V_j$, but 
    \[
    \vec{c_{ux}}(s)-\vec{c_{ux}}(s')\le 2\maxeff{x} < 4\maxeff{xy}
    \]
    
    in contradiction to \cref{itm:separated infix gaps V} of \cref{def:separated repeating infix}.
    We conclude that $k_{j,x}\ge -\maxeff{x}$ for all $j$, and similarly for $y$.

    We now proceed to the latter two items.
    Let $1\le j\le \ell$ and $s\in V_{j}$. Assume by way of contradiction that for every $s'$ such that $s'\runsto{x}s$ we have $s'\in V_{j'}$ for $j'>j$. Fix $s'$ that attains $\minweight(s'\runsto{x}s)$.
    By the gap criterion \cref{itm:separated infix gaps V}, since the gap is much larger than $2\maxeff{x}$ we have
    $\vec{c_{ux}}(s)\ge \vec{c_u}(s')-\maxeff{x}>\vec{c_{u}}(s)+\maxeff{x}\ge \vec{c_{u}}(s)+k_{j',x}$.
    But this is a contradiction since $\vec{c_{ux}}(s)=\vec{c_u}(s)+k_{j',x}$.
    
    Finally, we address the last item. We first prove the case of $k=1$.
    Let $s'\in V_{j'}$ such that $s'\runsto{x}s$. 
    Assume by way of contradiction that $j'<j$, then by the gap criterion \cref{itm:separated infix gaps V}, since the gap is much larger than $2\maxeff{x}$, we have 
    $\vec{c_{ux}}(s)\le \vec{c_u}(s')+\maxeff{x}<\vec{c_{u}}(s)-2\maxeff{x}$, but this is a contradiction to $k_{j,x}\ge -\maxeff{x}$ proved above.
    Thus, $j'\ge j$, as required.
    
    We now generalise this to any $k\le |S|\bigM$ by observing that the proofs of the two latter items above (with $k=1$ for the last) assume a gap of at most $2\maxeff{x}$. Note that after reading $x$, the gaps change by at most $2\maxeff{x}$ for every pair of states. Since we start with a gap of $4|S|\bigM \maxeff{x}$, after reading $x^{k-1}$, the gaps are still greater than $2\maxeff{x}$, so we can apply this argument repeatedly.
\end{proof}

\begin{proposition}
\label{prop: Positive SRI no negative cycles on x}
    Consider a Positive SRI $uxyv$. 
    In the notations of \cref{def:separated repeating infix}, for every $r\in B$ and $k\in \bbN$ we have $\minweight(r\runsto{x^k} r)\ge 0$.
\end{proposition}
\begin{proof}
    The following proof is a slight simplification of the one from~\cite{almagor2026determinization}, but is mostly similar.
    Consider a run  
    $\rho:r=r_0\runsto{x}r_1\runsto{x}r_2\cdots r_{k-1}\runsto{x}r_k=r$.
    By \cref{itm:separated infix partition} of \cref{def:separated repeating infix}, for every $0\le i\le k$ we have $r_i\in V_{j_i}$ for some $1\le j_i\le \ell$. 
    By \cref{prop:run characteristics of SRI}, we have $j_{i-1}\ge j_i$ for all $1\le i\le k$, since there are only runs from $V_j$ to ``lower'' parts.
    Since $r_0=r_k$, it therefore follows that all the states must be in the same part $V_j$.

    Let $k_{j,x}$ be the corresponding constant to $V_j$. Observe that for every $r',r''\in V_j$ it holds that $\minweight(r'\runsto{x} r'')\ge \vec{c_u}(r'')-\vec{c_u}(r')+k_{j,x}$ (otherwise we would have $\vec{c_{ux}}(r'')<\vec{c_u}(r'')+k_{j,x}$).
    Thus, as a telescopic sum, and since $r_k=r_0=r$, we conclude:
    \[ 
    \begin{split}
        &\weight(\rho)\ge \sum_{i=1}^k \minweight(r_{i-1}\runsto{x} r_i)\ge  \sum_{i=1}^k (\vec{c_u}(r_i)-\vec{c_u}(r_{i-1})+k_{j,x}) \\
        &= \vec{c_u}(r_k)-\vec{c_u}(r_0)+k\cdot k_{j,x}=\vec{c_u}(r)-\vec{c_u}(r)+k\cdot k_{j,x}=k\cdot k_{j,x}\ge 0
    \end{split}
    \]
\end{proof}
We now have that in a Positive SRI, for every $k\in \bbN$, the minimal weight that can be attained by a run over $x^k$ cannot be too small (as the run cannot contain negative cycles). Specifically, we have the following.
\begin{corollary}
\label{cor:increasing infix no very negative runs on x k}
Consider a Positive SRI $uxyv$. For every $k\in \bbN$ and run $\rho:B\runsto{x^k} B$ we have $\weight(\rho)\ge -\maxeff{x}|S|$.
\end{corollary}

Our overall approach henceforth is the following: given an SRI, we try to make it nicer, where Stable SRI are the ``nicest'', followed by Positive SRI and then Negative SRI. This niceness manifests as the ability to shorten an infix of the SRI.
We first show that Stable SRIs are indeed at least as nice as Positive ones.
\begin{proposition}
    \label{prop:stable SRI is positive}
    Every Stable SRI is a Positive SRI.
\end{proposition}
\begin{proof}
    Consider a Stable SRI $uxyv$, so $(\ghostTrans(s_0,u),x)$ is a stable cycle. Let $1\le j\le \ell$, we first claim that $k_{j,x}\ge 0$. 
    Assume by way of contradiction that $k_{j,x}<0$, we claim that there exists $s\in V_j$, $k\in \bbN$ and a run $\rho:s\runsto{x^k}s$ such that $\weight(\rho)<0$. Indeed, by the gaps in \cref{itm:separated infix gaps V} of \cref{def:separated repeating infix}, for every $s\in V_j$ we have
    \[ \begin{split}
        &\xconf(\vec{c_{u}},x^{2|S|\maxeff{x}})(s)=\vec{c_u}(s)+2|S|\maxeff{x}k_{j,x} \\&<\vec{c_u}(s)-2|S|\maxeff{x}
    \end{split}
    \]
    That is, we can reach very low values of $s$. By \cref{prop:run characteristics of SRI} there exists $s'\in V_j$ and a run that attains this decrease, i.e., $\pi: s'\runsto{x^{2|S|\maxeff{x}}} s$ such that $\weight(\pi)<-2|S|\maxeff{x}$. In particular, by the pigeonhole principle within $\pi$ there is a negative cycle on some $x^n$ (otherwise the weight of $\pi$ can decrease to at most $|S|\maxeff{x}$). 
    
    This, however, is a contradiction to the assumption that $(\ghostTrans(s_0,u),x)$ is stable, as we then have a negative minimal reflexive state.
    We conclude that $k_{j,x}\ge 0$ for every $1\le j\le \ell$, so $uxyv$ is a Positive SRI.
\end{proof}
The most naive type of SRIs are Stable Degenerate ones. Here, the setting is very simple, and as expected -- they induce degenerate behaviours. Specifically, they manifest by $x$ causing the exact configuration to repeat, and we can therefore omit $x$ without any effect, as follows.
\begin{proposition}
    \label{prop:degenerate SRI x repeats config}
    Consider a Degenerate Stable SRI $w=uxyv$, then $\vec{c_{u}}=\vec{c_{ux}}$.
\end{proposition}
\begin{proof}
    Since $uxyv$ is Degenerate, we have that $(\ghostTrans(s_0,u),x)$ is a degenerate stable cycle. 
    Let $1\le j\le \ell$, we claim that $k_{j,x}=0$. 
    Since $uxyv$ is positive, then $k_{j,x}\ge 0$. 
    By \cref{prop:run characteristics of SRI} and the pigeonhole principle, there exists $q\in V_{j}$ such that $q\runsto{x^k}q$ for some $k\le |S|$.
    By \cref{def:degenerate stable cycle}, this implies that there exists $r\in \MinRefStates(\ghostTrans(s_0,u),x)$ such that $q\runsto{x^{2\bigM}}r\runsto{x^{2\bigM}}q$ where $\minweight(r\runsto{x^{2\bigM}}r)=0$.
    By applying \cref{prop:run characteristics of SRI} twice (once from $q$ to $r$ and once from $r$ to $q$), we have that $r\in V_j$.

    Assume by way of contradiction that $k_{j,x}>0$, then by the gaps in \cref{itm:separated infix gaps V} of \cref{def:separated repeating infix}, we have that
    \[\minweight(r\runsto{2\bigM}r)= r+2\bigM k_{j,x}>0\]
    which is a contradiction.
    We conclude hat $k_{j,x}=0$ for every $j$. In particular, for every $q\in B$ we have $\vec{c_{ux}}(q)=\vec{c_{u}}(q)+0=\vec{c_u}(q)$, and we are done.
\end{proof}
In particular, we can ``safely'' remove $x$ from a Degenerate SSRI or GSRI:
\begin{corollary}
    \label{cor:degenerate SRI remove x maintain potential and charge}
    Consider a Degenerate Stable SRI $w=uxyv$, then $\pot(uxyv)=\pot(uyv)$ and $\charge(uxyv)=\charge(uyv)$.
\end{corollary}

The next result is a central and new contribution of this work. 
Recall from \cref{lem: shift reflexive cycle to stable cycle} that we can turn a reflexive cycle into a stable one by shifting on the minimal-slope run. We now show that intuitively, the ``positive'' $V_j$'s do not survive such a shift, as they increase too much. Technically, the claim is slightly more delicate.

We show that for SRI, if $V_1,\ldots,V_{\ell'}$ all have $k_{j,x}\ge 0$ up to some ${\ell'}\le \ell$, then they do not have any finite-value runs on the cactus letter resulting from the shift. Concretely, this is later applied in two settings -- one where all the $k_{j,x}$ are non-negative, and one where $k_{1,x}\ge 0$, and we consider only $V_1$.
\begin{lemma}
    \label{lem: shift on negative run in SRI kills positive sets}
    Consider an SRI $uxyv$ where $k_{j,x}\ge 0$ for all $1\le j\le \ell'$ for some $\ell'\le \ell$.
    Assume there is a minimal-slope run $\rho:s\runsto{x^k}s$ with $s\in \ghostTrans(s_0,u)$, $k\le |S|$ and $\weight(\rho)<0$. 
    Then $(S'',x')=\baseshift{(\ghostTrans(s_0,u),x^k)}{\rho}$ is a stable cycle, and for every state $s'\in V_j$ with $j\le \ell'$ we have $\minweight(\baseshift{s'}{\rho}\runsto{\alpha_{S'',x'}}S)=\infty$.
\end{lemma}
\begin{proof}
Intuitively, we proceed in four steps as depicted in \cref{fig:shift on neg kills pos sets proof}. First, we baseline shift $x^k$ on $\rho$, resulting in a stable cycle $(S'',x')$ as per \cref{lem: shift reflexive cycle to stable cycle}. We then consider a state $s$ in some $V_j$ with $j<\ell'$, i.e., one of the ``low and non-negative'' $V_j$'s. We view the shifted version of $s$, and consider what happens when it reads $\alpha_{S'',x'}$. 
If there is a finite-weight run, this means that this state is grounded, and in particular can reach a $0$-cycle on $x'^\bigM$. However, by \cref{prop:run characteristics of SRI}, the only transitions are to \emph{lower} $V_{j}$, which are also non-negative by the premise. In addition, a $0$-cycle after the shift corresponds to a negative cycle before the shift. We therefore found some lower $V_{j'}$ with a negative cycle on $x^{k\bigM}$, which contradicts the non-negativity. 
    \begin{figure}[H]
        \centering
        \includegraphics[width=0.9\linewidth]{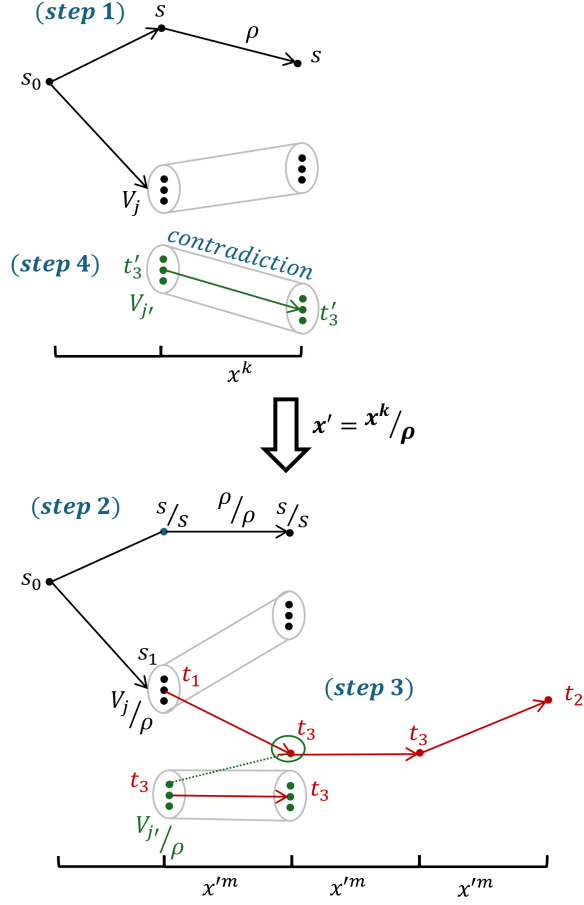}
        \caption{Proof of \cref{lem: shift on negative run in SRI kills positive sets}. In step 1 we consider the decreasing run $\rho$ over $x^k$. We baseline shift to it to obtain the figure on the right. In step 2 we consider some increasing $V_j$, and assume there is a finite-weight run from it (the red run). We then identify a state $t_3$ with a $0$-weight run. In step 3 we pull $t_3$ back to the left drawing ($t'_3$) so it becomes a negative run. In step 4 we reach a contradiction, since this negative run is lower than $V_j$.}
        \label{fig:shift on neg kills pos sets proof}
      
    \end{figure}
    
    We now turn to the precise details.
    Let $\rho:s\runsto{x^k}s$ with $s\in \ghostTrans(s_0,u)$ and $k\le |S|$ be a negative minimal-slope run.
    
    Note that $(\ghostTrans(s_0,u),x)$ is a reflexive cycle. Indeed, by \cref{def:separated repeating infix} we have $\booltrans(B,x)=B$ for $B=\booltrans(s_0,u)$, and by \cref{rmk:booltrans determines ghosttrans} the ghost-reachable states are determined by $B$.
    Let $(S'',x')=\baseshift{(\ghostTrans(s_0,u),x^k)}{\rho}$. By \cref{lem: shift reflexive cycle to stable cycle} $(S'',x')$ is a stable cycle. 
    
    Consider a state $s_1\in V_j$ with $j\le \ell'$, so that $k_{j,x}\ge 0$, let $t_1=\baseshift{s_1}{\rho}$ and assume by way of contradiction that $\minweight(t_1\runsto{\alpha_{S'',x}}t_2)<\infty$ for some $t_2\in S$. 
    Thus, by \cref{def:grounded pairs} we have $(t_1,t_2)\in \GroundPairs(S'',x')$, so there is some $t_3\in S''$ such that $t_1\runsto{x'^\bigM}t_3\runsto{x'^{\bigM}}t_3\runsto{x'^{\bigM}}t_2$, and $\minweight(t_3\runsto{x'^\bigM}t_3)=0$. Denote such a $0$-weight run $\tau:t_3\runsto{x'^\bigM}t_3$. 
    Intuitively, we reach a contradiction by baseline shifting $S''$ back to $S'$, so that $\tau$ becomes a negative run from $V_{j'}$ for some $j'<j$, contradicting $k_{j',x}\ge 0$.
    
    More precisely, denote by $\rhobase:s_b\runsto{x}s_b$ the baseline run of $(\ghostTrans(s_0,u),x)$, then $\rhobase^k$ is the baseline run of $(\ghostTrans(s_0,u),x^k)$. By \cref{rmk:baseline shift is right absorbing} we have $\baseshift{(S'',x')}{\rhobase^k}=(\ghostTrans(s_0,u),x^k)$. 
    Note that $\weight(\baseshift{\tau}{\rho^\bigM})=0$ since $\weight(\tau)=\weight(\rho^\bigM)=0$, and trivially $\weight(\baseshift{\rho^\bigM}{\rho^\bigM})=0$. 
    Since $\weight(\rho)=\weight(\baseshift{\rho}{\rhobase^k})$ (since $\rhobase^k$ is constantly $0$), then by \cref{cor:baseline shift maintains gaps} we also have
    \[\weight(\baseshift{\tau}{\rhobase^{k\bigM}})=\weight(\baseshift{\rho}{\rhobase^k})=\weight(\rho)<0\]
    But $\baseshift{\tau}{\rhobase^{k\bigM}}:s_3\runsto{x^k\bigM}s_3$ for some $s_3$ such that $s_1\runsto{x^{k\bigM}}s_3$ and therefore by \cref{prop:run characteristics of SRI} we have $s_3\in V_{j'}$ with $j'\le j\le \ell'$.

    But due to the gaps in \cref{itm:separated infix gaps V} of \cref{def:separated repeating infix} and since $k\le |S|$, we have $\minweight(s_3\runsto{x^{\bigM k}}s_3)\ge k_{j',x}\cdot k\cdot \bigM\ge 0$, since $k_{j',x}\ge 0$. This is a contradiction to $\weight(\baseshift{\tau}{\rhobase^{k\bigM}})<0$.

    We conclude that there is no transition of finite weight on $\alpha_{S'',x'}$ from any state in $V_1,\ldots, V_{\ell'}$.
\end{proof}

We proceed to study Stable SRI. Here, we show that we can ``bud'' a cactus letter from $x$, which allows us to shorten $x$ to a single letter, as well as entirely remove $y$, while not decreasing the potential and not increasing the charge. As usual, there is a caveat that there might be a type-0 witness. Specifically, we have the following.
\begin{lemma}[Cactus Budding on Stable SRI]
\label{lem:nondegenerate SRI budding increases potential decreases charge} 
    Consider a Non-degenerate Stable SRI $w=uxyv$ and let $w'=u\alpha_{\ghostTrans(s_0,u),x}v$, then the following hold:
    \begin{enumerate}
        \item For every $t\in S$ it holds that $\minweight(s_0\runsto{w} t)\le \minweight(s_0\runsto{w'} t)$.
        \item $\charge(w)\ge \charge(w')$.
        \item Either there is a type-0 witness, or $\pot(w)\le \pot(w')$.
    \end{enumerate}
\end{lemma}
\begin{proof}
Conceptually, the proof is an adaptation of the analogue in~\cite{almagor2026determinization}. Technically, however, the proof is supported by the careful choice of bounds in \cref{def:separated repeating infix}, which are different from those in~\cite{almagor2026determinization}.

Denote $S'=\ghostTrans(s_0,u)$, we therefore consider replacing the infix $xy$ with $\alpha_{S',x}$.
\paragraph*{Step 1: From $y$ to $x^*$}
The first step is to replace $y$ with many iterations of $x$. Intuitively, we imagine that $x$ is already replaced with $\alpha_{S',x}$, and we unfold it.

Consider the word $\unfold(u,\alpha_{S',x},v\wr F)=ux^{2\bigM M_0}v$ for some $F>2\maxeff{u\alpha_{S',x}yv}$. 
By \cref{rmk:increasing repetitions in unfolding} we can further assume that $M_0> \maxeff{xy}$.

We show the following properties:
\begin{itemize}
    \item For every $t\in S$ it holds that $\minweight(s_0\runsto{uxyv} t)\le \\\minweight(s_0\runsto{ux^{2\bigM M_0}v} t)$.
    \item $\charge(uxyv)\ge \charge(ux^{2\bigM M_0}v)$.
    \item $\pot(uxyv)\le \pot(ux^{2\bigM M_0}v)$.
\end{itemize}

We start with the first property.
The intuitive idea is as follows. We repeat the infix $x$ many times. Since $uxyv$ is in particular a Positive SRI (by \cref{prop:stable SRI is positive}), for each $V_j$ in the partition, this either causes all states to increase, or stay in place. However, this is only true as long as the $V_j$ sub-configurations stay far from one another. At some point, it could be that e.g., $V_2$ rises a lot, while $V_3$ stays in place, and when the weights ``cross'', the states in $V_3$ now become minimal. 
Now, while state in $V_3$ grow when reading $x$, it may still be the case that $V_3$ generates \emph{negative} runs to states in $V_2$. Before pumping $x$, this had no effect. Now, however, it generates lower runs.

Nonetheless, these negative runs are ``local'', since $V_3$ cannot generate any negative cycle in $V_2$ (otherwise $V_2$ would have a negative cycle itself). Since the gap we start with is huge, we are guaranteed that even after these local decreasing runs, the overall configuration increases in comparison to the first $x$.
Moreover, even when removing $y$, the negative effect this might create is not large enough to overcome the gaps we create by pumping $x$, so we end up with a superior configuration. This idea is depicted in~\cref{fig:inc infix cactus unfolding} 

An important point to stress is that the pumped runs of the $V_j$ sets might indeed cross each other, but if this happens, it must be so far above the original configuration, that it maintains the overall increase.

\begin{figure}[H]
    \centering
    \includegraphics[width=0.8\linewidth]{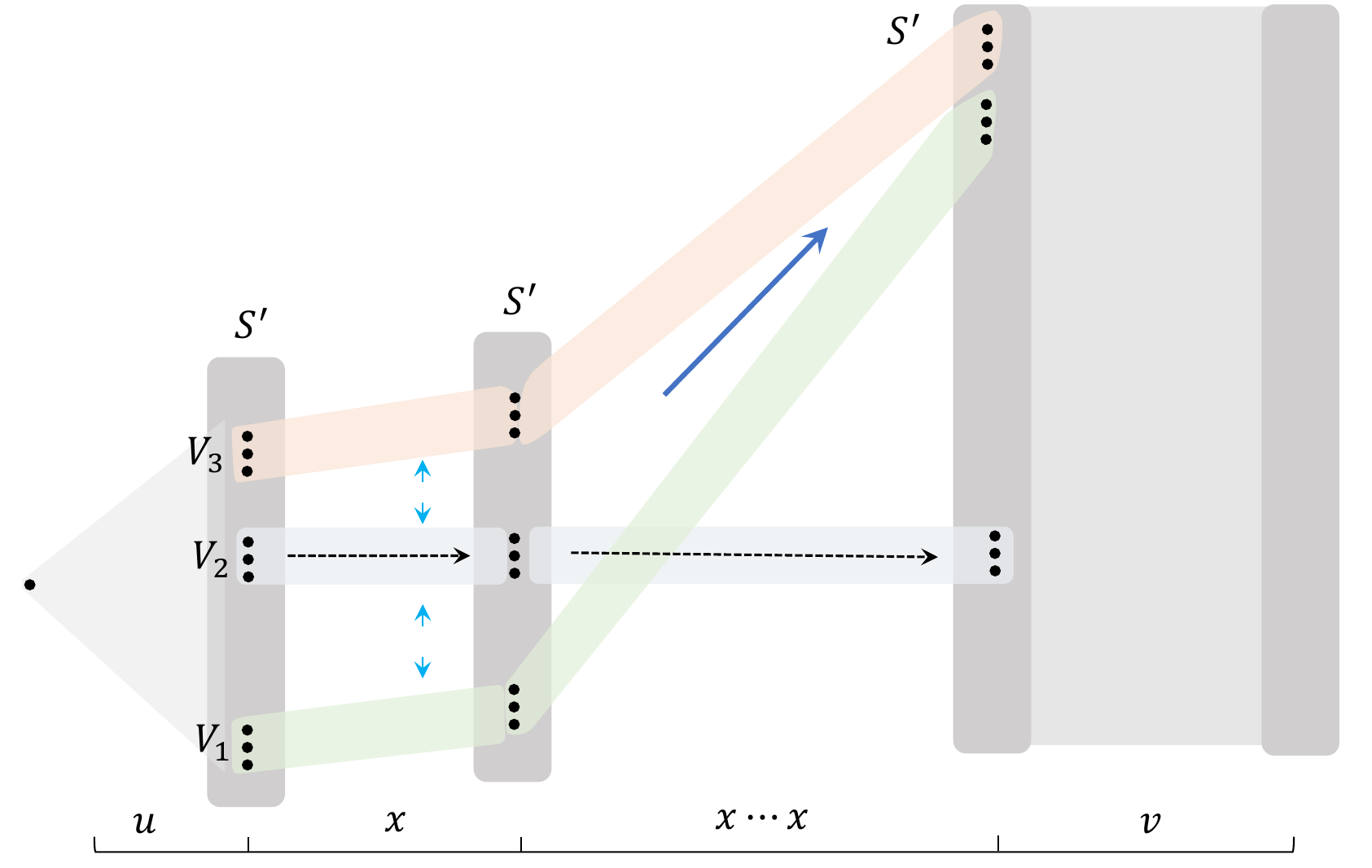}
    \caption{Replacing $y$ with $x^{2\bigM M_0}$. Visually, the increase seems ``steeper'', but it is the same slope as $x$, only much longer. Also notice that runs from $V_1$ eventually cross those of $V_2$ (and so some runs might be overtaken by $V_2$), but this happens way above the original $V_1$ states.}
    \label{fig:inc infix cactus unfolding}
\end{figure}

We now turn this wild hand-waving to a precise argument. 
We start by analysing runs on the words $uxy$ and $ux^{2\bigM M_0}$. Specifically, we show that 
$\xconf(s_0,uxy)\le \xconf(s_0,ux^{2\bigM M_0})$.
Using the notations of \cref{def:separated repeating infix}, 
since $\supp(\vec{c_u})=\supp(\vec{c_{ux}})=\supp(\vec{c_{uxy}})=B$, then by induction we also have that $B= \supp(\vec{c_{ux^{k}}})$ for any $k\in \bbN$, where $c_{ux^k}=\xconf(s_0,ux^k)$.

Consider some $r'\in B$, and assume $r'\in V_j$ then by \cref{itm:separated infix linear k} of \cref{def:separated repeating infix} we have that $k_{j,x},k_{j,y}$ satisfy  $k_{j,x}=0$ if and only if $k_{j,y}=0$ if and only if $k_{j,x}+k_{j,y}=0$. 

Let 
\[\rho:s_0\runsto{u}p_0\runsto{x}p_1\runsto{x}p_2\cdots p_{2\bigM M_0-1} \runsto{x} p_{2\bigM M_0}=r'\]
be a minimal-weight run. By \cref{prop:run characteristics of SRI}, for all $0\le i\le 2\bigM M_0$ we have that  $p_i\in V_{j'}$ for $j'\ge j$ 
(otherwise there is a transition from lower $V_{j''}$ to a higher one, which is a contradiction).
Denote $\vec{d_i}=\xconf(s_0,ux^i)$. We aim to show that restricted to $V_j$, the configuration $\vec{d_{2\bigM M_0}}$ is superior to $\vec{c_{uxy}}$.
Note that for every $0\le i<2\bigM M_0$ we have $\vec{d_{i-1}}\le \vec{d_i}$.
Indeed, by \cref{def:separated repeating infix} each $V_j$ can only (weakly) increase. This, however, is not enough to prove our claim, since once $i$ gets high enough, the gaps are no longer maintained and the $V_j$ sets may ``interfere'' with each other, and prevent certain states from strictly increasing. Then, removing $y$ may decrease things, which can be problematic.

We consider two cases.
\begin{itemize}
    \item If $p_i\in V_j$ for all $0\le i\le 2\bigM M_0$. In particular, the minimal run leading to each $p_i$ stems from $V_j$, when starting at $\vec{d_{i-1}}$. 
    We claim (by induction) that $\vec{d_i}(s)=\vec{c_u}(s)+ik_{j,x}$ for all $i$.
    
    The base case $i=0$ is trivial, since $\vec{d_0}=\vec{c_u}$. Assume correctness for $i$, we prove for $i+1$.
    By the induction hypothesis, we have that $\vec{d_i}(s)=\vec{c_u}(s)+ik_{j,x}$.
    By the observation above, we have that
    $\vec{d_{i+1}}(s)=\minweight_{\vec{d_i}}(V_j\runsto{x} s)$. 
    Since this expression depends only on states in $V_j$, then by the above we can restrict attention to $V_j$, and therefore write
    \[\vec{d_{i+1}}(s)=\minweight_{\vec{c_u}}(V_j\runsto{x}s)+ik_{j,x}=\vec{c_{ux}}(s)+ik_{j,x}=\vec{c_u}(s)+(i+1)k_{j,x}\]
    which completes the induction.

    Now, if $k_{i,j}=0$ then for every $s'\in V_j$ we have $\vec{d_{ 2\bigM M_0}}(s')=\vec{c_u}(s')=\vec{c_{ux}}(s')=\vec{c_{uxy}}(s')$ (since by \cref{def:separated repeating infix} we also have $k_{j,y}=0$). So the ``sub-configurations'' of $uxy$ and $ux^{2\bigM M_0}$ restricted to $V_j$ are equal. In particular, $\minweight(s_0\runsto{uxy} r')=\minweight(s_0\runsto{ux^{2\bigM M_0}} r')$.
    
    If $k_{j,x}>0$, then by our choice that $M_0> \maxeff{xy}$ and the inductive claim, we have 
    \[
    \vec{d_{2\bigM M_0}}(r')> \vec{c_u}(r')+2\bigM \maxeff{xy} k_{j,x}\ge \vec{c_u}(r')+\maxeff{xy}
    \]
    However, by \cref{prop:run characteristics of SRI} we have
    \[ \begin{split}
    &\vec{c_{uxy}}(r')\le \\&\vec{c_u}(r')+k_{j,x}+k_{k,y}\le \vec{c_u}(r')+\maxeff{x}+\maxeff{y}\\& = \vec{c_u}(r')+\maxeff{xy}
    \end{split}\]
    and we conclude that restricted to $V_j$, the sub-configuration of $ux^{2\bigM M_0}$ is superior to that of $uxy$. In particular, $\minweight(s_0\runsto{uxy} r')<\minweight(s_0\runsto{ux^{2\bigM M_0}} r')$.

    \item  The second case is when there exists $0\le i\le 2\bigM M_0$ such that $p_i\in V_{j'}$ for some $j'>j$. 
    Intuitively, in this case the run on $x^{2\bigM M_0}$ goes so high that it mingles with e.g., $V_{j+1}$. From there, it must remain much higher than any weight that $xy$ can attain starting from $V_j$ (and in particular the weight of reaching $r'$)
    
    By \cref{cor:increasing infix no very negative runs on x k}, for every $s\in V_{j''}$ with $j''>j$ and every $k\in \bbN$ we have $\minweight_{\vec{c_u}}(s\runsto{x^k} B)\ge \vec{c_u}(s)-\maxeff{x}|S|$. 
    As observed above, for every $i'$ we have $\vec{d_{i'-1}}\le \vec{d_{i'}}$, and therefore we in particular have 
    \[\minweight_{\vec{d_i}}(p_i\runsto{x^{2\bigM M_0-i}} B)\ge \vec{c_u}(s)-\maxeff{x}|S|\]
    but by \cref{itm:separated infix gaps V} of \cref{def:separated repeating infix} (for any $L\in \{\simpL,\generL\}$), for the maximal state $s'\in V_j$ we have
    \[
        \begin{split}
        &\vec{c_u}(s)-\maxeff{x}|S| > \\
        &\vec{c_u}(s')+4|S|\bigM\maxeff{xy} -\maxeff{x}|S| \gg \\
        &\vec{c_u}(s')+\maxeff{xy}\ge \\
        &\minweight_{\vec{c_u}}(V_j\runsto{xy} B)
        \end{split}
    \]
    Since $p_0\runsto{xy}r'$ and $p_0\in V_j$, it follows that $\minweight(s_0\runsto{uxy} r')< \minweight(s_0\runsto{ux^{2\bigM M_0}} r')$.

    We remark that this latter case cannot occur if $k_{j,x}=0$, but we do not use this fact.
\end{itemize}

Since the above works for every $j$ and every $r'$, we conclude that $\vec{c_{uxy}}\le \vec{d_{2\bigM M_0}}$. 
Therefore, we immediately have that
$\xconf(\vec{c_{uxy}},v)\le \xconf(\vec{d_{2\bigM M_0}},v)$.  
We can conclude the first property we want to prove, namely:
for every $t\in S$ it holds that $\minweight(s_0\runsto{uxyv} t)\le \minweight(s_0\runsto{ux^{2\bigM M_0}v} t)$.

Obtaining the second and third items is now easy. 
Observe that since both $x$ and $y$ are cycles on $B$, i.e., $\booltrans(B,x)=\booltrans(B,y)=B$, then 
\[\supp(\xconf(\vec{c_{uxy}},v))= \supp(\xconf(\vec{d_{2\bigM M_0}},v))\]
This brings us under the conditions of \cref{lem:potential and charge are monotone} (two support-equivalent configurations with a superiority relation and the same baseline). Applying the lemma with the suffix $v$, we have the required inequalities:
\[\charge(uxyv)=\charge(\xconf(\vec{c_{uxy}},v))\ge \charge(\xconf(\vec{d_{2\bigM M_0}},v))=\charge(ux^{2\bigM M_0}v)\]
\[\pot(uxyv)=\pot(\xconf(\vec{c_{uxy}},v))\le \pot(\xconf(\vec{d_{2\bigM M_0}},v))=\pot(ux^{2\bigM M_0}v)\]

\paragraph*{Step 2: from $x^*$ to $\alpha_{S',x}$}
Using the previous step, we now need to prove the following.
\begin{itemize}
    \item For every $t\in S$ it holds that $\minweight(s_0\runsto{ux^{2\bigM M_0}v}t)\le \\ \minweight(s_0\runsto{u\alpha_{S',x}v} t)$.
    \item $\charge(ux^{2\bigM M_0}v)\ge \charge(u\alpha_{S',x} v)$.
    \item $\pot(ux^{2\bigM M_0}v)\le \pot(u\alpha_{S',x} v)$.
\end{itemize}
Intuitively, we already repeat $x$ enough times so that all its behaviours are ``stabilised'', and can be replaced by $\alpha_{S',x}$.

Technically, all the work is already done by \cref{lem:unfolding maintains potential}, as follows.
Recall that $M_0$ is chosen so that $ux^{2\bigM M_0}v$ is an unfolding $\unfold(u,\alpha_{S',x},v \wr F)$ with $F$ large enough. 
We are therefore within the conditions of \cref{lem:unfolding maintains potential}. We then have from Item 1 therein that 
$\xconf(s_0,ux^{2\bigM M_0}v)\le \xconf(s_0,u\alpha_{S',x}v)$ (by applying the lemma with the suffix $v$). This in means that for every $t\in S$ we have 
\[
\begin{split}
&\minweight(s_0\runsto{ux^{2\bigM M_0}v} t)=\xconf(s_0,ux^{2\bigM M_0}v)(t)\le\\ 
&\xconf(s_0,u\alpha_{S',x}v)=\minweight(s_0\runsto{u\alpha_{S',x}v} t)
\end{split}
\]
concluding Item 1 of our proof.
This also readily implies that $\charge(ux^{2\bigM M_0}v)\ge \charge(u\alpha_{S',x} v)$, since the baseline remains the same, but the minimal run increases, concluding Item 2 of our proof.

Finally, still by \cref{lem:unfolding maintains potential} we have that either $\augA_\infty^\infty$ has a type-0 witness, or the potential inequality (Item 2 of our proof) follows immediately from Item 2 therein. 
\end{proof}

\subsection{SSRI -- A Toolbox}
\label{sec:SSRI toolbox}
We now present some SSRI-specific tools.
We first treat Negative SSRI.
As it turns out, these imply the existence of a type-1 witness. To show this we make heavy use of baseline shifts.
\begin{lemma}
\label{lem:negative SSRI to type 1 witness}
    If there exists a Negative SSRI, then there exists a type-1 witness.
\end{lemma}
\begin{proof}
    Consider a Negative SSRI $w=uxyv$ with the notations of \cref{def:separated repeating infix}. The proof proceeds in several steps, and we start with a high-level intuition.
    In Step 1 we show that the maximal dominant states at $\vec{c_u}$ and $\vec{c_{ux}}$ are the same. This is simple, since the ordering of the states induced by $\vec{c_u}$ and $\vec{c_{ux}}$ is the same.

    In Step 2 we start our path towards a witness by flattening the prefix $u$, and showing that we can assume this does not affect the potential and the dominant states. It does make the ghost states reachable, though.

    In Step 3 we consider the minimal slope run on the $x$ infix (rather, on any $x^k$), and baseline shift on this run, thus inducing a stable cycle $x'$ as per \cref{lem: shift reflexive cycle to stable cycle}. We also shift the prefix $u$ so that it glues properly to the shifted cycle. We then keep track again of the maximal dominant state, and show that it corresponds to the shift of the dominant state at $\vec{c_u}$, and we point to the suffix that causes it to be dominant.

    Finally, in Step 4 we show that the resulting tuple of (prefix,stable cycle,suffix) forms a type-1 witness. This requires a lot of point-of-view changes between the runs before and after the baseline shift. It also heavily relies on the gaps in the SRI definition.

    \paragraph*{Step 1: Identify the maximal dominant states.}
    From \cref{def:separated repeating infix} we have that $\pot(u)\le \pot(ux)$. We claim that $\maxdom(\vec{c_u})=\maxdom(\vec{c_{ux}})$. Indeed, observe that $\vec{c_u}$ and $\vec{c_{ux}}$ induce the same ordering on the state weights. That is, $\vec{c_u}(q)\le \vec{c_{u}}(q')$ if and only if $\vec{c_{ux}}(q)\le \vec{c_{ux}}(q')$. Since maximal dominance is determined only by the ordering (\cref{def:dominant state}), it follows that the maximal dominant states are equal.

    Fix such a maximal dominant state $s$ and let $V_d$ for some $1\le d\le \ell$ its subset in the partition of $B$ (so $s\in V_d$). 
    Notice that $k_{d,x}\ge 0$, since $\pot(u)\le \pot(ux)$ (and the baseline run remains at $0$). This becomes important in the following, when we shift the baseline to a negative run, making the dominant state strictly increase.
    
    \paragraph*{Step 2: Flatten $u$.}
    Recall that in a witness (\cref{def:witness}), the prefix is in $(\Gamma_0^0)^*$. We therefore flatten the prefix $u$, to obtain a prefix over this alphabet. 
    Let $u'=\flatten(u\wr 4\bigM\cdot \maxeff{w})$ and write $\vec{c_{u'}}=\xconf(s_0,u')$. We call upon \cref{cor:flatttening maintains potential} which gives one of two outcomes: either there is a witness of type-0, or $\pot(u)=\pot(u')$. In the former case we are done, since a type-0 witness is in particular type-1. We therefore proceed under the latter. 

    By \cref{lem:flattening essence} we have that $\vec{c_u}(q)=\vec{c_{u'}}(q)$ for every $q\in B$, and $\vec{c_{u'}}(q)-\vec{c_{u'}}(p)>2\bigM\maxeff{w}$ for all $q\in \ghostTrans(s_0,u)\setminus B$ and $p\in B$. 
    That is, flattening does not interfere with the weights of the $u$-reachable states ($B$), and the remaining $u'$-reachable states $\ghostTrans(s_0,u)\setminus B$ get very high weights.

    We claim that $\pot(ux)=\pot(u'x)$. Intuitively, the newly-reachable states $\ghostTrans(s_0,u)\setminus B$ on $u'$ cannot generate dominating states, as this would mean $\pot(u')>\pot(u)$. We make this precise.    
    
    Observe that all runs from $\ghostTrans(s_0,u)\setminus B$ on $x$ at $\vec{c_{u'}}$ attain weights much higher than runs from $B$. Thus, $s$ is still dominant in $\vec{c_{u'x}}=\xconf(\vec{c_{u'}},x)$. So we only need to show that it is maximal dominant. To this end, if $\hat{s}\in \ghostTrans(s_0,u)\setminus B$ is maximal dominant in $\vec{c_{u'x}}$ then $\hat{s}$ is also dominant in $\vec{c_{u'}}$ (by concatenating its separating suffix to $x$). But $\vec{c_{u'}}(\hat{s})>\vec{c_{u'}}(s)$, so $s$ is not maximal dominant in $\vec{c_{u'}}$, contradicting $\pot(u)=\pot(u')$.

    We thus have the same maximal-dominant state $s\in V_d$ in $\vec{c_u},\vec{c_{u'}},\vec{c_{ux}},\vec{c_{u'x}}$.

    \paragraph*{Step 3: shift $x$ to a stable cycle $x'$.} 
    Denote $S'=\ghostTrans(s_0,u)=\ghostTrans(s_0,u')=\booltrans(s_0,u')$.
    Since $uxyv$ is a Negative SRI there exists some $1\le i\le \ell$ such that $k_{i,x}<0$. We claim that $\minslope(S',x)<0$. Indeed, we can find a reflexive run $\rho:q\runsto{x^m}q$ with $m\le |S|$ and $q\in V_i$ and weight $mk_{i,x}<0$ by the pigeonhole principle.
    Note, however, that negative runs might also occur in $S'\setminus B$, and that the states there are not nicely partitioned as in $B$. 
    
    Let $\rho:s'\runsto{x^k}s'$ be a minimal slope run from $S'$ as per \cref{prop: reflexive cycle negative slope short} with $k\le |S|$, and denote $(S'',x')=\baseshift{(S',x^k)}{\rho}$. By \cref{lem: shift reflexive cycle to stable cycle} $(S'',x')$ is a stable cycle. 
    In order to properly ``glue'' $(S'',x')$ to $u'$, we need to shift $u'$ as well. Since $s'\in S'=\booltrans(s_0,u')$, there is a run $\mu:s_0\runsto{u'}s'$. We then obtain $u''=\baseshift{u'}{\mu}$, so $\booltrans(s_0,u'')=S''$ and we can consider $u''x'$. 

    Let $\vec{c_{u''}}=\xconf(s_0,u'')$ and $\vec{c_{u''x'}}=\xconf(s_0,u''x')$. We want to show that $s''=\baseshift{s}{s'}$ is maximal dominant in both configurations. For $\vec{c_{u''}}$ this readily follows from the fact that dominance and weight-difference within configurations are invariant under baseline shifts (\cref{prop: dominance invariant to baseline shifts}). For $\vec{c_{u''x'}}$, recall that $x'=\baseshift{x^k}{\rho}$ for some $k$. If we are lucky and $k=1$, then this also follows from \cref{prop: dominance invariant to baseline shifts}.    
    However, since we might have $1<k\le |S|$ we need to work slightly harder.

    To this end, it is enough to show that $s$ is maximal dominant at $\xconf(\vec{c_{u'}},x^k)$, i.e., before the shift, and then apply \cref{prop: dominance invariant to baseline shifts}. Fortunately, this is simple to show given the conditions of \cref{def:separated repeating infix}: recall that the gap between $V_j$ and $V_{j'}$ for $j\neq j'$ is at least $4\bigM \maxeff{xy}\gg 2|S|\maxeff{x}$. Thus, for every $s\in V_j$ we have
    $\minweight(S'\runsto{x^k}s)=\minweight(V_j\runsto{x^k}s)=\vec{c_{u'}}(s)+k\cdot k_{j',x}$
    due to \cref{itm:separated infix linear k} of \cref{def:separated repeating infix}.

    By~\cref{prop:run characteristics of SRI} we have $|k_{j,x}|\le \maxeff{x}$, and therefore $|k\cdot k_{j',x}|\le |S|\maxeff{x}\ll 2\bigM\maxeff{xy}$. Thus, the runs stemming from $V_j$ on $x^k$ do not ``interfere'' with any run from other $V_{j'}$. We conclude that $\xconf(\vec{c_{u'}},x^k)$ has the same weight ordering as $\vec{c_{u'}}$, and therefore $s$ is maximal-dominant.
    
    Thus, we have that $s''=\baseshift{s}{s'}$ is maximal-dominant in $\baseshift{\vec{c_{u''}}}{s'}$ and $\baseshift{\vec{c_{u''x'}}}{s'}$. In particular, by \cref{def:dominant state} there exists $z''\in (\simpCacRebJump)^*\genCacRebCac$ such that $\minweight(s''\runsto{z''}S)<\infty$, for every $r$ such that $\vec{c_{u''x'}}(r)< \vec{c_{u''x'}}(s'')$ we have $\minweight(r\runsto{z''}S)=\infty$, and $\vec{c_{u''x'}}(s'')$ is maximal with this property.

    \paragraph*{Step 4: $u'',\alpha_{S'',x'},z''$ is a type-1 witness.}
    Since $(S'',x')$ is a stable cycle, we can consider the cactus letter $\alpha_{S'',x'}$. Note that $x'\in \simpCacReb$, and therefore $\alpha_{S'',x'}\in \simpCacRebCac$. Thus, $u'',\alpha_{S'',x'},z''$ is in the correct ``format'' of a type-1 witness (see \cref{def:witness}). We claim that it is indeed a witness.

    First, the seamless run on $ux$ extends through the various baseline shifts to a seamless run on $u''x'$. Next, by the definition of $z''$ we have \[\minweight(s_0\runsto{u''x'z''}S)\le \minweight(s_0\runsto{u''x'}s''\runsto{z''}S)<\infty\]
    It remains to show that $\minweight(s_0\runsto{u''\alpha_{S'',x'}z''}S)=\infty$.
    
    Assume by way of contradiction that $\tau':s_0\runsto{u''}s'_1\runsto{\alpha_{S'',x'}}s'_2\runsto{z''}s'_3$ has $\weight(\tau)<\infty$.
    By definition, it holds that $(s'_1,s'_2)\in \GroundPairs(\alpha_{S'',x'})$. Thus (by \cref{def:grounded pairs}), there exists some state $r$ such that $s'_1\runsto{x'^\bigM} r'\runsto{x'^\bigM}s'_2$ and $\minweight(r'\runsto{x'^\bigM}r')=0$.
    Recall that $x'=\baseshift{x^k}{\rho}$. In particular, every run  $\eta':p'\runsto{x'}q'$ can be written as $\eta'=\baseshift{\eta}{\rho}$ for some run $\eta:p\runsto{x^k}q$ and $p'=\baseshift{p}{\rho}$ and $q'=\baseshift{q}{\rho}$ (note that the latter is because $\rho$ is a cycle, but essentially $q'$ is the shift of the \emph{last} state in $\rho$, not the first as per \cref{def: baseline shift on states sets and config}).

    Denote $\pi':r'\runsto{x'^\bigM}r'$ with $\weight(\pi')=0$, then we can write write $r'=\baseshift{r}{\rho}$ and $\pi'=\baseshift{\pi}{\rho^m}$ so that $\pi:r\runsto{x^{k\bigM}}r$. Since $\weight(\rho)<0$, we have $\weight(\pi)<0$ (by \cref{cor:baseline shift maintains gaps}).
    We now consider the baseline shift of the $V_i$ sets: denote $V'_i=\baseshift{V_i}{\rho}$ for every $1\le i\le \ell$, as well as $C=\baseshift{\ghostTrans(s_0,u')\setminus B}{\rho}$ the shift of the ghost states. By \cref{prop: dominance invariant to baseline shifts} the gaps between these shifted sets are maintained in $\vec{c_{u''}}$. 

    We now split to cases according to the relative weight $\vec{c_{u''x'}}(s'_2)$ and $\vec{c_{u''x'}}(s'')$.
    \begin{itemize}
        \item If $\vec{c_{u'x''}}(s_2)<\vec{c_{u'x''}}(s'')$ then        
        by the dominance of $s''$ we have that $\minweight(s_2\runsto{z''}S)=\infty$, in contradiction to $\weight(\tau')<\infty$.
        \item If $s_2\in V'_d$ (recall that $s''\in V'_d$ as $s\in V_d$, where $s$ is the corresponding maximal dominant state in $\vec{c_{ux}}$), then $s_2\in V_d$. By \cref{prop:run characteristics of SRI}
        we have $s_1\in V_d$. Since $s_1\runsto{x^{k\bigM}}r$ and $r\runsto{x^{k\bigM}}s_2$, we have $r\in V_d$ as well (otherwise either $s_1$ or $s_2$ would be in some ``lower'' $V_i$).
        However, $\minweight(r\runsto{x^{k\bigM}}r)\le \weight(\pi)<0$, which contradicts our observation from Step 1 that $k_{d,x}\ge 0$.

        Intuitively, this can be summed as: if $s_2$ is very close to the dominant state, then it cannot belong to a grounded pair, since $V'_d$ is strictly increasing.

        \item If $s_2\in V'_j$ for some $j>d$ (i.e., ``way above'' $s''$), assume with out loss of generality that it has the minimal value $\vec{c_{u''x'}}(s_2)$ with the property that $\minweight(s_2\runsto{z''}S)<\infty$. Then $s_2$ is dominant, which contradicts the fact that $s''$ is maximal dominant in $\vec{u''x'}$.
    \end{itemize}

    We conclude that $\minweight(s_0\runsto{u''\alpha_{S'',x'}z''}S)=\infty$, so $u'',\alpha_{S'',x'},z''$ is a type-1 witness, and we are done.
\end{proof}

We now turn our attention to Positive SSRI. 
We observe that a Positive SSRI does not allow negative cycles on $x^k$ for any $k$. Note that negative non-cyclic runs on $x$ may occur. For example, a state in $V_7$ may have a negative run on $x$, but only to a state in e.g., $V_3$, which is much lower, and therefore this negative run is not seamless, and thus does not really ``matter'' (as long as $x$ is not pumped too much, which becomes a concern in \cref{lem:pos SSRI to stable}).

The following central lemma is our main tool when using Positive SSRI. Intuitively, it shows that a Positive SSRI in fact induces a stable cycle on $x$, and is therefore a Stable SSRI. This relies on our novel \cref{lem: shift on negative run in SRI kills positive sets}.
A caveat is that the above might not hold, but in such a case we have a type-0 witness, in the sense of \cref{cor:flatttening maintains potential}.

\begin{lemma}[From Positive SSRI to Stable SSRI]
\label{lem:pos SSRI to stable} 
    Consider a Positive SSRI $w=uxyv$, then either $(\ghostTrans(u),x)$ is a stable cycle or there is a type-0 witness.
\end{lemma}
\begin{proof}
    We start with an intuitive overview: if $(\ghostTrans(s_0,u),x)$ is not a stable cycle, then there is some negative-slope run on some $x^k$. However, since $x$ is part of a Positive SSRI it does not admit such a reachable cycle. 
    Thus, this cycle must stem from a ghost state. 
    By applying \cref{lem: shift on negative run in SRI kills positive sets} we obtain a stable cycle where none of the reachable states have finite-weight transitions on this new cactus letter. Intuitively, we can now flatten $u$ and get arbitrarily high dominant states. Then, combined with \cref{cor:flatttening maintains potential} 
    this is a 0-type witness.

    Assume that $(\ghostTrans(s_0,u),x)$ is not a stable cycle (otherwise we are done).  
    Note that $(\ghostTrans(s_0,u),x)$ is a reflexive cycle. Indeed, by \cref{def:separated repeating infix} we have $\booltrans(B,x)=B$ for $B=\booltrans(s_0,u)$, and by \cref{rmk:booltrans determines ghosttrans} the ghost-reachable states are determined by $B$.
    
    Let $\rho:s\runsto{x^k}s$ with $s\in \ghostTrans(s_0,u)$ and $k\le |S|$ be a minimal-slope run as per \cref{prop: reflexive cycle negative slope short}. 
    By \cref{lem: shift reflexive cycle to stable cycle} we have $\slope(\rho)< 0$ (otherwise $(\ghostTrans(s_0,u),x)$ is already a stable cycle). Therefore, $\weight(\rho)<0$. 
    By \cref{prop: Positive SRI no negative cycles on x}, this means that $s\notin B$. That is, $s\in \ghostTrans(s_0,u)\setminus \booltrans(s_0,u)$, so there are ghost states after reading $u$. Denote $G_u=\ghostTrans(s_0,u)\setminus B$.

    By \cref{lem: shift on negative run in SRI kills positive sets} we have that $(S'',x')=\baseshift{(\ghostTrans(s_0,u),x^k)}{\rho}$ is a stable cycle, and that for every $s\in B$ we have \\$\minweight(\baseshift{s}{\rho}\runsto{\alpha_{S'',x'}}S)=\infty$. Indeed, this follows since the SSRI is positive and $B=V_1\cup \ldots\cup V_\ell$ (i.e., $\ell'=\ell$ in the notation of \cref{lem: shift on negative run in SRI kills positive sets}). Denote $B'=\baseshift{B}{\rho}$.
    
    By \cref{def:separated repeating infix} we have $|x|\le \frac1\bigM \simpL(\depth(x)+1)$ and since $k\le |S|\le \bigM$ we have $|x'|=|x^k|\le \simpL(\depth(x)+1)$. In addition, we have $\depth(x')=\depth(x)$ and since $x\in \simpCac^*$ then due to the baseline shift, $x'$ may contain rebase letters, and therefore $x'\in \simpCacReb$. Thus, $\alpha_{S'',x'}\in \simpCacRebCac$.

    We now unfold the prefix $u$ as follows: consider the word $u'=\flatten(u\wr \maxeff{w}\cdot 4\bigM)$, and denote $\vec{c_u'}=\xconf(s_0,u')$. Then $G_u\cup B=\supp(\vec{c_u'})$ and by \cref{lem:flattening essence} for every $s'\in B$ and $s''\in G_u$ we have $\vec{c_u'}(s'')-\vec{c_u'}(s')> \maxeff{w}\cdot 2\bigM$. Put simply -- all states in $G_u$ are much higher than those in $B$.
    We now employ \cref{cor:flatttening maintains potential} with two possible outcomes:
    if $\augA_\infty^\infty$ has a type-0 witness, we conclude the proof. 
    Otherwise, we have that $\pot(u)=\pot(u')$.     
    This, in turn, means that every $s''\in G_u$ is not dominant at $\vec{c_u'}$. Indeed, we have $\pot(u)\le \maxeff{u}$, but $\vec{c_u'}(s'')>\maxeff{w}\cdot 2\bigM$, i.e., is far too high.

    We now reach a contradiction to the above by finding a dominant state in $G_u$, as follows.
    Consider the word $u'\cdot \jl_{s_b\to s}\cdot \alpha_{S'',x'}$. Observe that
    $\weight(s_0\runsto{u'}s\runsto{\jl_{s_b\to s}}t\runsto{\alpha_{S'',x'}}t)<\infty$ where $t$ is the baseline-shifted version of $s$ on $\rho$ (i.e., $\baseshift{\rho}{\rho}:t\runsto{x'}t$). Indeed, $s$ is reachable from $s_0$ on the flattened $u'$, then the jump letter has weight $0$, and $\baseshift{\rho}{\rho}$ yields the transition $t\runsto{\alpha_{S'',x'}}t$. 
    Recall, however, that $s\in G_u$ and that there are no transitions from $B'$ on $\alpha_{S'',x'}$. This means that there is no transition from $B$ on $\jl_{s_b\to s}\cdot \alpha_{S'',x'}$ (since the jump letter exactly takes us from $B$ to $B'$).
    We conclude that either $s$, or some lower state in $G_u$, is dominant, which is a contradiction.
\end{proof}

\subsection{GSRI -- A Toolbox}
\label{sec:GSRI toolbox}
We continue by developing some GSRI-specific tools.

We start by observing that for GSRI, the minimal set in the partition (namely $V_1$, whose in particular responsible for $\charge$), does not generate negative cycles on $x$. This can be viewed as a counterpart of \cref{prop: Positive SRI no negative cycles on x}.
\begin{proposition}
\label{prop: general SRI no negative cycles on x V0}
    Consider a GSRI $uxyv$. 
    In the notations of \cref{def:separated repeating infix}, for every $r\in V_1$ and $k\in \bbN$ we have $\minweight(r\runsto{x^k} r)\ge 0$. In particular $k_{1,x}\ge 0$.
\end{proposition}
\begin{proof}
    Assume by way of contradiction there is some $r\in V_1$ and $k\in \bbN$ such that $\minweight(r\runsto{x^k} r) < 0$.
    By the pigeonhole principle there is some $k' \le |S|$ such that $\minweight(r\runsto{x^{k'}} r) < 0$.
    However, observe that in a GSRI, the minimal runs after $u$ and after $ux$ end in a state in $V_1$. Since $\charge(u)\ge \charge(ux)$ we have $k_{1,x}\ge 0$. Therefore, we have 
    $\xconf(\vec{c_u},x^{k'})(r)= \vec{c_u}+k' k_{1,x}\ge \vec{c_u}$, so $\minweight(r\runsto{x^{k'}} r) \ge  0$, in contradiction.
\end{proof}

The next lemma is the central property of GSRI, and is the counterpart of \cref{lem:pos SSRI to stable}. Conceptually, it is what links the charge to the potential.
\begin{lemma}[From GSRI to Stable GSRI or High Potential]
\label{lem:GSRI to stable or high potential} 
    Consider a GSRI $w=uxyv$, then either $(\ghostTrans(s_0,u),x)$ is a stable cycle or there is a word $w'\in (\Gamma_0^0)^*$ such that $\pot(w') \ge \budH$.
\end{lemma}
\begin{proof}
    We start with an intuitive overview: if $(\ghostTrans(s_0,u),x)$ is not a stable cycle, then there is some minimal negative-slope run $\rho$ on some $x^k$. By \cref{prop: general SRI no negative cycles on x V0} this cycle cannot stem from a $V_1$ state, so it stems from a higher (possibly ghost) state. 
    By \cref{lem: shift on negative run in SRI kills positive sets} we can shift on $\rho$ to obtain a cactus letter $\alpha_{S'',x'}$ that has no finite weight transition from $V_1$.
    
    By the gap property (\cref{itm:separated infix gaps V} in \cref{def:separated repeating infix}) this state is \emph{much} higher than $V_1$.
    We flatten $u$ to a word $u'$ such that the ghost states are also much higher than $\budH$ above $V_1$.

    Using a jump letter to $S''$ and $\alpha_{S'',x'}$ as a suffix, we get that the maximal dominant state after reading $u'$ is not in $V_1$, and is therefore more than $\budH$ above it. Finally, we can shift so that the baseline run is in $V_1$, and we get potential greater than $\budH$.
    \begin{figure}[H]
        \centering
        \includegraphics[width=0.9\linewidth]{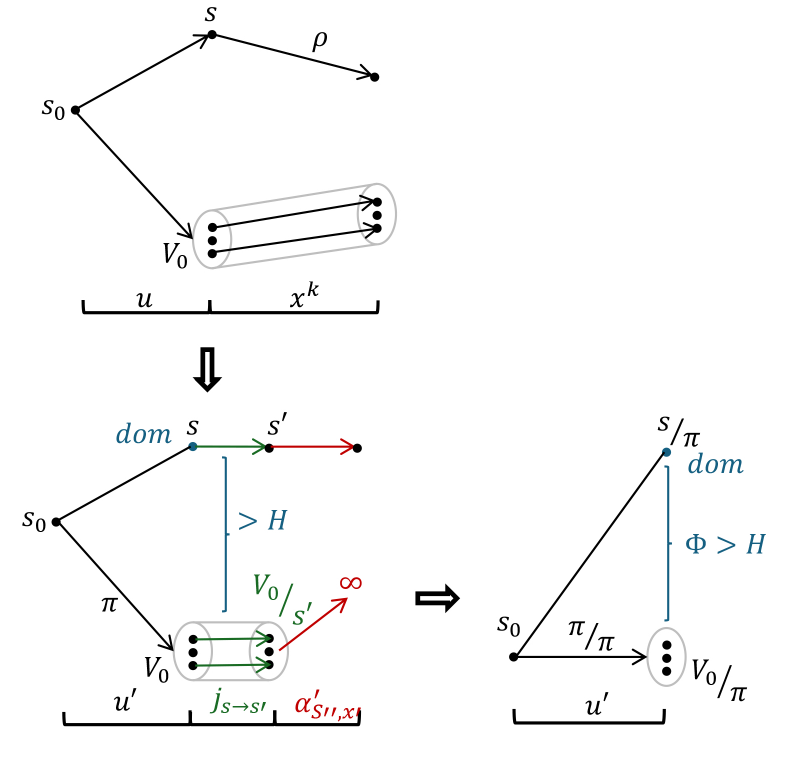}
        \caption{Proof outline of \cref{lem:GSRI to stable or high potential}. We start with the negative run $\rho$ and baseline-shift on it. This casuses $V_1$ to jump to $\infty$ (by \cref{lem: shift on negative run in SRI kills positive sets}. The GSRI gaps are larger than $\budH$, so by baseline-shifting again to a minimal run to $V_1$ and using a jump letter, we obtain a high-potential configuration.}
        \label{fig:pos GSRI to high potential}
    \end{figure}

    Assume that $(\ghostTrans(s_0,u),x)$ is not a stable cycle (otherwise we are done).  Therefore, by \cref{lem: shift reflexive cycle to stable cycle} there is some $s\in \ghostTrans(s_0,u)$ and a minimal-slope run $\rho:s\runsto{x^k}s$ such that $k\le |S|$  and $\weight(\rho)<0$.
    By \cref{prop: general SRI no negative cycles on x V0} we have $k_{1,x}\ge 0$, so we can apply \cref{lem: shift on negative run in SRI kills positive sets} with $\ell'=0$. We obtain the following: $(S'',x')=\baseshift{(\ghostTrans(s_0,u),x^k)}{\rho}$ is a stable cycle and for every $s\in V_1$ we have $\minweight(\baseshift{s'}{\rho}\runsto{\alpha_{S'',x'}}S)=\infty$.

    By \cref{def:separated repeating infix} we have $|x|\le \frac1\bigM \generL(\depth(x)+1)$ and since $k\le |S|\le \bigM$ we have $|x^k|\le \generL(\depth(x)+1)$, and keep in mind that $x\in \genCac^*$. Note that $x'=\baseshift{x^k}{\rho}$, so $\depth(x')=\depth(x)$, and by the above $|x'|\le \generL(\depth(x')+1)$. 
    Due to the baseline shift, $x'$ may contain rebase letters, and therefore $x'\in \genCacReb^*$. 
    Thus, $\alpha_{S'',x'}\in \genCacRebCac$.

    We now flatten $u$ as to $u'=\flatten(u\wr 2\budH\cdot 4\bigM\cdot \maxeff{w})$. Denote $\vec{c_{u'}}=\xconf(s_0,u')$ and observe that $\supp(\vec{c_{u'}})=\ghostTrans(s_0,u)$. Moreover, for every $r\in V_1$ and state $r_1\notin V_1$ we have that $\vec{c_{u'}}(r_1)-\vec{c_{u'}}(r)>\budH$. Indeed, if $r_1\in V_j$ for some $j>1$ then this holds by the gap criterion (\cref{itm:separated infix gaps V} in \cref{def:separated repeating infix}), and if $r_1\in \ghostTrans(s_0,u)\setminus B$ then this holds due to the flattening constants (i.e., \cref{lem:flattening essence}). 
    Let $r_b$ be the baseline state of $B$ (namely before the shift on $\rho$, which changes the baseline to $s$).
    Consider the word $u'\cdot \jl_{r_b\to s}$, then $\ghostTrans(s_0,u'\jl_{r_b\to s})=S''$, since we reach $\ghostTrans(s_0,u)$ after reading $u'$, and then the jump letter simply changes the baseline component without any other effect.

    But then we have that 
    $\minweight(s_0\runsto{u'\jl_{r_b\to s}}V_1 \runsto{\alpha_{S'',x'}}S)=\infty$, whereas \[
    \begin{split}
    &\minweight(s_0\runsto{u'\jl_{r_b\to s}}S'' \runsto{\alpha_{S'',x'}}S) \\
    &\le \minweight(s_0\runsto{u'}s\runsto{\jl_{r_b\to s}}\baseshift{s}{s} \runsto{\alpha_{S'',x'}} S)<\infty
    \end{split}
    \]
    since $\rho:s\runsto{x^k}s$, so $\baseshift{\rho}{\rho}:\baseshift{s}{s}\runsto{x'}\baseshift{s}{s}$, and in particular we can read $\alpha_{S'',x'}$ from  $\baseshift{s}{s}$.
    In particular, since all states above $V_1$ are more than  $\budH$ above it, for the maximal dominant state $r'$ in $\vec{c_{u'}}$ we have $\vec{c_{u'}}(r')-\max\{\vec{c_{u'}}(r)\mid r\in V_1\}>\budH$. Crucially, note that this suffix is legitimate for the definition of potential (\cref{def:potential}), since as we argue above we have $\alpha_{S'',x'}\in \genCacRebCac$.

    Finally, we perform a baseline shift on a seamless run $\pi:s_0\runsto{u'}r$ to the maximal state in $V_1$, then by \cref{prop: dominance invariant to baseline shifts} we have $\pot(\baseshift{u'}{\pi})>\budH$, and we are done.
\end{proof}

\section{Existence of SSRI}
\label{sec:existence of SSRI}
The tools we develop in \cref{sec:separated repeating infix} show that various useful conclusions can be reached if an SRI is found. Our next challenge is to show that SRI can indeed be found, under suitable conditions. In this section show this for SSRI. In \cref{sec:existence of GSRI} we do the same for GSRI. We remark that the structure and technical details of both sections is almost identical, with the main exception being the bounded-growth property of $\pot$ versus the absence of it for $\charge$. However, due to many other tiny differences, we opt for duplicating the arguments, rather than presenting a unified but heavily-parametrised version that captures both.

The high-level approach is the following. We consider a word decomposed as $w_1w_2w_3$, with $\pot(w_1)\le \pot(w_1w_2)$.
We then separate to two cases: either the potential increases a lot on $w_2$ (dubbed \emph{high amplitude}), or it stays in some bounded ``band'' (dubbed \emph{bounded amplitude}). 
By placing some constraints on the depth and length of $w_2$. We show that if it is long enough, but not too long, we can decompose it into many smaller infixes with a similar behaviour to the original $w_2$, but with respect to shorter length. 

In tandem, we introduce a new quantity into our analysis -- a set $I$ of seamless runs that are far away from each other, called \emph{independent runs}. Specifically, we focus on the size $i=|I|$ and use it to parametrise our definitions (this is the same $i$ that appears already in the functions of \cref{sec:effective cactus}).

We now ask whether every run along $w_2$ remains close to some run in $I$. If this is the case, then we have a \emph{cover}, and we use the large gaps between the runs in $I$ to induce an SSRI. Otherwise, some run gets very far from all the runs in $I$, and we use it to induce a new independent run, thus increasing $i$ and going one step lower in the induction. 
The latter step involves very careful management of the length, depth, and amplitude of the word, and the size of $i$. It is where all the pieces of the recursive functions of \cref{sec:effective cactus} fit in.

For SSRI, we focus on words over $\simpCac$ decomposed as $w_1w_2w_3$. The intuition behind this decomposition is the following: $w_1$ is a prefix leading to some set of states with very large gaps between some of them (possibly). 
$w_2$ is a long infix that maintains a separation between independent runs. $w_3$ is just a suffix.

\begin{remark}[Implicit Decomposition to $(w_1,w_2,w_3)$]
    \label{rmk:w1w2w3 is a decomposition}
    When we write $w_1w_2w_3$ we assume that the decomposition is explicit (instead of writing $(w_1,w_2,w_3)$). 
\end{remark}

We start with the technical details. The following definitions pertain to any word, and are reused in \cref{sec:existence of GSRI}. They introduce the concept of independent runs, and measure how close these runs are to the other states.
\begin{definition}[Independent runs with gap $G$]
    \label{def:independent run}
    Consider a word $w_1w_2w_3$ and two seamless runs $\rho_1:s_0\runsto{w_1}p_1\runsto{w_2}q_1$ and $\rho_2:s_0\runsto{w_1}p_2\runsto{w_2}q_2$. For $G\in \bbN$ we say that $\rho_1$ and $\rho_2$ are \emph{independent with respect to $w_2$ with gap $G$} if for every prefix $u$ of $w_2$ we have $|\weight(\rho_1[{w_1u}])-\weight(\rho_2[{w_1u}])|\ge G$.

    A set $I$ of seamless runs on $w_1w_2w_3$ is \emph{independent with gap $G$} if every two distinct runs in $I$ are independent with respect to $w_2$ with gap $G$. We denote such a set with its natural ordering $I=\{\rho_1<\rho_2<\ldots<\rho_i\}$.
\end{definition}

\begin{definition}[States near an independent run]
    \label{def:near run}
    Consider a word $w=xy$ and a seamless run $\rho:s_0\runsto{x}s_1\runsto{y}s_2$.
    Let $b\in \bbN$. Denote $\vec{c_x}=\xconf(s_0,x)$ and define
    \[\near_b(\rho,x)=\{(s,k)\mid k=|\vec{c_x}(s)-\vec{c_{x}}(s_1)|\le b\}\]
We usually think of $\near_b$ as a sub-configuration $\{-b,\ldots,b,\bot\}^{S}$ where $\bot$ means that a state is either not reachable, or far from $\rho$.
\end{definition}
Given some infix $y$ and a run $\rho$ on it, we now look at the behaviour of the states before and after reading $y$ with respect to how close they are to $\rho$, as well as the overall tendency of $\rho$ (increasing, decreasing, or with weight $0$). This is later used to find repeating infixes.
\begin{definition}[$\difftype$]
    \label{def:diff type}
    Consider a word $w=xy$ and a run $\rho:s_0\runsto{xy}s_1$. Let $b\in \bbN$. 
    We define 
    \[\difftype_b(\rho,x,y)=(\near_b(\rho,x),\sign(\weight(\rho[y])),\near_b(\rho,xy))\]
    We lift this to a set of independent runs $I=\{\rho_1<\ldots<\rho_i\}$ by defining 
    \[\difftype_b(I,x,y)=(\near_b(\rho_j,x),\sign(\weight(\rho_j[y])),\near_b(\rho_j,xy))_{j=1}^i\]
        
    We denote the set of all possible $\difftype_b$ vectors (over a set of runs) by $\difftypeset_b$.
\end{definition}
Intuitively, $\difftype_b(I,x,y)$ contains the following information: for each $\rho\in I$ what is the relative weight of the states near $\rho$ before and after reading $y$, and whether $\rho$ itself increased or decreased upon reading $y$. 
\begin{remark}[Size of $\difftypeset$]
\label{rmk:size of difftypeset}
    Let $i=|I|$ and observe that 
    \[\difftype_b(\rho,x,y)\in \{-b,\ldots,b,\bot\}^S\times \{-1,0,1\}\times \{-b,\ldots,b,\bot\}^S\]
    In particular, we have 
    \[|\difftypeset_b(I,x,y)|\le ((2b+2)^{|S|}\cdot 3\cdot (2b+2)^{|S|})^i= 3^i\cdot (2b+2)^{2|S|i}\] 

    Note that for $b=\simpLfuncCover{d+1}{i}$ it holds that $|\difftypeset_b(I,x,y)|\le \simpLfuncTypes{d+1}{i}$. Similarly, for $b=\generLfuncCover{d+1}{i}$ we have \\ $|\difftypeset_b(I,x,y)|\le \generLfuncTypes{d+1}{i}$.
\end{remark}

We now focus on SSRI, and in particular work over the alphabet $\simpCac$ and focus on $\pot$.
Throughout this section, we mainly look at words with non-decreasing potential:
\begin{definition}[$\pot$-Increasing Word]
\label{def:potential increasing word}
A word $w_1w_2w_3$ is \emph{$\pot$-increasing} if there is a seamless baseline run on $w_1w_2w_3$ and $\pot(w_1)\le \pot(w_1w_2)$.
\end{definition}
Our next definition limits words to a certain length based on their depth and the number of independent runs associated with them by a given set $I$. 
\begin{definition}[$(d,i)$-Fit word]
    \label{def:d i fit}
    Consider a word $w_1w_2w_3\in (\simpCac)^*$ and a set $I$ of $G$-independent runs. The word $w_1w_2w_3$ (with $I$) is \emph{$(d,i)$-fit} if the following hold.
    \begin{itemize}
        \item $w_1w_2w_3$ is $\pot$-increasing.
        \item $\depth(w_2)\le d$.
        \item $|w_2|\le \frac{1}{\extraSize}\simpLfuncLength{d+1}{i}$ where $i=|I|$.
    \end{itemize} 
\end{definition}
We now specialise $(d,i)$-fit words to ones with certain significant properties relating to the behaviour of $\pot$ and the gaps between the runs in $I$.
\begin{definition}[$(d,i)$ Well-Separated and High/Bounded\\-Amplitude word]
\label{def:well separated word}
\label{def:high amplitude word}
\label{def:bounded amplitude word}
A word $w_1w_2w_3\in (\simpCac)^*$ with set $I$ is \emph{$(d,i)$ well separated} if $w_1w_2w_3$ is $(d,i)$-fit and $I$ has gap at least
\[\frac12 \simpLfuncCover{d+1}{i-1}=4\bigM^2(\simpLfuncMaxW{d}\simpLfuncLength{d}{1}\simpLfuncLength{d+1}{i}+\simpLfuncCover{d+1}{i})\] 
If $w_1w_2w_3$ is well separated, then it is:
\begin{itemize}
    \item \emph{$(d,i)$ high amplitude} if $\pot(w_1w_2)-\pot(w_1)>\simpLfuncAmp{d+1}{i}$.
    \item \emph{$(d,i)$ bounded amplitude} if $|\pot(w_1x)-\pot(w_1)|\le \simpLfuncAmp{d+1}{i}$ for every prefix $x$ of $w_2$, and $|w_2|=\frac{1}{\extraSize}\simpLfuncLength{d+1}{i}$.
\end{itemize}
\end{definition}
Note that by \cref{lem:general bounded growth potential} we know the potential increases in a bounded fashion, and therefore high-amplitude words cannot be very short. However, bounded amplitude a-priori holds also for short words, hence the length restriction in \cref{def:bounded amplitude word}.

Our next batch of definitions pertain to decompositions of the $w_2$ part of a word into subwords. They are depicted in \cref{fig:decompositions}

\begin{figure}[H]
  \begin{center}
    \includegraphics[width=0.9\textwidth]{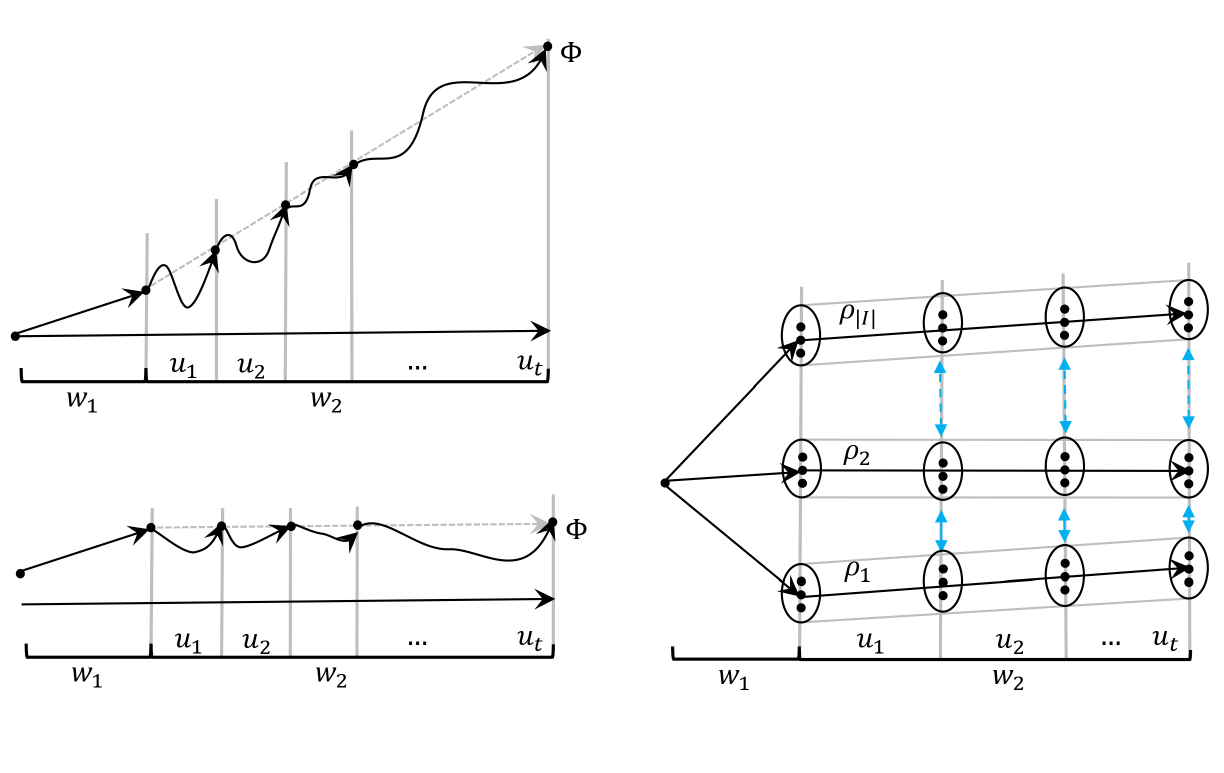}
  \end{center}
  \caption{In a $\pot$-increasing decomposition (top left), the potential gradually increases, and within each segment, it never surpasses the potential at its end. In a $\pot$-bounded decomposition (bottom left), the potential keeps its value at the end of the segments, and certain points within the segments never surpass it, and in cover decomposition (right), the states within the sub-configurations at the end of each segment are close to one of the independent runs. }

  \label{fig:decompositions}
\end{figure}

\begin{definition}[$\pot$-increasing $(d,i)$-decomposition]
\label{def: pot increasing decomposition}
Consider a $(d,i)$ high-amplitude word $w_1w_2w_3$ with a decomposition $w_2=u_0u_1\cdots u_t,u_{t+1}$. Denote $v_j=u_0\cdot u_1\cdots u_j$ for every $0\le j\le t$. The decomposition is a \emph{$\pot$-increasing $(d,i)$-decomposition} if the following hold.
\begin{enumerate}
    \item $t=\ramsey(\simpLfuncTypes{d+1}{i},3)$.
    \item For every $0\le j<j'\le t$ we have $\pot(w_1v_j)<\pot(w_1v_{j'})$.
    \item For every $1\le j\le t$ and prefix $x$ of $u_j$ we have $\pot(w_1v_{j-1}x)\le \pot(w_1v_j)$.
    \item $|u_0|\ge \simpLfuncLength{d+1}{i+1}$.
    \item For every $1\le j\le t$ we have $|u_i|\ge \simpLfuncLength{d}{1}$.
\end{enumerate}
\end{definition}

Intuitively, a $\pot$-increasing decomposition has the following properties: there are enough segments to apply some Ramsey argument (Item 1), the potential when moving from one end of a segment to the next end of a segment strictly increases (Item 2), while in all prefixes between such ends the potential might decrease, but never surpasses the end of the segment (Item 3). 
The length of the prefix is enough to apply an inductive argument on depth $d+1$ (Item 4), and all other segments are long enough for depth $d$ (Item 5). These length constraints become clear when we use the definition. Note that the last segment $u_{t+1}$ is just padding to cover all of $w_2$, and has no restrictions.

We now show that under the right circumstances, such a decomposition exists.
\begin{proposition}[From high amplitude to $\pot$-increasing decomposition]
    \label{prop: pot high amplitude to pot increasing decomposition}
    Every $(d,i)$ high-amplitude word has a $\pot$-increasing $(d,i)$-decomposition.
\end{proposition}
\begin{proof}
    Consider a $(d,i)$ high-amplitude word $w_1w_2w_3$ as per \cref{def:high amplitude word}, then $\pot(w_1w_2)-\pot(w_1)> \simpLfuncAmp{d+1}{i}$. 
    \cref{cor:effective bounded growth potential} shows that the potential $\pot$ has a bounded growth rate. Specifically, with each letter the potential can grow by at most $2\simpLfuncMaxW{d}$. 

    Thus, intuitively, in order to increase the potential from $\pot(w_1)$ to $\pot(w_1w_2)$, there must be many prefixes of $w_2$ at which the potential increases slightly above its previous values. By splitting at those indices, we obtain the decomposition. We now make this precise.

    Define $u_0$ to be the shortest prefix of $w_2$ such that 
    \[\pot(w_1u_0)\ge \pot(w_1)+4\simpLfuncMaxW{d}\cdot \simpLfuncLength{d+1}{i+1}\] 
    We then proceed inductively: for every $1\le i\le \ramsey(\simpLfuncTypes{d+1}{i},3)$, let $w_2=u_0u_1\cdots u_{i-1}z$, we define $u_i$ to be the shortest prefix of $z$ such that \[\pot(w_1u_0\cdots u_i)\ge \pot(w_1)+4\cdot i\cdot \simpLfuncMaxW{d}\cdot \simpLfuncLength{d}{1}\]

    We claim that all the words above exist and that $u_i\neq \epsilon$ for all $i$.
    Indeed, $\pot(w_1w_2)-\pot(w_1)\ge \simpLfuncAmp{d+1}{i}$, and observe from \cref{def: simpL func Amp} that
    \[\simpLfuncAmp{d+1}{i}\ge 4\cdot \simpLfuncMaxW{d}\cdot \simpLfuncLength{d+1}{i+1}\]
    so $u_0$ exists, and by the bounded growth of $\pot$ (\cref{cor:effective bounded growth potential}) we also have $|u_0|\ge \simpLfuncLength{d+1}{i+1}$.
    
    Similarly, from \cref{def: simpL func Amp} we also have
    \[\simpLfuncAmp{d+1}{i}\ge 8\cdot \ramsey(\simpLfuncTypes{d+1}{i},3)\cdot \simpLfuncMaxW{d}\cdot \simpLfuncLength{d}{1}\]
    so that all the $u_i$ exist, and by \cref{cor:effective bounded growth potential} we have that 
    \[
    \pot(w_1u_0\cdots u_{i+1})-\pot(w_1u_0u_1\cdots u_i)\ge 4\cdot \simpLfuncMaxW{d}\cdot \simpLfuncLength{d}{1}
    \]
    so $|u_{i+1}|\ge  \simpLfuncLength{d}{1}$. 

    Finally, denote $v_j$ as per \cref{def: pot increasing decomposition}, then since each $u_j$ is chosen as a minimal prefix, we trivially have that for every prefix $x$ of $u_j$ it holds that $\pot(w_1v_{j-1}x)\le \pot(w_1v_j)$. 
    This concludes all the desired properties of the decomposition.
\end{proof}

Dually to \cref{def: pot increasing decomposition}, we present a decomposition that corresponds to bounded-amplitude words (\cref{def:bounded amplitude word}). The specific requirements of the segments differ significantly, though.
\begin{definition}[$\pot$-bounded $(d,i)$-decomposition]
    \label{def: pot bounded decomposition}
    Consider a $(d,i)$ bounded-amplitude word $w_1w_2w_3$ with a decomposition $w_2=u_0u_1\cdots u_tu_{t+1}$. Denote $v_j=u_0\cdot u_1\cdots u_j$ for every $0\le j\le t$, and let $E=\frac{1}{\extraSize}\simpLfuncLength{d}{1}$.
    The decomposition is a \emph{$\pot$-bounded $(d,i)$-decomposition} if the following hold.
    \begin{enumerate}
        %\item $t=  \ramsey(\simpLfuncTypes{d+1}{i},3)+\simpLfuncLength{d}{1}+\simpLfuncLength{d+1}{i+1}+1$.
        \item $t=\ramsey(\simpLfuncTypes{d+1}{i},3)+1$.
        \item For every $0<j<j'\le t$ we have $\pot(w_1v_j)=\pot(w_1v_{j'})$.
        \item For every $0\le j\le t$ we have that $|u_j|$ is a multiple of $E$. 
        \item $|u_0|\ge \simpLfuncLength{d+1}{i+1}$.
        \item For every $0<j\le t$ and prefix $x$ of $u_j$ with $|x|$ multiple of $E$ we have $\pot(w_1v_{j}x)\le \pot(w_1v_{j+1})$.
        %\item $|u_1|\ge \simpLfuncLength{d}{1}+\simpLfuncLength{d+1}{i+1}$.
    \end{enumerate}
\end{definition}
Analogously to \cref{prop: pot high amplitude to pot increasing decomposition}, we show that every bounded-amplitude word can be decomposed. The proof, however, is more intricate. Intuitively, this is because we no longer have the increasing potential to decompose by, so we need to introduce the decomposition indices ``artificially'', but with great care.

\begin{proposition}[From bounded amplitude to $\pot$-bounded decomposition]
\label{prop: pot bounded amp to bounded decomposition}
Every $\pot$-bounded word has a $\pot$-bounded $(d,i)$-decomposition.
\end{proposition}
\begin{proof}
    Consider a $\pot$-bounded word $w_1w_2w_3$. By \cref{def:bounded amplitude word} we have that $|w_2|=\frac{1}{\extraSize}\simpLfuncLength{d+1}{i}$ and for every prefix $x$ of $w_2$ it holds that 
    \[
    \pot(w_1)-\simpLfuncAmp{d+1}{i}\le \pot(w_1x)\le \pot(w_1)+\simpLfuncAmp{d+1}{i}
    \]
    Let $E=\frac{1}{\extraSize}\simpLfuncLength{d}{1}$ be as in \cref{def: pot bounded decomposition}.
    We start with a preliminary decomposition of $w_2=u_0x_0\cdots x_{t'}$ as follows.  
    Observe the coarse lower bound
    \[
    \begin{split}
    &|w_2|=\frac{1}{\extraSize}\simpLfuncLength{d+1}{i}=\extraSize\cdot E \cdot \simpLfuncLength{d+1}{i+1}\cdot \\
    &(\ramsey(\simpLfuncTypes{d+1}{i},3)+2)^{2\simpLfuncAmp{d+1}{i}+1}\\
    &\gg E\cdot (\simpLfuncLength{d+1}{i+1}+ (\ramsey(\simpLfuncTypes{d+1}{i},3)+2)^{2\simpLfuncAmp{d+1}{i}+1})
    \end{split}
    \]
    We can take the prefix $u_0$ such that $|u_0|=E\cdot \simpLfuncLength{d+1}{i+1}$, and define $x_0,\ldots, x_{t'}$ such that $|x_j|=E$ for every $0 \le j< t'$ (so that we cannot decompose further, i.e., $t'$ is maximal). 
    We have in particular that $t'\gg (\ramsey(\simpLfuncTypes{d+1}{i},3)+2)^{2\simpLfuncAmp{d+1}{i}+1}$.
    Henceforth, we only consider positions at the end of each $x_j$. 
        
    Denote $T=\ramsey(\simpLfuncTypes{d+1}{i},3)+1$. 
    We start with an intuition for the proof, where the difficult part is to ensure requirement 5 of \cref{def: pot bounded decomposition}.
    We initially look for $T+1$ positions along $w_2$ where the potential is $\pot(w_1)+\simpLfuncAmp{d+1}{i}$. If we find such positions, we use the infixes between them as the decomposition, and we are done (as this is the highest possible potential, so requirement 5 certainly holds). 
    
    Otherwise, there are at most $T$ such positions, and since $t'$ is very large there is a long infix of $w_2$ where $\pot$ remains at most $\pot(w_1)+\simpLfuncAmp{d+1}{i}-1$ (at the ends of the $x_j$s). We proceed inductively, either finding $T+1$ positions, or lowering the upper bound. Due to the bounded amplitude, this terminates and we find $T+1$ such positions. We then use these positions to construct the decomposition.
    We now turn to formalise this.

    Consider a constant $-\simpLfuncAmp{d+1}{i}\le  P\le \simpLfuncAmp{d+1}{i}$ and an infix $w'_2=x_{j}\cdots x_{j'}$ for some $j\le j'$ of $w_2$. We say that $w'_2$ is $P$-upper bounded if $\pot(w_1u_0x_0\cdots x_{j''})\le \pot(w_1)+P$ for every $j\le j''\le j'$ (i.e., every prefix of this infix, restricted to the $x_j$s).
    From the pigeonhole principle, if $w'_2$ is $P$-upper bounded, then one of the following holds:
    \begin{enumerate}
        \item $w'_2$ has at $T+1$ prefixes $z_0,z_1,\ldots,z_T$ where $\pot(w_1u_0z_k)=\pot(w_1)+P$ for every $k\ge 0$ (where the prefixes are concatenations of the $x_j$s).
        \item There is an infix $w''_2$ of $w_2$ that is $P-1$-upper bounded, and $|w''_2| \ge E\cdot \frac{t'}{T+2}$.
    \end{enumerate}
    Moreover, if $P=-\simpLfuncAmp{d+1}{i}$ then we trivially have the first case, provided $|w'_2|\ge E\cdot (T+1)$.
    Recall that $t'\ge  (T+2)^{2\simpLfuncAmp{d+1}{i}+1}$ and that $w_2$ is $\simpLfuncAmp{d+1}{i}$-upper bounded. 
    We can therefore apply the reasoning above inductively, starting from $P=\simpLfuncAmp{d+1}{i}$. Since we divide the length by $T+1$ at each step, then in the worst case we reach $P=-\simpLfuncAmp{d+1}{i}$ with length $T+2$ left, after $2\simpLfuncAmp{d+1}{i}$ repetitions.
    
    By implicitly adding to $u_0$ the infix until the first position in this set, we therefore found a decomposition $w_2=u_0u_1\cdots u_{T}$ where $\pot(w_1u_0\cdots u_j)=\pot(w_1u_0\cdots u_{j'})$ for every $0<j<j'\le T$ and for every prefix $x$ of $u_{j}$ with $|x|$ multiple of $E$, we have $\pot(w_1u_0\cdots u_{j-1}x)\le \pot(w_1u_0\cdots u_{j})$. Moreover, $T\ge t$, so we can arbitrarily choose $t$ of the infixes, and leave the rest in $u_{t+1}$.
\end{proof}

The next property we consider is a \emph{cover}. Intuitively, a decomposition is a cover if at every end point of an infix, all the states are close to one of the independent runs. Then, tracking the states near independent runs amounts to tracking all states. This idea is depicted in \cref{fig:decompositions}.

\begin{definition}[Cover decomposition]
    \label{def: cover decomposition}
    Consider a $(d,i)$ well-separated word $w_1w_2w_3$ (see \cref{def:well separated word}) with a set of independent runs $I$ and a decomposition $w_2=u_0u_1\cdots u_tu_{t+1}$. 
    Denote $v_j=u_0\cdot u_1\cdots u_j$ and $\vec{c_j}=\xconf(s_0,v_j)$ for every $0\le j\le t$. Let $b= \simpLfuncCover{d+1}{i}$ with $|I|=i$.
    
    The decomposition is a \emph{cover} if for every $0< j\le t$ and state $s\in \supp(\vec{c_j})$ there exists $\rho\in I$ such that $\near_b(\rho,w_1v_j)(s)\neq \bot$. That is, $s$ is near $\rho$ after reading $w_1v_j$.
\end{definition}

Our final decomposition lemma is that if the word is either $\pot$-increasing or $\pot$-bounded, then we either have an SRI, or we can reduce the depth of the decomposition. This is the main tool we use inductively in \cref{sec:well separation to SSRI}.
\begin{proposition}
    \label{prop:pot cover decomp induces SRI or lower depth}
    Consider a $(d,i)$ well-separated word $w_1w_2w_3$ with a cover decomposition $w_2=u_0u_1\cdots u_t,u_{t+1}$ (\cref{def: cover decomposition}).
    Assume the word is either high-amplitude and the decomposition is $\pot$-increasing, or the word is bounded-amplitude and the decomposition is $\pot$-bounded (\cref{def:high amplitude word,def: pot increasing decomposition,def: pot bounded decomposition}). 
    Then one of the following holds:
    \begin{enumerate}
        \item $w_1w_2w_3=uxyv$ where $uxyv$ is an SSRI with $\depth(x)=d$ and $xy$ is an infix of $w_2$.
        \item $w_1w_2w_2=w'_1w'_2w'_3$ where $w'_1w'_2w'_3$ is a $(d-1,1)$-fit well-separated word (\cref{def:well separated word}) and $w'_2$ is an infix of $w_2$ with $|w'_2|=\frac{1}{\extraSize}\simpLfuncLength{d}{1}$.
    \end{enumerate}
\end{proposition}
\begin{proof}
    Let $v_j$ denote $u_0\cdots u_j$ as per \cref{def: pot increasing decomposition,def: pot bounded decomposition}.
    We begin by identifying a structure that repeats twice along the decomposition. Let $b=2\simpLfuncCover{d+1}{i}$, and consider a complete graph on the vertices $\{1,\ldots,t\}$. We define a colouring of the edges as follows. For $1\le j<j'\le t$ write $u_{j,j'}=u_{j}\cdots u_{j'-1}$, colour the edge $\{j,j'\}$ with $\difftype_b(w_1v_{j-1},u_{j,j'})$. 
    
    By \cref{rmk:size of difftypeset}, the number of colours we use is bounded by $\simpLfuncTypes{d+1}{i}$, and by \cref{def: pot increasing decomposition,def: pot bounded decomposition} we have $$t\ge \ramsey(\simpLfuncTypes{d+1}{i},3)$$ Thus, there exists a monochromatic 3-clique in our graph, i.e., there are $j<j'<j''$ such that 
    \begin{align*}
    \difftype_b(w_1v_{j-1},u_{j,j'}) =  &\difftype_b(w_1v_{j'-1},u_{j',j''})\\ =  &\difftype_b(w_1v_{j-1},u_{j,j''})
    \end{align*}
    Write $w_1w_2w_3=uxyv$ where $u=w_1v_{j-1}$, $x=u_{j,j'}$, $y=u_{j',j''}$ and $z=u_{j''+1}\cdots u_tw_2$. Clearly $xy$ is an infix of $w_2$.
    
    We now separate to two cases. 
    \paragraph*{If $\depth(x)=d$.} In this case we claim that $uxyv$ is an SSRI. Indeed, we go over the requirements of \cref{def:separated repeating infix,def: SSRI and GSRI}:
    \begin{enumerate}
        \item $\supp(\vec{c_u})=\supp(\vec{c_{ux}})=\supp(\vec{c_{uxy}})=B$: since the decomposition is a cover, then as noted in \cref{def: cover decomposition}, every reachable state is near some independent run, and in particular the equivalent $\difftype_b$ implies the same support.
        \item In the case of a $\pot$-increasing decomposition (\cref{def: pot increasing decomposition}), $\pot(u)\le \pot(ux)\le \pot(uxy)$ follows immediately. Then, the fact that $\pot(u)=\pot(ux)$ iff $\pot(ux)=\pot(uxy)$ is actually implied by the definition, in particular by Item 4 below. 

        In the case of a $\pot$-bounded decomposition, we already have $\pot(u)= \pot(ux)= \pot(uxy)$. Note that this is the only place where we treat the two types of decomposition separately.
        \item $|x|\le \frac{1}{\bigM}\simpL(\depth(x)+1)=\frac{1}{\bigM}\simpLfuncLength{d+1}{1}$: this is immediate from the requirement in \cref{def:d i fit} that $|w_2|\le \frac{1}{\extraSize}\simpLfuncLength{d+1}{i}\le \frac{1}{\bigM}\simpLfuncLength{d+1}{1}$ (recall that the inductive definition of $\simpLfuncLength{d+1}{i}$ is with increasing $i$).
        \item We partition $B=V_1\cup\ldots \cup V_i$ according to the independent runs. That is, set $V_j=\supp(\near_b(\rho_j,x))$ for $b=\simpLfuncCover{d+1}{i}$ where $I=\{\rho_1<\ldots<\rho_i\}$. 
        By \cref{def: cover decomposition} this is indeed a partition.
        \begin{itemize}
            \item The gaps between $V_j$ and $V_{j+1}$ are higher than \\ $4|S|\bigM\maxeff{xy}$: by \cref{def:well separated word} the set $I$ has gap at least 
            $4\bigM^2(\simpLfuncMaxW{d}\simpLfuncLength{d}{1}\simpLfuncLength{d+1}{i}+\simpLfuncCover{d+1}{i})$
            Notice that $|xy|\le |w_2|<\simpLfuncLength{d+1}{i}$ (by \cref{def:d i fit}), so \[\maxeff{xy}\le \simpLfuncMaxW{d}\simpLfuncLength{d+1}{i}\le \simpLfuncMaxW{d}\simpLfuncLength{d}{1}\simpLfuncLength{d+1}{i}\] 
            Thus, the minimal possible gap of a state in $V_j$ from a state in $V_{j+1}$ is at least
            \[
            \begin{split}
            &4\bigM^2(\maxeff{xy}+\simpLfuncCover{d+1}{i})-2b\\
            &=4\bigM^2(\maxeff{xy}+\simpLfuncCover{d+1}{i})-2\simpLfuncCover{d+1}{i}\\
            &\gg 4|S|\bigM\maxeff{xy}
            \end{split}
            \]
            \item The existence of $k_{j,x}, k_{j,y}$ and their sign equality is immediate from $\difftype_b(w_1v_{j-1},u_{j,j'})= \\\difftype_b(w_1v_{j'-1},u_{j',j''})$. Indeed, this means that the configurations around each independent run are identical, and that the sign of the run on $x$ and on $y$ is the same.
        \end{itemize}
        \item The seamless baseline run exists by \cref{def:potential increasing word}.
    \end{enumerate}
    This concludes that $uxyv$ is an SSRI.

    \paragraph*{If $\depth(x)<d$.} In this case we treat the two types of decomposition separately. 
    The case of a $\pot$-increasing decomposition is simple: we just restrict attention to an infix of $w_2$, as follows.
    Define $w'_2$ to be the suffix of $x$ of length $\frac{1}{\extraSize}\simpLfuncLength{d}{1}$, and let $w'_1,w'_3$ be the prefix and suffix such that $w_1w_2w_3=w'_1w'_2w'_3$. Note that there is such a suffix since $|x|\ge \simpLfuncLength{d}{1}$ as a concatenation of $u_j$'s, and by \cref{def: pot increasing decomposition}.
    Still by \cref{def: pot increasing decomposition} we have $\pot(w'_1)\le \pot(w'_1w'_2)$ (viewing $w'_2$ as a prefix of $u_{j'}$). Take $I'=\{\rho_1\}$ (i.e., arbitrarily choose a single independent run from $I$) so $w'_1w'_2w'_3$ with $I'$ is a $(d-1,1)$ fit word, as per \cref{def:d i fit}.

    The case of $\pot$-bounded decomposition is slightly more delicate. We still take $w'_2$ to be the suffix of $x$ of length $\frac{1}{\extraSize}\simpLfuncLength{d}{1}$. This is possible since each part in the decomposition is of length multiple of $\frac{1}{\extraSize}\simpLfuncLength{d}{1}$ (see \cref{def: pot bounded decomposition}). Moreover, Property 5 in \cref{def: pot bounded decomposition} guarantees that $\pot(w'_1)\le \pot(w'_1w'_2)$ (although between them the potential may increase, but this does not concern us).
    We can now proceed as above by taking $I'=\{\rho_1\}$, and we are done.
\end{proof}

\subsection{From Well Separation to SSRI}
\label{sec:well separation to SSRI}
We now have the tools to show that given a well-separated word, we can obtain an SSRI from it. To do so, we develop an effective version of the \emph{Zooming Technique} introduced in~\cite{almagor2026determinization}. We remark that this is where the convoluted inductive definitions of \cref{sec:effective cactus} play a central role.
\begin{lemma}
    \label{lem: d i fit to SSRI or lower depth}
    Consider a $(d,i)$-fit well-separated word $w_1w_2w_3$ where
    % \gatodo{either $\pot(w_1w_2)-\pot(w_1)\ge \simpLfuncAmp{d+1}{i}$ or}
    $|w_2|= \frac{1}{\extraSize}\simpLfuncLength{d+1}{i}$, then one of the following holds:
    \begin{enumerate}
        \item $w_1w_2w_3=uxyv$ where $uxyv$ is an SSRI with $\depth(x)=d$ and $xy$ is an infix of $w_2$.
        \item $w_1w_2w_2=w'_1w'_2w'_3$ where $w'_1w'_2w'_3$ is a $(d-1,1)$-fit well-separated word (\cref{def:well separated word}) and $w'_2$ is an infix of $w_2$ with $|w'_2|=\frac{1}{\extraSize}\simpLfuncLength{d}{1}$.
    \end{enumerate}
\end{lemma}
\begin{proof}
    Let $I=\{\rho_1<\ldots<\rho_i\}$ be the set of independent runs associated with $w_1w_2w_3$ (\cref{def:independent run}). 
    Consider the behaviour of $\pot$ along $w_2$ as per \cref{def:bounded amplitude word}. 
    We claim there are two cases: either $w_1w_2w_3$ is $(d,i)$ bounded amplitude, or it has an infix that is high amplitude.
    Indeed, if $w_1w_2w_3$ is not bounded potential, then for some prefix $x$ of $w_2$ we have $|\pot(w_1x)-\pot(w_1)|>\simpLfuncAmp{d+1}{i}$. If $\pot(w_1x)-\pot(w_1)>\simpLfuncAmp{d+1}{i}$ then $x$ is an infix as required. If $\pot(w_1)-\pot(w_1x)>\simpLfuncAmp{d+1}{i}$ then since $\pot(w_1)\le \pot(w_1w_2)$ (by \cref{def:potential increasing word}) the suffix $z$ such that $w_2=xz$ is a high-amplitude infix of $w_2$. 

    We can therefore assume without loss of generality that $w_1w_2w_3$ itself is either $(d,i)$ bounded amplitude or $(d,i)$ high amplitude.
    We then employ either \cref{prop: pot bounded amp to bounded decomposition} in the former case, or \cref{prop: pot high amplitude to pot increasing decomposition} in the latter, and obtain a $(d,i)$ decomposition as per either \cref{def: pot bounded decomposition} or \cref{def: pot increasing decomposition}, respectively. I.e., we can write $w_2=u_0u_1\cdots u_tu_{t+1}$.
    
    Our goal is now to use \cref{prop:pot cover decomp induces SRI or lower depth} to conclude the lemma. However, this requires the decomposition to also be a cover (\cref{def: cover decomposition}), which is not necessarily the case. This is where Zooming comes in.
    
    We proceed by reverse induction on $i$.
    The base case is $i=|S|$. In this case, every state belongs to one of the independent runs at each step. Therefore, the decomposition is a cover already with $b=0$, and in particular with $b=\simpLfuncCover{d+1}{|S|}=1$ (by \cref{def: simpL func Cover}). We can therefore use \cref{prop:pot cover decomp induces SRI or lower depth} and we are done.

    We proceed to the induction case, where $1\le i<|S|$. Consider the decomposition of $w_2$. If it is already a cover, then we can again use \cref{prop:pot cover decomp induces SRI or lower depth} and we are done. 
    
    Assume henceforth that this is not the case. In the notations of \cref{def: cover decomposition} this means that there is some $0< j\le  t$ and state $s\in \supp(\vec{c_j})$ that is more than $\simpLfuncCover{d+1}{i}$ away from all the independent runs in $I$. More precisely, for every run $\rho\in I$, write $\rho[w_1v_j]:s_0\runsto{v_j}s'$, then $|\vec{c_j}(s)-\vec{c_j}(s')|>\simpLfuncCover{d+1}{i}$.

    We now ``zoom in'' on $s$, and show that we can use it to construct a new independent run on an infix of $w_2$, so we can use the induction hypothesis with $i+1$ (see \cref{fig:zooming}). This, however, needs to be done carefully in order to account for the depth, length, and cover.

\begin{figure}[H]
  \begin{center}
    \includegraphics[width=1\textwidth]{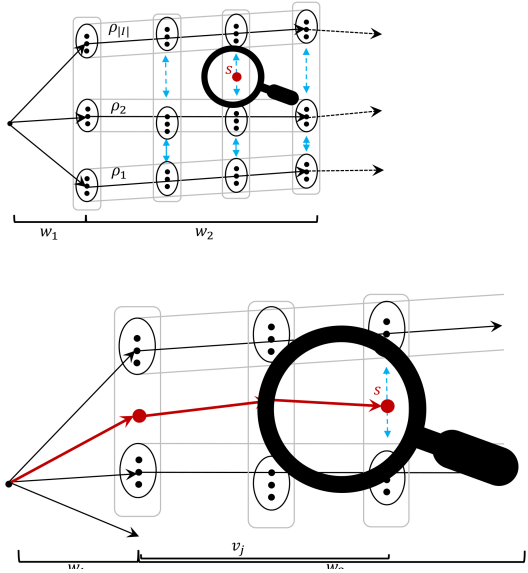}
  \end{center}
  \caption{The state $s$ is not near any of the independent runs. We can then zoom in on $s$, and construct a new independent run that reaches it.}

  \label{fig:zooming}
\end{figure}

    Let $\pi:s_0\runsto{w_1v_j}s$ be a seamless minimal-weight run to $s$. In particular, $\weight(\pi)=\vec{c_j}(s)$. Decompose $w_1v_j=z_1z_2$ where $|z_2|=\frac{1}{\extraSize}\simpLfuncLength{d+1}{i+1}$. Such a decomposition is possible due to the lower bound on $|u_0|$ in both \cref{def: pot increasing decomposition} and \cref{def: pot bounded decomposition}.
    
    We claim that $\pi[z_2]$ remains at least $\frac12\simpLfuncCover{d+1}{i+1}$ away from all independent runs along this suffix.
    Indeed, this is a direct consequence of \cref{prop: relating cover length and maxW}, since the maximal change in weight that can occur during $\pi[z_2]$ is 
    \[2\maxeff{|z_2|}\le |z_2|\cdot \simpLfuncMaxW{d}\cdot 2= \frac{1}{\extraSize}\simpLfuncLength{d+1}{i+1}\cdot 2\]
    Thus, the word $z_1z_2z_3$ with $z_3=\epsilon$ is $(d,i+1)$-fit well separated word (\cref{def:well separated word}), and by the induction hypothesis it satisfies the conditions of the lemma, so we are done.
\end{proof}

Finally, consider a potential-increasing word $w_1w_2w_3\in (\simpCac)^*$ with $|w_2|= \frac{1}{\extraSize}\simpLfuncLength{|S|}{1}$. By \cref{rmk:depth vs length bound} we have $\depth(w_2)\le |S|-1$. By choosing an arbitrary seamless run on it, we have that it is $(|S|-1,1)$-fit well-separated. Indeed, for $i=1$ well-separation is trivial (see \cref{def:well separated word}).
We can apply \cref{lem: d i fit to SSRI or lower depth} so that either we identify an SSRI, or we rewrite $w_1w_2w_3=w'_1w'_2w'_3$ as a $(|S|-2,1)$-fit well-separated word. We continue inductively until at some point an SSRI must be found (since the induction cannot go below depth $0$). We therefore have the following:
\begin{corollary}
\label{cor: potential increasing to SSRI}
    Consider a potential-increasing word $w_1w_2w_3\in (\simpCac)^*$ with $|w_2|= \frac{1}{\extraSize}\simpLfuncLength{|S|}{1}$, then $w_1w_2w_3=uxyv$ where $uxyv$ is an SSRI and $xy$ is an infix of $w_2$.
\end{corollary}

\section{Existence of GSRI}
\label{sec:existence of GSRI}
We now turn our attention to GSRI and their existence. This section has the same structure as \cref{sec:existence of SSRI}, and the proofs follow mostly the same patterns.
The main difference is that GSRI are related to $\charge$, which does not have bounded growth. In addition, the definitions are based on general length functions (\cref{sec: general length bound function}), which introduces slight differences. 

We use the notions of independent runs, gaps and $\difftype$ introduced in \cref{def:independent run,def:near run,def:diff type}. In this section, we work over $(\genCac)^*$. 
\begin{definition}[$\charge$-Decreasing Word]
\label{def:charge decreasing word}
A word $w_1w_2w_3$ is \emph{$\charge$-decreasing} if there is a seamless baseline run on $w_1w_2w_3$ and $\charge(w_1)\ge \charge(w_1w_2)$.
\end{definition}

\begin{definition}[$\charge-(d,i)$-Fit word]
    \label{def:charge d i fit}
    Consider a word $w_1w_2w_3\in (\genCac)^*$ and a set $I$ of $G$-independent runs. The word $w_1w_2w_3$ (with $I$) is \emph{$\charge-(d,i)$-fit} if the following hold.
    \begin{itemize}
        \item $w_1w_2w_3$ is $\charge$-decreasing.
        \item $\depth(w_2)\le d$.
        \item $|w_2|\le \frac{1}{\extraSize}\generLfuncLength{d+1}{i}$ where $i=|I|$.
    \end{itemize} 
\end{definition}
We again specialise to high and bounded amplitude.
\begin{definition}[$\charge-(d,i)$ Well-Separated and High / Bounded-Amplitude word]
\label{def:charge well separated word}
\label{def:charge high amplitude word}
\label{def:charge bounded amplitude word}
A word $w_1w_2w_3\in (\genCac)^*$ with set $I$ is \emph{$\charge-(d,i)$ well separated} if $w_1w_2w_3$ is $\charge-(d,i)$-fit and $I$ has gap at least
\[\begin{split}
&\frac12 \generLfuncCover{d+1}{i-1}=\\
&4\bigM^2(\generLfuncMaxW{d}\generLfuncLength{d}{1}\generLfuncLength{d+1}{i}+\generLfuncCover{d+1}{i}+\budH)
\end{split}\] 
If $w_1w_2w_3$ is well separated, then it is:
\begin{itemize}
    \item \emph{$\charge-(d,i)$ high amplitude} if $\charge(w_1)-\charge(w_1w_2)> \generLfuncAmp{d+1}{i}$.
    \item \emph{$\charge-(d,i)$ bounded amplitude} if $|\charge(w_1)-\charge(w_1x)|\le \generLfuncAmp{d+1}{i}$ for every prefix $x$ of $w_2$, and $|w_2|=\frac{1}{\extraSize}\generLfuncLength{d+1}{i}$.
\end{itemize}
\end{definition}
We proceed with the decomposition of $\charge$-high-amplitude words. The intuition is similar to that of \cref{def: pot increasing decomposition}, referring to decreasing $\charge$ instead of increasing $\pot$.
\begin{definition}[$\charge$-decreasing $(d,i)$-decomposition]
\label{def: charge decreasing decomposition}
Consider a $(d,i)$ high-amplitude word $w_1w_2w_3$ with a decomposition $w_2=u_0u_1\cdots u_tu_{t+1}$. Denote $v_j=u_0\cdot u_1\cdots u_j$ for every $0\le j\le t$. The decomposition is a \emph{$\charge$-decreasing $(d,i)$-decomposition} if the following hold.
\begin{enumerate}
    \item $t=\ramsey(\generLfuncTypes{d+1}{i},3)$.
    \item For every $0\le j<j'\le t$ we have $\charge(w_1v_j)>\charge(w_1v_{j'})$.
    \item For every $1\le j\le t$ and prefix $x$ of $u_j$ we have $\charge(w_1v_{j-1}x)\ge \charge(w_1v_j)$.
    \item $|u_0|\ge \generLfuncLength{d+1}{i+1}$.
    \item For every $1\le j\le t$ we have $|u_i|\ge \generLfuncLength{d}{1}$.
\end{enumerate}
\end{definition}
In contrast to \cref{prop: pot high amplitude to pot increasing decomposition}, being $\charge$ high amplitude is not a sufficient condition for having a decomposition. Indeed, it may be the case that the charge has unbounded growth, which then might cause the word to be very short. Fortunately, if this is the case, then \cref{lem: charge decrease to high potential} already gives us unbounded potential, which puts us in the SSRI setting.
\begin{proposition}[From $\charge$-high amplitude to $\charge$-decreasing decomposition]
    \label{prop: charge high amplitude to charge decreasing decomposition}
    Either there is a word $w' \in (\Gamma_0^0)^*$ such that $\pot(w')>\budH$, or every every $\charge-(d,i)$ high-amplitude word has a $\charge$-decreasing $(d,i)$-decomposition.
\end{proposition}
\begin{proof}
    If there exists a word $w' \in (\Gamma_0^0)^*$ such that $\pot(w')>\budH$, we are done. Assume therefore that this is not the case. By \cref{lem: charge decrease to high potential} (viewed in the contrapositive), this means that for every word $u\in (\genCac)^*$ and $\sigma\in \genCac$ with $\depth(\sigma)=d$ we have $\charge(u)-\charge(u\sigma)\le \budH+2\generLfuncMaxW{d}$.
    Consider a $(d,i)$ high-amplitude word $w_1w_2w_3$ as per \cref{def:charge high amplitude word}, then $\charge(w_1w_2)-\charge(w_1)< -\generLfuncAmp{d+1}{i}$.
    
    Thus, intuitively, in order to decrease the charge from $\charge(w_1)$ to $\charge(w_1w_2)$, there must be many prefixes of $w_2$ at which the charge decreases slightly below its previous values. By splitting at those indices, we obtain the decomposition. We now make this precise.

    Define $u_0$ to be the shortest prefix of $w_2$ such that \[\charge(w_1u_0)+4(2\generLfuncMaxW{d}+\budH)\cdot \generLfuncLength{d+1}{i+1}\le \charge(w_1)\] 
    We then proceed inductively: for every $1\le i\le \ramsey(\generLfuncTypes{d+1}{i},3)$, let $w_2=u_0u_1\cdots u_{i-1}z$, we define $u_i$ to be the shortest prefix of $z$ such that \[\charge(w_1u_0\cdots u_i)+4\cdot i\cdot (2\generLfuncMaxW{d}+\budH)\cdot \generLfuncLength{d}{1}\le \charge(w_1)\]

    We claim that all the words above exist and that $u_i\neq \epsilon$ for all $i$.
    Indeed, $\charge(w_1)-\charge(w_1w_2)\ge \generLfuncAmp{d+1}{i}$, and observe from \cref{def: genL func Amp} that
    \[\generLfuncAmp{d+1}{i}\ge 4\cdot (2\generLfuncMaxW{d}+\budH)\cdot \generLfuncLength{d+1}{i+1}\]
    so $u_0$ exists, and by the bounded decrease of $\charge$ we also have $|u_0|\ge \generLfuncLength{d+1}{i+1}$.
    
    Similarly, from \cref{def: genL func Amp} we also have
    \[\generLfuncAmp{d+1}{i}\ge 8\cdot \ramsey(\generLfuncTypes{d+1}{i},3)\cdot (2\generLfuncMaxW{d}+\budH)\cdot \generLfuncLength{d}{1}\]
    so that all the $u_i$ exist, and by the bounded decrease we have that 
    \[
    \charge(w_1u_0u_1\cdots u_i)-\charge(w_1u_0\cdots u_{i+1})\ge 4\cdot (2\generLfuncMaxW{d}+\budH)\cdot \generLfuncLength{d}{1}
    \]
    so $|u_{i+1}|\ge  \generLfuncLength{d}{1}$. 

    Finally, denote $v_j$ as per \cref{def: charge decreasing decomposition}, then since each $u_j$ is chosen as a minimal prefix, we trivially have that for every prefix $x$ of $u_j$ it holds that $\charge(w_1v_{j-1}x)\ge \charge(w_1v_j)$. 
    This concludes all the desired properties of the decomposition.
\end{proof}
We now proceed to study decompositions of $\charge$ bounded amplitude words. In this case, the possibly unbounded growth of $\charge$ plays no role, and therefore this part is nearly identical to its $\pot$ counterpart in \cref{sec:existence of SSRI}.
\begin{definition}[$\charge$-bounded $(d,i)$-decomposition]
    \label{def: charge bounded decomposition}
    Consider a $\charge-(d,i)$ bounded-amplitude word $w_1w_2w_3$ with a decomposition $w_2=u_0u_1\cdots u_tu_{t+1}$. Denote $v_j=u_0\cdot u_1\cdots u_j$ for every $0\le j\le t$, and let $E=\frac{1}{\extraSize}\generLfuncLength{d}{1}$.
    The decomposition is a \emph{$\charge$-bounded $(d,i)$-decomposition} if the following hold.
    \begin{enumerate}
        \item $t=\ramsey(\generLfuncTypes{d+1}{i},3)+1$.
        \item For every $0<j<j'\le t$ we have $\charge(w_1v_j)=\charge(w_1v_{j'})$.
        \item For every $0\le j\le t$ we have that $|u_j|$ is a multiple of $E$. 
        \item $|u_0|\ge \generLfuncLength{d+1}{i+1}$.
        \item For every $0<j\le t$ and prefix $x$ of $u_j$ with $|x|$ multiple of $E$ we have $\charge(w_1v_{j}x)\ge \charge(w_1v_{j+1})$.
    \end{enumerate}
\end{definition}
\begin{proposition}[From bounded amplitude to $\charge$-bounded decomposition]
\label{prop: charge bounded amp to bounded decomposition}
Every $\charge$-bounded word has a $\charge$-bounded $(d,i)$-decomposition.
\end{proposition}
\begin{proof}
    Consider a $\charge$-bounded word $w_1w_2w_3$. By \cref{def:charge bounded amplitude word} we have that $|w_2|=\frac{1}{\extraSize}\generLfuncLength{d+1}{i}$ and for every prefix $x$ of $w_2$ it holds that 
    \[
    \charge(w_1)-\generLfuncAmp{d+1}{i}\le \charge(w_1x)\le \charge(w_1)+\generLfuncAmp{d+1}{i}
    \]
    Let $E=\frac{1}{\extraSize}\generLfuncLength{d}{1}$ be as in \cref{def: charge bounded decomposition}.
    We start with a preliminary decomposition of $w_2=u_0x_0\cdots x_{t'}$ as follows.  
    Observe the coarse lower bound
    \[
    \begin{split}
    &|w_2|=\frac{1}{\extraSize}\generLfuncLength{d+1}{i}= \\& 
    \extraSize\cdot E \cdot \\
    &\generLfuncLength{d+1}{i+1}\cdot(\ramsey(\generLfuncTypes{d+1}{i},3)+2)^{2\generLfuncAmp{d+1}{i}+1}\\
    &\gg E\cdot \\
    &(\generLfuncLength{d+1}{i+1}+ 
    (\ramsey(\generLfuncTypes{d+1}{i},3)+2)^{2\generLfuncAmp{d+1}{i}+1})
    \end{split}
    \]
    We can take the prefix $u_0$ such that $|u_0|=E\cdot \generLfuncLength{d+1}{i+1}$, and define $x_0,\ldots, x_{t'}$ such that $|x_j|=E$ for every $0 \le j< t'$ (so that we cannot decompose further, i.e., $t'$ is maximal). 
    We have in particular that $t'\gg (\ramsey(\generLfuncTypes{d+1}{i},3)+2)^{2\generLfuncAmp{d+1}{i}+1}$.
    Henceforth, we only consider positions at the end of each $x_j$. 
        
    Denote $T=\ramsey(\generLfuncTypes{d+1}{i},3)+1$. 
    We start with an intuition for the proof, where the difficult part is to ensure requirement 5 of \cref{def: charge bounded decomposition}.
    We initially look for $T+1$ positions along $w_2$ where the charge is $\charge(w_1)-\generLfuncAmp{d+1}{i}$. If we find such positions, we use the infixes between them as the decomposition, and we are done (as this is the lowest possible charge, so requirement 5 certainly holds). 
    
    Otherwise, there are at most $T$ such positions, and since $t'$ is very large there is a long infix of $w_2$ where $\charge$ remains at least $\charge(w_1)-\generLfuncAmp{d+1}{i}+1$ (at the ends of the $x_j$s). We proceed inductively, either finding $T+1$ positions, or increasing the lower bound. 
    Due to the $\charge$-bounded amplitude, this terminates and we find $T+1$ such positions. We then use these positions to construct the decomposition.
    We now turn to formalise this.

    Consider a constant $-\generLfuncAmp{d+1}{i}\le  P\le \generLfuncAmp{d+1}{i}$ and an infix $w'_2=x_{j}\cdots x_{j'}$ for some $j\le j'$ of $w_2$. We say that $w'_2$ is $P$-lower bounded if $\charge(w_1u_0x_0\cdots x_{j''})\ge \charge(w_1)-P$ for every $j\le j''\le j'$ (i.e., every prefix of this infix, restricted to the $x_j$s).
    From the pigeonhole principle, if $w'_2$ is $P$-lower bounded, then one of the following holds:
    \begin{enumerate}
        \item $w'_2$ has at $T+1$ prefixes $z_0,z_1,\ldots,z_T$ where $\charge(w_1u_0z_k)=\charge(w_1)-P$ for every $k\ge 0$ (where the prefixes are concatenations of the $x_j$s).
        \item There is an infix $w''_2$ of $w_2$ that is $P-1$-lower bounded, and $|w''_2| \ge E\cdot \frac{t'}{T+2}$.
    \end{enumerate}
    Moreover, if $P=-\generLfuncAmp{d+1}{i}$ then we trivially have the first case, provided $|w'_2|\ge E\cdot (T+1)$.
    Recall that $t'\ge  (T+2)^{2\generLfuncAmp{d+1}{i}+1}$ and that $w_2$ is $\generLfuncAmp{d+1}{i}$-lower bounded. 
    We can therefore apply the reasoning above inductively, starting from $P=\generLfuncAmp{d+1}{i}$. Since we divide the length by $T+1$ at each step, then in the worst case we reach $P=-\generLfuncAmp{d+1}{i}$ with length $T+2$ left. 
    
    By implicitly adding to $u_0$ the infix until the first position in this set, we therefore found a decomposition $w_2=u_0u_1\cdots u_{T}$ where $\charge(w_1u_0\cdots u_j)=\charge(w_1u_0\cdots u_{j'})$ for every $0<j<j'\le T$ and for every prefix $x$ of $u_{j}$ with $|x|$ multiple of $E$, we have $\charge(w_1u_0\cdots u_{j-1}x)\le \charge(w_1u_0\cdots u_{j})$. Moreover, $T\ge t$, so we can arbitrarily choose $t$ of the infixes, and leave the rest in $u_{t+1}$.
\end{proof}
Next, we define a cover, and show that the existence of a covered decomposition implies the existence of a GSRI, or a reduction of the depth of the word. Note that here the addition of $\budH$ comes into play.
\begin{definition}[Cover decomposition for $\charge$]
    \label{def: charge cover decomposition}
    Consider a $\charge-(d,i)$ well-separated word $w_1w_2w_3$ (see \cref{def:charge well separated word}) with a set of independent runs $I$ and a decomposition $w_2=u_0u_1\cdots u_tu_{t+1}$. 
    Denote $v_j=u_0\cdot u_1\cdots u_j$ and $\vec{c_j}=\xconf(s_0,v_j)$ for every $0\le j\le t$. Let $b= \generLfuncCover{d+1}{i}$ with $|I|=i$.
    
    The decomposition is a \emph{cover} if for every $0< j\le t$ and state $s\in \supp(\vec{c_j})$ there exists $\rho\in I$ such that $\near_b(\rho,w_1v_j)(s)\neq \bot$. That is, $s$ is near $\rho$ after reading $w_1v_j$.
\end{definition}
\begin{proposition}
    \label{prop:charge cover decomp induces GSRI or lower depth}
    Consider a $\charge-(d,i)$ well-separated word $w_1w_2w_3$ with a cover decomposition $w_2=u_0u_1\cdots u_t,u_{t+1}$.
    Assume the word is either $\charge$ high amplitude and the decomposition is $\charge$-decreasing, or the word is $\charge$ bounded amplitude and the decomposition is $\charge$-bounded (\cref{def:charge high amplitude word,def: charge decreasing decomposition,def: charge bounded decomposition}). 
    Then one of the following holds:
    \begin{enumerate}
        \item $w_1w_2w_3=uxyv$ where $uxyv$ is a GSRI with $\depth(x)=d$, and $xy$ is an infix of $w_2$.
        \item $w_1w_2w_2=w'_1w'_2w'_3$ where $w'_1w'_2w'_3$ is a $\charge-(d-1,1)$-fit well-separated word (\cref{def:charge well separated word}) and $|w'_2|=\frac{1}{\extraSize}\simpLfuncLength{d}{1}$.
    \end{enumerate}
\end{proposition}
\begin{proof}
    Let $v_j$ denote $u_0\cdots u_j$ as per \cref{def: charge decreasing decomposition,def: charge bounded decomposition}.
    We begin by identifying a structure that repeats twice along the decomposition. Let $b=2\generLfuncCover{d+1}{i}$, and consider a complete graph on the vertices $\{1,\ldots,t\}$. We define a colouring of the edges as follows. For $1\le j<j'\le t$ write $u_{j,j'}=u_{j}\cdots u_{j'-1}$, colour the edge $\{j,j'\}$ with $\difftype_b(w_1v_{j-1},u_{j,j'})$. 
    
    By \cref{rmk:size of difftypeset}, the number of colours we use is bounded by $\generLfuncTypes{d+1}{i}$, and by \cref{def: charge decreasing decomposition,def: charge bounded decomposition} we have \\ $t\ge \ramsey(\generLfuncTypes{d+1}{i},3)$. Thus, there exists a monochromatic 3-clique in our graph, i.e., there are $j<j'<j''$ such that 
    \[\begin{split} &\difftype_b(w_1v_{j-1},u_{j,j'})= \difftype_b(w_1v_{j'-1},u_{j',j''})\\& = \difftype_b(w_1v_{j-1},u_{j,j''})\end{split}\]
    Write $w_1w_2w_3=uxyv$ where $u=w_1v_{j-1}$, $x=u_{j,j'}$, $y=u_{j',j''}$ and $z=u_{j''+1}\cdots u_tw_2$. Clearly $xy$ is an infix of $w_2$.
    
    We now separate to two cases. 
    \paragraph*{If $\depth(x)=d$.} In this case we claim that $uxyv$ is a GSRI. Indeed, we go over the requirements of \cref{def:separated repeating infix,def: SSRI and GSRI}:
    \begin{enumerate}
        \item $\supp(\vec{c_u})=\supp(\vec{c_{ux}})=\supp(\vec{c_{uxy}})=B$: since the decomposition is a cover, then as noted in \cref{def: charge cover decomposition}, every reachable state is near some independent run, and in particular the equivalent $\difftype_b$ implies the same support.
        \item In the case of a $\charge$-decreasing decomposition (\cref{def: charge decreasing decomposition}), $\charge(u)\ge \charge(ux)\ge \charge(uxy)$ follows immediately. Then, the fact that $\charge(u)=\charge(ux)$ iff $\charge(ux)=\charge(uxy)$ is actually implied by the definition, in particular by Item 4 below. 

        In the case of a $\charge$-bounded decomposition, we already have $\charge(u)= \charge(ux)= \charge(uxy)$. Note that this is the only place where we treat the two types of decomposition separately.
        \item $|x|\le \frac{1}{\bigM}\generL(\depth(x)+1)=\frac{1}{\bigM}\generLfuncLength{d+1}{1}$: this is immediate from the requirement in \cref{def:charge d i fit} that $|w_2|\le \frac{1}{\extraSize}\generLfuncLength{d+1}{i}\le \frac{1}{\bigM}\generLfuncLength{d+1}{1}$ (recall that the inductive definition of $\generLfuncLength{d+1}{i}$ is with increasing $i$).
        \item We partition $B=V_1\cup\ldots \cup V_i$ according to the independent runs. That is, set $V_j=\supp(\near_b(\rho_j,x))$ for $b=\generLfuncCover{d+1}{i}$ where $I=\{\rho_1<\ldots<\rho_i\}$. 
        By \cref{def: charge cover decomposition} this is indeed a partition.
        \begin{itemize}
             \item The gaps between $V_j$ and $V_{j+1}$ are higher than \\$4|S|\bigM\maxeff{xy}+\budH$: by \cref{def:charge well separated word} the set $I$ has gap at least 
            $4\bigM^2(\generLfuncMaxW{d}\generLfuncLength{d}{1}\generLfuncLength{d+1}{i}+\generLfuncCover{d+1}{i}+\budH)$
            Notice that $|xy|\le |w_2|<\generLfuncLength{d+1}{i}$ (by \cref{def:charge d i fit}), so \[
            \begin{split}
                &\maxeff{xy}\le \generLfuncMaxW{d}\generLfuncLength{d+1}{i}\le \\&\generLfuncMaxW{d}\generLfuncLength{d}{1}\generLfuncLength{d+1}{i}
            \end{split}
            \] 
            Thus, the minimal possible gap of a state in $V_j$ from a state in $V_{j+1}$ is at least
            \[
            \begin{split}
            &4\bigM^2(\maxeff{xy}+\simpLfuncCover{d+1}{i}+\budH)-2b\\
            &=4\bigM^2(\maxeff{xy}+\simpLfuncCover{d+1}{i}+\budH)-2\simpLfuncCover{d+1}{i}
            \gg  \\&4|S|\bigM\maxeff{xy}+\budH
            \end{split}
            \]
            
            \item The existence of $k_{j,x}, k_{j,y}$ and their sign equality is immediate from $\difftype_b(w_1v_{j-1},u_{j,j'})=  \\\difftype_b(w_1v_{j'-1},u_{j',j''})$. Indeed, this means that the configurations around each independent run are identical, and that the sign of the run on $x$ and on $y$ is the same.
        \end{itemize}
        \item The seamless baseline run exists by \cref{def:charge decreasing word}.
    \end{enumerate}
    This concludes that $uxyv$ is an GSRI.

    \paragraph*{If $\depth(x)<d$.} In this case we treat the two types of decomposition separately. 
    The case of a $\charge$-decreasing decomposition is simple: we just restrict attention to an infix of $w_2$, as follows.
    Define $w'_2$ to be the suffix of $x$ of length $\frac{1}{\extraSize}\generLfuncLength{d}{1}$, and let $w'_1,w'_3$ be the prefix and suffix such that $w_1w_2w_3=w'_1w'_2w'_3$. Note that there is such a suffix since $|x|\ge \generLfuncLength{d}{1}$ as a concatenation of $u_j$'s, and by \cref{def: charge decreasing decomposition}.
    Still by \cref{def: charge decreasing decomposition} we have $\charge(w'_1)\ge \charge(w'_1w'_2)$ (viewing $w'_2$ as a prefix of $u_{j'}$). Take $I'=\{\rho_1\}$ (i.e., arbitrarily choose a single independent run from $I$) so $w'_1w'_2w'_3$ with $I'$ is a $\charge-(d-1,1)$ fit word, as per \cref{def:charge d i fit}.

    The case of $\charge$ bounded decomposition is slightly more delicate. We still take $w'_2$ to be the suffix of $x$ of length $\frac{1}{\extraSize}\generLfuncLength{d}{1}$. This is possible since each part in the decomposition is of length multiple of $\frac{1}{\extraSize}\generLfuncLength{d}{1}$ (see \cref{def: charge bounded decomposition}). Moreover, Property 5 in \cref{def: charge bounded decomposition} guarantees that $\charge(w'_1)\ge \charge(w'_1w'_2)$ (although between them the charge may decrease, but this does not concern us).
    We can now proceed as above by taking $I'=\{\rho_1\}$, and we are done.
\end{proof}

\subsection{From Well Separation to GSRI}
\label{sec:charge well separation to GSRI}
We now have the tools to show that given a $\charge$ well-separated word, we can obtain a GSRI from it. To do so, we use a similar zooming argument as in \cref{sec:well separation to SSRI}. Note that the caveat of \cref{prop: charge high amplitude to charge decreasing decomposition} now pops up again, so we cannot always obtain a GSRI -- sometimes we show the potential can be large instead.
\begin{lemma}
    \label{lem: charge d i fit to GSRI or lower depth or high potential}
    Either there is a word $w' \in (\Gamma_0^0)^*$ such that $\pot(w')>\budH$, or for every $\charge-(d,i)$-fit well-separated word $w_1w_2w_3$ where 
    $|w_2|= \frac{1}{\extraSize}\generLfuncLength{d+1}{i}$, one of the following holds:
    \begin{enumerate}
        \item $w_1w_2w_3=uxyv$ where $uxyv$ is an GSRI with $\depth(x)=d$ and $xy$ is an infix of $w_2$.
        \item $w_1w_2w_2=w'_1w'_2w'_3$ where $w'_1w'_2w'_3$ is a $\charge-(d-1,1)$-fit well-separated word (\cref{def:charge well separated word}) and $w'_2$ is an infix of $w_2$ with $|w'_2|=\frac{1}{\extraSize}\generLfuncLength{d}{1}$.
    \end{enumerate}
\end{lemma}
\begin{proof}
    Assume there is no word  $w' \in (\Gamma_0^0)^*$ such that $\pot(w')>\budH$ (otherwise we are done). Consider therefore a $\charge-(d,i)$-fit well-separated word $w_1w_2w_3$. 
    Let $I=\{\rho_1<\ldots<\rho_i\}$ be the set of independent runs associated with $w_1w_2w_3$ (\cref{def:independent run}). 
    Consider the behaviour of $\charge$ along $w_2$ as per \cref{def:charge bounded amplitude word}. 
    We claim there are two cases: either $w_1w_2w_3$ is $\charge-(d,i)$ bounded amplitude, or it has an infix that is $\charge$ high amplitude.
    Indeed, if $w_1w_2w_3$ is not $\charge$ bounded, then for some prefix $x$ of $w_2$ we have $|\charge(w_1)-\charge(w_1x)|>\generLfuncAmp{d+1}{i}$. If $\charge(w_1)-\charge(w_1x)>\generLfuncAmp{d+1}{i}$ then $x$ is an infix as required. If $\charge(w_1x)-\charge(w_1)>\generLfuncAmp{d+1}{i}$ then since $\charge(w_1)\ge \charge(w_1w_2)$ (by \cref{def:charge decreasing word}) the suffix $z$ such that $w_2=xz$ is a $\charge$ high-amplitude infix of $w_2$. 

    We can therefore assume without loss of generality that $w_1w_2w_3$ itself is either $\charge-(d,i)$ bounded amplitude or $\charge-(d,i)$ high amplitude.
    We then employ either \cref{prop: charge bounded amp to bounded decomposition} in the former case, or \cref{prop: charge high amplitude to charge decreasing decomposition} in the latter, and obtain a $(d,i)$ decomposition as per either \cref{def: charge bounded decomposition} or \cref{def: charge decreasing decomposition}, respectively. I.e., we can write $w_2=u_0u_1\cdots u_tu_{t+1}$. Note that crucially, in order to use \cref{prop: charge high amplitude to charge decreasing decomposition}, we rely on the assumption that there are no high-potential words (above $\budH$), assumed above.
    
    Our goal is now to use \cref{prop:charge cover decomp induces GSRI or lower depth} to conclude the lemma. However, this requires the decomposition to also be a cover (\cref{def: charge cover decomposition}), which is not necessarily the case. This is where Zooming comes in.
    
    We proceed by reverse induction on $i$.
    The base case is $i=|S|$. In this case, every state belongs to one of the independent runs at each step. Therefore, the decomposition is a cover already with $b=0$, and in particular with $b=\generLfuncCover{d+1}{|S|}=1$ (by \cref{def: genL func Cover}). We can therefore use \cref{prop:charge cover decomp induces GSRI or lower depth} and we are done.

    We proceed to the induction case, where $1\le i<|S|$. Consider the decomposition of $w_2$. If it is already a cover, then we can again use \cref{prop:charge cover decomp induces GSRI or lower depth} and we are done. 
    
    Assume henceforth that this is not the case. In the notations of \cref{def: charge cover decomposition} this means that there is some $0< j\le  t$ and state $s\in \supp(\vec{c_j})$ that is more than $\generLfuncCover{d+1}{i}$ away from all the independent runs in $I$. More precisely, for every run $\rho\in I$, write $\rho[w_1v_j]:s_0\runsto{v_j}s'$, then $|\vec{c_j}(s)-\vec{c_j}(s')|>\generLfuncCover{d+1}{i}$.

    We now ``zoom in'' on $s$, and show that we can use it to construct a new independent run on an infix of $w_2$, so we can use the induction hypothesis with $i+1$. This, however, needs to be done carefully in order to account for the depth, length, and cover.

    Let $\pi:s_0\runsto{w_1v_j}s$ be a seamless minimal-weight run to $s$. In particular, $\weight(\pi)=\vec{c_j}(s)$. Decompose $w_1v_j=z_1z_2$ where $|z_2|=\frac{1}{\extraSize}\generLfuncLength{d+1}{i+1}$. Such a decomposition is possible due to the lower bound on $|u_0|$ in both \cref{def: charge decreasing decomposition} and \cref{def: charge bounded decomposition}.
    
    We claim that $\pi[z_2]$ remains at least $\frac12\generLfuncCover{d+1}{i+1}$ away from all independent runs along this suffix.
    Indeed, this is a direct consequence of \cref{prop: relating cover length and maxW}, since the maximal change in weight that can occur during $\pi[z_2]$ is 
    \[2\maxeff{|z_2|}\le |z_2|\cdot \generLfuncMaxW{d}\cdot 2= \frac{1}{\extraSize}\generLfuncLength{d+1}{i+1}\cdot 2\]
    Thus, the word $z_1z_2z_3$ with $z_3=\epsilon$ is $\charge-(d,i+1)$-fit well separated word (\cref{def:charge well separated word}), and by the induction hypothesis it satisfies the conditions of the lemma, so we are done.
\end{proof}

Finally, consider a charge-decreasing word $w_1w_2w_3\in (\genCac)^*$ with $|w_2|= \frac{1}{\extraSize}\generLfuncLength{|S|}{1}$. Recall from \cref{rmk:depth vs length bound} that we have $\depth(w_2)\le |S|-1$. By choosing an arbitrary seamless run on it, we have that it is $(|S|-1,1)$-fit well-separated. Indeed, for $i=1$ well-separation is trivial (see \cref{def:charge well separated word}).
We can apply \cref{lem: charge d i fit to GSRI or lower depth or high potential} so that either there is a word $w'\in (\simpCac)^*$ with $\pot(w')> \budH$, or we identify a GSRI, or we rewrite $w_1w_2w_3=w'_1w'_2w'_3$ as a $\charge-(|S|-2,1)$-fit well-separated word. 
We continue inductively (assuming we did not find a $w'$ as above) until at some point a GSRI must be found (since the induction cannot go below depth $0$). We therefore have the following:
\begin{corollary}
\label{cor: charge decrease to GSRI or high potential}
    Either there is a word $w' \in (\Gamma_0^0)^*$ such that $\pot(w')>\budH$, or for every charge-decreasing word $w_1w_2w_3\in (\genCac)^*$ with $|w_2|= \frac{1}{\extraSize}\generLfuncLength{|S|}{1}$ we have that $w_1w_2w_3=uxyv$ where $uxyv$ is a GSRI and $xy$ is an infix of $w_2$.
\end{corollary}

\section{Large Gaps Imply Undeterminisability}
\label{sec:gaps to undeterminisability}
We are now ready to our main result, namely that if $\cA$ has a $B$-gap witness for large enough $B$, then it is undeterminisable. This readily implies the decidability of determinisation, and provides a simple algorithm (albeit with horrible complexity).
\begin{theorem}
    \label{thm:main large gap to undet}
    Consider a WFA $\cA$, then $\cA$ is undeterminisable if and only if there is a $B$-gap witness with $B=\generLfuncAmp{|S|}{0}$ where $|S|$ is the size of $\augA$.
\end{theorem}
We prove the theorem in the remainder of this section.

In the trivial direction, if $\cA$ is undeterminisable then by \cref{thm:det iff bounded gap} it has unboundedly large gap witnesses, in particular there is some $B$ gap witness for $B=\generLfuncAmp{|S|}{0}$.

We turn to the interesting direction, and proceed in steps.
\paragraph*{Step 1: From a Gap witness to a $\charge$-Decreasing Word.}
Consider a $B$-gap witness $xy$ for $B=\generLfuncAmp{|S|}{0}$ as per \cref{thm:det iff bounded gap}. Let $\rho:q_0\runsto{xy}Q$ be a seamless minimal weight run on $xy$ and $\pi:q_0\runsto{x}q$ be a minimal weight run on $x$, then $\weight(\rho[x])-\weight(\pi)>B$. 

Intuitively, we view this gap witness as a word read by $\augA$, and select the run $\rho$ as baseline. This way, the run $\pi$ decreases significantly below $0$, leading to a very high $\charge$. Since at the end of the run $\rho$ is minimal, we get a significant decrease in $\charge$. \cref{fig:big gap to big charge} depicts this idea. We now formalise this idea.

We view $\rho$ as a word read by $\augA$, and let $\rho':s_0\runsto{\rho[x]}s_1\runsto{\rho[y]}s_2$ be its baseline run. In particular, $\rho'$ has constant weight $0$ (see \cref{obs:baseline runs have weight 0}).
We then consider the run $\pi':s_0\runsto{\rho[x]}q_1$ obtained from $\pi$ with the baseline component $\rho[x]$. By the construction of $\augA$ we have 
\[\weight(\rho'[\rho[x]])-\weight(\pi')=\weight(\rho[x])-\weight(\pi)>B\]
Since $\weight(\rho)$ is constantly $0$, we have $\weight(\pi')<-B$. In particular, $\charge(\rho[x])<-B$.
However, since $\rho$ is a minimal weight run on $xy$, then $\rho'$ is a minimal weight run on $\rho[xy]=\rho[x]\rho[y]$. It follows that $\charge(\rho[x]\rho[y])=0$. 
Thus, we have that $\rho[x]\rho[y]$ is a charge-decreasing word, and moreover $\charge(\rho[x])-\charge(\rho[x]\rho[y])>B=\generLfuncAmp{|S|}{0}$

\begin{figure}[H]
  \begin{center}
    \includegraphics[width=0.9\textwidth]{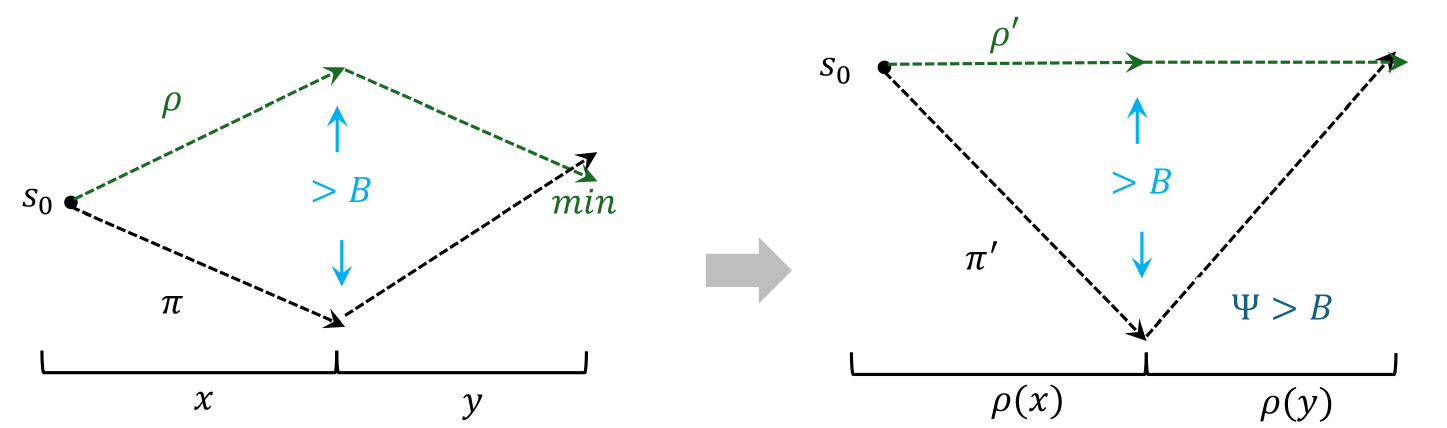}
  \end{center}
  \caption{A $B$-gap witness and its counterpart in $\augA$, where the baseline run is the higher one. The lower run then induces a large charge.}
  \label{fig:big gap to big charge}

\end{figure}

\paragraph*{Step 2: From $\charge$-Decreasing to $\pot$-Increasing}
We now consider $\augA_\infty^\infty$ with the alphabet $\genCac$. 
We know that there is some word, namely $\rho[x]\rho[y]$, with a significant charge drop. We now let $w_1w_2$ be a \emph{minimal length} word over $(\genCac)^*$ such that $\charge(w_1)-\charge(w_1w_2)>\generLfuncAmp{|S|}{0}$. Note that we now allow cactus letters, and also require length minimality. Also recall by \cref{def:L bounded cactus letters} that $\depth(w_1w_2)\le |S|-1$, since it does not use any degenerate letters.

We now ask whether there is a word $w'\in (\Gamma_0^0)^*$ such that $\pot(w')> \budH$. If there is, we proceed to the next step.
We continue with this step assuming by way of contradiction that there is no such word.

Intuitively, we now use this assumption in combination with \cref{lem: charge decrease to high potential} and obtain that the charge has bounded decrease over $\genCac$. Since $w_1w_2$ has a huge charge drop, then $w_2$ must be very long. We use this with \cref{cor: charge decrease to GSRI or high potential} to show that $w_1w_2$ is a GSRI. Then, we use our GSRI toolbox to show that we can find a shorter witness with the same charge decrease, contradicting minimality. We now formalise this intuition.

Let $d=\depth(w_2)\le |S|-1$. By \cref{lem: charge decrease to high potential} and the assumption that there is no word $w'$ with $\pot(w')>\budH$, we have that $\charge$ has an upper-bounded decrease with each letter, namely at most $\budH+2\generLfuncMaxW{|S|-1}$. Applying this to all of $w_2$, we have
\[\generLfuncAmp{|S|}{0}<\charge(w_1)-\charge(w_1w_2)\le |w_2|(\budH+2\generLfuncMaxW{|S|-1})\]
By plugging in \cref{def: genL func Amp} and rearranging, we obtain
\[|w_2|\gg \frac{1}{\extraSize}\generLfuncLength{|S|}{1}\]

By the minimality of $w_1w_2$, we have that $\charge(w_1)-\charge(w_1z)<\charge(w_1)-\charge(w_1w_2)$ for every strict prefix $z$ of $w_2$ (otherwise we can shorten to $w_1z$). Therefore, by taking the suffix $w'_2$ of $w_2$ such that $|w'_2|=\frac{1}{\extraSize}\generLfuncLength{|S|}{1}$ and defining $w'_1$ such that $w_1w_2=w'_1w'_2$ we have that $w'_1w'_2$ is a charge-decreasing word. Indeed, by the minimality we have \[\charge(w_1)-\charge(w'_1)<\charge(w_1)-\charge(w_1w_2)=\charge(w_1)-\charge(w'_1w'_2)\]
which implies $\charge(w'_1)>\charge(w'_1w'_2)$.

We are therefore within the premise of \cref{cor: charge decrease to GSRI or high potential}. Moreover, since we assume there is no $w'\in (\Gamma_0^0)^*$ with $\pot(w')>\budH$, we get that $w_1w_2=w'_1w'_2=uxyv$ where $uxyv$ is a GSRI and $xy$ is an infix of $w'_2$ (and thus of $w_2$). Write $w_1w_2=uxyv=w_1zxyv$.
We now apply \cref{lem:GSRI to stable or high potential}, and get that $uxyv$ is a Stable GSRI (see \cref{def:flavours of inc inf}). We separate to cases:
\begin{itemize}
    \item If $uxyv$ is a Degenerate Stable GSRI, we use \cref{cor:degenerate SRI remove x maintain potential and charge} and get that $\charge(uxyv)=\charge(uyv)$, but then $\charge(w_1)-\charge(w_1zyv)=\charge(w_1)-\charge(w_1w_2)$, while $|w_1zyv|<|w_1w_2|$, contradicting the minimality of $w_1w_2$.
    \item If $uxyv$ is a Non-degenerate Stable GSRI, we use \cref{lem:nondegenerate SRI budding increases potential decreases charge} and get that $\charge(uxyv)\ge \charge(u\alpha_{S',x}v)$. Similarly to the case above, we can write $uxyv=w_1zxyv$, so we have $\charge(w_1)-\charge(w_1zyv)\le \charge(w_1)-\charge(w_1z\alpha_{S',x}v)$, while $|w_1\alpha_{S',x}v|<|w_1w_2|$ (since $|\alpha_{S',x}|=1$ but $|xy|\ge 2$).
    This is again a contradiction to the minimality of $w_1w_2$.
\end{itemize}
In either case, we reach a contradiction. Thus, there is a word $w'\in (\Gamma_0^0)^*$ such that $\pot(w')> \budH$.

\paragraph*{Step 3: Increasing Potential to Witness}
Intuitively, we now repeat Step 2 with $\pot$ and SSRI instead of $\charge$ and GSRI, and this results in a witness, as follows.

By the existence of $w'$, let $w\in (\simpCac)^*$ be the shortest word for which $\pot(w)>\budH$. Note that we again require minimality, now over $\simpCac$.
We claim that $|w|\ge \frac{1}{\extraSize}\simpLfuncLength{|S|}{1}$. Indeed, by \cref{cor:effective bounded growth potential} we have that $\pot$ has bounded-increase on every letter. Specifically, we have
\[\budH< \pot(w)\le |w|\cdot 2\simpLfuncMaxW{|S|-1}\]
Expanding $\budH=\simpLfuncAmp{|S|}{0}$ by \cref{def: simpL func Amp} and rearranging we obtain
\[|w|\gg\frac{1}{\extraSize}\simpLfuncLength{|S|}{1}\]
By minimality, we have $\pot(w)>\pot(z)$ for every prefix $z$ of $w$. In particular, write $w=w_1w_2$ with $|w_2|=\frac{1}{\extraSize}\simpLfuncLength{|S|}{1}$, then we are within the premise of \cref{cor: potential increasing to SSRI}. 
We can therefore write $w_1w_2=uxyv$ where $uxyv$ is an SSRI.

We now split to cases according to the flavour of this SSRI (as per \cref{def:flavours of inc inf}). In each case, if we reach a witness then we proceed to Step 4.
\begin{itemize}
    \item If $uxyv$ is a Negative SSRI, then by \cref{lem:negative SSRI to type 1 witness} there is a type-1 witness in $\augA_\infty^\infty$. We proceed to Step 4.
    \item If $uxyv$ is not a Negative SSRI, then it is a Positive SRI. By \cref{lem:pos SSRI to stable}, either there is a type-0 witness, in which case we proceed to Step 4, or $uxyv$ is a Stable SSRI, and we proceed to the next item.
    \item If $uxyv$ is a Stable SSRI, we split to cases again.
    \begin{itemize}
        \item If $uxyv$ is a Degenerate SSRI, then by \cref{cor:degenerate SRI remove x maintain potential and charge} we have $\pot(uxyv)=\pot(uyv)$, contradicting the minimality of $w$.
        \item If $uxyv$ is a Non-degenerate SSRI, then by \cref{lem:nondegenerate SRI budding increases potential decreases charge}, either there is a type-0 witness, in which case we proceed to Step 4, or $\pot(uxyv)\le \pot(u\alpha_{S',x}v)$. The latter, however, contradicts the minimality of $w$, since $|\alpha_{S',x}|=1$ but $|xy|\ge 2$.
    \end{itemize}
\end{itemize}
We conclude that in all cases we either reach a contradiction, or we have a witness.

\paragraph*{Step 4: Witness to Undeterminisability}
By \cref{lem:witness implies nondet}, since there is a witness, then $\augA_\infty^\infty$ is undeterminisable. By \cref{thm:aug A inf inf determinisable iff A determinisable}, this means $\cA$ is undeterminisable. This concludes the proof of \cref{thm:main large gap to undet}. \hfill \qed

\subsection{Algorithm and Complexity}
\label{sec:algorithm and complexity}
\cref{thm:main large gap to undet} gives rise to a naive algorithm for determinisability, made up of two semi-algorithms: given WFA $\cA$, one semi-algorithm iterates over all deterministic WFAs and checks whether one is equivalent to $\cA$ (which is decidable by \cite{Almagor2020Whatsdecidableweighted}). The other iterates over all gap witnesses, and checks if one is found with gap $\generLfuncAmp{|S|}{1}$ (after computing this value).

This algorithm offers no improvement over the results of \cite{almagor2026determinization}. However, thanks to the concrete bound we obtain in \cref{thm:main large gap to undet}, we can now do much better and give a very simple algorithm. This is based on the following:
\begin{theorem}
\label{thm: det construction for B bounded gap}
    Consider a WFA $\cA$ and $B\in \bbN$. We can construct a deterministic WFA $\cA|_B$ such that $\cA$ is equivalent to $\cA|_B$ if and only if $\cA$ has $B$-bounded gaps.

    Moreover, $\cA|_B$ can be computed in time $B^{O(|\cA|)}$ (and also has this description size), where $|\cA|$ is the description size of $\cA$.
\end{theorem}
The proof of this theorem is folklore (e.g., \cite{filiot2017delay}) and appears explicitly in \cite{almagor2026determinization}. Intuitively, the idea is for $\cA|_B$ to keep track of the minimal run, and all runs up to $B$ above it. When a run goes more than $B$ above, it is ``forgotten''.

Given \cref{thm: det construction for B bounded gap}, a determinisability algorithm is now simple: given a WFA $\cA$, compute $B=\generLfuncAmp{|S|}{0}$, construct $\cA|_B$, and check whether $\cA$ is equivalent to $\cA|_B$. If it is, it is determinisable, and $\cA|_B$ is a deterministic equivalent. Otherwise, it is not determinisable by \cref{thm:main large gap to undet}.

It remains to analyse the complexity of this algorithm. To this end, we show in \cref{sec:complexity} that $\generLfuncAmp{|S|}{0}$ is a primitive-recursive function, and in fact lies in $\gregF_6$, the $6$-th level of the fast-growing hierarchy (Grzegorczyk hierarchy)~\cite{grzegorczyk1953some,tourlakis1984computability,schmitz2016complexity}. Since this astronomical complexity ``swallows'' the exponent $\cA$, we have that the overall complexity is in $\gregF_6$.

\section{Growth Rate of $\generLfuncAmp{|S|}{0}$}
\label{sec:complexity}
In this section we establish the complexity of the functions we define in \cref{sec:effective cactus}. 
Due to the multi-parameter recursive nature of these functions, it is perhaps not surprising that they grow \emph{very} fast. Concretely, we place them all in various levels of the \emph{fast growing hierarchy} (also known as the \emph{Grzegorczyk hierarchy}). More precisely, we place them in levels between $\gregF_2$ to $\gregF_6$, with $\generLfuncAmp{|S|}{0}$ in $\gregF_6$. 
In a coarser lens, this places all of them as primitive recursive (much lower than Ackermann), but non-elementary. 

Since we do not assume familiarity with the fast-growing Hierarchy, we summarise briefly the properties we use. We refer to~\cite{grzegorczyk1953some,tourlakis1984computability,schmitz2016complexity} for further details.

\subsection{Properties of the Fast Growing Hierarchy}
\label{sec:intro to grzegorczyk}
For every $a\in \bbN$ we denote by $\gregF_a$ the $a$-th class of fast-growing functions (corresponding to $\mathcal{E}^{a+1}$ in the Grzegorczyk hierarchy). The functions in this class are of the form $f:\bbN^k\to \bbN$ for arbitrary $k$. Specifically, $\gregF_2$ is exactly the class ELEMENTARY.

For $a\ge 2$, the class $\gregF_a$ is closed under standard arithmetic operators, as well as exponentiation. Moreover, there are two special operations that can be applied (but may increase the level):
\paragraph*{Substitution:} Given functions $f(x_1,\ldots,x_m)$ and $g_1(y_1,\ldots,y_n)$,$\ldots$ , $g_m(y_1,\ldots,y_n)$, we define the substitution    
\[
\begin{split}
&f'(y_1,\ldots,y_n)=\\
&f(g_1(y_1,\ldots,y_n),\ldots,g_m(y_1,\ldots,y_n))
\end{split}
\]
If $f,g_1,\ldots,g_m\in \gregF_a$ for $a\ge 2$, then $f'\in \gregF_a$ as well~\cite{tourlakis1984computability}.

\paragraph*{Bounded Primitive Recursion:}
Given functions $f(x_1,\ldots,x_m,y_1,y_2)$ and $g(x_1,\cdots,x_m)$ we define $\prim(f,g):\bbN^{m+1}\to \bbN$ as
\[
\begin{split}
&\prim(f,g)(x_1,\ldots,x_m,y)=\\
&\begin{cases}
    g(x_1,\ldots,x_m) & y=0\\
    f(x_1,\ldots,x_m,y,\prim(f,g)(x_1,\ldots,x_m,y-1)) & y>0
\end{cases}
\end{split}
\]
If $f,g\in \gregF_a$ then $\prim(f,g)\in \gregF_{a+1}$. Thus, primitive recursion increases the class by at most one~\cite{tourlakis1984computability}.

\subsection{Upper Bounds for Effective Functions}
\label{sec:functions are in fast growing}
Our goal is to give upper bounds to our functions of \cref{sec:effective cactus} that place them in various levels of the fast growing hierarchy. Note that the main parameter now is the size of the input. In the following, we denote by $n$ the size of $\augA$ (i.e., $n=|S|$). This is not precisely the size of the input $\cA$. Indeed, it is exponential in the size of $\cA$. Nonetheless, this suffices for our analysis, since the levels we deal with are closed under exponentiation anyway.

\begin{table}[h]
\centering
\begin{tabular}{|c|c|c|}
\hline
\textbf{Original function} & \textbf{Upper bound function}\\ \hline
\simpLfuncMaxW{d} & W(n,d,Ld) \\ \hline%& (2n \cdot n!)^nL_d^n
\simpLfuncCover{d+1}{i} & C(n,d,i,Ld,Li)  \\ \hline%& 8(n!)^2W(n,d,i,Ld,Li)Ld(Li)^n
\simpLfuncTypes{d+1}{i} & T(n,d,i,Ld,Li) \\ \hline%3^n\cdot (2C(n,d,i,Ld,Li)+2)^{2n^2}
\simpLfuncAmp{d+1}{i} & A(n,d,i,Ld,Li) \\ \hline%& 32n^2(2W(n,d-1,i,Ld,Li))(\ramsey(T(n,d,i,Ld,Li),3)+Li+Ld)
\simpLfuncLength{d+1}{i} & L(n,d,i,Ld,Li) \\ \hline%& 32n^2(Ld \cdot Li)(\ramsey(T(n,d,i,Ld,Li),3)+2)^{A(n,d,i,Ld,Li)+1} 
\generLfuncMaxW{d} & W'(n,d,Ld) \\ \hline%& (2n \cdot n!)^nL_d^n
\generLfuncCover{d+1}{i} & C'(n,d,i,Ld,Li)  \\ \hline%& 8(n!)^2W(n,d,i,Ld,Li)Ld(Li)^n
\generLfuncTypes{d+1}{i} & T'(n,d,i,Ld,Li) \\ \hline%3^n\cdot (2C(n,d,i,Ld,Li)+2)^{2n^2}
\generLfuncAmp{d+1}{i} & A'(n,d,i,Ld,Li) \\ \hline%& 32n^2(2W(n,d-1,i,Ld,Li))(\ramsey(T(n,d,i,Ld,Li),3)+Li+Ld)
\generLfuncLength{d+1}{i} & L'(n,d,i,Ld,Li) \\ \hline%& 32n^2(Ld \cdot Li)(\ramsey(T(n,d,i,Ld,Li),3)+2)^{A(n,d,i,Ld,Li)+1} 
% Row 6, Col 1 & Row 6, Col 2 \\ \hline
% Row 7, Col 1 & Row 7, Col 2 \\ \hline
\end{tabular}
\caption{Each function we define with a corresponding upper-bound function.}
\label{tab:upper bound functions}
\end{table}
%\gatodo{if we removed the formal proof for $\generLfuncLength{n}{0}$ do we still need $A',T',L',...$ upper bounds in the table ~\cref{tab:upper bound functions} ?}
In~\cref{tab:upper bound functions} we associate with each function we consider a corresponding upper bound function. Note that the upper bounds have additional parameters, which we now turn to explain.
\begin{itemize}
    \item $n$ represents the size of $\augA$, as explained above.
    \item $d$ is the depth of a cactus letter.
    \item $i$ almost represents the number of independent runs, but not exactly: since primitive recursion only works with decreasing variables, and the number of independent runs increases, we instead think of $i$ as $n-j$ where $j$ is the number of independent runs.
    \item $Ld$ represents $\simpLfuncLength{d}{1}$ (or $\generLfuncLength{d}{1}$ for the general functions).
    \item $Li$ represents $\simpLfuncLength{d+1}{i+1}$ (or $\generLfuncLength{d+1}{i+1}$ for the general functions).
\end{itemize}

We now tackle each function, show how to define its corresponding upper bound function, and place the upper bound in the hierarchy. For the most part, this is straightforward, with some issues arising for functions that depend on e.g., $\simpLfuncLength{d}{1}$, since for those we first need to bound this respective component.

\subparagraph*{$\simpLfuncMaxW{d}$.}
Recall that 
$\simpLfuncMaxW{d}=2\bigM \simpLfuncMaxW{d-1}\simpLfuncLength{d+1}{1}$. We therefore define
\[W(n,d,Ld)=\begin{cases}
    n & d=0\\
    (2n\cdot n!)\cdot W(n,d-1,Ld)\cdot  Ld & d>0
\end{cases}\]
Clearly if $\simpLfuncLength{d}{1}\le Ld$ then $\simpLfuncMaxW{d}\le W(n,d,Ld)$.

By expanding the inductive definition we readily have
\[W(n,d,Ld)\le (2n\cdot n!)^n Ld^n+n\]
which is in $\gregF_2$ as an elementary function. \hfill $\triangle\triangle\triangle$

\subparagraph*{$\simpLfuncCover{d+1}{i}$.}
Recall that 
\[\begin{split}
&\simpLfuncCover{d+1}{i}=\\
&8\bigM^2(\simpLfuncMaxW{d}\simpLfuncLength{d}{1}\simpLfuncLength{d+1}{i+1}+\simpLfuncCover{d+1}{i+1})
\end{split}\]
We therefore define
\[
\begin{split}
    &C(n,d,i,Ld,Li)=\\
    &\begin{cases} 
        1 &  i = 0 \\
        8(2n \cdot n!)^2\cdot (W(n,d,Ld)\cdot Ld \cdot Li + C(n,d,i-1,Ld,Li)) &  i > 0 
        \end{cases}
\end{split}
\]
And clearly this is an upper bound.
Again by expanding the inductive definition (noticing this is in fact a non-homogenous linear recurrence in $C(\cdots)$) we have
\[C(n,d,i,Ld,Li)\le (8(2n\cdot n!))^{2(n+1)}(W(n,d,Ld)\cdot Ld \cdot Li + 1)\]
Again, this an elementary function and is therefore in $\gregF_2$.
\hfill $\triangle\triangle\triangle$

\subparagraph*{$\simpLfuncTypes{d+1}{i}$}
Recall that
\[\simpLfuncTypes{d+1}{i}=3^i\cdot (2\simpLfuncCover{d+1}{i}+2)^{2|S|i}\]
We define
\[
T(n,d,i,Ld,Li)=3^i(2C(n,d,i,Ld,Li)+2)^{2n i}
\]
and note that the upper bound swaps $i$ with $n-i$. That is,
$\simpLfuncTypes{d+1}{i}\le T(n,d,n-i,\simpLfuncLength{d}{1},\simpLfuncLength{d+1}{i+1})$.

By closure under substitution and since $C$ is in $\gregF_2$, we also get that $T$ is in $\gregF_2$.
\hfill $\triangle\triangle\triangle$

\subparagraph*{$\simpLfuncAmp{d+1}{i}$}
Recall that 
\begin{align*}
    \simpLfuncAmp{d+1}{i} =& \extraSize(2\simpLfuncMaxW{d})(\ramsey(\simpLfuncTypes{d+1}{i},3)\\ &+\simpLfuncLength{d+1}{i+1}+\simpLfuncLength{d}{1})
    \end{align*}
We define
\[\begin{split}
&A(n,d,i,Ld,Li)=\\
&32(2n\cdot n!)^2(2W(n,d-1,Ld)\ramsey(T(n,d,i,Ld,Li),3)+Li+Ld)
\end{split}
\]
Note that again we swap $i$ with $n-i$ for the upper bound.
Also, again by substitution and the fact that the Ramsey function is elementary (single exponential, in fact), we have that $A$ is in $\gregF_2$.
\hfill $\triangle\triangle\triangle$

\subparagraph*{$\simpLfuncLength{d+1}{i}$}
Recall that 
 $\simpLfuncLength{d+1}{i}=\extraSize\cdot\simpLfuncLength{d}{1}\cdot\simpLfuncLength{d+1}{i+1}\cdot (\ramsey(\simpLfuncTypes{d+1}{i},3)+2)^{2\simpLfuncAmp{d+1}{i}+1}$.
We define
\[
\begin{split}
&L(n,d,i,Ld,Li)=\\
&32(2n\cdot n!)^2 (Ld\cdot Li\cdot \ramsey(T(n,d,i,Ld,Li),3)+2)^{2A(n,d,i,Ld,Li)+1}
\end{split}
\]
if $i>0$, and $0$  if $i=0$.

Again by substitution and the fact that the Ramsey function is elementary, we have that $L$ is in $\gregF_2$.
\hfill $\triangle\triangle\triangle$

\subsubsection*{Getting rid of $Ld$ and $Li$}
Our upper bounds thus far are not useful, as they assume $Ld$ and $Li$ as parameters, where they are in fact stand-in for recursive calls. In order to get rid of them, we use bounded primitive recursion, as follows.

We first get rid of $Li$. To this end, define
Define 
\[
        Len1 (n,d,i,Ld) = 
        \begin{cases} 
        L(n,d,0,Ld,0) &  i = 0 \\
        L(n,d,i,Ld,Len1 (n,d,i-1,Ld)) &  i > 0 
        \end{cases}
\]
By the definition of $L$, it is immediate by induction that  $\simpLfuncLength{d}{i}\le Len1(n,d,n-i,\simpLfuncLength{d-1}{1})$ (recalling that $n-i$ is the number of independent runs).

Further note that $Len1$ is in fact obtained as
\[\prim(L(n,d,i,Ld,Li),L(n,d,0,Ld,0))\]
Since $L(n,d,i,Ld,Li)\in \gregF_2$ and $L(n,d,0,Ld,0)\in \gregF_2$ (as we show above), then by the bounded primitive recursion property above, we get that $Len1\in \gregF_3$.

We now proceed to get rid of $Ld$. To this end, define
\[
        Len2 (n,d) = 
        \begin{cases} 
        Len1(n,1,n-1,1) &  d = 1 \\
        Len1(n,d,n-1,Len2 (n,d-1)) &  d > 1 
        \end{cases}
\]
Again, an immediate induction gives $Len2(n,d)=\simpLfuncLength{d}{0}$.

Since $Len2$ is in fact $\prim(Len1(n,d,n-1,Ld),Len1(n,0,n-1,1))$, then since $Len1\in \gregF_3$ we have $Len2\in \gregF_4$, and we conclude that $\simpLfuncLength{d}{0}\in \gregF_4$.

We now conduct the same analysis, mutatis-mutandis, with the general functions. The only difference is that now we have a dependency on $\budH=\generLfuncAmp{n}{0}\in \gregF_4$ (as a composition of $A\in \gregF_2$ with $\simpLfuncLength{n}{0}\in \gregF_4$). This means that all the upper bound functions are in $\gregF_4$, and getting rid of $L_i$ and $L_d$ using two recursions increases the class further to $\gregF_6$, in particular, $\generLfuncLength{n}{0}\in \gregF_6$, and therefore also

We conclude that $\generLfuncAmp{n}{0}\in \gregF_6$, which completes the complexity analysis.

%\bibliographystyle{plain} 
%\bibliography{main}

\end{document}